\documentclass[a4paper,reqno,fleqn]{amsbook}

\usepackage{amsmath}
\usepackage{amssymb}
\usepackage{graphicx}
\usepackage{geometry}
\usepackage{hyperref}

\hypersetup{
 bookmarks=true,		% show bookmarks bar?
 unicode=false,			% non-Latin characters in Acrobat’s bookmarks
 pdftoolbar=true,		% show Acrobat’s toolbar?
 pdfmenubar=true,		% show Acrobat’s menu?
 pdffitwindow=false,		% window fit to page when opened
 pdfstartview={FitH},		% fits the width of the page to the window
 pdftitle={Quantum entanglement in finite-dimensional Hilbert spaces},    % title
 pdfauthor={Szil{\'a}rd Szalay},	% author
 pdfsubject={Ph.D. Dissertation in Quantum Entanglement},	% subject of the document
 pdfcreator={pdflatex},		% creator of the document
 pdfproducer={vim},		% producer of the document
 pdfkeywords={Quantum} {Entanglement}, % list of keywords
 pdfnewwindow=true,		% links in new window
 colorlinks=true,		% false: boxed links; true: colored links
 linktoc=page,			% defines which part of an entry in the table of contents is made into a link
 linkcolor=blue,		% color of internal links	(red)
 citecolor=blue,		% color of links to bibliography
 filecolor=magenta,		% color of file links
 urlcolor=blue			% color of external links
}

\usepackage{enumitem}

\usepackage{rotating}
\usepackage{tabu}
\newcommand{\field}[1]{\mathbb{#1}}
\newcommand{\DscrGrp}[1]{\mathrm{#1}}
\newcommand{\LieGrp}[1]{\mathrm{#1}}
\newcommand{\LieAlg}[1]{\mathfrak{#1}}

\newcommand{\cc}[1]{#1^*}	% ha mar van superscript, akkor \cc{{.....}} a hasznalat, de ha csak subscript van, akkor ez afelé rakja
\newcommand{\cmpl}[1]{\overline{#1}}
\newcommand{\pcmpl}[1]{\overline{#1}}

\newcommand{\ve}[1]{\mathbf{#1}}
\newcommand{\vs}[1]{\boldsymbol{#1}}
\newcommand{\tpl}[1]{\mathbf{#1}}
\newcommand{\tpls}[1]{\boldsymbol{#1}}
\newcommand{\vvs}[1]{\underline{\boldsymbol{#1}}}
\newcommand{\sr}[1]{\mathbf{#1}{\boldsymbol{\sigma}}}
\newcommand{\src}[1]{\cc{\mathbf{#1}}{\boldsymbol{\sigma}}}
\newcommand{\Id}{\mathrm{I}}
\newcommand{\IId}{\mathcal{I}}
\newcommand{\cket}[1]{\vert #1 \rangle}
\newcommand{\bra}[1]{\langle #1 \vert}
\newcommand{\bracket}[1]{\langle #1 \rangle}
\newcommand{\defn}{\overset{\text{def.}}{\Longleftrightarrow}}

\newcommand{\nsubset}{\not\subset}

\newcommand{\cnvroof}[1]{{#1}^{\cup}}
\newcommand{\cncroof}[1]{{#1}^{\cap}}

\newcommand{\diffeom}{\cong}

\newcommand{\isom}{\cong}

\DeclareMathOperator{\rk}{rk}
\DeclareMathOperator{\tr}{tr}
\DeclareMathOperator{\Det}{Det}
\DeclareMathOperator{\Spect}{Spect}

\DeclareMathOperator{\Eigv}{Eigv}

\DeclareMathOperator{\Ad}{Ad}

\DeclareMathOperator{\Lin}{Lin}
\DeclareMathOperator{\BiLin}{BiLin}
\DeclareMathOperator{\Conv}{Conv}
\DeclareMathOperator{\Extr}{Extr}

\DeclareMathOperator{\ee}{e}
\newcommand{\dd}{\mathrm{d}}
\newcommand{\transp}{\mathrm{t}}

\newcommand{\td}{\tilde{d}}
\newcommand{\tg}{\tilde{g}}
\newcommand{\tw}{\tilde{w}}

\newcommand{\Bigset}[2]{\Bigl\{ #1 \;\Big\vert\; #2 \Bigr\}}

\providecommand{\abs}[1]{\lvert#1\rvert}
\providecommand{\norm}[1]{\lVert#1\rVert}
\providecommand{\bigabs}[1]{\bigl\lvert#1\bigr\rvert}
\providecommand{\bignorm}[1]{\bigl\lVert#1\bigr\rVert}

\providecommand{\Bignorm}[1]{\Bigl\lVert#1\Bigr\rVert}

\numberwithin{equation}{chapter}
\numberwithin{section}{chapter}
\numberwithin{figure}{chapter}
\numberwithin{table}{chapter}

\newenvironment{organization}
{The organization of this chapter is as follows.
 \renewcommand*{\descriptionlabel}[1]{In \textbf{section ##1},~}
 \begin{description}[rightmargin=\parindent,leftmargin=\parindent,labelsep=0pt]
 }
{\end{description}}

\newenvironment{organizationc}
{\renewcommand*{\descriptionlabel}[1]{In \textbf{chapter ##1},~}
 \begin{description}[rightmargin=\parindent,leftmargin=\parindent,labelsep=0pt]
 }  
{\end{description}}
\newenvironment{remarks}
{Now, we list some remarks and open questions.
 \begin{enumerate}[label=(\roman*),ref=\roman*,rightmargin=\parindent,leftmargin=\parindent,itemindent=0pt, topsep=12pt]
 }
{\end{enumerate}}
\newtheorem{thm}{Theorem}[section]

\newtheorem{prop}[thm]{Proposition}

\newtheorem{alg}[thm]{Algorithm}

\begin{document}

%*******************************************************************************
%	Front matter
%
\frontmatter

% valami gebasz van a pdf tartalomjegyzekbeli linkekkel a frontmatterben

%*******************************************************************************
%	Custom titlepage
%
\cleardoublepage 
\thispagestyle{empty}

\begin{center}
\LARGE
\textbf{Quantum entanglement\\[6pt]
in finite-dimensional Hilbert spaces}\\
\smallskip
\normalsize
by\\
\smallskip
\LARGE
\textbf{Szil{\'a}rd Szalay}\\ % M.~Sc.~Physics ??

\medskip
\normalsize
\textbf{Dissertation}\\[6pt]
presented to the Doctoral School of Physics of the\\[2pt]
Budapest University of Technology and Economics\\[2pt]
in partial fulfillment of the requirements for the degree of\\[6pt]
\textbf{Doctor of Philosophy in Physics}\\

\bigskip
\normalsize
\begin{tabular}[t]{rl}
Supervisor: & Dr.~P{\'e}ter P{\'a}l L{\'e}vay\\
            & research associate professor\\
            & Department of Theoretical Physics\\
            & Budapest University of Technology and Economics
\end{tabular}\\
\bigskip
\includegraphics[width=6cm]{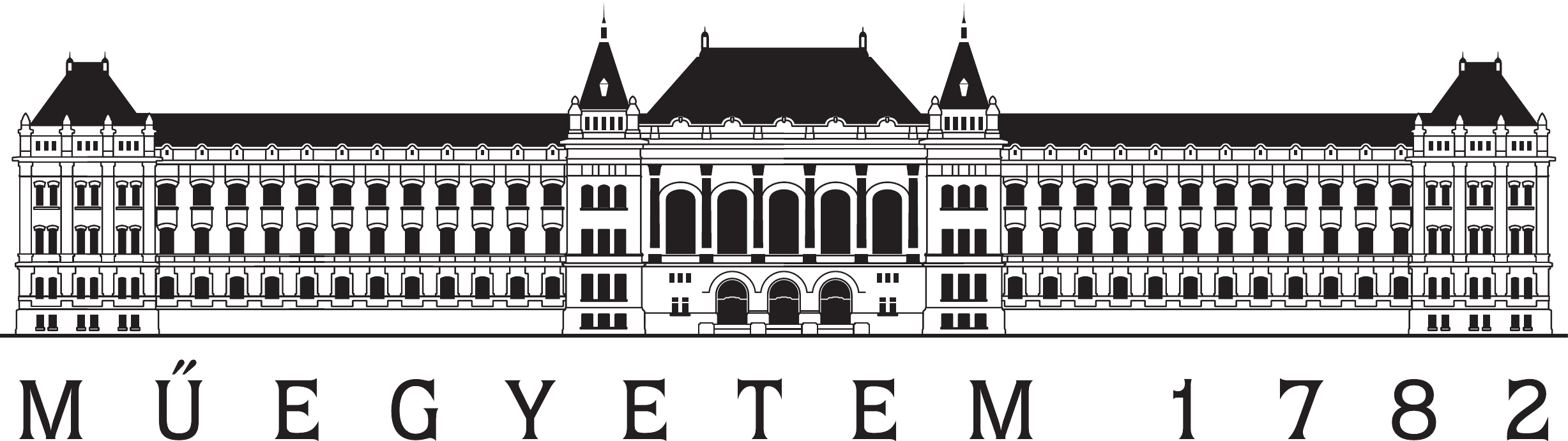}\\
\smallskip
\LARGE
2013
\end{center}

\pagebreak

%*******************************************************************************
%    Dedication.  If the dedication is longer than a line or two,
%    remove the centering instructions and the line break.
%
\thispagestyle{empty}
\hbox{}\cleardoublepage 

\vspace*{13.5pc}
\begin{center}
%  Dedication text (use \\[2pt] for line break if necessary)
\textit{\Small
  To my wife, daughter and son.}
\end{center}

%\cleardoublepage
%    Change page number to 6 if a dedication is present.
%\setcounter{page}{4}	% this causes trouble for the hyperref package

%*******************************************************************************
%	Custom abstractpage
%
\thispagestyle{empty}
\hbox{}\cleardoublepage 

\begin{center}
\begin{minipage}[t]{0.8\textwidth}
\Small
\textsc{Abstract.} In the past decades, 
quantum entanglement has been recognized to be the basic resource in quantum information theory.
A fundamental need is then the understanding its \emph{qualification} and its \emph{quantification}:
Is the quantum state entangled,
and if it is, then how much entanglement is carried by that?
These questions introduce the topics
of \emph{separability criteria} and \emph{entanglement measures}, 
both of which are based on the issue of \emph{classification of multipartite entanglement}.
In this dissertation,
after reviewing these three fundamental topics for finite dimensional Hilbert spaces,
I present my contribution to knowledge.
My main result is the elaboration of the 
\emph{partial separability classification of mixed states}
of quantum systems composed of arbitrary number of subsystems of Hilbert spaces of arbitrary dimensions.
This problem is simple for pure states,
however, for mixed states it has not been considered in full detail yet.
I give not only the \emph{classification} but also
\emph{necessary and sufficient criteria} for the classes,
which make it possible to determine to which class a mixed state belongs.
Moreover, these criteria are given by the vanishing of quantities \emph{measuring} entanglement.
Apart from these,
I present some side results
related to the entanglement of mixed states.
These results are obtained in the learning phase of my studies
and give some illustrations and examples.

% PACS numbers !

\end{minipage}
\end{center}

%*******************************************************************************
%	Table of Contents
%
%\setcounter{tocdepth}{2} % only for editing
%\setcounter{tocdepth}{1} % default vaule
\tableofcontents

%*******************************************************************************
%	Unnumbered chapters (preface, acknowledgments, etc.)
%

\cleardoublepage
\phantomsection
\chapter*{Acknowledgements}

This work would not have been possible without the help of several people,
whom I would like to mention here.

First and foremost, it is a pleasure to thank my adviser \emph{P{\'e}ter L{\'e}vay} 
for his supervising through the years of learning and research,
and for giving me independence to pursue research on the ideas that came across my mind.
His helpful discussions together with his insight and passion for research have always been inspiring.

I would like to extend my gratitude to some of my other teachers as well,
\emph{D{\'e}nes Petz}, \emph{Tam{\'a}s Geszti} and \emph{Tam{\'a}s Matolcsi},
the lectures and books of whom were guides of great value in studying quantum mechanics and mathematical physics.

I am grateful to
\emph{L{\'a}szl{\'o} Szunyogh}, the head of the Department of Theoretical Physics,
and \emph{Gy{\"o}rgy Mih{\'a}ly}, the head of the Doctoral School of Physics,
as well as \emph{M{\'a}ria Vida}, my administrator,
for the flexible, effective and helpful attitude for administrative issues,
supporting my studies to a large extent. 
My Ph.D.~studies were partially supported by
    the \emph{New Hungary Development Plan}  (project ID: T{\'A}MOP-4.2.1.B-09/1/KMR-2010-0002),
    the \emph{New Sz{\'e}chenyi Plan of Hungary} (project ID: T{\'A}MOP-4.2.2.B-10/1--2010-0009)
and the \emph{Strongly correlated systems research group} of the \emph{``Momentum'' program of the Hungarian Academy of Sciences}
(project ID: 81010-00).
%(project IDs: T{\'A}MOP-4.2.1/B-09/1/KMR-2010-0002 and T{\'A}MOP-4.2.2/B-10/1-2010-0009).

I am gerateful to \emph{my parents} for supproting my studies financially and in principles as well.
I would not be succesful without this.

Last but not least,
I thank my wife, \emph{M{\'a}rta}, for her faithful love and everlasting support,
providing the affectionate and peaceful atmosphere
which is an essential condition of any absorbed research.
I would like to dedicate this piece of work to her and to our children.

\cleardoublepage
\phantomsection
\chapter*{Certifications in hungarian}

\noindent
Alul{\'i}rott \textbf{Szalay Szil{\'a}rd} kijelentem,
hogy ezt a doktori {\'e}rtekez{\'e}st magam k{\'e}sz{\'i}tettem
{\'e}s abban csak a megadott forr{\'a}sokat haszn{\'a}ltam fel.
Minden olyan r{\'e}szt, amelyet sz{\'o} szerint 
vagy azonos tartalommal, de {\'a}tfogalmazva m{\'a}s forr{\'a}sb{\'o}l {\'a}tvettem,
egy{\'e}rtelm{\H u}en, a forr{\'a}s megad{\'a}s{\'a}val megjel{\"o}ltem.

\noindent
Budapest, 2013. február 14.

\vspace{48pt}
\begin{tabbing}
\hspace{240pt}\=\underline{\phantom{mmmmmmmmmmmmmmmmmmm}}\\
\>\hspace{24pt}Szalay Szil{\a'a}rd\\
%\>\hspace{24pt}()
\end{tabbing}

\vspace{48pt}

\noindent
Alul{\'i}rott \textbf{Szalay Szil{\'a}rd} hozz{\'a}j{\'a}rulok
a doktori {\'e}rtekez{\'e}sem interneten t{\"o}rt{\'e}n{\H o} 
korl{\'a}toz{\'a}s n{\'e}lk{\"u}li nyilv{\'a}noss{\'a}gra hozatal{\'a}hoz.
%az al{\'a}bbi form{\'a}ban:\\
%- \underline{korl{\'a}toz{\'a}s n{\'e}lk{\"u}l} \\
%- el{\'e}rhet{\H o}s{\'e}g a fokozat oda{\'i}t{\'e}l{\'e}s{\'e}t k{\"o}vet{\H o}en 2 {\'e}v m{\'u}lva, korl{\'a}toz{\'a}s n{\'e}lk{\"u}l\\
%- el{\'e}rhet{\H o}s{\'e}g a fokozat oda{\'i}t{\'e}l{\'e}s{\'e}t k{\"o}vet{\H o}en 2 {\'e}v m{\'u}lva, csak magyarorsz{\'a}gi c{\'i}mr{\H o}l

\noindent
Budapest, 2013. február 14.

\vspace{48pt}
\begin{tabbing}
\hspace{240pt}\=\underline{\phantom{mmmmmmmmmmmmmmmmmmm}}\\
\>\hspace{24pt}Szalay Szil{\a'a}rd\\
%\>\hspace{24pt}()
\end{tabbing}

%Budapest University of Technology and Economics\\
%Dr. P{\'e}ter P{\'a}l L{\'e}vay\\
%
%certificate

%This is to certify that this Thesis
%entitled \emph{Quantum Entanglement in Low-Dimensional Hilbert Spaces}
%and submitted by \emph{Szil{\'a}rd Szalay}
%for award of Ph.D. Degree of the Institute, 
%embodies original work done by him under my supervision.

%\begin{tabbing}
%\noindent
%\hspace{240pt}\=\underline{\phantom{mmmmmmmmmmmmmmmmmmm}}\\
%\>\hspace{24pt}Dr. P{\a'e}ter P{\a'a}l L{\a'e}vay\\
%\>\hspace{24pt}(supervisior)
%\end{tabbing}
%\vspace{48pt}
%\begin{tabbing}
%\hspace{240pt}\=\underline{\phantom{mmmmmmmmmmmmmmmmmmm}}\\
%\>\hspace{24pt}Szil{\a'a}rd Szalay\\
%\>\hspace{24pt}()
%\end{tabbing}

\cleardoublepage
\phantomsection
\chapter*{List of publications}
\label{chap:publist}

The research articles [\ref{pub:ferm}], [\ref{pub:sepcrit}], [\ref{pub:deg6}] and [\ref{pub:partsep}] are covered by this thesis.
The research articles [\ref{pub:attrdist}] and [\ref{pub:attrconc}] are the results of another research project
done in the related field of Black Hole / Qubit correspondence.
The publications are listed in chronological order.

\vspace{22pt}

\begin{enumerate}[label={[}\arabic*{]}, ref=\arabic*,leftmargin=72pt,topsep=72pt,itemsep=6pt]
 \item \label{pub:ferm} 
  \textbf{Szil\'ard Szalay,} P\'eter L\'evay, Szilvia Nagy, J\'anos Pipek,\\
  \textit{A study of two-qubit density matrices with fermionic purifications,}\\
  \href{http://iopscience.iop.org/1751-8121/41/50/505304}
  {J.~Phys.~A \textbf{41}, 505304 (2008)}
  (\href{http://arxiv.org/abs/0807.1804}
  {arXiv: 0807.1804 [quant-ph]})
 \item \label{pub:attrdist}
  P\'eter L\'evay, \textbf{Szil\'ard Szalay,}\\
  \textit{Attractor mechanism as a distillation procedure,}\\
  \href{http://link.aps.org/doi/10.1103/PhysRevD.82.026002}
  {Phys.~Rev.~D \textbf{82}, 026002 (2010)}
  (\href{http://arxiv.org/abs/1004.2346}
  {arXiv: 1004.2346 [hep-th]})
 \item \label{pub:attrconc}
  P\'eter L\'evay, \textbf{Szil\'ard Szalay,}\\
  \textit{$STU$ attractors from vanishing concurrence,}\\
  \href{http://link.aps.org/doi/10.1103/PhysRevD.83.045005}
  {Phys.~Rev.~D \textbf{84}, 045005 (2011)}
  (\href{http://arxiv.org/abs/1011.4180}
  {arXiv: 1011.4180 [hep-th]}) 
 \item \label{pub:sepcrit}
  \textbf{Szil\'ard Szalay,}\\
  \textit{Separability criteria for mixed three-qubit states,}\\
  \href{http://link.aps.org/doi/10.1103/PhysRevA.83.062337}
  {Phys.~Rev.~A \textbf{83}, 062337 (2011)}
  (\href{http://arxiv.org/abs/1101.3256}
  {arXiv: 1101.3256 [quant-ph]})
 \item \label{pub:deg6}
  \textbf{Szil\'ard Szalay,}\\
  \textit{All degree 6 local unitary invariants of $k$ qudits,}\\
  \href{http://iopscience.iop.org/1751-8121/45/6/065302/}
  {J.~Phys.~A \textbf{45}, 065302 (2012)}
  (\href{http://arxiv.org/abs/1105.3086}
  {arXiv: 1105.3086 [quant-ph]})
 \item \label{pub:partsep}
  \textbf{Szil\'ard Szalay,} Zolt\'an K\"ok\'enyesi\\
  \textit{Partial separability revisited: Necessary and sufficient criteria,}\\
  \href{http://link.aps.org/doi/10.1103/PhysRevA.86.032341}
  {Phys.~Rev.~A \textbf{86}, 032341 (2012)}
  (\href{http://arxiv.org/abs/1206.6253}
  {arXiv: 1206.6253 [quant-ph]})
\end{enumerate}

\cleardoublepage
\phantomsection
\chapter*{Thesis statements}
\label{chap:statements}

In the past decades,
quantum entanglement has been recognized to be the basic resource in quantum information theory.
A fundamental need is the understanding of its \emph{qualification} and its \emph{quantification}:
Is the state entangled,
and in this case how much entanglement is carried by it?
These questions introduce the topics
of \emph{separability criteria} and \emph{entanglement measures},
both of which are based on the problem of \emph{classification of multipartite entanglement}.
In the following thesis statements
I present my contribution to these three issues.

\vspace{22pt}

\begin{enumerate}[label=\Roman*.,ref=\Roman*,rightmargin=24pt,leftmargin=24pt,itemsep=24pt]

\item \label{statement:ferm}
I study a $12$-parameter family of two-qubit mixed states,
arising from a special class of two-fermion systems with four single particle states 
or alternatively from a four-qubit state vector with amplitudes arranged in an antisymmetric matrix.
I obtain a local unitary canonical form for those states.
By the use of this
I calculate two famous \emph{entanglement measures} which are 
the Wooters concurrence and the negativity in a closed form. %, and study their behavior.
%In particular, 
I obtain bounds on the negativity for given Wootters concurrence,
which are strictly stronger than those for general two-qubit states.
I show that the relevant entanglement measures satisfy 
the generalized Coffman-Kundu-Wootters formula of distributed entanglement.
I give an explicit formula for the residual tangle as well.\\[6pt]
%The geometry of such density matrices is elaborated in some detail.
%In particular, an explicit form for the Bures metric is given.
%
The publication belonging to this thesis statement is [\ref{pub:ferm}] of the list on page \pageref{chap:publist}.\\[6pt]
The main references belonging to this thesis statement are 
%\cite{LevayNagyPipekTwoFermions,AudenaertetalNegConc,VerstraeteetalNegConc,CKWThreetangle,OsborneVerstraeteMonogamy}.
\cite{LevayNagyPipekTwoFermions,VerstraeteetalNegConc,CKWThreetangle,OsborneVerstraeteMonogamy}.

\item \label{statement:deg6}
Local unitary invariance is a fundamental property of all \emph{entanglement measures}.
I study quantities having this property for general multipartite systems.
%I study the local unitary invariant quantities of multipartite systems.
%Local unitary invariance is a property which is a fundamental one for all \emph{entanglement measures}.
In particular, 
I give explicit index-free formulas for
all the algebraically independent local unitary invariant polynomials up to degree six,
for finite dimensional multipartite pure and mixed quantum states.
I carry out this task by the use of graph-technical methods,
which provide illustrations for this rather abstract topic.\\[6pt]
The publication belonging to this thesis statement is [\ref{pub:deg6}] of the list on page \pageref{chap:publist}.\\[6pt]
The main references belonging to this thesis statement are \cite{HWLUA,HWWLUA,PetiLUA1,PetiLUA23}.

\item \label{statement:sepcrit}
I study the noisy GHZ-W mixture
and demonstrate some \emph{necessary but not sufficient criteria}
for different \emph{classes of separability} of these states.
I find that
the partial transposition criterion of Peres %\cite{PeresCrit}
and the criteria of G\"uhne and Seevinck %\cite{GuhneSevinckCrit} 
dealing directly with matrix elements
are the strongest ones for different separability classes of this two-parameter state.
I determine a set of entangled states of positive partial transpose.
I also give constraints on three-qubit entanglement classes related to the pure SLOCC-classes,
and I calculate the Wootters concurrences of the two-qubit subsystems.\\[6pt]
The publication belonging to this thesis statement is [\ref{pub:sepcrit}] of the list on page \pageref{chap:publist}.\\[6pt]
The main references belonging to this thesis statement are \cite{PeresCrit,GuhneSevinckCrit}.

\item \label{statement:partsep}
I elaborate the \emph{partial separability classification} of mixed states of quantum systems
composed of arbitrary number of subsystems of Hilbert spaces of arbitrary dimensions.
This extended classification is complete in the sense of partial separability
and gives $1+18+1$ partial separability classes in the tripartite case
contrary to the formerly known $1+8+1$.
I also give \emph{necessary and sufficient criteria} for the classes 
by the use of convex roof extensions of functions defined on pure states.
I show that these functions can be defined so as to be entanglement-monotones,
which is another fundamental property of all \emph{entanglement measures}.\\[6pt]
The publication belonging to this thesis statement is [\ref{pub:partsep}] of the list on page \pageref{chap:publist}.\\[6pt]
The main references belonging to this thesis statement are \cite{DurCiracTarrach3QBMixSep,DurCiracTarrachBMixSep,SeevinckUffinkMixSep}.

\item \label{statement:threeqb}
For the case of three-qubit systems,
by the use of the Freudenthal triple system approach of three-qubit pure state entanglement,
I obtain a set of functions on pure states,
whose convex roof extensions
give \emph{necessary and sufficient criteria} for the partial separability classification.
%In the case of three-qubit systems,
%the Freudenthal triple system approach of entanglement of three-qubit pure states
%allows us to define a set of functions on pure states
%whose convex roof extensions
%give \emph{necessary and sufficient criteria} for the partial separability classification.
These functions have some advantages over the ones defined in the general construction,
which is given in the previous thesis statement.
Moreover, these functions fit naturally
for a special three-qubit classification
which arises as the combination of the partial separability classification
with the classification obtained by Ac{\'i}n et.~al.~for three-qubit mixed states.\\[6pt]
The publication belonging to this thesis statement is [\ref{pub:partsep}] of the list on page \pageref{chap:publist}.\\[6pt]
The main references belonging to this thesis statement are \cite{BorstenetalFreudenthal3QBEnt,DurCiracTarrach3QBMixSep,DurCiracTarrachBMixSep,Acinetal3QBMixClass,SeevinckUffinkMixSep}.

\end{enumerate}

\cleardoublepage
\phantomsection
\chapter*{Prologue}

% Mikrol akarok itt irni?
% - altalanos szempontok: legyen ertheto nem-szakmai kozonsegnek
% - kitekintes, nagyobb kontextusba helyezes, alkalmazasok
% - a vegen jutunk el az entanglement-hez, es ahhoz pedig a Chapter I. elejen kell irni bevezetot. 
%

The laws of quantum mechanics proved to be very successful 
in the description and prediction of the behaviour of the microworld.
Among these predictions, however, there were some very surprising ones
which are in connection with the description of composite quantum systems.
In the formalism of quantum mechanics,
the so called \emph{entangled (or inseparable) states} of composite systems appear naturally,
while the understanding of the correlations of the physical qantities measured on the subsystems
of a system being in an entangled state
is a challenge for the mind.
Namely, these correlations arise from the quantum mechanical interactions between the subsystems,
and they can not be modelled classically,
these are the manifestations of the entirely quantum behaviour of the nature.
Entanglement theory is therefore
a deep and fundamental field of central importance,
lying in the very basics of the understanding of the physical world.

An interesting twist of the story is that
these nonclassical correlations can be used for nonclassical solutions of classical, moreover, of nonclassical tasks,
leading to the idea of quantum computation \cite{FeynmanSimPhysComp}.
%%
%This is 
These nonclassical computational and information theoretical methods
are the subject of the emerging field of quantum information theory, 
which is the extension of the classical information theory for quantum systems,
dealing with these quantum correlations \cite{NielsenChuang}.
The significance of this relatively new field of science is hallmarked, among other things, by the Wolf Prize in Physics in this year.

In the scope of quantum information theory,
there are entirely nonclassical, information theoretical tasks
(such as quantum communication with
super-dense coding, quantum teleportation, quantum key distribution, quantum cryptography, quantum error correction)
and also classical computational tasks 
(such as quantum algorithms for factoring numbers, for quantum search, and for further tasks.)
What is really fascinating is that quantum algorithms significantly outperform the best known classical algorithms for the same tasks,
moreover, they are able to solve some problems in polynomial time, %(which is regarded efficient,)
which problems can not be solved in polynomial time by the known classical algorithms.

During the run of all the above quantum protocols,
the basic resource expended is entanglement,
that is, composite quantum systems being in entangled states.
A fundamental need is then the studying of the characterization of entanglement,
which is the main concern of this dissertation.
Although the entanglement which is used for quantum information processing tasks
is presented mostly in maximally entangled Bell pairs of two qubits,
but the structure of entanglement is far richer than that of two-qubit pure states.
We will consider some aspects of this issue in the present dissertation,
here and now we just want to emphasize that
the rich structure of multipartite entanglement might provide a lot of opportunities, 
which are still far from being explored and utilized.

The utilization of even the bipartite entanglement is by no means an easy job.
Quantum mechanics works in microscopic scales,
and, due to the environmental decoherence,
the manifestations of this particular behaviour are hard to reach.
Effects of entanglement are studied in many-body systems as well,
but an important color in the picture is that the experimental manipulation
of individual quantum objects is not out of reach,
as is also illustrated by the Nobel Prize in Physics in last year.

% itt kell egymashoz fuzni a temakat !!!!!
% itt technikai kifejezesek nelkul, de motivacioval
% de mehet bele eredmeny, konkluzio, megjegyzes is.
%\pagebreak

The organization of this dissertation is as follows:
\begin{organizationc}
\item[\ref{chap:QM}]
we give a brief review on the fundamental topics of quantum entanglement which we deal with.
We introduce the main notions and notational conventions
and attempt to cover the whole material which will be used in the following chapters.
Our main concerns are about
the qualification of entanglement, that is, 
deciding about a given state whether it is entangled or separable;
and the quantification of entanglement, that is, 
defining quantities characterizing the ``amount of entanglement'' carried by a given state,
doing this in some motivated way.
Of course, if we have some evaluated quantities in hand which give the amount of entanglement,
then the decision of entangledness is solved as well,
but we usually do not have such opportunity
and even the decision of entangledness leads to a hard optimization problem.
The situation is more complicated in multipartite systems,
where many different kinds of entanglement arise.
In the following chapters we present our contributions to knowledge in these fields.
\item[\ref{chap:Ferm}]
we start with a special two-qubit system.
Qubit systems are of particular importance 
because, on the one hand, qubits are the elementary building blocks of applications in quantum information theory,
on the other hand, they have a simple mathematical structure leading to explicit results
in the quantification of entanglement.
Apart from that, systems of bigger size can be embedded into multiqubit systems.
For the special family of two-qubit states we deal with,
we evaluate explicitly some measures of entanglement,
and investigate some relations among those.\\
The material of this chapter covers thesis statement \ref{statement:ferm}.
\item[\ref{chap:Deg6}]
we continue with a quite general construction of some quantites characterizing quantum states,
a construction which is independent of the size of the subsystems.
These quantities share the invariance property of the most detailed characterization of entanglement,
so these might provide a natural language for the characterization 
and even for the quantitative description of entanglement.\\
The material of this chapter covers thesis statement \ref{statement:deg6}.
\item[\ref{chap:SepCrit}]
after the investigations of the previous two chapters, concerning the characterization of quantum states by quantities in some sense,
we turn to the problem of the decision of entangledness.
In the literature there are numerous conditions for this.
For the use of these conditions, 
various quantities have to be evaluated for a given state.
Unfortunately, these quantities are given only implicitly in the most of the cases,
and those ones which can be evaluated explicitly
result in sufficient but not necessary criteria of entanglement only.
Here we show some of the criteria of this kind at work,
considering a particular example of a family of three-qubit states.\\
The material of this chapter covers thesis statement \ref{statement:sepcrit}.
\item[\ref{chap:PartSep}]
after the particular examples of the previous chapter, we consider the partial separability problem in general.
The partial separability treat every subsystem as a fundamental unit, regardless of its size or even of the number of its components,
and concerns the existence of entanglement among the subsystems only.
We extend the usual classification of partial separability
and formulate also necessary and sufficient criteria for the decision of different kinds of entanglement.
These criteria are given in terms of quantities measuring entanglement. 
The use of these necessary and sufficient criteria leads
to untractable hard optimization problems in general,
so these criteria can only be used for special families of states,
similarly to other necessary and sufficient criteria.
However, our criteria have the advantage of reflecting clearly the structure of partial separability, 
and they work in a similar way for all classes.
We work out the tripartite case, then we give the general definitions for arbitrary number of subsystems.\\
The material of this chapter covers thesis statement \ref{statement:partsep}.
\item[\ref{chap:ThreeQB}]
after the general constructions of the previous chapter,
we turn to the particular system of three qubits again.
In this case, thanks to a beautiful mathematical coincidence,
another set of quantities can be written for the formulation of the
necessary and sufficient criteria given in the previous chapter.
Although these quantities are not measures of entanglement,
but they fit not only for the partial separability classification
but also for a more interesting classification of three-qubit states which goes a bit beyond partial separability.\\
The material of this chapter covers thesis statement \ref{statement:threeqb}.
%\item[\ref{chap:Summary}]
%we present a summary of our works.
\end{organizationc}

%*******************************************************************************
%	Main Matter
%
\mainmatter

%*******************************************************************************
%	Numbered chapters
%

%\include{dissertation_notations}	% nem kell
%

\chapter{Quantum entanglement}
\label{chap:QM}

In quantum systems, correlations having no counterpart in classical physics arise.
Pure states showing these strange kinds of correlations are called entangled ones \cite{Horodecki4,BengtssonZyczkowski},
and the existence of these states has so deep and important consequences
that Schr{\"o}dinger has identified entanglement to be the characteristic trait of quantum mechanics \cite{SchrodingerEnt,Schrodinger2}.

Historically, the nonlocal behaviour of entangled states of bipartite systems was the main concern first.
Einstein, Podolsky and Rosen in their famous paper \cite{EPRpaper} showed that
under the assumption of locality,
entanglement gives rise to some ``elements of reality'',
that is, values of physical quantities exactly known without measurements,
about which quantum mechanics does not know,
since it gives only statistical answers.
Therefore quantum mechanics is incomplete,
and there may exist variables, hidden for quantum mechanics, which determine the outcomes of the measurements uniquely. 
What is more interesting, is that any hidden-variable model of quantum mechanics is essentially nonlocal \cite{BellOnEPR},
which is the famous, experimentally testable result of Bell.
Nowadays, it is widely accepted that quantum mechanics is a complete, but statistical theory,
and only the composite system possesses values of physical quantities,
it is not possible to ascribe values of physical quantities of local subsystems prior to measurements \cite{BellOnEPR}.

Recently, the focus of attention in entanglement theory changed                                      
from locality issues to more general forms of nonclassical behaviour \cite{Horodecki4}.
As was mentioned in the Prologue,
the nonclassical behaviour of entangled quantum states has
far-reaching consequences manifested in quantum information theory,
which is the theory of nonclassical correlations together with applications \cite{NielsenChuang,CavesQISEmergingNoMore}.

In this dissertation, we encounter mixed states rather than pure ones,
since the former ones play much more important roles in entanglement theory than the latter ones,
because of multiple reasons.
The majority of methods in quantum information theory, as well as the issues concerning locality,
generally use pure entangled states, which can easily be prepared and which are easy to use to obtain nonclassical results.
However, in a laboratory one can not get rid of the interaction with the environment perfectly,
thus the separable compound state of the system and the environment evolves into an entangled one,
the prepared pure state of the system evolves into a noisy, mixed one.
This was a practical reason for studying mixed state entanglement, 
however, theoretical ones are much more important.
First, in the case of multipartite systems 
even if the state of the whole system is pure,
the states of its bipartite subsystems are generally mixed ones,
which is a hallmark of entanglement in itself \cite{SchrodingerEnt,Schrodinger2}.
Moreover, the understanding of classicality in the language of correlations
can also be done only for mixed states even in the bipartite case \cite{DornerVedralCorrelations}.

The definition of entanglement and separability of mixed states
was given first by Werner \cite{WernerSep}.
%A mixed state is separable if it is not a mixture of product states
In this paper, he also constructed famous examples for mixed states which are entangled
and still local in the sense that a local hidden variable model can be constructed
for that, describing the usual projective measurements.
So we could think that from the point of view of nonclassicality, 
entanglement does not grasp the nonclassical behaviour perfectly.
%So from this point of view, not all entangled states are of the same importance.
However, an important result, came from quantum information theory, disprove this.
Namely, every entangled state can be used for some nonclassical task
\cite{MasanesEntanglementExtractFromAllNonsep,MasanesAllEntUseful,MasanesAllEntHiddenNonlocal}.
So, for mixed states, nonlocality is considered only as a stronger manifestation of nonclassicality,
but entanglement is still important from the point of view of nonclassicality.

In this chapter, 
we give a brief review of the fundamental topics we deal with
in quantum mechanics \cite{NeumannFoundQM,PetzLinAnal,BengtssonZyczkowski}
and quantum entanglement \cite{Horodecki4}, with some connections to quantum information theory \cite{PetzQInfo,NielsenChuang,PreskillNotes}.
We introduce the main notions together with the notational conventions,
and we attempt to cover the whole material which will be used in the following chapters.
We will see that entanglement in itself is a direct consequence of the formalism of the mathematical description of quantum mechanics.
%%%%%%%%%%%%%%%%%%%%%%%%%%%%
Because of the reasons above,
we follow a treatment from the point of view of mixed states.
This has advantages and also disadvantages.
Usually, quantum mechanics is built upon the primary role of pure states,
resulting in an inductive, better motivated and historically faithful treatment,
in the course of which mixed states arise as ensembles or states of subsystems of entangled systems.
Here we give a reverse treatment, which is an axiomatic, deductive and less motivated one, usual in entanglement theory,
in the course of which pure states arise as special cases of mixed states.
%%%%%%%%%%%%%%%%%%%%%%%%%%%%

\begin{organization}
\item[\ref{sec:QM.QuantSys}] 
we start with recalling the general description of singlepartite quantum systems
(section \ref{subsec:QM.QuantSys.Descr})
together with the characterization of the mixedness of the states of those
in the terms of entropic quantities
(section \ref{subsec:QM.QuantSys.Mixedness}).
The most important differences between classical and quantum systems
appear in these very basic topics.
We also give the detailed description of a single qubit,
which is the simplest quantum system
(section \ref{subsec:QM.QuantSys.Qubit}).
\item[\ref{sec:QM.Ent}]
after the issues of singlepartite systems in the previous section,
we turn to the description of compound systems and entanglement.
First, we review the general non-unitary operations on open quantum systems
arising from the quantum interaction inside the bipartite composite of the system with its environment
(section \ref{subsec:QM.Ent.Ops}),
then some basics about the entanglement in bipartite and multipartite systems 
(sections \ref{subsec:QM.Ent.2Part} and \ref{subsec:QM.Ent.NPart}),
and finally, the important point where these two topics meet each other, 
which is the so called distant lab paradigm
(section \ref{subsec:QM.Ent.LO}).
\item[\ref{sec:QM.EntMeas}]
after the basics of entanglement in the previous section,
we turn to issues 
related to the characterization of entanglement in 
some particular few-partite systems.
First we review some tools for the quantification of bipartite entanglement
(section \ref{subsec:QM.EntMeas.EntMeas}),
then we consider the pure and mixed states of general bipartite
(sections \ref{subsec:QM.EntMeas.2Pure} and \ref{subsec:QM.EntMeas.2Mix})
and two-qubit systems 
(section \ref{subsec:QM.EntMeas.2QBPure} and \ref{subsec:QM.EntMeas.2QBMix}).
The structure of multipartite entanglement is much more complex,
we just review some important results for the case of
three-qubit pure and mixed states
(sections \ref{subsec:QM.EntMeas.3QBPure} and \ref{subsec:QM.EntMeas.3QBMix}),
and of four-qubit pure states 
(section \ref{subsec:QM.EntMeas.4QBPure}).
\end{organization}

%*******************************************************************************
%*******************************************************************************
\section{Quantum systems}
\label{sec:QM.QuantSys}

In the most part of this dissertation, 
we deal with quantum states rather than physical quantities themselves.
By state we mean in general something
what determines the values of measurable physical quantities in some sense.
In classical mechanics, 
the (pure) state of the system is represented by a point 
in a subset of a $2f$ dimensional real vector space,
or more precisely in a simplectic manifold, called \emph{phase space},
where $f$ denotes the number of the degrees of freedom.
In principle, the values of all physical quantities are completely determined by the actual phase point,
so physical quantities are then represented by functions on this space.
The case of quantum mechanics is more subtle.
Instead of the real finite dimensional phase space
we have a complex separable Hilbert space,
the rays of that are regarded as (pure) quantum states.
Moreover, the values of physical quantities are \emph{not} determined by the quantum state,
only distributions of them.
%
%In the following, we review the description of quantum systems
%with 

%*******************************************************************************
\subsection{Description of quantum systems}
\label{subsec:QM.QuantSys.Descr}
The mathematical foundations of quantum mechanics are due to von Neumann \cite{NeumannFoundQM}.
Let $\mathcal{H}$ be the complex Hilbert space corresponding to a \emph{quantum system.}
In the whole of this dissertation, we consider systems having finite dimensional Hilbert space only.
The dimension of the Hilbert space is denoted by $d$.
In the classical scenario, this corresponds to the discrete phase space of $d$ points.
The system in the particular case when $d=2$ is called \emph{qubit}.
This case is not only the most simple but also a very exceptional one,
there are many mathematical coincidences which hold only in two dimensions.
We will see some manifestations of them in the following.

The \emph{dynamical variables} of the quantum system, also called \emph{observables},
are represented by normal operators acting on $\mathcal{H}$,
\begin{equation*}
%\label{eq:A}
\mathcal{A}(\mathcal{H}) = \Bigset{A\in\Lin(\mathcal{H})}{AA^\dagger=A^\dagger A}.
\end{equation*}
Operators of this kind admit the spectral decomposition
\begin{equation*}
A=\sum_ia_i\cket{\alpha_i}\bra{\alpha_i}, \qquad\text{where}\qquad \bracket{\alpha_{i'}\vert\alpha_i}=\delta^{i'}_{i},
\end{equation*}
which is of fundamental importance for the structure of the theory.
As we will see, the discrete eigenvalues represent the discrete outcomes of the measurments,
which is how quantum mechanics describe the quantized phenomena of the microworld.
The dynamical variables in quantum mechanics are usually inherited from the classical mechanics,
where they take real values.
In this case the quantum mechanical dynamical variables are represented by self-adjoint operators,
having real eigenvalues.
(Sometimes, only these operators are called observables.)
Another note is that there is a freedom in the choice of the Hilbert space,
as far as the considered observables can be represented on that.

The \emph{state}
of the quantum system is represented 
by a self-adjoint positive semidefinite operator acting on $\mathcal{H}$,
which is normalized, which means in this context that its trace is equal to $1$.
%By this means, these operators become 
These operators are called statistical operators, or \emph{density operators.}
The set of the states is denoted by $\mathcal{D}\equiv\mathcal{D}(\mathcal{H})$,
which is then%
%%%%%%%%%%%%%%%%%%%%%%%%
\footnote{Strictly speaking,
the states are the probability measures on the lattice of subspaces of the Hilbert space \cite{FayTorosKvLogika},
and the set of them is isomorphic to $\mathcal{D}(\mathcal{H})$ only for $d>2$, 
which is Gleason's theorem \cite{GleasonThm}.
In the pathological $d=2$ case there are probability measures
to which density operators can not be assigned.
We often consider qubits, but we deal only with density operators,
and, inaccurately, by states we mean density operators only.}
%%%%%%%%%%%%%%%%%%%%%%%%
\begin{equation*}
%\label{eq:D}
\mathcal{D}(\mathcal{H}) = 
\Bigset{\varrho\in\Lin(\mathcal{H})}{\varrho^\dagger =\varrho,\; \varrho\geq0,\; \tr\varrho=1}.
\end{equation*}
%This set has remarkable structures, which are 
The self-adjoint operators form a vector space over the field of real numbers.
This vector space can also be endowed with an inner product and also a metric.
%the so called Hilbert-Schmidt megtric:
%\begin{equation*}
%
%\end{equation*}
The operators of unit trace forms an affin subspace in that,
while the positive semidefinite operators form a cone,
which is convex.
$\mathcal{D}(\mathcal{H})$ is then the intersection of these two,
so it is a convex set in the affin subspace of unit trace
in the real vector space of self-adjoint operators acting on $\mathcal{H}$.
By virtue of this, the dimension of $\mathcal{D}(\mathcal{H})$ is $d^2-1$.
The $\pi$ extremal points of $\mathcal{D}(\mathcal{H})$ are of the form $\pi=\cket{\psi}\bra{\psi}$,
where $\cket{\psi}\in\mathcal{H}$ is normalized, $\norm{\psi}^2=1$.
They are called \emph{pure states},
and they form a $2d-2$-dimensional submanifold of $\mathcal{D}(\mathcal{H})$,
denoted with  $\mathcal{P}(\mathcal{H})$.
Contrary to the classical scenario, here we have continuously many pure states even for qubits.
The set of states is the convex hull of the pure states $\mathcal{D}(\mathcal{H})=\Conv\bigl(\mathcal{P}(\mathcal{H})\bigr)$,
in other words, every state can be formed by the convex combination of pure states,
\begin{equation}
\label{eq:mixing}
\varrho=\sum_{j=1}^m p_j\pi_j, 
%\cket{\psi_j}\bra{\psi_j},
\end{equation}
where the $m$-tuple $\tpl{p}=(p_1,\dots,p_m)$ of convex combination coefficients is positive
and normalized with respect to the $1$-norm, $\norm{p}_1=\sum_jp_j=1$.
The set of such $m$-tuples, the $m-1$-simplex, is denoted with $\Delta_{m-1}\subset\field{R}^m$.
The principle of measurement, given in the following paragraphs,
enables us to consider this $\tpl{p}$ as a discrete probability distribution.
If the convex combination is not trivial then the state is called \emph{mixed state},
and its interpretation is that the system is in the pure state $\pi_j$ with probability $p_j$.
If an ensemble of quantum systems being in pure states $\pi_j$ with mixing weights $p_j$ is given,
then random sampling results in such a distribution.
Note that here, contrary to the classical scenario,
the pure states have intrinsic structure,
so a mixed quantum state is not only a probability distribution
but a probability distribution together with directions in the Hilbert space.

A \emph{(generalized) measurement} on the system is given by a set of measurement operators 
\begin{subequations}
\label{eq:Meas}
\begin{equation*}
\Bigset{M_i\in\Lin(\mathcal{H})}{i=1,\dots,m,\;\sum_iM_i^\dagger M_i=\Id}.
\end{equation*}
A \emph{selective measurement} has $m$ outcomes,
resulting in the $m$ post-measurement states:
\begin{equation}
\label{eq:selMeas}
\varrho\qquad\longmapsto\qquad\varrho_i'=\frac{M_i\varrho M_i^\dagger}{\tr M_i\varrho M_i^\dagger},
\qquad\text{with probability}\qquad q_i=\tr M_i\varrho M_i^\dagger.
\end{equation}
(The $\sum_iM_i^\dagger M_i=\Id$ resolution of identity ensures that $\sum_iq_i=1$.)
Here we have physical access to the $\varrho_i$ outcome states of the measurement,
under which we mean that we are able to execute different quantum operations on the different outcome systems.
Note that the probabilistic nature of the measurements 
is an inherent property of quantum mechanics,
it does not come from that the measurement devices are inaccurate and sometimes miss the right output.
Quite the contrary, these principles of quantum measurements are formulated with ideal mesurement devices.
Another point here is that 
the linearity of the trace in the $q_i$ probabilities
allows us to consider the (\ref{eq:mixing})
convex combination of pure states as a statistical mixture of states,
since the probabilities of the measurement outcomes arise from a weighted average of that of pure states.

The other main difference between the classical and quantum measurement is
that the measurement inherently affects, disturbes the state of the system.
If we carry out the measurement but forget about which outome we got,
that is, we form the mixture of the post-measurement states,
which is the result of a \emph{non-selective measurement}, we get
\begin{equation}
\label{eq:nselMeas}
\varrho\qquad\longmapsto\qquad\varrho'=\sum_iq_i\varrho_i'=\sum_iM_i\varrho M_i^\dagger,
\end{equation}
\end{subequations}
which is not equal to the original state in general.
Physically, the measurement device interacts with the system, and this interaction can not be neglected.
%Which is a bit disturbing is that this phenomenon is inherently contained in the formalism,

In the special case of the \emph{von Neumann measurement},
which is the archetype of measurements,
the measurement operators $M_i=P_i$ are projectors of orthogonal supports, 
$P_i=P_i^\dagger$, $P_iP_{i'}=\delta_{ii'}P_i$.
In this case, the repeated measurements give the same outcome.
The $P_i$ projectors arise as the spectral projectors of an observable $A$,
and the measured value of the observable in the case of the $i$th outcome of the measurement 
is the $a_i$ eigenvalue corresponding to the eigensubspace
onto which $P_i$ projects.
The expectation value of the measurement is then
\begin{equation}
\bracket{A}\equiv\sum_ia_iq_i=\tr A\varrho,
\end{equation}
in this sense the state $\varrho$ defines a linear functional on the observables.
In the next section we will see how the~(\ref{eq:Meas}) generalized measurement arises.

If the $q_i=\tr M_i\varrho M_i^\dagger$ measurement statistics is the only thing of interest,
then it is enough to deal with the $E_i=M_i^\dagger M_i$ positive operators
instead of the $M_i$ measurement operators.
The set $\{E_i\}$ is called
Positive Operator Valued Measure (POVM),
and the maps $\varrho\mapsto\tr E_i\varrho$, determining the measurement statistics, are linear functionals on the states.
This makes the use of POVMs much more convenient than that of the measurement operators.

%Let $A$ be an observable, with the nondegenerated eigenvalues $a_i$ and spectral projections $P_i=\cket{\alpha_i}\bra{\alpha_i}$
%\begin{equation*}
%A=\sum_i a_i P_i.
%\end{equation*}
%The orthogonal spectral projections give rise to a measurement with $M_i=P_i$, which is called \emph{von Neumann-measurement}.
%Because of the linearity of trace,
%the convex combination---mixing---of pure states 
%manifests itself in the weighted average in the probabilities and expectation values.
%(This is where the statistical operator got its name from.)
%Indeed, if the system is in a pure state 
%$\varrho=\cket{\psi}\bra{\psi}$, 
%then we get the value $a_i$ with probability 
%$q_i=\bracket{\alpha_i\vert\varrho\vert\alpha_i}=\vert\bracket{\alpha_i\vert\psi}\vert^2$,
%and the state after the measurement collapses into
%$\varrho_i=\cket{\alpha_i}\bra{\alpha_i}$.
%The expectation value of the observable is 
%$\bracket{A}=\sum_i q_ia_i=\bracket{\psi\vert A\vert\psi} = \tr{\varrho A}$.
%On the other hand, if the system is in the mixed state 
%$\varrho=\sum p_j\cket{\psi_j}\bra{\psi_j}$,
%then we get the value $a_i$ with probability 
%$q_i=\bracket{\alpha_i\vert\varrho\vert\alpha_i}=\sum_jp_j\vert\bracket{\alpha_i\vert\psi_j}\vert^2$,
%and the state after the measurement collapses into 
%$\varrho_i=\cket{\alpha_i}\bra{\alpha_i}$.
%The expectation value of the observable is 
%$\bracket{A}=\sum_i q_ia_i= \sum_jp_j\bracket{\psi_j\vert A\vert\psi_j} = \tr{\varrho A}$.

The linear structure in the underlying Hilbert space is also important.
If the state is pure, 
sometimes we deal with the \emph{state vector} $\cket{\psi}\in\mathcal{H}$
instead of the rank one density matrix $\cket{\psi}\bra{\psi}\in\mathcal{P}(\mathcal{H})\subset\mathcal{D}(\mathcal{H})$.
In this case, we regard the pure state in the Hilbert space as the phase-equivalence class of the state vector.
Let $\{\cket{i}\mid i=1,\dots,d\}$ be an orthonormal basis in $\mathcal{H}$,
sometimes called computational basis, 
then the state vector can be written as%
%%%%%%%%%%%%%%%%%%%%%%%%
\footnote{The indices of the basis run sometimes from $0$ to $d-1$,
especially in the elements of quantum information theory,
where this practice is rather convenient.
But note that in this case \emph{all} indices, 
even those of the convex combination coefficients in (\ref{eq:mixing}),
should run from zero,
because Schr{\"o}dinger's mixture theorem couples together these two kinds of summations,
as we will see in  (\ref{eq:SchMixtureThm}) in the next subsection.} 
%%%%%%%%%%%%%%%%%%%%%%%%
\begin{equation*}
\cket{\psi}=\sum_{i=1}^d\psi^i\cket{i},\qquad\text{where}\qquad \psi^i=\bracket{i|\psi}\in\field{C}.
\end{equation*}
We use the convention for coefficients with lower indices $\cc{(\psi^i)}=\psi_i$,
which are the $\bracket{\psi|i}$ coefficients of the $\bra{\psi}=\cket{\psi}^*\in\mathcal{H}^*$ dual vector.%
%%%%%%%%%%%%%%%%%%%%%%%%
\footnote{In the finite dimensional case, 
the $\bracket{\cdot|\cdot}$ inner product identifies $\mathcal{H}^*$ with $\mathcal{H}$,
and we denote this identification with the star: $*:\mathcal{H}\to\mathcal{H}^*$, $\cket{\psi}^*=\bra{\psi}$,
and since $\mathcal{H}^{**}\isom\mathcal{H}$ in the finite dimensional case, $\bra{\psi}^*=\cket{\psi}$.
This can be extended to tensors as well.
For example $\theta=\theta^i_{\phantom{i}j}\cket{i}\otimes\bra{j}\in\mathcal{K}\otimes\mathcal{H}^*$
(the $\otimes$ sign is often omitted in the case of tensors of this kind),
we have $\theta^*=(\theta^*)_i^{\phantom{i}j}\bra{i}\otimes\cket{j}\in\mathcal{K}^*\otimes\mathcal{H}$,
leading to $(\theta^*)_i^{\phantom{i}j}=\cc{(\theta^i_{\phantom{i}j})}$,
which is denoted simply with $\theta_i^{\phantom{i}j}$ through the identification.
Note, however, that the indices of tensors can not be uppered and lowered independently,
since $\bracket{\cdot|\cdot}$ is conjugate-linear in the first position.
%Since the composition of linear transformations
Linear operations act from the left, that is, $\Lin(\mathcal{H}\to\mathcal{K})\isom \mathcal{K}\otimes\mathcal{H}^*$.
We have also the transposition, which is the natural operation 
$\transp:\mathcal{K}\otimes\mathcal{H}^*\to \mathcal{H}^*\otimes\mathcal{K}$, 
$\cket{i}\otimes\bra{j}\mapsto\bra{j}\otimes\cket{i}$.
This is defined without the inner product,
it simply interchanges the Hilbert spaces,                                                              
so it can act independently on pairs of indices.
Later, more general partial transpositions, reshufflings and general permutations of Hilbert spaces will also be used.
For convenience, we have also the hermitian transpostion 
$\dagger=*\circ\transp:\mathcal{K}\otimes\mathcal{H}^*\to \mathcal{H}\otimes\mathcal{K}^*$,
$\cket{i}\otimes\bra{j}\mapsto\cket{j}\otimes\bra{i}$
for the action of linear operations on the dual.
For further details in tensor algebraic constructions, see part 2.~in \cite{MatolcsiSpacetime},
with slightly different notations.} 
%%%%%%%%%%%%%%%%%%%%%%%%

The Hilbert space $\mathcal{H}$ is closed under \emph{complex-linear combination}
$\cket{\psi}=\sum_jc_j\cket{\psi_j}$,
which is called \emph{superposition} in this context.
This makes the Hilbert space 
and also $\mathcal{P}(\mathcal{H})$ a much more interesting place 
than the classical phase space,
and in multipartite systems this is responsible for entanglement.
On the other hand,
the space of states $\mathcal{D}$ is closed under \emph{convex combination}
$\varrho=\sum_jp_j\cket{\psi_j}\bra{\psi_j}$,
which is called \emph{mixing.}
A fundamental difference between these two constructions is the possibility of interference.
The measurement probabilities in the first and second cases are
\begin{align*}
q_i&=\tr \bigl(M_i\cket{\psi}\bra{\psi}M^\dagger_i\bigr)=\Bignorm{\sum_jc_j M_i\cket{\psi_j}}^2,\\
q_i&=\tr \bigl(M_i\varrho M^\dagger_i\bigr)=\sum_jp_j\Bignorm{M_i\cket{\psi_j}}^2.
\end{align*}
In the first case, contrary to the second one, $q_i$ can be zero even if the vectors $M_i\cket{\psi_j}$ are nonzero,
which is a manifestation of the famous phenomenon of quantum interference.

If the system is in a pure state $\pi=\cket{\psi}\bra{\psi}\in\mathcal{P}$,
and we consider a von Neumann measurement
with the measurement operators being the orthogonal spectral projectors of a nondegenerate observable,
$M_i=\cket{\alpha_i}\bra{\alpha_i}$,
then we get back \emph{Born's Rule}
\begin{equation*}
q_i=\abs{\bracket{\alpha_i|\psi}}^2.
\end{equation*}
The square in that, together with the interference of different measurement outcomes
could have been the first indications
that there is a Hilbert space somewhere in the grounds of the mathematical description of quantum mechanics.
On the other hand, we can consider this measurement
as a transformation of the complex $\psi^i=\bracket{\alpha_i|\psi}$ superposition coefficients
to the real $q_i=\abs{\psi^i}^2$ mixing weights.
%In this sense, the measurement sweeps the state of interference.
In this sense, the measurement washes away the interference.

The probabilities of the outcomes of the measurements are given by the trace, or the inner product,
both of them are invariant under the action of the unitary group%$\LieGrp{U}(\mathcal{H})$.%
%%%%%%%%%%%%%%%%
\footnote{This is, of course, not a coincidence. 
The trace is the \emph{natural} linear map from $\Lin(\mathcal{H})\isom\mathcal{H}\otimes\mathcal{H}^*$ to $\field{C}$,
and $\mathcal{H}^*$ is \emph{naturally} identified with $\mathcal{H}$ by the inner product of the Hilbert space,
and the unitary group is the invariance group of the inner product, by definition.}
%%%%%%%%%%%%%%%%
$\LieGrp{U}(\mathcal{H})$,
which is, after fixing an orthonormal basis, isomorphic with the classical matrix group $\LieGrp{U}(d)$.
For $\cket{\psi}\mapsto\cket{\psi'}=U\cket{\psi}$ with an $U \in \LieGrp{U}(\mathcal{H})$,
we have the same group action on the states and the observables
\begin{align*}
%\cket{\psi}\qquad\longmapsto\qquad U\cket{\psi}\\
\varrho\qquad&\longmapsto\qquad \varrho'= U\varrho U^\dagger,\\
 A     \qquad&\longmapsto\qquad  A'     = U A      U^\dagger,\\
 M_i   \qquad&\longmapsto\qquad  M_i'   = U M_i    U^\dagger,
\end{align*}
since all of them are elements in $\Lin(\mathcal{H})\isom\mathcal{H}\otimes\mathcal{H}^*$.
One can see, which is desired in physics,
that only the description may depend on the chosen coordinates in $\mathcal{H}$,
%not the physics itself.
not the measurement statistics.

There is another role of unitary transformations 
besides the coordinate transformation in the Hilbert space,
which is the time evolution.
In this case the state and the observables are transformed differently,
that is, their ``relative angle'' in $\Lin(\mathcal{H})$ changes.
In quantum mechanics, the time evolution of the \emph{state} of an \emph{isolated quantum system} 
is described by a unitary transformation
$\varrho(0)\mapsto\varrho(t) = U(t)\varrho(0) U(t)^\dagger$,
while this time the observables are independent of time, hence not transformed (Schr{\"o}dinger picture).
This evolution operator can be obtained from the \emph{von Neumann equation}%
%%%%%%%%%%%%%%%%%%%%%%%%
\footnote{In this dissertation we use metric system in which $\hbar = 1$.}
%%%%%%%%%%%%%%%%%%%%%%%%
\begin{equation*}
\frac{\partial\varrho(t)}{\partial t} = -i\bigl[H(t),\varrho(t)\bigr]
\end{equation*}
given with the self-adjoint observable $H\in\mathcal{A}(\mathcal{H})$ corresponding to the energy of the system,
called Hamiltonian,
as the time-ordered operator
\begin{equation*}
U(t)=\mathcal{T}\exp\Biggl(-i\int_0^t H(t')\dd t'\Biggr).
\end{equation*}
This reduces to $U(t)=\exp\bigl(-iHt\bigr)$ if $H$ does not depend on time.

%*******************************************************************************
\subsection{The mixedness of a state}
\label{subsec:QM.QuantSys.Mixedness}

A good summary on the mixedness of the quantum states can be found in \cite{BengtssonZyczkowski}.
The decomposition of a mixed state into the ensemble of pure states is,
contrary to the classical case, far from unique.
In general, the state $\varrho$ can be written 
with the ensemble 
\begin{equation*}
\Bigset{\bigl(p_j,\cket{\psi_j}\bra{\psi_j}\bigr)}{\tpl{p}\in\Delta_{m-1},\;\norm{\psi_j}^2=1}
\end{equation*}
as
\begin{equation*}
\varrho=\sum_{j=1}^m p_j\cket{\psi_j}\bra{\psi_j}.
\end{equation*}
The spectral decomposition defines, however, a unique ensemble. %, which is called the \emph{eigendecomposition}.
It consists of the spectrum and the \emph{orthogonal} spectral projections, 
\begin{equation*}
\Bigset{\bigl(\lambda_i,\cket{\varphi_i}\bra{\varphi_i}\bigr)}{\tpls{\lambda}\in\Delta_{d-1},\;\bracket{\varphi_{i'}|\varphi_i}=\delta^i_{i'}},
\end{equation*}
giving
\begin{equation*}
\varrho=\sum_{i=1}^d \lambda_i\cket{\varphi_i}\bra{\varphi_i}.
\end{equation*}
There is an elegant theorem,
called Schr\"odinger's mixture theorem \cite{SchrodingerMixtureThm} or
Gisin-Hughston-Jozsa-Wootters lemma \cite{GisinMixtureThm,HughstonJozsaWoottersMixtureThm},
which gives all the possible decompositions of a density matrix.
It relates them to the spectral decomposition in the following way:
\begin{equation}
\label{eq:SchMixtureThm}
\sqrt{p_j}\cket{\psi_j} = \sum_{i=1}^d U^j_{\phantom{j}i}\sqrt{\lambda_i}\cket{\varphi_i},
\end{equation}
where $U^j_{\phantom{j}i}$s are the entries of an $m\times d$ matrix with orthonormal columns, $U^\dagger U=\Id_d$.
The meaning of this matrix of coefficients is clarified later 
from the point of view of pure states of bipartite systems (section \ref{subsec:QM.EntMeas.2Pure}).
The set of such matrices is a compact complex manifold
$V_d(\field{C}^m)\diffeom \LieGrp{U}(m)/\LieGrp{U}(m-d)$,
which is called Stiefel manifold.

Since we have the quantum state as a mixture of pure states, % (\ref{eq:mixing}),
moreover, as the same mixture for different ensembles of pure states,
as a natural question arises, how mixed a state is then?
The mixedness of a state is given by the notion of majorization.
First we invoke
the notion of \emph{majorization for discrete probability distributions.}
For two probability distributions % of length $m$
$\tpl{p}=(p_1,\dots,p_m)\in\Delta_{m-1}$ and $\tpl{q}=(q_1,\dots,q_m)\in\Delta_{m-1}$,
$\tpl{p}$ is majorized by $\tpl{q}$,
denoted with the symbol $\preceq$, 
with the following definition: 
\begin{equation}
\label{eq:major}
\tpl{p}\preceq\tpl{q}\qquad\defn\qquad
\sum_{i=1}^k p_i^\downarrow \leq \sum_{i=1}^k q_i^\downarrow  \quad \forall k=1,2,\dots,m,
\end{equation}
where $\downarrow$ in the superscript means decreasing order.
%(For $k=m$ the inequality turns to equality, since both sides of it are equal to $1$.
%If the length of $\tpl{p}$ and $\tpl{q}$ differs, one can add some zeroes to the shorter one.)
%(For example let $\tpl{p}^\downarrow:=(1/2,1/8,\dots)$ and
%$\tpl{q}^\downarrow:=(1/3,1/3,\dots)$,
%then $\tpl{p}\npreceq\tpl{q}$ and $\tpl{q}\npreceq\tpl{p}$.)
%This is why the majorization gives rise merely to a \emph{partial} ordering.
The majorization is clearly \emph{reflexive}
($\tpl{p}\preceq \tpl{p}$)
and \emph{transitive}
(if $\tpl{p}\preceq \tpl{q}$ and $\tpl{q}\preceq \tpl{r}$ then $\tpl{p}\preceq \tpl{r}$)
but the \emph{antisymmetry}
(if $\tpl{p}\preceq \tpl{q}$ and $\tpl{q}\preceq \tpl{p}$ then $\tpl{p}=\tpl{q}$)
holds only in a restricted manner: if $\tpl{p}\preceq \tpl{q}$ and $\tpl{q}\preceq \tpl{p}$ then $\tpl{p}^\downarrow=\tpl{q}^\downarrow$.
On the other hand,
it is clear that $\tpl{p}\npreceq \tpl{q}$ does not imply $\tpl{q}\preceq \tpl{p}$,
in other words there exist pairs of probability distributions which we can not compare by majorization.
Hence the majorization defines a \emph{partial order} on the set of probability distributions \emph{up to permutations}.

With respect to majorization,
the set of discrete probability distributions contains a greatest and a smallest element.
One can check that all
$\tpl{p}\in\Delta_{m-1}$ majorize the uniform distribution 
and all $\tpl{p}$ is majorized by the distribution containing only one element,
\begin{equation*}
(1/m,1/m,\dots)\preceq \tpl{p}\preceq (1,0,\dots).
\end{equation*}
It is generally accepted to 
use the \emph{mathematical definition} of majorization 
for the comparsion of disorderness (mixedness) of discrete probability distributions.
%Employing the notion of majorization we can compare the amount of disorder (mixedness)
%contained in different probability distributions.
If $\tpl{p}\preceq \tpl{q}$ then we can say that $\tpl{p}$ is more disordered (more mixed) than $\tpl{q}$, 
or equivalently, $\tpl{q}$ is more ordered (more pure) than $\tpl{p}$,
but, as was mentioned before, 
there are pairs of distributions for which their rank of disorder can not be compared in this sense.
%
%Note that we do not have a primary notion of
%what we think about the disorderness in an ordinary sense,
%we have only this \emph{mathematical definition} for that.
%However, this definition is motivated by the compatibility

Real-valued functions defined on probability distributions
and compatible with majorization
are of particular importance.
An $f:\Delta_{m-1}\to\field{R}$ is \emph{Schur-concave}
if
\begin{equation}
\label{eq:ScurCnc}
\tpl{p}\preceq \tpl{q}\qquad\Longrightarrow\qquad f(\tpl{p})\geq f(\tpl{q}).
\end{equation}
%while it is \emph{Schur-concave}, if $-f$ is Schur-convex.
Schur concavity is the definitive property of all \emph{(generalized) entropies,}
which means that
if a distribution is more mixed than the other 
then it has greater entropy.
%and
%\begin{equation*}
%f(1/m,1/m,\dots)\geq f(\tpl{p})\geq f(1,0,\dots).
%\end{equation*}
% entropies tells us something about mixedness.
Note that the entropies can be compared for all pairs of distributions,
not only for those which can be compared by majorization,
so comparsion of mixedness by entropies is not the same as by majorization.

%$\ell_q$ norm of the distribution is Schur convex

The most basic entropy is the \emph{Shannon entropy}
\begin{subequations}
\begin{equation}
\label{eq:Shannon}
H(\tpl{p})=-\sum_j p_j\ln p_j,
\end{equation}
having the strongest properties among all entropies.
The \emph{R{\'e}nyi entropy} is defined as
\begin{equation}
\label{eq:Renyi}
H^\text{R}_q(\tpl{p})=\frac{1}{1-q}\ln\sum_j p_j^q,\qquad q>0,
\end{equation}
which is a generalization of the Shannon entropy in the sense that
$\lim_{q\to1}H^\text{R}_q(\tpl{p})=H(\tpl{p})$.
Its other limits are also notable.
For $q\to0^+$, this is the logarithm of the number of nonzero $p_j$s, known as \emph{Hartley entropy}
\begin{equation}
\label{eq:Hartley}
H^\text{R}_0(\tpl{p}):=\lim_{q\to0^+}H^\text{R}_q(\tpl{p})=\ln|\{j\mid p_j\neq0\}|.
\end{equation}
For $q\to\infty$, it converges to the \emph{Chebyshev entropy}
\begin{equation}
\label{eq:Chebyshev}
H^\text{R}_\infty(\tpl{p}):=\lim_{q\to\infty}H^\text{R}_q(\tpl{p})=-\ln p_\text{max}.
\end{equation}
The \emph{Tsallis entropy} is defined as
\begin{equation}
\label{eq:Tsallis}
H^\text{Ts}_q(\tpl{p})=\frac{1}{1-q}\Bigl(\sum_jp_j^q-1\Bigr),\qquad q>0,
\end{equation}
\end{subequations}
which is, contrary to the \emph{R{\'e}nyi entropy},
a non-additive generalization of the Shannon entropy,
$\lim_{q\to1}H^\text{Ts}_q(\tpl{p})=H(\tpl{p})$.
Note that the Tsallis entropy is the linear leading term in the 
power-series of the R{\'e}nyi entropy,%
%%%%%%%%%%%%%%%%%%%%%%%%
\footnote{Remember that $\ln x = (x-1) - \frac12(x-1)^2 + \frac13(x-1)^3 - \dots +\frac{(-1)^k}{k}(x-1)^k+\dots$
for $0<x$, the role of which is played by 
%the $q$-norm raised to the $q$th power $\norm{\tpl{p}}^2=\sum_jp_j^q$
$\sum_jp_j^q$.}
%%%%%%%%%%%%%%%%%%%%%%%%
this is why Tsallis entropy is sometimes called \emph{linear entropy}.

How to generalize the above conceptions to the quantum case?
A quantum state can be formed as a mixture from different ensembles,
so the $\tpl{p}$ mixing weights are not inherent properties of it.
However, the spectrum of a state is not only well-defined,
but, 
thanks to Schr{\"o}dingers mixture theorem~(\ref{eq:SchMixtureThm})
and the Hardy-Littlewood-P{\'o}lya lemma \cite{BengtssonZyczkowski},
it also majorizes every other mixing weights.
So the spectrum is special from the point of view of mixedness,
and the \emph{majorization of density matrices} is defined 
%by the density matrix itself without any reference to the decompositions
via the corresponding majorization of their spectra,
%$\varrho\preceq\omega$ by definition when
\begin{equation}
\varrho\preceq\omega\qquad\defn\qquad
\Spect\varrho\preceq\Spect\omega.
\end{equation}
By virtue of this, we can compare the mixedness of density matrices.
On the other hand, because of Schur concavity,
the entropy of the spectrum is smaller than that of any other mixing weights.
Now, if the quantum entropies of a state are defined
as the classical entropies of the spectrum,
then they are Schur concave in the sense of the majorization of density matrices.
Moreover,
the entropies above for a quantum state can be written by the density matrix itself
without any reference to the decompositions.

The quantum generalization of the Shannon entropy is
called \emph{von Neumann entropy}
\begin{subequations}
\begin{equation}
\label{eq:Neumann}
S(\varrho)=-\tr(\varrho\ln\varrho)=H\bigl(\Spect\varrho\bigr),
\end{equation}
the \emph{quantum R{\'e}nyi entropy} is
\begin{equation}
\label{eq:qRenyi}
S^\text{R}_q(\varrho)=\frac{1}{1-q}\ln\tr\varrho^q=H^\text{R}_q\bigl(\Spect\varrho\bigr),\qquad q>0,
\end{equation}
while its limits, the \emph{quantum Hartley entropy} is
\begin{equation}
\label{eq:qHartley}
S^\text{R}_0(\varrho):=\lim_{q\to0^+}S^\text{R}_q(\varrho)=\ln\rk\varrho=H^\text{R}_0\bigl(\Spect\varrho\bigr),
\end{equation}
which is the logarithm of the rank of $\varrho$,
and the \emph{quantum Chebyshev entropy}
\begin{equation}
\label{eq:qChebyshev}
S^\text{R}_\infty(\varrho):=\lim_{q\to\infty}S^\text{R}_q(\varrho)=-\ln \max\Spect\varrho
=H^\text{R}_\infty\bigl(\Spect\varrho\bigr).
\end{equation}
The other family,
the \emph{quantum Tsallis entropy} is
\begin{equation}
\label{eq:qTsallis}
S^\text{Ts}_q(\varrho)=\frac{1}{1-q}\bigl(\tr\varrho^q-1\bigr)=H^\text{R}_q\bigl(\Spect\varrho\bigr),\qquad q>0.
\end{equation}
\end{subequations}
An advantage of the Tsallis and R{\'e}nyi entropies is
that they are easy to evaluate for integer $q\geq2$ parameters,
when only matrix powers have to be calculated instead of the entire spectrum.

All of the above quantum entropies vanish for pure states
and reach their maxima for %$\ln d$ for 
\begin{equation}
\label{eq:whitenoise}
\varrho=\frac{1}{d}\Id,
\end{equation}
having the uniform distribution as its spectrum.
This state is sometimes called white noise,
because in this state 
all outcomes of a measurement of a nondegenerate observable
occur with equal probability.
%because in this state
%the system can be measured 
%to be in every pure state of a mutually orthogonal ensemble 
%with equal probability.
%(Such an ensemble is the quantum analog of a set of mutually exclusive events.)

Some other quantities are also in use for the characterization of mixedness.
For example the \emph{purity}
\begin{subequations}
\begin{equation}
\label{eq:purity}
P(\varrho)=\tr\varrho^2,
\end{equation}
the \emph{participiation ratio}
\begin{equation}
\label{eq:partratio}
R(\varrho)=\frac1{\tr\varrho^2},
\end{equation}
which can be interpreted as an effective number of pure states in the mixture,
and the so called \emph{concurrence-squared}
\begin{equation}
\label{eq:conc2}
%C^2(\varrho)=\frac{d}{d-1}\left((\tr \varrho)^2 - \tr \varrho^2\right).
C^2(\varrho)=\frac{d}{d-1}S^\text{Ts}_2(\varrho)=\frac{d}{d-1}\left(1 - \tr \varrho^2\right),
\end{equation}
\end{subequations}
the latter is normalized, $0\leq C^2(\varrho)\leq1$.
All of them are related to the $q=2$ quantum Tsallis (or quantum R{\'e}nyi) entropy,
which is in connection with Euclidean distances in $\mathcal{D}(\mathcal{H})$ \cite{BengtssonZyczkowski}.

%if we prefer to deal with \emph{polynomials} in the $\psi_{ijk}$ and $\cc{\psi}_{ijk}$ coefficients.
%This is the non-trivial polynomial of the lowest degree which is also an entropy, i.e.~Schur-concave,
%so tells us something about mixedness.

The Shannon or von Neumann entropies are widely used in classical and quantum statistical physics,
while their generalizations are often considered unphysical or useless,
mainly because of, e.g., the non-additivity (non-extensivity) of the Tsallis entropies.
An interesting observation of us is that in entanglement theory,
contrary to statistical physics,
the generalized entropies often prove to be more useful than the original one.
We will see in section \ref{subsec:SepCrit.2Part.Entr}
that for a family of three-qubit states,
the generalized entropies for high parameters $q$
give stronger conditions of entanglement.
Here the R{\'e}nyi and Tsallis entropies lead to the same conditions for the same parameters $q$.
Another, more sophisticated example for the usefulness of Tsallis entropies
can be found in section \ref{sec:ThreeQB.GenThreePart},
where it is shown that 
the additive definition of some of the indicator functions for tripartite systems
can be given only by generally non-additive entropies.
In this case the subadditivity seems to be more important than the additivity.
We should mention here also that
Tsallis entropies are sometimes used even in thermodynamics.
For example,
non-extensive thermodynamical models are developed
for the modelling of the behaviour of the quark-gluon plasma
produced in heavy-ion collisions,
see in \cite{BiroNonExt} and in the references therein.

%******************************************************************************
\subsection{The states of a qubit}
\label{subsec:QM.QuantSys.Qubit}

As an example, consider the case of a qubit, $d=\dim\mathcal{H}=2$. %$\mathcal{H}\isom\field{C}^2$.
Spin degree of freedom of particles having $1/2$ spin,
or polarization degree of freedom of photons are the typical physical systems
%to which two dimensional Hilbert spaces are assigned.
whose states are described by qubits.
%the state of which are described by qubits.
In $\mathcal{A}(\mathcal{H})$ we have the linearly independent $\sigma_0,\sigma_1,\sigma_2,\sigma_3$ Pauli operators,
which are self-adjoint, $\tr\sigma_\nu=2\delta_{\nu0}$
and obeying the well-known Pauli algebra
\begin{equation}
\label{eq:Pauli.alg}
\sigma_0\sigma_0=\sigma_0,\qquad
\sigma_0\sigma_i=\sigma_i\sigma_0=\sigma_i,\qquad
\sigma_i\sigma_j=\delta_{ij}\sigma_0+i\sum_k\varepsilon_{ijk}\sigma_k,
\end{equation}
where $i,j,k\in\{1,2,3\}$,
and $\varepsilon_{ijk}$ denotes the parity of the permutation $ijk$ of $123$
if $i$, $j$ and $k$ are different, othervise $\varepsilon_{ijk}$ is zero.
Any density operator can be written
with these in the form
\begin{equation}
\label{eq:qubit}
\varrho=\frac{1}{2}\left(\sigma_0+\sr{x}\right),
\end{equation}
where the $\ve{x}\in\mathbb{R}^3$ \emph{Bloch vector} parametrizes the state,
and we use the shorthand notation $\sr{x}=x^1\sigma_1+x^2\sigma_2+x^3\sigma_3$.
The characteristic equation of $\varrho$
\begin{equation*}
\lambda^2 - \lambda \tr\varrho + \det\varrho = 0
\end{equation*}
allows us to obtain the eigenvalues,
and by the use of the Cayley-Hamilton theorem
\begin{equation*}
\varrho^2 - \varrho \tr\varrho + \Id \det\varrho = 0
\end{equation*}
together with the (\ref{eq:Pauli.alg}) algebra of Pauli operators
we have that $\norm{\ve{x}}^2=1-4\det\varrho$,
which allows us to write the eigenvalues
in geometrical terms
\begin{equation*}
\lambda_\pm
=\frac{1}{2}\left(1\pm\sqrt{1-4\det\varrho}\right)
=\frac{1}{2}\left(1\pm\norm{\ve{x}}\right).
\end{equation*}
This tells us that $\varrho$ is a proper quantum state of a qubit
if and only if $\norm{\ve{x}}^2\leq1$,
while $\varrho$ is a pure state 
if and only if $\norm{\ve{x}}^2=1$.
So, for qubits, we have the space of states $\mathcal{D}(\mathcal{H})\diffeom B^3$,
and its extremal points, i.e.~the set of pure states $\mathcal{P}(\mathcal{H})\diffeom S^2$,
which are called \emph{Bloch ball} and \emph{Bloch sphere} in this context (figure \ref{fig:BlochBall}).
The $\norm{\ve{x}}^2=0$ center of the ball is the $1/2\Id$ white noise.%
%%%%%%%%%%%%%%%%%%%%%%%%
\footnote{So, for qubits, $\sigma_0=\Id$, 
but note that this does not hold for the higher dimensional representations of the Pauli algebra.}
%%%%%%%%%%%%%%%%%%%%%%%%
Note that in this case the whole boundary of $\mathcal{D}(\mathcal{H})$ is extremal.
This does not hold for $d>2$, as can easily be seen by counting the dimensions.%
%%%%%%%%%%%%%%%%%%%%%%%%
\footnote{There are some results on the geometry of the state space of a qutrit ($d=3$),
in which case the Gell-Mann matrices can be used \cite{SarbickiBengtssonQutritGeom,GoyalGeomGenBlochQutrit}.
For general $d$, the suitable generators of $\LieGrp{SU}(d)$ can be used.}
%%%%%%%%%%%%%%%%%%%%%%%%

The only self-adjoint observable in this case is generally of the form% 
%%%%%%%%%%%%%%%%%%%%%%%%
\footnote{Adding $\sigma_0$ only shifts the eigenvalues and leave the eigensubspaces invariant.}
%%%%%%%%%%%%%%%%%%%%%%%%
$A_\ve{u}=\sr{u}$ with $\ve{u}\in\field{R}^3$.
Because of the  (\ref{eq:Pauli.alg}) Pauli algebra,
$\{\frac12\sigma_1,\frac12\sigma_2,\frac12\sigma_3\}$
obey the 
commutation relations of the angular momentum%
%%%%%%%%%%%%%%%%%%%%%%%%
\footnote{Hence they represent the $\LieAlg{su}(2)$ Lie-algebra of $\LieGrp{SU}(2)$.}
%%%%%%%%%%%%%%%%%%%%%%%%
\begin{equation*}
\Bigl[\frac12\sigma_i,\frac12\sigma_j\Bigr]=i\sum_k\varepsilon_{ijk} \frac12\sigma_k,
\end{equation*}
so if $\norm{\ve{u}}^2=1$ then 
$A_\ve{u}$ represents the observable corresponding to a spin measurement in $\hbar/2$ units, along the direction $\ve{u}$.
%........................miert mondjuk, hogy ez imp.mom iranyokat reprezental?.........................
As before, we have that $A_\ve{u}$ has the eigenvalues $\pm\norm{\ve{u}}$.
If we multiply the eigenvalue equation
\begin{equation*}
A_\ve{u}\cket{\alpha_\pm(\ve{u})}=\pm\norm{\ve{u}}\cket{\alpha_\pm(\ve{u})}
\end{equation*}
with $\bra{\alpha_\pm(\ve{u})}\sigma_i$ from the left, 
using (\ref{eq:Pauli.alg}) we have for the real part that
\begin{equation*}
\bra{\alpha_\pm(\ve{u})}\sigma_i\cket{\alpha_\pm(\ve{u})} = \pm\frac{u_i}{\norm{\ve{u}}}.
\end{equation*}
This gives meaning to the $\cket{\alpha_\pm(\ve{u})}$ eigenvectors,
these represent the pure sates corresponding to the $\pm\ve{u}$ spin direction.
%which means that $\sigma_i$ operators represent unit vectors of orthogonal directions.
%On the other hand, by linearity, $\sr{u}$ represents a vector.

Now a state $\varrho$ given in (\ref{eq:qubit}) with the Bloch vector $\ve{x}$ 
and an observable $A_\ve{x}$ have the same eigensubspaces,
and if $\varrho$ is pure % ($\norm{\ve{x}}=1$),
then it can easily be checked that 
\begin{equation}
\varrho=\cket{\alpha_\pm(\ve{x})}\bra{\alpha_\pm(\ve{x})},\qquad\text{if and only if $\norm{\ve{x}}=1$}.
\end{equation}
Therefore we can assign physical meaning to the points of the Bloch sphere
through the expectation values of a measurement:
if $\ve{x}\in S^2$  then
$\varrho$ is the sate corresponding to the  $\ve{x}$ spin direction.
Note that this does not hold for points \emph{inside} the Bloch ball,
they represent statistical mixtures of pure points instead.

%%%%%%%%%%%%%%%%%%%%%%%%
\begin{figure}
 \includegraphics{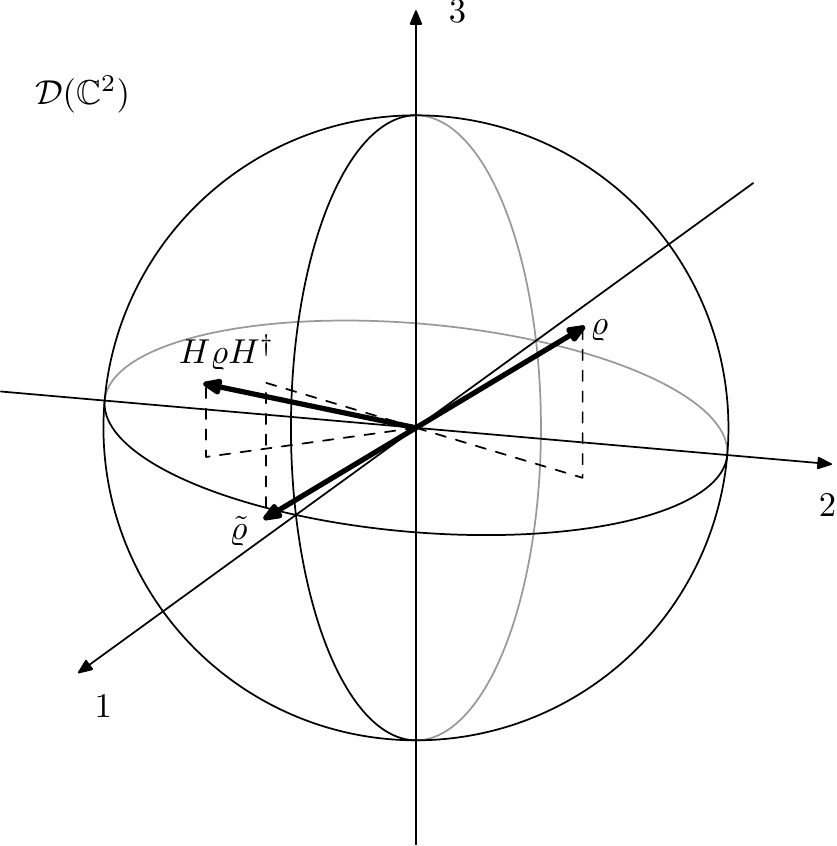}
 \caption{Bloch ball, the state space of a qubit. 
The effects of spin flip (\ref{eq:sflipRho}) 
and discrete Fourier transformation (\ref{eq:Fourier})
are also depicted.}
\label{fig:BlochBall}
\end{figure}
%%%%%%%%%%%%%%%%%%%%%%%%

The mixedness of the state $\varrho$ can be written by, e.g., the concurrence-squared~(\ref{eq:conc2})
\begin{equation}
\label{eq:conc2qubit}
C^2(\varrho)=1-\norm{\ve{x}}^2 = 4\det\varrho.
\end{equation}
The eigenvalues and all the entropies can be expressed with this quantity.
When we deal with qubits,
it is useful to use logarithm to the base $2$ in the definition of the quantum-R\'enyi entropies,
then they range from $0$ to $1$,
and the von Neumann entropy is said to be measured in \emph{qubits}.
It can be expressed with the concurrence as
\begin{subequations}
\begin{equation}
\label{eq:qubitentr}
S(\varrho)=h\left(\frac{1}{2}\Bigl(1+\sqrt{1-C^2(\varrho)}\Bigr)\right)
\end{equation}
with the binary entropy function
\begin{equation}
\label{eq:binentr}
h(x)=-x\log_2x-(1-x)\log_2(1-x).
\end{equation}
\end{subequations}
It is clear from (\ref{eq:conc2qubit}) that the concurrence-squared $C^2(\varrho)$
(and the von Neumann entropy $S(\varrho)$ as well)
is an $\LieGrp{U}(1)\times\LieGrp{SL}(2,\field{C})$-invariant quantity for qubits.%
%%%%%%%%%%%%%%%%%%%%%%%%
\footnote{Note that the normalization of a state vector and that of a density matrix
are not invariant under the action of $\LieGrp{U}(1)\times\LieGrp{SL}(2,\field{C})$.
For unnormalized distributions, 
the $S^\text{Ts}_q(\varrho)=1/(1-q)\bigl(\tr\varrho^q-(\tr\varrho)^q\bigr)$
definition of Tsallis entropies has to be used instead of (\ref{eq:qTsallis}),
leading to $C^2(\varrho)=2\bigl((\tr\varrho)^q - \tr\varrho^q\bigr)$ for qubits.}
%%%%%%%%%%%%%%%%%%%%%%%%
since entropies are invariant only under unitaries in general.

Note that all of the above constructions were carried out
\emph{without} an explicit representation for the Pauli operators.
This abstract approach is very useful in the derivation of the nonlinear Bell inequalities,
which are recalled in section \ref{subsec:QM.Ent.2Part}, and of which multipartite extension is used in section \ref{subsec:SepCrit.3Part.Spin}.
Now, after choosing an orthonormal basis $\{\cket{0},\cket{1}\}$, the Hilbert space $\mathcal{H}\isom\field{C}^2$,
and for further use we introduce the usual matrix representation of the Pauli operators
\begin{equation}
\label{eq:Pauli.mx}
\sigma_0=\begin{bmatrix}1&0\\0&1\end{bmatrix}=\Id,\qquad
\sigma_1=\begin{bmatrix}0&1\\1&0\end{bmatrix},\qquad
\sigma_2=\begin{bmatrix}0&-i\\i&0\end{bmatrix},\qquad
\sigma_3=\begin{bmatrix}1&0\\0&-1\end{bmatrix},
\end{equation}
which are called Pauli matrices.
But there is a matrix of particular importance which \emph{has to be given explicitly},
\begin{equation}
\label{eq:epsilon}
\varepsilon=\begin{bmatrix}0&1\\-1&0\end{bmatrix},
\end{equation}
its $\varepsilon_{ij}$ entry is the parity of the permutation $ij$ of $01$ if $i$ and $j$ are distinct, othervise zero.
With this, the linear transformation $\varepsilon=\varepsilon_{ij} \bra{i}\otimes\bra{j}
\in\mathcal{H}^*\otimes\mathcal{H}^*\isom\Lin(\mathcal{H}\to\mathcal{H}^*)$
gives another identification of $\mathcal{H}$ with $\mathcal{H}^*$,
which is a basis-dependent one.%
%%%%%%%%%%%%%%%%%%%%%%%%
\footnote{Note that we use here a convention different from the one which is used
in the representation theory of the Lorentz group on Dirac and Weyl spinors,
where there are two Hilbert spaces, carrying the left-handed and right handed representations,
having undotted and dotted indices,
and $\varepsilon$ is used for lowering and uppering indices in both Hilbert spaces.
%where $\varepsilon\in\Lin(\mathcal{H}_\text{L}\to\mathcal{H}_\text{L}^*)$ or $\Lin(\mathcal{H}_\text{R}\to\mathcal{H}_\text{R}^*)$
%and indices are lowered by that in $\mathcal{H}_\text{L}$ or $\mathcal{H}_\text{R}$,
%while the indices on $\mathcal{H}_\text{R}$ are dotted ones.  
Instead of these,
we have upper and lower indices on $\mathcal{H}$ and $\mathcal{H}^*$, respectively,
and $\varepsilon\in\Lin(\mathcal{H}\to\mathcal{H}^*)$ and $\varepsilon^*\in\Lin(\mathcal{H}^*\to\mathcal{H})$
are always written out explicitly,
and $\varepsilon^{ii'}\equiv(\varepsilon^*)^{ii'}=\cc{(\varepsilon_{ii'})}=\varepsilon_{ii'}$.
This convention is more convenient when the ``default'' group action is that of $\LieGrp{U}(2)$
instead of $\LieGrp{SL}(2,\field{C})$, which latter represents the Lorentz group on two dimensional Hilbert spaces.
Note again, however, that in this case changing index positions can only be done for all indices collectively.}
%%%%%%%%%%%%%%%%%%%%%%%%
Let $\bra{\tilde{\psi}}=\varepsilon\cket{\psi}$, then
\begin{subequations}
\begin{equation}
\label{eq:sflipPsi}
%\bra{\tilde{\psi}}=\varepsilon\cket{\psi},\qquad\text{and}\qquad
\cket{\psi}=\psi^i\cket{i}
\qquad\longmapsto\qquad
\cket{\tilde{\psi}}=\bra{\tilde{\psi}}^*=\cc{(\varepsilon_{ii'})}\cc{{\psi^{i'}}}\cket{i} = \varepsilon^{ii'}\psi_{i'}\cket{i},
\end{equation}
where the $\tilde{\;}=*\circ\,\varepsilon$ notation is used.
The corresponding operation on $\mathcal{D}(\mathcal{H})$
\begin{equation}
\label{eq:sflipRho}
%\varrho\longmapsto\tilde{\varrho}=\sigma_2\cc{\varrho}\sigma_2^\dagger
\begin{split}
\varrho=\varrho^{i}_{\phantom{i}j}\cket{i}\otimes \bra{j}
\qquad\longmapsto\qquad
\tilde{\varrho}=(\varepsilon\varrho\varepsilon^\dagger)^*
&=(\varepsilon_{ii'}\varrho^{i'}_{\phantom{i'}j'}\varepsilon^{jj'}\bra{i}\otimes \cket{j})^*\\
=\cc{(\varepsilon_{ii'})}\cc{(\varrho^{i'}_{\phantom{i'}j'})}\cc{(\varepsilon^{jj'})}\cket{i}\otimes \bra{j}
&=\varepsilon^{ii'}\varrho_{i'}^{\phantom{i'}j'}\varepsilon_{jj'}\cket{i}\otimes \bra{j}
\end{split}
\end{equation}
\end{subequations}
where the $\tilde{\;}=*\circ\Ad_\varepsilon$ notation is used,
results in the $\ve{x}\mapsto -\ve{x}$ space inversion in $\field{R}^3$.
This operation is called \emph{spin-flip} (figure \ref{fig:BlochBall}).
Note that this is an antilinear operation on $\mathcal{H}$ and $\Lin(\mathcal{H})$.
Antilinear operations are in connection with the time reversal in quantum mechanics.
Indeed, a spin changes sign for time reversal, but not for space inversion, being an axial-vector.
(Space inversion is not even contained in $\LieGrp{SO}(3)\subset\LieGrp{SO}(3,1)$, 
the group of space rotations, represented on $\Lin(\mathcal{H})$,
whose double cover $\LieGrp{SU}(2)\subset\LieGrp{SL}(2,\field{C})$ is represented on $\mathcal{H}$.)

The characteristic property of $\varepsilon$ is the very special transformation behaviour
\begin{equation}
\label{eq:epstraf}
A^\transp\varepsilon A =(\det A) \varepsilon
\end{equation}
for any $A\in\Lin(\mathcal{H})$,
leading to the invariance under $\LieGrp{SL}(2,\field{C})$.
This makes $\varepsilon$ suitable for obtaining Lorentz-invariant combinations from Weyl spinors.
Although, in nonrelativistic entanglement theory Lorentz transformations are not involved,
but $\LieGrp{SL}(2,\field{C})$ comes into the picture in a different way,
making the structure $\varepsilon$ still important.
A sign of this is that the determinant can also be written with the spin-flip given by $\varepsilon$,
leading to
\begin{equation}
\label{eq:qbconcsq}
C^2(\varrho)=4\det\varrho=2\tr\varrho\tilde{\varrho}.
\end{equation}
This characterizes not only the mixedness of a qubit,
but, as we will see, its entanglement with its environment.
And this is not the only case in which $\varepsilon$ comes into the picture,
it appears again and again along the issues of the entanglement of qubits.
We will meet it in sections \ref{subsec:QM.EntMeas.2QBPure}, \ref{subsec:QM.EntMeas.2QBMix},
\ref{subsec:QM.EntMeas.3QBPure}, \ref{subsec:QM.EntMeas.4QBPure}
and in section \ref{subsec:ThreeQB.Pure.SLOCC3QBLSL} in connection with the
Freudenthal triple system approach of three-qubit entanglement \cite{BorstenetalFreudenthal3QBEnt}.

Another operation, which is important in quantum information theory,
is the discrete Fourier transformation
\begin{subequations}
\begin{equation}
\label{eq:Fourier}
\varrho\qquad\longmapsto\qquad H\varrho H^\dagger
\end{equation}
given by the unitary involution having the matrix
\begin{equation}
\label{eq:Hadamard}
H=\frac1{\sqrt2}\begin{bmatrix}
 1 & 1 \\
 1 &-1
\end{bmatrix},
\end{equation}
\end{subequations}
which is a Hadamard matrix.
This results in a $(x^1,x^2,x^3)^\transp\mapsto(x^3,-x^2,x^1)^\transp$ rotation (figure \ref{fig:BlochBall}).

%\ve{x}&=\begin{bmatrix}
%r\sin(\vartheta)\cos(\varphi)\\
%r\sin(\vartheta)\sin(\varphi)\\
%r\cos(\vartheta)
%\end{bmatrix},
%\qquad 0\leq r \leq1,\quad 0\leq\vartheta\leq\pi,\quad 0\leq\varphi\leq2\pi,\\
%\qquad \longleftrightarrow \qquad
%\varrho&=\frac{1}{2}\begin{bmatrix}
%1+r\cos(\vartheta) & r\ee^{-i\varphi}\sin(\vartheta)\\
%r\ee^{i\varphi}\sin(\vartheta) & 1-r\cos(\vartheta)
%\end{bmatrix}=\cket{\Psi}\bra{\Psi},
%\end{align}

%\Psi
%%\cket{\Psi}
%=\begin{bmatrix}
%\ee^{i\varphi/2}\cos(\vartheta/2)\\
%\ee^{i\varphi/2}\sin(\vartheta/2)
%\end{bmatrix} ,\qquad \left[\cket{\Psi}\right]\in\field{C}\mathrm{p}^1.

%*******************************************************************************
\section{Composite systems and entanglement}
\label{sec:QM.Ent}

In the classical scenario, the phase space of a composite system
arises as the \emph{direct product} of the phase spaces of the subsystems.
%However, in quantum mechanics we require that the pure states are directions in a Hilbert space,
%which is closed under superposition.
%direkt összeg is tud szuperponalodni!!!
In quantum mechanics, however, the Hilbert space of a bipartite composite quantum system 
arises as the \emph{tensor product} of the Hilbert spaces of the subsystems.
As we will see, this structure along with the superposition is responsible for entanglement.

%So direct product is not enough, moreover, we have to step back to the Hilbert space
%and have
%As we will see, this structure is responsible for the entanglement.

For two subsystems, the tensor product of the Hilbert spaces of the subsystems is 
$\mathcal{H}\equiv\mathcal{H}_{12}=\mathcal{H}_1\otimes\mathcal{H}_2$
and $d_1=\dim\mathcal{H}_1$, $d_2=\dim\mathcal{H}_2$ and $d\equiv d_{12}=\dim\mathcal{H}=d_1d_2$,
and $\tpl{d}=(d_1,d_2)$ denotes the $2$-tuple of the local dimensions.%
%%%%%%%%%%%%%%%%%%%%%%%%
\footnote{If we have the computational bases
$\{\cket{i}\mid i=1,\dots,d_1\}\subset\mathcal{H}_1$ and
$\{\cket{j}\mid j=1,\dots,d_2\}\subset\mathcal{H}_2$ of the subsystems,
then the computational basis of the composite system is
$\{\cket{i}\otimes\cket{j}\mid i=1,\dots,d_1,j=1,\dots,d_2\}\subset\mathcal{H}$,
the element of which is often abbreviated as $\cket{i}\otimes\cket{j}\equiv\cket{ij}$.}
%%%%%%%%%%%%%%%%%%%%%%%%
The set of states $\mathcal{D}\equiv\mathcal{D}_{12}\equiv\mathcal{D}(\mathcal{H}_{12})$
is defined in the same way as for singlepartite systems,
while the sets of states of the subsystems are 
$\mathcal{D}_1\equiv\mathcal{D}(\mathcal{H}_1)$ and 
$\mathcal{D}_2\equiv\mathcal{D}(\mathcal{H}_2)$.
The sets of pure states of the composite system and those of its subsystems are denoted with
$\mathcal{P}\equiv\mathcal{P}_{12}\equiv\mathcal{P}(\mathcal{H}_{12})$,
and 
$\mathcal{P}_1\equiv\mathcal{P}(\mathcal{H}_1)$ and
$\mathcal{P}_2\equiv\mathcal{P}(\mathcal{H}_2)$.
The \emph{reduced states} of $\varrho\in\mathcal{D}$ are given by the partial trace operation,
$\varrho_1=\tr_2\varrho\in\mathcal{D}_1$ and $\varrho_2=\tr_1\varrho\in\mathcal{D}_2$,
which is the quantum analogy of obtaining marginal distributions.
On the other hand, a \emph{purification} of a state $\varrho_1\in\mathcal{D}_1$
is a $\pi\in\mathcal{P}$ pure state of the composite system
from which $\varrho_1$ arises as the reduced state, that is, $\varrho_1=\pi_1\equiv\tr_2\pi$.
Such purification exists for all $\varrho_1$
if the other Hilbert space $\mathcal{H}_2$ is big enough, that is, $d_2\geq d_1$.

Composite systems in quantum mechanics appear basically in two main respects.
Namely,
when a composite of subsystems playing equal roles is investigated (entanglement theory),
and 
when the composite system is regarded as the compound of a system with its environment (theory of open quantum systems).
These two cases have the same mathematical description,
the difference is physical: we can not execute quantum operations on the environment.
Of course, these two fields are strongly interrelated,
the distinction is made with respect to their main concerns only.
%In the latter case, we denote the Hilbert spaces of the system and the environment 
%with $\mathcal{H_\text{S}}$ and $\mathcal{H}_\text{Env}$, respectively.
In the following, we review the general non-unitary operations on open quantum systems,
some basics about the entanglement of bipartite and multipartite systems,
and the important point where these two meet each other, which is the so called distant lab paradigm.

The main reference on entanglement is \cite{Horodecki4}.

%\subsection{Geometry}
%\label{subsec:QM.Ent.Geom}
%
%For a $\varrho_1\in\mathcal{D}_1$ one can always find a pure state in a bigger Hilbert space ..................
%
%kellene: Hilbert-Schmidt inner product, ortogonalitas, geometria...
%
%Bures metrika,

%2 qubitra (Peter) vagy az nem ide?

%Its pure states are given by the vector
%\begin{equation*}
%\cket{\psi}=\sum_{i,j=1}^{d_1,d_2}\psi^{ij}\cket{ij},
%\end{equation*}
%and the reduced density matrices can be written as
%\begin{align*}
%\varrho_1&=\sum_{i,i'=1}^{d_1} (\psi\psi^\dagger)^i_{\phantom{i}i'}\cket{i}\bra{i'}\\
%\varrho_2&=\sum_{j,j'=1}^{d_2} (\cc{(\psi^\dagger\psi)})^j_{\phantom{j}j'}\cket{j}\bra{j'}
%\end{align*}

\subsection{Operations on quantum systems}
\label{subsec:QM.Ent.Ops}
Here we outline the treatment of open quantum systems, following section 2.~of \cite{SagawaNotes}.
The most general operations on quantum states
can be formulated by the use of \emph{completely positive maps}.
These are the positive maps $\Phi\in \Lin\bigl(\Lin(\mathcal{H})\to\Lin(\mathcal{H}')\bigr)$ 
for which the map extended with identity $\Phi\otimes\IId\in 
\Lin\bigl(\Lin(\mathcal{H}\otimes\mathcal{H}_\text{Env})\to\Lin(\mathcal{H}'\otimes\mathcal{H}_\text{Env})\bigr)$ is also positive
for an arbitrary dimensional Hilbert space $\mathcal{H}_\text{Env}$ corresponding to the environment.
(Note that in this general treatment the change of the Hilbert space is also allowed.
For example, an ancillary system can be coupled to the original system, 
in which case $\mathcal{H}'=\mathcal{H}\otimes\mathcal{H}_\text{Anc}$, 
or it can also be dropped, in which case $\mathcal{H}=\mathcal{H}'\otimes\mathcal{H}_\text{Anc}$.)
These maps preserve the positivity of the state not only of the system
but also of the compound of the system and its environment---or the reservoir, or the rest of the world---hence
the physically relevant transformations must be of this kind.
The representation theorem of Kraus states that $\Phi$ is completely positive
if and only if it can be written by the Kraus operators $M_j\in\Lin(\mathcal{H}\to\mathcal{H}')$ as
\begin{equation*}
\Phi(\varrho)=\sum_jM_j\varrho M^\dagger_j, 
\end{equation*}
which is called the Kraus form of $\Phi$.
On the other hand, a completely positive $\Phi$ should preserve the trace
to be a proper transformation on quantum states.
To handle selective measurements, we have to allow a map to decrease the trace too.
A completely positive map is trace-preserving if and only if
\begin{equation*}
\sum_jM^\dagger_jM_j=\Id,
\end{equation*}
and trace non-increasing if and only if
\begin{equation*}
\sum_jM^\dagger_jM_j\leq\Id.
\end{equation*}

In the light of these,
the most general operations on quantum states
can be given by the so called \emph{stochastic} quantum operation
\begin{equation*}
\varrho\qquad\longmapsto\qquad\varrho'=\frac{\Phi(\varrho)}{\tr\Phi(\varrho)}
\end{equation*}
with the trace non-increasing completely positive map $\Phi$.
The operation takes place with probability $q=\tr\Phi(\varrho)$,
which is equal to one if and only if $\Phi$ is trace-preserving.
A trace-preserving completely positive map is called \emph{deterministic} quantum operation,
or \emph{quantum channel}.

We have the following physical prototypes of deterministic quantum operations:
\begin{subequations}
\label{eq:protoCP}
\begin{align}
\label{eq:protoCP.anc}
\Phi(\varrho)&=\varrho\otimes\pi_\text{Anc},\\ %  \cket{\psi_\text{Anc}}\bra{\psi_\text{Anc}},\\
\label{eq:protoCP.U}
\Phi(\varrho)&=U\varrho U^\dagger,\\
\label{eq:protoCP.tr}
\Phi(\varrho)&=\tr_{\mathcal{H}_\text{Anc}}\varrho,
\end{align}
that is, adding an uncorrelated ancilla,
unitary time-evolution, and
throwing away a subsystem, respectively.
The prototype of stochastic quantum operations is a selective von Neumann measurement
with a projector $P_i$
\begin{equation}
\label{eq:protoCP.proj}
\Phi(\varrho)=P_i \varrho P_i^\dagger.
\end{equation}
\end{subequations}
This means that we throw away all but the $i$th output state,
which is a special case of the \emph{postselection} operation.
If we have only one copy of the state 
then the operation takes place with probability $q=\tr\Phi(\varrho)$,
othervise the protocol fails.
On the other hand, 
if the state is present in multiple copies, then only a fraction of the copies, proportional to $q$,
is left after the operation.
We have these two physical interpretations of the decreasing of the trace.

Since the physically relevant transformations are the completely positive ones,
the outputs of a measurement are related in general to the input by such transformations.
So a generalized measurement is given by a set of completely positive maps $\{\Phi_i\}$,
which are given by the Kraus operators $\{M_{ij}\}$ as
\begin{subequations}
\begin{equation}
\Phi_i(\varrho)=\sum_jM_{ij}\varrho M_{ij}^\dagger,\qquad\text{with probability}\qquad q_i=\tr\Phi_i(\varrho),
\end{equation}
and which all are trace non-increasing, $\sum_{j}M_{ij}^\dagger M_{ij} \leq \Id$.
However, there is a constraint that the whole non-selective measurement $\Phi = \sum_i\Phi_i$, acting as
\begin{equation}
\Phi(\varrho)=\sum_i q_i \frac{\Phi_i(\varrho)}{\tr\Phi_i(\varrho)} = \sum_i \sum_jM_{ij}\varrho M_{ij}^\dagger,
\end{equation}
\end{subequations}
has to be trace preserving,
$\sum_{ij}M_{ij}^\dagger M_{ij} = \Id$.

It can be shown that 
every measurement described by such a $\{\Phi_i\}$
can be written on an extended system as
\begin{equation}
\Phi_i(\varrho)=\tr_{\mathcal{H}_\text{Anc}}\bigl((\Id\otimes P_i)U(\varrho\otimes\pi_\text{Anc})U^\dagger(\Id\otimes P_i)^\dagger\bigr),
\end{equation}
where $\{P_i\}$ is a complete set of projection operators having orthogonal support, $\sum_i P_i=\Id$, $P_iP_j=\delta_{ij}P_i$.
So we have the physical interpretation for the trace non-increasing completely positive maps
corresponding to the measurement outputs:
A generalized measurement arises as a von Neumann measurement on an ancilla, which interacts with the original system.
In other words, every selective generalized measurement (non-unitary stochastic operation) can be constructed 
by the use of the elementary, physically motivated steps (\ref{eq:protoCP.anc})--(\ref{eq:protoCP.proj}).
It can also be shown that if the $P_i$ projectors are of rank one,
which is the case of measurement with nondegenerate observable,
then all $\Phi_i$s are given by only one Kraus operator each. That is, $M_{ij}=M_i$ for all $j$,
so $\Phi_i(\varrho)=M_{i}\varrho M_{i}^\dagger$,
and we get back the (\ref{eq:Meas}) formulas of generalized measurements.

On the other hand, we get the trace-preserving operation
by summing up $\Phi = \sum_i\Phi_i$,
which results in that every trace preserving completely positive map
can be written in an extended system as
\begin{equation}
\Phi(\varrho)=\tr_{\mathcal{H}_\text{Anc}}\bigl(U(\varrho\otimes\pi_\text{Anc})U^\dagger\bigr).
\end{equation}
So we have the physical interpretation for the non-unitary evolution of a system:
It can be modelled by a unitary evolution of an extended system.
In other words, every non-selective generalized measurement (non-unitary deterministic operation) can be constructed 
by the use of the elementary, physically motivated steps (\ref{eq:protoCP.anc})--(\ref{eq:protoCP.tr}).

\subsection{Entanglement in bipartite systems}
\label{subsec:QM.Ent.2Part}

Now, we consider a composite quantum system of two subsystems.
The central notion here is that of separable states, which is defined to be
 the convex sum of tensor products of states
\begin{subequations}
\label{eq:sep}
\begin{equation}
\label{eq:sepDecomp}
\varrho\in\mathcal{D}_\text{sep}\qquad\defn\qquad
\varrho=\sum_j p_j'\varrho_{1,j}\otimes \varrho_{2,j},
\end{equation}
where $\varrho_{1,j}\in \mathcal{D}_1$ and $\varrho_{2,j}\in \mathcal{D}_2$,
and  $\tpl{p}'\in\Delta_{m'-1}$, as usual.
The motivation of this definition is that
states of this kind can be prepared locally,
with the use of classical correlations only \cite{WernerSep}. 
Due to the positivity restriction of the $p_j'$s, $\mathcal{D}_\text{sep}$ is a proper subset of $\mathcal{D}$,
and the elements of the set $\mathcal{D}\setminus\mathcal{D}_\text{sep}$ are called \emph{entangled states}.
That is, entangled states can not be written as a convex combination of product states,
which is another plausible motivation, 
since classical joint probability distributions can always be written
as a convex combination of product distributions.
The set $\mathcal{D}_\text{sep}$ is a convex one,
and, since the dimensions of the Hilbert spaces of the subsystems are finite, 
we can rewrite its elements (\ref{eq:sep}) as%
%%%%%%%%%%%%%%%%%%%%%%%%
\footnote{Here we use that 
$\Lin(\mathcal{H}_1)\otimes\Lin(\mathcal{H}_2)
\isom\mathcal{H}_1\otimes\mathcal{H}_1^*\otimes\mathcal{H}_2\otimes\mathcal{H}_2^*
\isom\mathcal{H}_1\otimes\mathcal{H}_2\otimes\mathcal{H}_1^*\otimes\mathcal{H}_2^*
\isom\Lin(\mathcal{H}_1\otimes\mathcal{H}_2)$.}
%%%%%%%%%%%%%%%%%%%%%%%%
\begin{equation}
\label{eq:sepPureDecomp}
\varrho\in\mathcal{D}_\text{sep}\qquad\Longleftrightarrow\qquad
\varrho=\sum_j p_j
%\cket{\psi_{1,j}}\bra{\psi_{1,j}}\otimes 
%\cket{\psi_{2,j}}\bra{\psi_{2,j}}.
\bigl(\cket{\psi_{1,j}}\otimes\cket{\psi_{2,j}}\bigr)
\bigl( \bra{\psi_{1,j}}\otimes \bra{\psi_{2,j}}\bigr)
\end{equation}
\end{subequations}
with the different weights $\tpl{p}\in\Delta_{m-1}$.
Hence the set of separable states is the convex hull of separable pure states $\pi=\cket{\psi}\bra{\psi}$,
arising from tensor product vectors of the form $\cket{\psi}=\cket{\psi_1}\otimes\cket{\psi_2}$.
The set of these is denoted with $\mathcal{P}_\text{sep}$.
The reduced states of a separable pure state are pure ones,
$\pi_1=\cket{\psi_1}\bra{\psi_1}$ and
$\pi_2=\cket{\psi_2}\bra{\psi_2}$.
Due to superposition, not all vectors in $\mathcal{H}_1\otimes\mathcal{H}_2$ are of this kind.
In fact, almost all vectors are not of this kind, so they are called entangled ones.
As we will see, the reduced states of an entangled pure state are mixed ones,
contrary to the classical case, where the marginals of a pure joint probability distribution are pure ones.
This was the first embarrassing observation about entanglement,
made by Einstein, Podolsky, Rosen and Schr\"odinger:
even if we know exactly the state of the whole system,---i.e.~it is in a pure state,
which contains all the information that quantum mechanics can provide about the system---the
possible (pure) states of the subsystems are known only with some probabilities \cite{EPRpaper,SchrodingerEnt,Schrodinger2}.
(And what is worse, the ensemble of these pure states is not even unique.)
%.........to which refer the well-known statement, that
%`the whole is greater than the sum of its parts'.
This means that
if we have an ensemble of systems prepared \emph{identically} in a pure entangled state
then we can not choose such measurements on a subsystem which leads to 
pure measurement statistics $\tpl{q}_1=(1,0,0,\dots)$.

%As we will see, subsystems of composite systems being in identical pure states
%behave
%which is the entanglement itself (section \ref{}).
% composite systems in identical states, subsystems in a nontrivial ensemble

As an extremal example, consider one of the (pure) Bell-states of two qubits, given by the state vector
\begin{equation}
\label{eq:B}
\cket{\text{B}}=\frac{1}{\sqrt 2}\bigl(\cket{00}+\cket{11}\bigr).
\end{equation}
Its reduced states $\varrho_1=1/2\Id$, $\varrho_2=1/2\Id$ are maximally mixed.
But this is only one part of the story.
If a selective measurement is carried out on both subsystems 
of a system being in the state $\cket{\text{B}}\bra{\text{B}}$
with the observables $A\otimes\Id$ and $\Id\otimes B$ with $A=\sigma_3$, $B=\sigma_3$, 
($1/2$-spin mesurements along the z axis)
then the outcomes of the two measurements are maximally correlated.
Moreover, 
after the selective measurement (\ref{eq:selMeas}) on the first subsystem only (with $A\otimes\Id$),
the whole state become a separable pure one,
the reduced states of both subsystems are changed to pure ones,%
%%%%%%%%%%%%%%%%%%%%%%%%
\footnote{Note that a non-selective measurement (\ref{eq:nselMeas})
modifies the state of only that subsystem on which it is performed, 
independently of the state.}
%%%%%%%%%%%%%%%%%%%%%%%%
so the state of the second subsystem is determined without measurement.
What was really embarrassing with this is that 
this happens instantaneously even if 
the measurements are spatially separated.
%which is the so called Einstein-nonlocality,
This nonlocal behaviour of entangled states 
was called ``spuckhafte Fernwirkung'' (spooky action at a distance) by Einstein.
Note that this nonlocality can not be used for superluminal signalling,
because the outcome of the first measurement is trully random.
This observation is called ``no-signalling''.

%This is a real problem from the point of view of classical physics,
%(under which we mean the relativistic mechanics together with electrodynamics
%and also the state of the art theory of gravity, the general relativity)
%%since we have never observed any nonlocal behaviour
%sinve we never have any observation or theoretical prediction of nonlocal behaviour.
%But this contradiction is resolved if we consider the quantum state as 
%our knowledge about the system (statistical interpretation)
%instead of a real physical quantity or property (Copenhagen interpretation)
%ez nem igy van: stat vs ortodox: csak statisztikusan jellemzi a rendszereket, vagy az egyedi rendszert is jellemzi

% itt most akkor

Here we have to take a short detour and pose the question:
What does it mean that the correlation contained in the Bell-state is considered entirely nonclassical?
%To elucidate this 
In the case of mixed states this happens even in the classical scenario.
If the joint state of the two subsystems is correlated%
%%%%%%%%%%%%%%%%%%%%%%%%
\footnote{A joint probability distribution
(state of classical compound system) is correlated
if it does not arise as a product of the probability distributions of the subsystems,
$p_{ij}\neq p_ip_j$.} 
%%%%%%%%%%%%%%%%%%%%%%%%
then we can obtain information about one subsystem (its state is then updated)
by performing a measurement on the other one.
In the quantum scenario, if a state is correlated%
%%%%%%%%%%%%%%%%%%%%%%%%
\footnote{A density operator of a compound system is correlated
if it does not arise as a tensor product of the density operators of the subsystems
$\varrho\neq\varrho_1\otimes\varrho_2$.
Note that this does not mean entanglement, 
the set of uncorrelated states is a proper subset of the separable states.}
%%%%%%%%%%%%%%%%%%%%%%%%
then the same happens, which is then not implausible.
What is interesting is that in the quantum scenario this happens even if the state of the system is pure.
In the classical scenario, if a composite system is in a pure state then the subsystems are in pure states as well,
and such state is completely uncorrelated.
In the quantum scenario, however,
the subsystems of a system being in a pure state can be in mixed states,
which means then a new kind of correlation, which is not classical.

Maybe the most famous topic in connection with entanglement is the topic of the \emph{Bell inequalities} \cite{BellOnEPR,PreskillNotes}.
It is related to the problem of the existence of a local hidden variable model %(LHVM)
for the description of a quantum system, 
which could be a possibility to avoid nonlocality.
Namely, it can be possible that 
quantum mechanics does not provide a complete description of the physical world,
and there are variables, hidden for quantum mechanics, which determine the outcomes of the measurements uniquely.
From a complete theory containing these hidden variables,
the probabilities inherent in quantum mechanics arise in the sense 
that the preparation of a quantum state does not fix the value of the hidden variable uniquely,
and quantum mechanics arises as an effective theory.
Now, the locality principle is not violated if 
the value of these hidden variables are fixed during the interaction of the parties,
and they are not affected by each other after the subsystems are moved far away from each other, and considered to become isolated.
The key discovery here, found by Bell \cite{BellOnEPR},
is that if the measurement statistics are determined by a Local Hidden Variable Model (LHVM)
then constraints on the statistics of correlation experiments can be obtained.
Instead of Bell's original inequality,
we show here a simplified version,
using only two dichotomic%
%%%%%%%%%%%%%%%%%%%%%%%%
\footnote{Dichotomic measurements are measurements having only two outcomes.}
%%%%%%%%%%%%%%%%%%%%%%%%
measurements on each site,
proposed by Clauser, Horne, Shimony and Holt (CHSH) \cite{CHSHExperimentLHVM}.
Let these be denoted with $A$ and $A'$ in the first subsystem,
and $B$ and $B'$ in the second one,
all of them have the outcomes $\pm1$.
%
%If the outcomes of the measurements are determined by an LHVM then
%we can write the expectation value of the joint measurement of the product of $A$ and $B$
%with respect to the local hidden variable as $\bracket{AB}$.
%and write the following combination using all pairs of the local observables
%\begin{equation*}
%\bracket{AB+AB'+A'B-A'B'}=\bracket{A(B+B')}+\bracket{A'(B-B')}
%\end{equation*}
%
We denote the expectation value with respect to the hidden variable with $\bracket{\cdot}$,
then such a constraint is given by the CHSH inequality
\begin{equation}
\text{LHVM}\qquad\Longrightarrow\qquad
\abs{\bracket{AB+AB'+A'B-A'B'}}\leq 2.
\end{equation}
In quantum mechanics, we have the archetype of dichotomic measurements,
which is the measurement of the spin of a spin-$1/2$ particle (section \ref{subsec:QM.QuantSys.Qubit}).
Let the unit vectors describing the directions of the measurements be denoted with
$\ve{a},\ve{a}',\ve{b},\ve{b}'\in S^2\subset\field{R}^3$
then the observable of the correlation experiment is the following
\begin{equation*}
S_\text{CHSH}=
%S_{\hat{\ve{a}},\hat{\ve{a}}',\hat{\ve{b}},\hat{\ve{b}}'}=
 \ve{a}\boldsymbol{\sigma}\otimes\ve{b}\boldsymbol{\sigma}
+\ve{a}\boldsymbol{\sigma}\otimes\ve{b}'\boldsymbol{\sigma}
+\ve{a}'\boldsymbol{\sigma}\otimes\ve{b}\boldsymbol{\sigma}
-\ve{a}'\boldsymbol{\sigma}\otimes\ve{b}'\boldsymbol{\sigma}.
\end{equation*}
Then the CHSH inequality takes the form
\begin{equation}
\text{$\varrho$ admits LHVM}\qquad\Longrightarrow\qquad
% \abs{\bracket{S_{\hat{\ve{a}},\hat{\ve{a}}',\hat{\ve{b}},\hat{\ve{b}}'}}}
%=\abs{\tr\bigl(\varrho S_{\hat{\ve{a}},\hat{\ve{a}}',\hat{\ve{b}},\hat{\ve{b}}'}\bigr)}\leq 2 
 \abs{\bracket{S_\text{CHSH}}}
=\abs{\tr\bigl(\varrho S_\text{CHSH}\bigr)}\leq 2 
\quad\text{for all settings}.
\end{equation}
However, it is known that in quantum mechanics
there are states and measurement settings for which this bound can be violated,
hence the predictions of quantum mechanics can not be obtained by a
local hidden variable model.
The experiments confirm the predictions of quantum mechanics,
although there exists loopholes
because of the insufficient efficiency of the detectors.

The Bell inequalities and their connection existence of local hidden variable models
is a deep and widely studied question \cite{GisinBellQuestions,MasanesAllEntHiddenNonlocal}, 
with heavy physical and philosophical consequences \cite{ESzabo,SchlosshauerKoflerZeilingerSnapshot}.
There is, however, another application of the Bell inequalities,
which is the detection of entanglement.
First of all it can be shown that for pure states,
entanglement is necessary and sufficient for the possibility of finding measurement settings
for which the CHSH inequality is violated,
or alternatively,
the CHSH inequality hold for all measurement settings if and only if the state is separable,
\begin{subequations}
\begin{equation}
\label{eq:critBellPure}
\pi\in\mathcal{P}_\text{sep}\qquad\Longleftrightarrow\qquad
\abs{\tr\bigl(\pi S_\text{CHSH}\bigr)}\leq 2 \quad\text{for all settings}.
\end{equation}
However, Werner showed that this does not hold for mixed states,
here we have only that
\begin{equation}
\label{eq:critBellMixed}
\varrho\in\mathcal{D}_\text{sep}\qquad\Longrightarrow\qquad
\abs{\tr\bigl(\varrho S_\text{CHSH}\bigr)}\leq 2 \quad\text{for all settings}.
\end{equation}
That is, there are mixed states, 
which are although entangled
but still can be modelled by an LHVM \cite{WernerSep}.
%However, not only projective measurements have to be able to be modelled by LHVM
But these are not the last words about
the connection between LHVMs of different measurement scenarios and entanglement of mixed states.
This is a widely studied and very interesting issue, 
however, it is out of the scope of this dissertation.%
%%%%%%%%%%%%%%%%%%%%%%%%
\footnote{For a recent summary for this topic we refer to section IV.C.1 of \cite{Horodecki4}.}
%%%%%%%%%%%%%%%%%%%%%%%%

What is important for us is that 
(\ref{eq:critBellMixed}) can be used for the detection of entanglement,
which is a central problem of this dissertation.
It is a difficult question to decide whether a state is separable or not,
that is, whether the decomposition given in (\ref{eq:sep}) exists or not.
The condition (\ref{eq:critBellMixed}) serves also for this purpose.
Namely, its negation states that 
if we can find a measurement setting for which the CHSH-bound is violated, then the state is entangled.
\begin{equation}
\label{eq:critBellMixedNeg}
\varrho\notin\mathcal{D}_\text{sep}\qquad\Longleftarrow\qquad
\quad\text{there exists setting giving}\quad
\abs{\tr\bigl(\varrho S_\text{CHSH}\bigr)}\nleq 2.
\end{equation}
\end{subequations}
Unfortunately, this is only a sufficient but not necessary criterion of entanglement,
that is, there are entangled states for which there does not exist such measurement setting
for which the entangledness can be detected by this method.
Moreover, another difficulty shows up here too, which accompanies us all along,
which is the issue of optimization over a huge manifold.
In this case, for a given state, we have to find a measurement setting, 
which leads to the violation of the inequality.
In other cases, other kinds of optimizations have to be done,
which makes the detection of entanglement difficult 
even if we have necessary and sufficient criteria of entanglement.

%A hermitian operator $W$ is a \emph{witness operator} for a convex compact set $\mathcal{C}\subset\mathcal{D}_N$
%if $0\leq\tr W\sigma$ for all $\sigma\in\mathcal{C}$
%and there exists $\varrho\notin\mathcal{C}$ for which $\tr W\varrho<0$.
%Hence the negativity of expectation value of the observable $W$
%bears witness that the state does not belong to $\mathcal{C}$.

Another way of detecting entanglement is the use of witness operators \cite{HorodeckiPosMapWitness}.
A witness operator is, by definition, an $W\in\Lin(\mathcal{H})$ self-adjoint observable
which has nonnegative expectation value for all separable states,
but there exists at least one entangled state for which the expectation value is negative.
In other words, a witness operator defines a $\mathcal{D}\to\field{R}$ linear functional
$\varrho\mapsto\tr W\varrho$,
the kernel of which, which is a hyperplane in $\Lin(\mathcal{H})$,
cuts into $\mathcal{D}$ but not into $\mathcal{D}_\text{sep}$
(figure \ref{fig:witness}).
%\begin{align*}
%$\forall\varrho\in\mathcal{D}_\text{sep}$ $\tr\bigl(\varrho W\bigr)\geq0$,
%$\exists\varrho\in\mathcal{D}\setminus\mathcal{D}_\text{sep}$ $\tr\bigl(\varrho W\bigr)<0$,
%\end{align*}
A corollary of the Hahn-Banach theorem
is that for every given entangled state there exists a witness operator which detects it \cite{HorodeckiPosMapWitness,BengtssonZyczkowski}.
In this sense, witnessing gives rise to a necessary and sufficient condition of entanglement,
leading to
\begin{subequations}
\label{eq:witnessing}
\begin{equation}
\varrho\in\mathcal{D}_\text{sep}\qquad\Longleftrightarrow\qquad
\bracket{W}\equiv\tr W\varrho\geq0\quad\text{for all withesses $W$}.
\end{equation}
The characterization of the convex set of separable states $\mathcal{D}_\text{sep}$
by witness operators (that is, by supporting hyperplanes)
is a hard problem that have not been solved yet.
So, for the decision of separability of a given state
the problem of optimization is still exists,
since one has to find a witness which detects the entanglement of that given state.
We can get necessary but not sufficient condition for separability
using only an insufficient set of witnesses
\begin{equation}
\varrho\in\mathcal{D}_\text{sep}\qquad\Longrightarrow\qquad
\bracket{W}\equiv\tr W\varrho\geq0\quad\text{for some withesses $W$}.
\end{equation}
\end{subequations}

One can obtain, for example, the witness corresponding to the CHSH correlation-experiment
\begin{equation*}
W_\text{CHSH}=2\Id\otimes\Id \pm S_\text{CHSH}.
\end{equation*}
This is actually a family of witnesses parametrized by the
$\ve{a},\ve{a}',\ve{b},\ve{b}'$
measurement settings.
As we have from (\ref{eq:critBellMixed}),
there are entangled states which can not be detected by CHSH inequality of any settings,
this means that the arrangement of the hyperplanes given by $W_\text{CHSH}$
is not strict enough to clip around $\mathcal{D}_\text{sep}$ perfectly.

%%%%%%%%%%%%%%%%%%%%%%%%
\begin{figure}
 \includegraphics{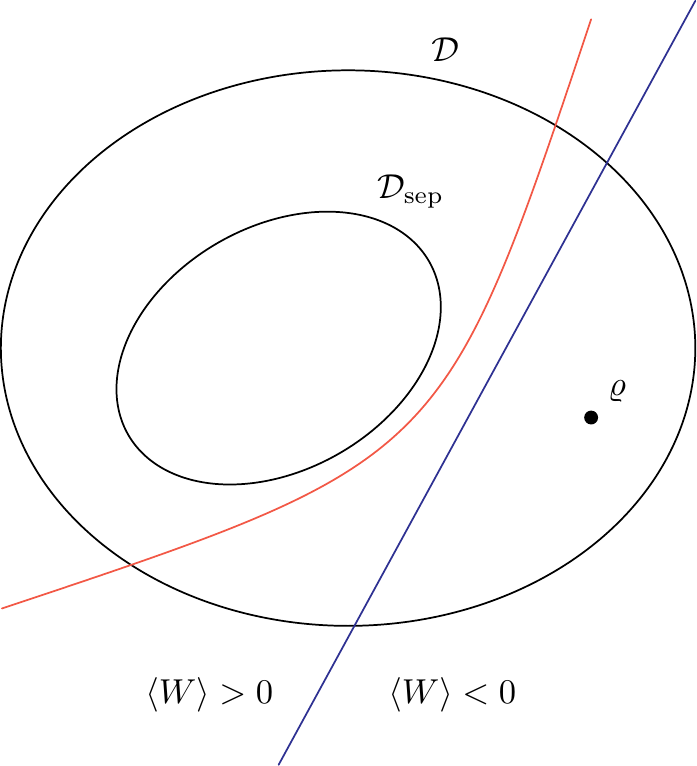}
%\setlength{\unitlength}{0.6\textwidth}
%\begin{picture}(1,1)
%\end{picture}
 \caption{Detection of entanglement of an entangled state $\varrho$. 
The blue straight line represents the kernel of $\bracket{W}$ 
given by a witness $W$ in the space os states $\mathcal{D}$.
Hence the witness identifies all states to be entangled for which $\bracket{W}<0$,
and $\bracket{W}\geq0$ for all separable states (\ref{eq:witnessing}).
The red line represents the border of the domain
in which separable states have to be found,
given by other nonlinear necessary but not sufficient criteria of separability,
for example the quadratic Bell-inequalities.}
\label{fig:witness}
\end{figure}
%%%%%%%%%%%%%%%%%%%%%%%%

The linear witnessing, which is the reformulation of the linear CHSH-inequalities
is not sufficient for the detection of entanglement,
but there is, however, a nonlinear extension of this criterion,
which proved to be a necessary and sufficient one for the two-qubit case.
These are called quadratic Bell inequalities \cite{UffinkSeevinckNonlinBell},
although their connection to the existence of LHVM is not clear.
%$\bracket{A\otimes B - A'\otimes B'}^2+\bracket{A'\otimes B + A\otimes B'}^2=
%(\bracket{A}^2+\bracket{A'}^2)(\bracket{B}^2+\bracket{B'}^2)$
%for separable pure states.
To obtain these inequalities, from now, consider measurements on each site along \emph{orthogonal} directions,
moreover, let us introduce a third observable on each site, $A''$ and $B''$,
being orthogonal to the previous two.
From the theory of spin-measurements then we have that
$\bracket{A}^2+\bracket{A'}^2+\bracket{A''}^2=1$ and
$\bracket{B}^2+\bracket{B'}^2+\bracket{B''}^2=1$
for pure states of subsystems (section \ref{subsec:QM.QuantSys.Qubit}).
Because of these, it is straightforward to check that
with the definition of the following bipartite correlation-observables
\begin{align*}
I &= \frac12 (\Id\otimes \Id + A''\otimes B''),&\qquad
I'&= \frac12 (\Id\otimes \Id - A''\otimes B''),\\
X &= \frac12 (A\otimes B - A'\otimes B'),&\qquad
X'&= \frac12 (A\otimes B + A'\otimes B'),\\
Y &= \frac12 (A'\otimes B + A\otimes B'),&\qquad
Y'&= \frac12 (A'\otimes B - A\otimes B'),\\
Z &= \frac12 (A''\otimes\Id + \Id\otimes B''),&\qquad
Z'&= \frac12 (A''\otimes\Id - \Id\otimes B''),
\end{align*}
the following holds \emph{for separable pure states}
\begin{equation*}
 \bracket{X}^2+\bracket{Y}^2
=\bracket{X'}^2+\bracket{Y'}^2
=\bracket{I}^2-\bracket{Z}^2
=\bracket{I'}^2-\bracket{Z'}^2.
\end{equation*}
The separable states are the convex combinations of separable pure states,
and using convexity arguments,
the following holds for all separable states
\begin{equation}
\label{eq:quadrBell}
\begin{split}
\varrho\in\mathcal{D}_\text{sep}\qquad\Longrightarrow\qquad
\bigl\{\bracket{X}^2+\bracket{Y}^2,&\bracket{X'}^2+\bracket{Y'}^2\bigr\}
\leq
\bigl\{\bracket{I}^2-\bracket{Z}^2, \bracket{I'}^2-\bracket{Z'}^2\bigr\}\\
&\text{for all orthogonal settings}.
\end{split}
\end{equation}
Moreover, these inequalities turned out to be necessary and sufficient ones for separability \cite{UffinkSeevinckNonlinBell},
but this observation can not be directly generalized for subsystems of arbitrary dimensional Hilbert spaces.
In contrast to these, since $S_\text{CHSH} = 2(X+Y)$ with the new observables, 
the original linear inequality condition (\ref{eq:critBellMixed}) takes the form
\begin{equation*}
\varrho\in\mathcal{D}_\text{sep}\qquad\Longrightarrow\qquad 
\abs{\tr\varrho(X+Y)}\equiv\abs{\bracket{X}+\bracket{Y}}\leq1.
\end{equation*}
However, note that in the quadratic case only orthogonal spin observables are used.
Moreover, if the orientations of the $\{A,A',A''\}$ and $\{B,B',B''\}$ sets of local spin measurements are both right-handed (or both left-handed), then
$\{I,X,Y,Z\}$ and also $\{I',X',Y',Z'\}$ obey the right-handed Pauli algebra (\ref{eq:Pauli.alg})
(or a left-handed one, featuring $-\varepsilon_{ijk}$ instead of $\varepsilon_{ijk}$).
In this case it holds for all states that
$\bracket{X}^2+\bracket{Y}^2+\bracket{Z}^2\leq \bracket{I}^2$ and
$\bracket{X'}^2+\bracket{Y'}^2+\bracket{Z'}^2\leq \bracket{I'}^2$,
with equality only for pure states,
which lead to a generalization for $n$ qubits \cite{SeevinckUffinkMixSep},
which will be used in section \ref{subsec:SepCrit.3Part.Spin}.

An important advantage of the separability criteria presented so far
is that they are formulated in the terms of measurable quantities,
so they are ready to be used in a laboratory.
However, the optimization still has to be carried out by the tuning of the measurement settings.
There are other criteria, which are theoretical ones,
under which we mean that the full tomography of the state is needed,
and the criteria are checked on a computer.
A famous criterion of this latter kind
was formulated by Peres \cite{PeresCrit}, involving partial transpose.%
%%%%%%%%%%%%%%%%%%%%%%%%
\footnote{For bipartite density matrices,
the partial transposition with respect to the first subsystem is given by
$\transp_1:\Lin(\mathcal{H}_1\otimes\mathcal{H}_2)\isom
\mathcal{H}_1\otimes\mathcal{H}_2\otimes\mathcal{H}_1^*\otimes\mathcal{H}_2^*\to
\mathcal{H}_1^*\otimes\mathcal{H}_2\otimes\mathcal{H}_1\otimes\mathcal{H}_2^*\isom
\Lin(\mathcal{H}_1^*\otimes\mathcal{H}_2)$,
$\cket{i}\otimes\cket{j}\otimes\bra{k}\otimes\bra{l}\mapsto
\bra{k}\otimes\cket{j}\otimes\cket{i}\otimes\bra{l}$.}
%%%%%%%%%%%%%%%%%%%%%%%%
%and $\;^{\transp_1}$ means transposition on the first subsystem,
%which is although basis-dependent, but the positivity of $\omega^{\transp_1}$ is not.
If a bipartite state is separable
then the partial transposition on subsystem $1$, being linear, 
acts on the $\varrho_{1,i}$s of the decomposition given in equation~(\ref{eq:sepDecomp}).
The transposition does not change the eigenvalues of a self-adjoint matrix,
so $(\varrho_{1,i})^\transp$s are also proper density matrices.
Hence the partial transpose of a separable density matrix is also a density matrix,%
%%%%%%%%%%%%%%%%%%%%%%%%
\footnote{Its eigenvalues are not the same in general as the ones of the original matrix,
but they are also nonnegative ones, and they sum up to one.
On the other hand, it is clear that no matter which subsystem is transposed, 
$\varrho^{\transp_1}\geq 0\;\Leftrightarrow\;\varrho^{\transp_2}=(\varrho^{\transp_1})^\transp\geq 0$.}
%%%%%%%%%%%%%%%%%%%%%%%%
\begin{subequations}
\begin{equation}
\label{eq:critPeres}
\varrho\in\mathcal{D}_\text{sep}\qquad\Longrightarrow\qquad
\varrho^{\transp_1}\geq 0.
\end{equation}
And, what is important, the partial transpose of a general density matrix is not necessarily positive,
so the negation of the implication above can be used for the detection of entanglement:
If $\varrho^{\transp_1}$ is not positive then $\varrho$ is entangled,
%\begin{equation}
%\varrho\notin\mathcal{D}_\text{sep}\qquad\Longleftarrow\qquad
%\varrho^{\transp_1}\ngeq 0,
%\end{equation}
while there still exist entangled states of positive partial transpose (PPTES).
This criterion proved to be a very strong one, as can also be seen in chapter \ref{chap:SepCrit}.
Moreover, it is necessary and sufficient for states of qubit-qubit and qubit-qutrit systems \cite{HorodeckiPosMapWitness}
\begin{equation}
\label{eq:critPeres6}
\varrho\in\mathcal{D}_\text{sep}\qquad\Longleftrightarrow\qquad
\varrho^{\transp_1}\geq 0 \qquad\text{if $d=d_1d_2\leq6$}.
\end{equation}
\end{subequations}
So in this case there do not exist any PPTESs.
Another important advantage of this criterion is that there is no need of optimization to use it.

The partial transposition criterion has a generalization,
in which general positive (but not completely positive) maps act on a subsystem
\cite{HorodeckiPosMapWitness}.
There are several other criteria of separability even for more than two subsystems.
We will review some of them in chapter~\ref{chap:SepCrit}.

\subsection{Multipartite systems}
\label{subsec:QM.Ent.NPart}

A bipartite mixed state can be either separable or entangled,
depending on the existence of a decomposition given by equation~(\ref{eq:sep}).
However, the structure of separability classes can be very complex even for three subsystems.
To get an adequate generalization of equation~(\ref{eq:sep}), 
we recall the definitions of $k$-separability and $\alpha_k$-separability,
as was given in \cite{SeevinckUffinkMixSep}.
Note that, however, a more complete generalization
can be given, which is one of our results in chapter~\ref{chap:PartSep}.

Consider an $n$-partite system with Hilbert space 
$\mathcal{H}\equiv\mathcal{H}_{12\dots n}=\mathcal{H}_1\otimes\mathcal{H}_2\otimes\dots\otimes\mathcal{H}_n$,
and denote the full set of states for this system as 
$\mathcal{D}\equiv\mathcal{D}_{12\dots n}\equiv\mathcal{D}(\mathcal{H}_{12\dots n})$, as before.
Let $L=\{1,2,\dots,n\}$ be the set of the labels of the singlepartite subsystems,
then a $K\subseteq L$ subset defines an arbitrary subsystem.
For the partial separability,
let $\alpha_k=L_1|L_2|\dots|L_k$ denote a $k$-partite split,
%(here we use a simplified notation for partitions usual in literature).
that is a \emph{partition} of the labels of singlepartite subsystems
$L$ into $k\leq n$ disjoint nonempty subsets $L_r$.
A density matrix is \emph{$\alpha_k$-separable}, 
i.e.~separable under the \emph{particular} $k$-partite split $\alpha_k$,
if and only if it can be written as a convex combination of product states
with respect to the split $\alpha_k$.
We denote the set of these states with $\mathcal{D}_{\alpha_k}$, that is,
\begin{subequations}
\label{eq:sepak}
\begin{equation}
\varrho\in\mathcal{D}_{\alpha_k}\qquad\defn\qquad
\varrho=\sum_jp_j'\otimes_{r=1}^k\varrho_{L_r,j},
\end{equation}
where $\varrho_{L_r,j}\in\mathcal{D}_{L_r}\equiv\mathcal{D}(\mathcal{H}_{L_r})$,
and  $\tpl{p}'\in\Delta_{m'-1}$, as usual. 
$\mathcal{D}_{\alpha_k}$ is a convex set, and we can rewrite its elements as
\begin{equation}
\varrho\in\mathcal{D}_{\alpha_k}\qquad\Longleftrightarrow\qquad
\varrho=\sum_jp_j 
%\otimes_{r=1}^k\cket{\psi_{L_r,j}}\bra{\psi_{L_r,j}}.
\bigl(\otimes_{r=1}^k\cket{\psi_{L_r,j}}\bigr)
\bigl(\otimes_{r'=1}^k\bra{\psi_{L_{r'},j}}\bigr).
\end{equation}
\end{subequations}
Hence $\mathcal{D}_{\alpha_k}$ is the convex hull of the partially separable pure states $\pi=\cket{\psi}\bra{\psi}$,
which arise from the tensor product vector $\otimes_{r=1}^k\cket{\psi_{L_r}}$.
The special case when $\alpha_k=\alpha_n\equiv 1|2|\dots|n$, the state  is called \emph{fully separable},
\begin{equation}
\label{eq:fullsep}
\varrho\in\mathcal{D}_{1|2|\dots|n}\equiv\mathcal{D}_\text{sep}\qquad\Longleftrightarrow\qquad
\varrho=\sum_jp_j'\varrho_{1,j}\otimes\varrho_{2,j}\otimes\dots\otimes\varrho_{n,j}.
\end{equation}
Again, states of this kind can be prepared locally,
using classical communication only.

More generally, \emph{for a given} $k$ we can consider states which can be written as a mixture of
$\alpha_k^j$-separable states for generally different $\alpha_k^j$ splits.
These states are called \emph{$k$-separable states} and denoted as $\mathcal{D}_\text{$k$-sep}$, that is,
\begin{subequations}
\label{eq:sepk}
\begin{equation}
\varrho\in\mathcal{D}_\text{$k$-sep}\qquad\defn\qquad
\varrho=\sum_jp_j'\otimes_{r=1}^k\varrho_{L_r^j,j},
\end{equation}
where $\varrho_{L_r^j,j}\in\mathcal{D}_{L_r^j}$ 
and in this case the $\alpha_k^j=L_1^j|L_2^j|\dots|L_k^j$ $k$-partite splits can be different for different $j$s.
Again, $\mathcal{D}_\text{$k$-sep}$ is a convex set, and we can rewrite its elements as
\begin{equation}
\varrho\in\mathcal{D}_\text{$k$-sep}\qquad\Longleftrightarrow\qquad
\varrho=\sum_jp_j 
%\otimes_{r=1}^k\cket{\psi_{L_r^j,j}}\bra{\psi_{L_r^j,j}}.
\bigl(\otimes_{r=1}^k\cket{\psi_{L_r^j,j}}\bigr)
\bigl(\otimes_{r'=1}^k\bra{\psi_{L_{r'}^j,j}}\bigr).
\end{equation}
\end{subequations}
Hence $\mathcal{D}_\text{$k$-sep}$ is the convex hull of all the $k$-partite separable pure states.
The motivation for the introduction of the sets of states of these kinds is that
to mix a $k$-separable state
we need at most only $k$-partite entangled pure states.

Since $\mathcal{D}_\text{$(k+1)$-sep}\subset\mathcal{D}_\text{$k$-sep}$,
the notion of $k$-separability gives rise to a natural hierarchic ordering of the states.
The full set of states is $\mathcal{D}\equiv\mathcal{D}_\text{1-sep}$,
and we call elements of $\mathcal{D}_\text{$k$-sep}\setminus\mathcal{D}_\text{$k+1$-sep}$
(i.e.~the $k$-separable but not $k+1$-separable states)
``\emph{$k$-separable entangled}''.
We call the $n$-separable states (\ref{eq:fullsep}) \emph{fully separable}
%the $2$-separable states \emph{biseparable}
and the $1$-separable entangled states \emph{fully entangled}.

Clearly, $\mathcal{D}_{\alpha_k}$ is a convex set,
and so is $\mathcal{D}_\text{$k$-sep}$, because it is the convex hull of
the union of $\mathcal{D}_{\alpha_k}$-s for a given $k$.
Note that these definitions allow
a $k$-separable state not to be $\alpha_k$-separable for any particular split $\alpha_k$,
and
a state which is $\alpha_k$-separable for all $\alpha_k$ partitions not to be $k+1$-separable.
The existence of such states was counterintuitive,
since for pure states, if, e.g., a tripartite pure state is separable under any $a|bc$ bipartition then it is fully separable.
For mixed states, however, explicit examples can be constructed.
(Using  a method dealing with \emph{unextendible product bases},
Bennett et.~al.~have constructed a three-qubit state
which is separable for all $\alpha_2$ but not fully separable \cite{BennettetalUPB}.
Another three-qubit example can be found in \cite{Acinetal3QBMixClass}.)

Let us now consider the case of three subsystems,
then we have the partitions
$\alpha_1=123$,
$\alpha_2=1|23$,
$\alpha_2'=2|13$,
$\alpha_2''=3|12$,
$\alpha_3=1|2|3$. 
With this, 
adopting the notations of \cite{SeevinckUffinkMixSep},
the classes of separability of mixed tripartite states are as follows (figure~\ref{fig:3part}).

\begin{figure}
 \includegraphics{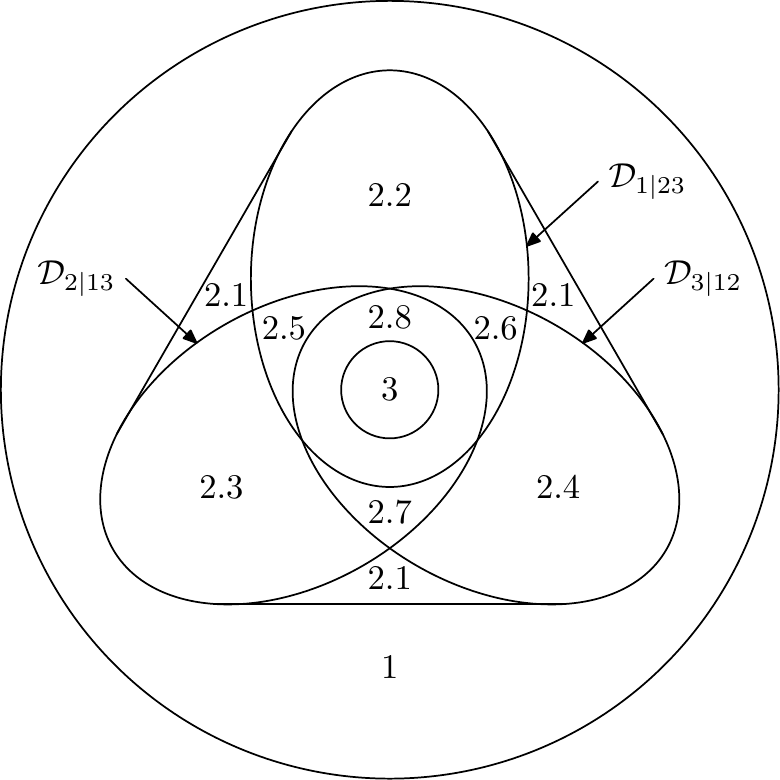}
 \caption{Separability classes for three subsystems.}
\label{fig:3part}
\end{figure}

\textit{Class 3:} This is the set of fully separable three-qubit states,
$\mathcal{D}_\text{$3$-sep}=\mathcal{D}_{1|2|3}$.
\textit{Classes 2.1--2.8:} These are the disjoint subsets of $2$-separable entangled states
$\mathcal{D}_\text{$2$-sep}\setminus\mathcal{D}_\text{$3$-sep}$.
Classes 2.2--2.8
can be obtained by the set-theoretical intersections of
$\mathcal{D}_{1|23}$,
$\mathcal{D}_{2|13}$ and
$\mathcal{D}_{3|12}$, as can be seen in figure~\ref{fig:3part}.
For example
\textit{Class 2.8} is $(\mathcal{D}_{1|23}\cap\mathcal{D}_{2|13}\cap\mathcal{D}_{3|12})\setminus \mathcal{D}_{1|2|3}$
(states that can not be mixed without the use of bipartite entanglement,
but can be mixed by the use of bipartite entanglement within only one bipartite subsystem,
it does not matter which one).
\textit{Class 2.7} is $(\mathcal{D}_{2|13}\cap\mathcal{D}_{3|12})\setminus\mathcal{D}_{1|23}$
(states that can be mixed by the use of bipartite entanglement within only the $12$ or $13$ subsystems,
but can not be mixed by the use of bipartite entanglement within only the $23$ subsystem).
\textit{Class 2.2} is $\mathcal{D}_{1|23}\setminus(\mathcal{D}_{2|13}\cup\mathcal{D}_{3|12})$
(states that can be mixed by the use of bipartite entanglement within only the $23$ subsystem,
but can not be mixed by the use of bipartite entanglement within only the $12$ or $13$ subsystem).
On the other hand,
the union of the sets of $\alpha_2$-separable states is not a convex one,
it is a proper subset of its convex hull $\mathcal{D}_\text{$2$-sep}$.
This defines \textit{Class 2.1} as $\mathcal{D}_\text{$2$-sep}\setminus (\mathcal{D}_{1|23}\cup\mathcal{D}_{2|13}\cup\mathcal{D}_{3|12})$,
that is, the set of states 
that are $2$-separable but can not be mixed by the use of bipartite entanglement
within only one bipartite subsystem.
However, we do not consider these states fully entangled
since they can be mixed without the use of tripartite entanglement.
\textit{Class 1:} This contains all the fully entangled states of the system:
$\mathcal{D}_\text{$1$-sep}\setminus\mathcal{D}_\text{$2$-sep}$.
This classification will be refined in section \ref{subsec:PartSep.Threepart.PSclasses}.

It is again a difficult question to decide to which class a given state belongs.
There are several criteria for the detection of these classes,
arising mostly as the generalizations of the bipartite criteria,
such as the the generalization of nonlinear Bell inequalities for multiqubit mixed states \cite{UffinkQuadBellMultipartEnt,SeevinckUffinkMixSep}.
We will review some of them in chapter~\ref{chap:SepCrit}.
We should mention, although we do not use that, that
another approach has also been worked out by the use of semidefinite programming
\cite{Dohertycrit1,Dohertycrit2,Dohertycrit3}.

%*******************************************************************************
\subsection{Global and local operations}
\label{subsec:QM.Ent.LO}
In closing this section, now we return to the operations performed on quantum states.
First of all, we have to make mention of the roles played by unitary transformations in the description of multipartite systems. 
We have the global unitary group, which is the unitary group $\LieGrp{U}(\mathcal{H})$ of the Hilbert space of the composite system,
and the local unitary group, which is a product of the unitary groups of the Hilbert spaces of the subsystems (with identified centers)
$\LieGrp{U}(\mathcal{H}_1)\times\LieGrp{U}(\mathcal{H}_2)\times\dots\times\LieGrp{U}(\mathcal{H}_n) \subset\LieGrp{U}(\mathcal{H})$.
First of all, unitaries have the role of basis-changes in Hilbert spaces.
In this sense, local unitary transformations are the basis-changes in the Hilbert spaces of the subsystems only,
and global unitary transformations are the basis-changes in the Hilbert space of the whole system.
Note that since not only the states but also the observables are transformed under basis-changes,
the whole description of composite systems is invariant under the action of both groups.
The other role played by unitaries is their generation of time-evolution.
Since in Schr{\"o}dinger picture only the states are transformed under time-evolution, but not the observables,
the measurement statistics are not invariant under such unitary transformations, which is indeed the desirable behaviour in such a situation.
But the important point here is that entanglement is \emph{not} invariant under \emph{global} unitary transformation of only the state.
For example, any density operator can be transformed into a diagonal form by the use of global unitary transformation,
and a diagonal density operator is obviously separable.
Unitaries are invertible, so any entangled state can be obtained from separable ones via suitable global unitary transformations.
And, indeed, this is the way of creating entanglement by the use of quantum interaction,
since such global unitaries are time-evolution operators arising from Hamiltonians which contain an interacting part.
Such operators have nontrivial action on the Hilbert spaces of at the most of two subsystems.
On the other hand, entanglement is clearly invariant under local unitary transformations,
as can be seen from the definitions of different kinds of separabilities (\ref{eq:sep}), (\ref{eq:sepak}) and (\ref{eq:sepk}).

There is another aspect of this issue, namely
the dependence of entanglement upon the choice of the tensor product structure.
If we start only with a Hilbert space $\mathcal{H}$
without specifying the tensor product structure on it,
that is, its ``decomposition'' into the $\mathcal{H}_a$ Hilbert spaces of its subsystems,
then it is meaningless to say that a state is entangled or separable,
since the tensor product structure is inherent in the definition of entanglement.
%entanglement is defined with respect to the tensor product structure,
(For example, such a decomposition can be given by
a $\{\cket{\psi_k}\mid k=1,\dots,d_1d_2\}$
orthonormal basis in $\mathcal{H}$,
if we specify which $\cket{\psi_{k}}$ is considered
which $\cket{\psi_i}\otimes\cket{\psi_j}$.)
In this sense, the tensor product structure (together with entanglement)
is invariant under local unitary transformations,
but not invariant under global unitary ones.
The important result here is that the tensor product structure
is induced by the algebra of operationally accessible interactions and observables \cite{ZanardiVirtualQSubsys,ZanardiLidarLloydTPS}.

The dependence of entanglement upon the tensor product structure is sometimes
erroneously regarded as a weakness of the notion of entanglement,
that is, ``entanglement is not an inherent property, but a relative one, depending on the observer''.
However, maybe not an accident that observer and observable are two different words.
This relativity of entanglement does not work in a same manner
as the relativity in, e.g.,~the theory of special relativity.
In the case of the latter, for a system there can exist more than one observers moving differently at the same time,
they coordinatize the same system in different ways (connected by Lorentz transformations)
but the physics is independent of the coordinatization, which is relative to the observers.
In quantum mechanics, however, the relativity of entanglement is a different kind of relativity.
Here the different laboratory equipments
leading to different local observable algebras
leading to different tensor product structures can not exist in the same time.
And this is not due to some technical difficulty.
The quantum measurement disturbes the system,
so different noncommuting observables (leading to different tensor product structures) can not be measured at the same time,
while in special relativity, the measurements do not disturb the system, 
so different observers can exist in parallel.

After these considerations,
an advanced thinking of entanglement is due to the \emph{distant laboratory} paradigm,
also called \emph{paradigm of LOCC},
which is the abbreviation of \emph{Local Operations and Classical Operations},
attempting to tackle the nonclassicality of quantum states \cite{BennettetalMixedStates}.
This is a natural class of operations suitable for manipulating entanglement,
where it is assumed that the subsystems can be manipulated locally, 
and they do not interact in a quantum mechanical sense,
only in a classical sense, which is modelled by classical communication.
%of any classical data, such as measurement results or random digits.
%resulting in classical correlation only.
This class of operations is motivated also from the point of view of quantum information technology,
because transferring classical information is cheap, 
since it is encoded in states of classical systems, hence can be amplified;
while transferring quantum information is expensive, 
since it is encoded in quantum states, which are very fragile.
However, the real point here is not ``economical'', but rather information theoretical: 
Classical communication, which means the transfer of classical systems,
can not convey quantum information.

A clear example for LOCC is the teleportation of an unknown pure quantum state \cite{Teleport}.
For the teleportation of a qubit state from Alice to Bob, given by $\cket{\varphi}$, % owned by Alice, %\in\mathcal{H}_1$ owned by Alice,
there is a need of a maximally entangled Bell-state $\cket{\text{B}}$ (\ref{eq:B}), %\in\mathcal{H}_2\otimes\mathcal{H}_3$,
shared between them previously. %Alice and Bob.
Furthermore, Alice and Bob are employed in distant laboratories and they are allowed to perform LOCC only,
where local operations are meant with respect to the $12$ and $3$ subsystems.
That is, $\mathcal{H}_A=\mathcal{H}_1\otimes\mathcal{H}_2$ corresponds to Alice's subsystem
and $\mathcal{H}_B=\mathcal{H}_3$ corresponds to Bob's subsystem.
So they have the shared state $\cket{\psi}=\cket{\varphi}\otimes\cket{\text{B}}$,
and they would like to perform the flip operator of the $13$ system,
that is, an operator acting in general as 
$\cket{\psi_1}\otimes\cket{\psi_2}\otimes\cket{\psi_3}\mapsto\cket{\psi_3}\otimes\cket{\psi_2}\otimes\cket{\psi_1}$.
Unfortunately, this is a nonlocal operation of course, which is not allowed.
But, the trick here is that this operator can be decomposed
for a sum of local operations, 
written as 
$1/2\sum_{i=0}^3\sigma_i\otimes\Id\otimes\sigma_i$
with the Pauli matrices (\ref{eq:Pauli.mx}),
resulting in
\begin{equation*}
\cket{\psi}=\cket{\varphi}\otimes\cket{\text{B}}
=\frac12\sum_{i=0}^3\cket{\text{B}_i}\otimes(\sigma_i\cket{\varphi}).
\end{equation*}
Here the maximally entangled Bell-states $\cket{\text{B}_i}=(\sigma_i\otimes\Id)\cket{\text{B}}$
constitute a complete orthonormal basis in $\mathcal{H}_1\otimes\mathcal{H}_2$.
This is just an equivalent writing of the original state, 
featuring the flip operator from the desired flipped state to the original one,
but, from this they can read out what to do.
First, Alice performs a selective von-Neumann measurement in the $12$ subsystem
given by the Bell-states%
%%%%%%%%%%%%%%%%%%%%%%%%
\footnote{This is an entangled measurement, but inside only Alice's subsystem $A=12$.
If Alice can measure only with observables acting on subsystems $1$ and $2$,
then se has to perform a CNOT operation previously,
which is nonlocal \emph{inside} the $12$ subsystem \cite{NielsenChuang}.
Note that performing such an operation is equivalent to generating a
maximally entangled Bell state from a separable state.
%Note that $\cket{\text{B}_{ij}}=(H\otimes\Id)C\cket{ij}$,
%with the Hadamard transformation (\ref{eq:Hadamard}) and the CNOT-gate $C$,
%with $C\cket{i}\otimes\cket{j}=\cket{i}\otimes\cket{i+j}$,
%with the binary indices $i$ and $j$, and the $+$ is understood $\mod2$.
%This means that Alice measures simply with the observable $\sigma_3$ on both of her subsystems.
}
%%%%%%%%%%%%%%%%%%%%%%%%
$P_i=\cket{\text{B}_i}\bra{\text{B}_i}$, resulting in
\begin{equation*}
\cket{\psi}\bra{\psi}\qquad\longmapsto\qquad \varrho_i'=\cket{\text{B}_i}\bra{\text{B}_i}\otimes (\sigma_i\cket{\varphi}\bra{\varphi}\sigma_i^\dagger)
\end{equation*}
with probability $q_i=1/4$.
Then Alice communicate the $i$ outcome of the measurement to Bob, 
who is then perform the corresponding $\sigma_i^{-1}=\sigma_i$ rotation locally on his subsystem,
which gets itself then into the teleported state
\begin{equation*}
\varrho_i'\qquad\longmapsto\qquad \varrho_i''=(\Id\otimes\Id\otimes\sigma_i) \varrho_i' (\Id\otimes\Id\otimes\sigma_i)^\dagger
=\cket{\text{B}_i}\bra{\text{B}_i}\otimes\cket{\varphi}\bra{\varphi}.
\end{equation*}
%and get $\cket{\varphi}\bra{\varphi}$.
Note that for a teleportation of one copy of a state,
roughly speaking, they carry out only the ``quarter part of the flip operation''
(one from the four terms in the sum),
but they had no prior knowledge of which one.
If they teleport multiple copies of a state, 
then the weighted average of the outcomes results in the overall operation
\begin{equation*}
\cket{\psi}\bra{\psi}\qquad\longmapsto\qquad \varrho''=\sum_{i=0}^3q_i \varrho_i''=
\frac12\Id\otimes\frac12\Id\otimes \cket{\varphi}\bra{\varphi},
\end{equation*}
so Alice is left with white noise, the state is fully separable,
the entanglement they have shared before the teleportation is used up.%
%%%%%%%%%%%%%%%%%%%%%%%%
\footnote{Of course Alice can be smart and perform a 
$(\sigma_i\otimes\Id)^{-1}=(\sigma_i\otimes\Id)$ rotation after each measurement on her subsystem 
so as to recover the entanglement in her subsystem,
$\cket{\text{B}}\bra{\text{B}}\otimes\cket{\varphi}\bra{\varphi}$.
But this is just the entanglement which is brought in by the entangled measurement.
%and this is located entirely at Alise's laboratory.
%In this sense, the resource is not the entanglement itself, but the entanglement shared between the distant labs EZ HULYESEG
On the other hand, although they performed the nonlocal flip operation \emph{on average} in this way,
but not in general, since this works only for the case when the $23$ subsystem is in a maximally entangled pure state.}
%%%%%%%%%%%%%%%%%%%%%%%%
This is an archetype of the use of entanglement as a \emph{resource} for manipulating quantum information.
%this is why we have discussed this rather well-known (but still rather interesting) .

An important note here is that teleportation can not be used for superluminal signalling either.
Although the  measurement causes the change of the state in Bob's laboratory instantaneously,
but he needs the classical information about the outcome of the measurement to get the state $\cket{\varphi}\bra{\varphi}$,
which arrives in a classical channel subluminally.
Without this information he has only white noise on average.

During the teleportation protocol only LOCC was used, beyond the sharing of entangled states done previously.
But the point of view of the LOCC paradigm came up even in the very basic grounds of entanglement theory.
Namely, even the (\ref{eq:sep}) definition of separability was also motivated by this approach \cite{WernerSep}.
Here the $\varrho_{1,j}\otimes\varrho_{2,j}$ states, prepared locally in distant laboratories, are uncorrelated for a given $j$,
and the mixing with the weights $p_j$ needs classical communication between the laboratories,
that is, we have to tell the preparing devices the outcome $j$ of a classical random number generator
realizing the probability distribution $\tpl{p}$. 
(To have entangled systems we need a preparing device 
in which the subsystems can interact in a controllable, or at least known way.)
%or collective selective measurement can take place.)
Therefore, neither local operations nor classical communication can not give rise to entanglement.
%which holds also in the general case, where such manipulations take place on general states,
%that is, whatever entanglement is, it can only decrease under LOCC.
On the other hand, LOCC can decrease entanglement, as we have seen in the case of the teleportation.

Although LOCC can not increase entanglement, but there are LOCC protocols
which obtain a number $m$ of pure maximally entangled $(\cket{\text{B}}\bra{\text{B}})^{\otimes m}$ Bell states (\ref{eq:B})
from a larger number $k$ of least entangled input states $\varrho^{\otimes k}$, which can be pure or mixed ones as well.
Such protocols are called \emph{entanglement distillation} protocols \cite{BennettetalDistillation,BennettetalMixedStates,ClarisseThesis}.
These are important methods for quantum information theory
since, for instance, they have the ability to recover some entanglement from states
which was originally maximally entangled, then shared between subsystems
and became mixed with some noise
due to the imperfect quantum channels in which environmental decoherence can not be avoided perfectly.
On the other hand, such protocols have important role also in quantum error correction \cite{BennettetalMixedStates,Horodecki4,NielsenChuang}.
Apart from practical reasons of these kinds,
distillation is important theoretically as well,
because it has turned out that there are entangled states from which no pure state entanglement can be distilled out.
This gives rise to a distinction between two kinds of entanglement,
which are called distillable and undistillable ones \cite{Horodecki3BoundEnt}.
The latter is often called bound entanglement,
and it is regarded then as a weaker form of entanglement.
(Of course, one can not distill entanglement out from separable states, since LOCC can not increase entanglement.)
An interesting result is that all entangled states having positive partial transpose (PPTES) are undistillable \cite{Horodecki3BoundEnt},
and the existence of undistillable states having non-positive partial transpose is still an open question \cite{Horodecki4,ClarisseThesis}.
It is also known that there are no such states for two qubits, moreover, all entangled two-qubit states are distillable \cite{Horodeckietal2qubitDist}.
It is usually hard to check whether a state of positive partial transpose is not separable,
there are few explicit examples of PPTESs in the literature
(see a list of references in section 1.2.4 of~\cite{ClarisseThesis}).
One of our results is to obtain a set of PPTESs (section \ref{subsec:SepCrit.3Part.Matrix}).
Note that in the tripartite case, all the states in Class 2.8 are PPTESs.

After these illustrations, now we turn to the general formalism of local operations \cite{Horodecki4,ChitambaretalWoodyLOCC}.
We list here some important classes of operations related to the structure of subsystems.
The completely positive maps are then denoted with $\Lambda$ instead of the general $\Phi$,
referring to that these are local in some sense.
%the structure of these are in connection with the structure of \emph{local} subsystems.

\emph{Local Operations}: The parties are allowed to execute quantum operations locally only,
and the operation is $\varrho\mapsto \Lambda(\varrho)/\tr\Lambda(\varrho)$ with
$\Lambda=\Lambda_1\otimes\dots \otimes\Lambda_n$.
That is, if the $\Lambda_r:\mathcal{D}_r\to\mathcal{D}_r$ completely positive maps are given by the
$\{M_{r,j_r}\}$ Kraus operators, then
\begin{subequations}
\begin{equation}
\Lambda(\varrho)=\sum_{j_1,j_2\dots,j_n} M_{1,j_1}\otimes M_{2,j_2}\otimes \dots \otimes M_{n,j_n}  
\;\varrho\; M_{1,j_1}^\dagger\otimes M_{2,j_2}^\dagger\otimes \dots \otimes M_{n,j_n}^\dagger.
\end{equation}
The Kraus operators obey
$\sum_{j_r} M_{r,j_r}^\dagger M_{r,j_r}=\Id$ or $\leq\Id$  for all $r$, %$\sum_{j_r} M_{r,j_r}^\dagger M_{r,j_r}\leq\Id$ for all $r$,
in the case of \emph{Deterministic} or \emph{Stochastic Local Operations}, respectively.
The former case is simply called Local Operations.

\emph{Local Operations with One-time Classical Communication}:
In this case,
there is a subsystem singled out,
on which a selective measurement is carried out.
Depending on the outcome of this measurement,
communicated to the other parties,
operations on the other subsystems are carried out.
Let the first subsystem be singled out.
Let moreover $\Lambda_1:\mathcal{D}_1\to\mathcal{D}_1$ be a completely positive map 
acting on that given by the $\{M_{1,j_1}\}$ Kraus operators.
Then there is a set of completely positive maps acting on every other subsystem.
Their elements are labelled by the $j_1$ outcomes of $\Lambda_1$,
that is, there are
$\Lambda_r^{j_1}:\mathcal{D}_r\to\mathcal{D}_r$ completely positive maps for all $j_1$
given by the $\{M_{r,j_r}^{j_1}\}$ Kraus operators.
Then the operation is $\varrho\mapsto \Lambda(\varrho)/\tr\Lambda(\varrho)$, with
\begin{equation}
\Lambda(\varrho)=\sum_{j_1\dots,j_n} M_{1,j_1}\otimes M_{2,j_2}^{j_1}\otimes \dots \otimes M_{n,j_n}^{j_1}
\;\varrho\; M_{1,j_1}^\dagger\otimes M_{2,j_2}^{j_1\dagger}\otimes \dots \otimes M_{n,j_n}^{j_1\dagger}.
\end{equation}
Again, the Kraus operators obey
$\sum_{j_1} M_{1,j_1}^\dagger M_{1,j_1}=\Id$ and
$\sum_{j_r} M_{r,j_r}^{j_1\dagger} M_{r,j_r}^{j_1}=\Id$ for all $r\neq1$ and for all $j_1$,
or
these are $\leq\Id$
%$\sum_{j_1} M_{1,j_1}^\dagger M_{1,j_1}\leq\Id$ and
%$\sum_{j_r} M_{r,j_r}^{j_1\dagger} M_{r,j_r}^{j_1}\leq\Id$ for all $r\neq1$ and for all $j_1$,
for all $r\neq1$ and for all $j_1$,
in the case of \emph{Deterministic} or \emph{Stochastic Local Operations with One-time Classical Communication}, respectively.
The former case is simply called Local Operations with One-time Classical Communication.

\emph{Local Operations with Classical Communication}:
This is the class of operations which arise as the arbitrary compositions of the operations of the above kinds.
Again, this can be either \emph{Deterministic Local Operations with Classical Communication},
depending on whether all the constituting operations are deterministic,
or \emph{Stochastic Local Operations with Classical Communication},
if at least one of the constituting operations are stochastic.
(The former case is simply called Local Operations with Classical Communication, abbreviated with LOCC,
while the latter one is abbreviated by SLOCC.)
The general writing of these operations is complicated, 
it can be found in, e.g., \cite{DonaldetalUniqunessEntMeasLOCC,GrudkaetalEntSwapBoxesLOCC}.

\emph{Separable Operations}:
The operations of this class can not be implemented locally in general,
however, it holds for their overall Kraus operators that they are tensor products.
That is, the operation is $\varrho\mapsto \Lambda(\varrho)/\tr\Lambda(\varrho)$, with
\begin{equation}
\Lambda(\varrho)=\sum_{j} M_{1,j}\otimes M_{2,j}\otimes \dots \otimes M_{n,j}
\;\varrho\; M_{1,j}^\dagger\otimes M_{2,j}^\dagger\otimes \dots \otimes M_{n,j}^\dagger.
\end{equation}
\end{subequations}
Here 
$\sum_j M_{1,j}^\dagger M_{1,j}\otimes M_{2,j}^\dagger M_{2,j}\otimes \dots \otimes M_{n,j}^\dagger M_{n,j} =\Id\otimes\Id\otimes\dots\otimes\Id$
or 
$\leq\Id\otimes\Id\otimes\dots\otimes\Id$
%$\sum_j M_{1,j}^\dagger M_{1,j}\otimes M_{2,j}^\dagger M_{2,j}\otimes \dots \otimes M_{n,j}^\dagger M_{n,j} \leq\Id$
in the case of \emph{Deterministic} or \emph{Stochastic Separable Operations}, respectively.
The former case is simply called Separable Operations.
From these restrictions on the Kraus operators it can be seen
that the set of SO contains set of LOCC,
while it can also be known that SO is a proper subset of LOCC \cite{BennettetalNonlocalityWithoutEnt},
while SSO is the same as SLOCC, up to probability of success \cite{Horodecki4}.
Although only (S)LOCC can be implemented locally in general,
%separable operations can not be implemented locally in general,
but (S)SO is also extensively used, since it has a much simpler form than (S)LOCC, 
%which makes the calculations more simple too. 
%which makes it possible to keep some calculations in a relatively easy way.
so from results concerning (S)SO one can draw conclusions 
concerning (S)LOCC.

Note that full-separability is preserved by the operations of the above kind,
which makes these operations important.
More generally, $\alpha_k$-separability can only be transformed into a finer%
%%%%%%%%%%%%%%%%%%%%%%%%
\footnote{The partition $\alpha_{k'}=L_1'|L_2'|\dots|L_{k'}'$ is finer than $\alpha_k=L_1|L_2|\dots|L_k$ if 
the $L_r$ subsets arise as the $L_r'$ subsets or the unions of those.}
%%%%%%%%%%%%%%%%%%%%%%%%
$\alpha_{k'}$-separability by these operations, %where $\alpha_{k'}$ is finer%
% than $\alpha_k$, 
hence partial separability can be preserved or increased.

Now, having the practically motivated LOCC and SLOCC classes of operations,
motivated classifications of states can also be defined.

\emph{LOCC classificaion:}
Two states are equivalent under LOCC (they are in the same LOCC class) by definition
if they can be transformed into each other \emph{with certainty} by the use of LOCC:
\begin{subequations}
\begin{equation}
\label{eq:LOCCequiv}
\varrho\sim_\text{LOCC}\varrho'\qquad\defn\qquad
%\varrho\approx\varrho'\qquad\defn\qquad
\exists \Lambda, \Lambda':\quad%\text{LOCC}
\varrho'=\Lambda(\varrho),\quad \varrho=\Lambda'(\varrho'),
\end{equation}
where $\Lambda$ and $\Lambda'$ are trace preserving completely positive maps implementing LOCC transformations.
For pure states, it turned out that
two states are equivalent under LOCC
if and only if they are equivalent under LU,
that is, they can be transformed into each other by LU (Local Unitary) transformations \cite{BennettetalEquivalences}:
\begin{equation}
\label{eq:LOCCequivPure}
\cket{\psi}\sim_\text{LOCC}\cket{\psi'}\qquad\Longleftrightarrow\qquad
%\cket{\psi}\approx\cket{\psi'}\qquad\Longleftrightarrow\qquad
\exists U_j\in\LieGrp{U}(\mathcal{H}_j):\quad
\cket{\psi'} = U_1\otimes U_2\otimes\dots\otimes U_n \cket{\psi}.
\end{equation}
\end{subequations}
So this gives the most fine grained classification scheme imaginable for pure states,
many continuous and discrete parameters are required to label the LOCC classes 
\cite{LindenPopescuOnMultipartEnt,AcinetalGenSchmidt3QB,Acinetal3QBPureCanon,Sudbery3qb,Kempe3qb}.
An important corollary is that 
the local spectra of LOCC-equivalent pure states are the same,
\begin{equation}
\pi\sim_\text{LOCC}\pi'\qquad\Longrightarrow\qquad \Spect \pi_K = \Spect \pi_K'\quad\text{for all $K\subseteq L$ subsystems}.
\end{equation}
(Note that the reverse is not true.)
From the point of view of quantum computational purposes,
two LOCC-equivalent pure states can be used for exactly the same task.
However, to our knowledge, 
there is no such practical criterion of LOCC-equivalence and LOCC classification for mixed states
as the LU-equivalence was for pure states.

\emph{SLOCC classificaion:}
A coarse-grained classification can be defined if we demand only the possibility of the transformation.
%which bring in selective dynamics.
Two states are equivalent under SLOCC (they are in the same SLOCC class) by definition
if they can be transformed into each other \emph{with nonzero probability} by the use of LOCC,
or, equivalently, if there are SLOCC transformations relating them:
\begin{subequations}
\begin{equation}
\label{eq:SLOCCequiv}
\varrho\sim_\text{SLOCC}\varrho'\qquad\defn\qquad
%\varrho\sim\varrho'\qquad\defn\qquad
\exists \Lambda, \Lambda': \quad%\text{LOCC}
\varrho'=\frac{\Lambda(\varrho)}{\tr\Lambda(\varrho)},\quad 
\varrho =\frac{\Lambda'(\varrho')}{\tr\Lambda'(\varrho')},
\end{equation}
where $\Lambda$ and $\Lambda'$ are trace non-increasing completely positive maps implementing SLOCC transformations.
For pure states, it turned out that
two states are equivalent under SLOCC
if and only if they are equivalent under LGL,
that is, they can be transformed into each other by LGL (Local General Linear) transformations%
%%%%%%%%%%%%%%%%%%%%%%%%
\footnote{Sometimes that is called ILO, standing for Invertible Local Operation \cite{DurVidalCiracSLOCC3QB},
but we prefer the uniform naming after the corresponding Lie groups.
On the other hand, it is enough to use only LSL (Local Special Linear) transformations,
that is, the $\LieGrp{SL}(\mathcal{H}_j)\subset\LieGrp{GL}(\mathcal{H}_j)$
subgroups because of the normalization.}
%%%%%%%%%%%%%%%%%%%%%%%%
\cite{DurVidalCiracSLOCC3QB}:
\begin{equation}
\label{eq:SLOCCequivPure}
%\cket{\psi}\sim_\text{SLOCC}\cket{\psi'}\qquad\Longleftrightarrow\qquad
\cket{\psi}\sim_\text{SLOCC}\cket{\psi'}\qquad\Longleftrightarrow\qquad
\exists G_j\in\LieGrp{GL}(\mathcal{H}_j):\quad
\cket{\psi'} = \frac{G_1\otimes G_2\otimes\dots\otimes G_n \cket{\psi}}
{\norm{G_1\otimes G_2\otimes\dots\otimes G_n \cket{\psi}}}.
\end{equation}
\end{subequations}
Since $\LieGrp{U}(\mathcal{H}_j)\subset\LieGrp{GL}(\mathcal{H}_j)$,
this gives a coarse grained classification scheme for pure states.
In some cases, including the three-qubit case, SLOCC classes of only finite number arise \cite{DurVidalCiracSLOCC3QB}.
An important corollary is that 
the local ranks of SLOCC-equivalent pure states are the same
\begin{equation}
\pi\sim_\text{SLOCC}\pi'\qquad\Longrightarrow\qquad \rk \pi_K = \rk \pi_K'\quad\text{for all $K\subseteq L$ subsystems}.
\end{equation}
(Note that the reverse is not true.)
This can be used for a coarse-grained classification involving SLOCC-invariant classes of finite number, 
even in the cases in which continuously many SLOCC classes arise \cite{LiLiSLOCC}.
From the point of view of quantum computational purposes,
two SLOCC-equivalent pure states can be used for the same task but with different probability of success.
Again, to our knowledge, 
there is no such practical criterion of SLOCC-equivalence and SLOCC classification for mixed states
as the LGL-equivalence was for pure states.

%*******************************************************************************
\section{Quantifying entanglement}
\label{sec:QM.EntMeas}

In the previous section we have introduced entanglement
together with some basic approaches
for the characterization of its structure,
such as separability classes, LOCC and SLOCC classes.
On the other hand, as was also illustrated by the quantum teleportation protocol,
entanglement is the basic resource of quantum information processing,
so its quantification is also a natural need.
For entanglement quantification one uses special real-valued functions on the states. 
In the light of the LOCC paradigm, we expect that these functions do not increase under LOCC
in order to express some quantity characterizing the amount of entanglement.
In this section we survey some of the important results in connection with this issue,
together with particular results for quantum systems of small numbers of subsystems.

%*******************************************************************************
\subsection{Entanglement measures}
\label{subsec:QM.EntMeas.EntMeas}

The most fundamental property of \emph{entanglement measures} \cite{PlenioVirmaniEntMeas}
is the monotonicity under LOCC \cite{HorodeckiEntMeas,VidalEntMon}.
A $\mu:\mathcal{D}\to\field{R}$ is \emph{(non-increasing) monotone under LOCC} if
\begin{subequations}
\label{eq:meas}
\begin{equation}
\label{eq:meas.mon}
\mu\bigl(\Lambda(\varrho)\bigr) \leq \mu(\varrho)
\end{equation}
for any LOCC transformation $\Lambda$,
which expresses that entanglement can not increase by the use of local operations
and classical communication.
Note that this implies \emph{LU-invariance} automatically,
\begin{equation}
\mu\bigl(U_1\otimes\dots\otimes U_n\varrho U_1^\dagger\otimes\dots\otimes U_n^\dagger\bigr) = \mu(\varrho),
\qquad U_j\in\LieGrp{U}(\mathcal{H}_j).
\end{equation}
A $\mu:\mathcal{D}\to\field{R}$ is \emph{non-increasing on average under LOCC} if
\begin{equation}
\label{eq:meas.average}
\sum_i p_i \mu(\varrho_i) \leq \mu(\varrho),
\end{equation}
where the LOCC is constitued as $\Lambda=\sum_i \Lambda_i$,
where the $\Lambda_i$s are the parts of the LOCC realizing the outcomes of selective measurements,
and $\varrho_i=\frac1{p_i}\Lambda_i(\varrho)$
with $p_i=\tr\Lambda_i(\varrho)$.
This latter condition is stronger than the former one
if the function is \emph{convex}:
\begin{equation}
\label{eq:meas.conv}
\mu\Bigl(\sum_i p_i \varrho_i\Bigr) \leq \sum_i p_i \mu(\varrho_i)
\end{equation}
\end{subequations}
for all ensembles $\{(p_i,\varrho_i)\}$,
which expresses that entanglement can not increase by mixing.
This is also a fundamental, and also plausible property,
since mixing is interpreted as forgetting some classical information
%about which state the system is in.
about the state in which the system is.
A $\mu:\mathcal{D}\to\field{R}$ is an \emph{entanglement-monotone}
if~(\ref{eq:meas.average}) and~(\ref{eq:meas.conv}) hold \cite{VidalEntMon}.
There is common agreement
that LOCC-monotonity~(\ref{eq:meas.mon}) is the only necessary postulate
for a function to be an \emph{entanglement measure} \cite{Horodecki4}.
However, the stronger condition~(\ref{eq:meas.average}) 
is often satisfied too,
and it is often easier to prove.

If $\mu$ is defined only for pure states,
$\mu:\mathcal{P}\to\field{R}$,
then only~(\ref{eq:meas.average}) makes sense, whose restriction is
\begin{equation}
\label{eq:averagePure}
\sum_j p_j \mu(\pi_j) \leq \mu(\pi).
\end{equation}
Here $\{(p_j,\pi_j)\}$ is the pure ensemble
generated by all the Kraus-operators of all $\Lambda_i$s
from the input state $\pi$.
Note that not all $\pi_j$ members of the ensemble are accessible physically,
only the outcomes of the LOCC, which are formed by partial mixtures of this ensemble \cite{HorodeckiEntMeas}.
Mathematically, however, the pure state ensemble can also be used, which makes some constructions much simpler.

There are also other properties beyond the monotonities above,
which are useful or convenient for the measuring of \emph{bipartite} entanglement.
For example, a plausible property is the \emph{discriminance}, 
that is,
$\mu:\mathcal{D}\to\field{R}$ function for bipartite states 
should vanish exactly for separable states,
\begin{subequations}
\begin{equation}
\label{eq:measdiscr.strong}
\varrho\in\mathcal{D}_\text{sep}\qquad\Longleftrightarrow\qquad \mu(\varrho)=0.
\end{equation}
But there are such functions for which only the \emph{weak discriminance} holds,
\begin{equation}
\label{eq:measdiscr.weak}
\varrho\in\mathcal{D}_\text{sep}\qquad\Longrightarrow\qquad \mu(\varrho)=0.
\end{equation}
We can regard the entanglement carried by the Bell state (\ref{eq:B}) as a unit of entanglement,
(that is, $1$ qubit)
%A normalization property can be also useful,
%that is, 
then a $\mu:\mathcal{D}\to\field{R}$ function for bipartite states is \emph{normalized} if
\begin{equation}
\label{eq:measnorm}
\mu(\cket{\text{B}}\bra{\text{B}}) = 1,
\end{equation}
\end{subequations}
which can be achieved by trivial rescaling, so we are not concerned with this one.
Normalization can be useful for comparing different measures.
%The additivity (or extensivity) is a property important for quantum information and coding theory.
%A $\mu:\mathcal{D}\to\field{R}$ function for bipartite states is \emph{additive}
%if additive for uncorrelated states,
%\begin{equation}
%\mu(\varrho_1\otimes\varrho_2)= \mu(\varrho_1) + \mu(\varrho_2)
%\end{equation}

The pure state entanglement measures can be extended to mixed states 
by the so called \emph{convex roof extension}
\cite{BennettetalMixedStates,UhlmannFidelityConcurrence,UhlmannConvRoofs,RothlisbergerLehmannLossNumericalConvRoof,LibCreme}.
It is motivated by the practical approach of the optimal mixing of the mixed state
from pure states, that is, using as little amount of pure state entanglement as possible.
For a continuous function $\mu:\mathcal{P}\to\field{R}$,
its convex roof extension $\cnvroof{\mu}:\mathcal{D}\to\field{R}$ is defined as
\begin{equation}
\label{eq:cnvroofext}
\cnvroof{\mu}(\varrho)=\min_{\sum_i p_i \pi_i=\varrho}  \sum_i p_i \mu(\pi_i),
\end{equation}
where the minimization
takes place over all $\{(p_i,\pi_i)\}$ pure state decompositions of $\varrho$. %:
%$0\leq p_i$, $\sum_i p_i=1$, $\sum_i p_i \pi_i=\varrho$.
It follows from Schr\"odinger's mixture theorem~(\ref{eq:SchMixtureThm}) that
the decompositions of a mixed state into an ensemble of $m$ pure states
are labelled by the elements of the \emph{compact} complex Stiefel manifold. 
On the other hand,
the Carath\'eodory theorem ensures that we need only \emph{finite} $m$,
or to be more precise $m \leq (\rk\varrho)^2 \leq d^2$, shown by Uhlmann \cite{UhlmannOptimalDecomp}.
These observations guarantee the existence of the minimum in~(\ref{eq:cnvroofext}).

Obviously, for pure states the convex roof extension is trivial,
\begin{equation}
\label{eq:cnvroofpure}
\cnvroof{\mu}(\pi)=\mu(\pi) \qquad\forall\pi\in\mathcal{P}.
\end{equation}
The convex roof extension of a function is convex~(\ref{eq:meas.conv}),
moreover, it is the largest convex function
taking the same values for pure states as the original function \cite{UhlmannOptimalDecomp}.
The convex roof extensions of pure state measures are good measures of entanglement,
because it can be proven \cite{VidalEntMon,HorodeckiEntMeas} that
if a function $\mu:\mathcal{P}\to\field{R}$ is non-increasing on average for pure states~(\ref{eq:averagePure}),
then its convex roof extension is also non-increasing on average for mixed states%
%%%%%%%%%%%%%%%%%%%%%%%%
\footnote{The $\Leftarrow$ implication is obvious.}
%%%%%%%%%%%%%%%%%%%%%%%%
(\ref{eq:meas.average})
\begin{equation}
\label{eq:averageConvRoof}
\sum_i p_i \mu(\pi_i) \leq \mu(\pi)
\quad\Longleftrightarrow\quad
\sum_i p_i \cnvroof{\mu}(\varrho_i) \leq \cnvroof{\mu}(\varrho).
\end{equation}
From this, $\cnvroof{\mu}(\varrho)$ is entanglement-monotone as well.
Because of these, for pure states~(\ref{eq:averagePure}) is called entanglement monotonicity.
The convex roof extension preserves the weak discriminance property (\ref{eq:measdiscr.weak}) 
and also the strong one (\ref{eq:measdiscr.weak}) if we additionally assume that $\mu\geq0$.
This is very useful because strong discriminance for mixed states can be used for the detection of entanglement,
which will exensively be done in chapters \ref{chap:PartSep} and \ref{chap:ThreeQB}.
Indeed, let $\mu:\mathcal{P}\to[0,\infty)$, then 
\begin{equation}
\label{eq:convroof.discr}
\begin{split}
\varrho\in\mathcal{D}_\text{sep}
\qquad&\Longleftrightarrow\qquad \exists \;\text{decomposition}\; \varrho=\sum_jp_j\pi_j \;\text{such that}\; \pi_j\in\mathcal{P}_\text{sep}\\
\qquad&\begin{array}{l}\Longrightarrow\\\Longleftrightarrow\end{array}\qquad \exists \;\text{decomposition}\; \varrho=\sum_jp_j\pi_j \;\text{such that}\; \mu(\pi_j)=0 \\
\qquad&\Longleftrightarrow\qquad \cnvroof{\mu}(\varrho)=0,
\end{split}
\end{equation}
where in the second implication the upper one is the weak discriminance (\ref{eq:measdiscr.weak})
and the lower one is the strong discriminance (\ref{eq:measdiscr.strong}),
while $\mu\geq0$ is necessary for the $\Leftarrow$ direction of the last implication.
On the other hand, the normalization (\ref{eq:measnorm}) property is obviously preserved by the convex roof extension.
%
%Another important property of convex roof extension is that it preserves also the invariance group of a function.
%Indeed, if $f(G\pi G^\dagger)=f(\pi)$ for a $G\in\Lin(\mathcal{H})$, then
%\begin{equation}
%\label{eq:convroof.invgroup}
%\begin{split}
%\cnvroof{f}(G\varrho G^\dagger)
%&\equiv\min\Bigset{\sum_jp_jf(\pi_j)}{\sum_jp_j\pi_j=G\varrho G^\dagger}\\
%&=     \min\Bigset{\sum_jp_jf(\pi_j)}{\sum_jp_jG^{-1}\pi_jG^{\dagger-1}=\varrho}\\
%&=     \min\Bigset{\sum_jp_jf(G\pi_j' G^\dagger)}{\sum_jp_j\pi_j'=\varrho}\\
%&=     \min\Bigset{\sum_jp_jf(\pi_j')}{\sum_jp_j\pi_j'=\varrho}
%\;\equiv\;\cnvroof{f}(\varrho),
%\end{split}
%\end{equation}
%where we switched over to the new variable $\pi_j'=G^{-1}\pi_j G^{\dagger-1}$.
%	THIS IS WRONG (meaningless)

%*******************************************************************************
\subsection{State vectors of bipartite systems}
\label{subsec:QM.EntMeas.2Pure}

Now, let $\mathcal{H}=\mathcal{H}_1\otimes\mathcal{H}_2$ the Hilbert space of bipartite systems
of dimension $\tpl{d}=(d_1,d_2)$.
A bipartite state vector $\cket{\psi}\in\mathcal{H}$,
having the general form
\begin{equation*}
\cket{\psi}=\sum_{i,j=1}^{d_1,d_2}\psi^{ij}\cket{i}\otimes\cket{j},
\end{equation*}
can be transformed to the so called Schmidt canonical form
\begin{subequations}
\label{eq:SchmidtDecompComp}
\begin{equation}
\label{eq:SchmidtDecomp}
\cket{\psi}=\sum_{i=1}^{d_\text{min}}\sqrt{\eta_i}\cket{\varphi_{1,i}}\otimes\cket{\varphi_{2,i}}
\end{equation}
by suitable local unitary transformations.
The existence of this form, also called Schmidt decomposition, makes the entanglement of pure states of bipartite systems simple.
%($d_\text{min}=\min\{d_1,d_2\}$)
Here $\{\cket{\varphi_{1,i}}\}$ and $\{\cket{\varphi_{2,i}}\}$
are sets of orthonormal vectors in $\mathcal{H}_1$ and $\mathcal{H}_2$.
The nonnegative $\eta_i$ numbers are called Schmidt coefficients,
and they sum up to one.
They form the spectra of the reduced states,
which are therefore the same for both of the subsystems,
the only difference can be the degenerancy of the zero eigenvalue,
since
\begin{equation}
\label{eq:SchmidtDecomp.Subsys}
\pi_1=\sum_{i=1}^{d_\text{min}}\eta_i\cket{\varphi_{1,i}}\bra{\varphi_{1,i}},\qquad
\pi_2=\sum_{i=1}^{d_\text{min}}\eta_i\cket{\varphi_{2,i}}\bra{\varphi_{2,i}},
\end{equation}
with the pure state $\pi=\cket{\psi}\bra{\psi}$ and its reduced states $\pi_1=\tr_2\pi$ and  $\pi_2=\tr_1\pi$,  as usual.
\end{subequations}
It is clear that 
$\cket{\psi}$ is separable if and only if it has only one non-zero Schmidt coefficient,
which is then equal to $1$.

On the other hand,
we can label the LU orbits in $\mathcal{H}$ by the Schmidt coefficients
(ordered non-increasingly)
so they are the only non-local parameters of a bipartite pure state.
Alternatively, we can form an equivalent set of LU-invariants, 
which can be calculated without the diagonalization of the local states, which is
\begin{equation}
\label{eq:2canonPureLUinvs}
I_q(\psi) = \tr\pi_1^q\equiv \sum_{i=1}^{d_\text{min}}\eta_i^q,\qquad q=1,2,\dots,d_\text{min}.
\end{equation}
%where $\pi=\cket{\psi}\bra{\psi}$.
($I_1(\psi)\equiv\norm{\psi}^2$ equals to $1$ if the state vector is normalized,
however, during the investigation of orbit structures in $\mathcal{H}$ under group actions,
this constraint is often relaxed.)
Since the LOCC equivalence is the same as the LU equivalence for pure states,
we have that two states can be \emph{interconverted} by the use of LOCC
if and only if they have the same invariants (\ref{eq:2canonPureLUinvs}).
A beautiful result here \cite{NielsenMaj} is that a condition can be given even for the LOCC \emph{convertibility}
by majorization (\ref{eq:major}):
\begin{equation}
\label{eq:NielsenMajor}
\text{$\exists\Lambda$ LOCC that}\;\cket{\psi'}=\Lambda(\cket{\psi})\qquad\Longleftrightarrow\qquad
\tpls{\eta}\preceq\tpls{\eta}'
\end{equation}
where $\tpls{\eta}$ and $\tpls{\eta}'$ are the $d_\text{min}$-tuples of 
Schmidt coefficients of $\cket{\psi}$ and $\cket{\psi'}$.
From this, the LU-equivalence follows as necessary and sufficient condition for LOCC interconvertibility.
Unfortunately, there is no such an elegant condition for the LOCC convertibility of pure states of systems composed of more than two subsystems,
%Unfortunately, the condition (\ref{eq:NielsenMajor}) for the LOCC convertibility
%can not be generalized for the case of more than two subsystems,
although the LU-equivalence is still a necessary and sufficient condition for LOCC interconvertibility (\ref{eq:LOCCequivPure}).

The LOCC equivalence classes have a bit too fine-grained structure,
characterized by the $d_\text{min}-1$ real parameters,
containing all the pieces of nonlocal information about the state.
What can be said about the SLOCC equivalence classes?
The Schmidt coefficients are not invariant under LGL transformations,
but their vanishings are that.
Therefore the Schmidt rank of the state,
which is the usual matrix rank of the reduced state
\begin{equation*}
\rk\psi = \rk \pi_1,
\end{equation*}
is invariant under LGL transformations,
moreover, it can easily be seen that there are $d_\text{min}$ SLOCC classes, which are characterized by the Schmidt rank.
The set of separable states is one of them, which is the set of states of $\rk\psi =1$,
and
%For $r=\rk\psi$
\begin{equation}
\cket{\psi_r}=\sum_{i=1}^r \frac1{\sqrt r}\cket{i}\otimes\cket{i}.
\end{equation}
gives a representative element for the $\rk\psi = r$ SLOCC class.%
%%%%%%%%%%%%%%%%%%%%%%%%
\footnote{It is not difficult to construct the LGL transformation
which maps a state $\cket{\psi}$ to this canonical form.}
%\footnote{It is illustrative to show the LGL transformation
%which maps a state $\cket{\psi}$ to this canonical form.}.
%%%%%%%%%%%%%%%%%%%%%%%%

We have seen that the Schmidt decomposition~(\ref{eq:SchmidtDecomp})
is very useful in the understanding of the structure of entanglement of bipartite pure states.
In additional,
we can also give illustration to Schr\"odinger's mixture theorem~(\ref{eq:SchMixtureThm}) 
by the use of that.
To this end, let us write the pure state in terms of state vectors in Schmidt canonical form
\begin{equation*}
\cket{\psi}\bra{\psi}=\sum_{ii'}\sqrt{\eta_i}\sqrt{\eta_{i'}}\cket{\varphi_{1,i}}\bra{\varphi_{1,i'}}\otimes\cket{\varphi_{2,i}}\bra{\varphi_{2,i'}},
\end{equation*}
now apply a unitary transformation on the second subsystem 
\begin{equation*}
\Id\otimes U\cket{\psi}\bra{\psi}\Id\otimes U^\dagger
=\sum_{ii'}\sqrt{\eta_i}\sqrt{\eta_{i'}}\cket{\varphi_{1,i}}\bra{\varphi_{1,i'}}\otimes U\cket{\varphi_{2,i}}\bra{\varphi_{2,i'}}U^\dagger,
\end{equation*}
and form the reduced state of the first subsystem
\begin{equation*}
\begin{split}
\tr_2\bigl(\Id\otimes U\cket{\psi}\bra{\psi}\Id\otimes U^\dagger\bigr)
&=\sum_{ii'}\sqrt{\eta_i}\sqrt{\eta_{i'}}\cket{\varphi_{1,i}}\bra{\varphi_{1,i'}} 
\sum_j \bra{\varphi_{2,j}}U\cket{\varphi_{2,i}}\bra{\varphi_{2,i'}}U^\dagger\cket{\varphi_{2,j}}\\
&=\sum_j \Bigl(\sum_iU^j_{\phantom{j}i}\sqrt{\eta_i}\cket{\varphi_{1,i}}\Bigr)\Bigl(\sum_{i'}\sqrt{\eta_{i'}}\cc{(U^j_{\phantom{j}i'})}\bra{\varphi_{1,i'}}\Bigr),
\end{split}
\end{equation*}
where the decomposition vectors in the form of (\ref{eq:SchMixtureThm}) appear in the parentheses.
So we can think of 
the freedom in the decomposition of a density matrix
as the unitary freedom in the additional Hilbert space of the purification.
%
%the freedom in the decomposition of a density matrix
%is equivalent to the unitary freedom on the additional Hilbert space of the purification.

%The Schmidt rank of a (pure) state equals $1$ if and only if the state is separable
%$\cket{\psi}=\cket{\psi_1}\otimes\cket{\psi_2}$.
%In this case, the reduced states 
%$\varrho_1=\cket{\psi_1}\bra{\psi_1}$,
%$\varrho_2=\cket{\psi_2}\bra{\psi_2}$
%are pure ones.
%On the other hand, the reduced states of an entangled pure state are mixed,

As we have also seen, 
the reduced states of an entangled pure state are mixed ones,
which is interpreted as
we know the possible pure states of the subsystem only with some probabilities.
So the uncertainty in the pure states of the subsystem presents itself 
for the quantification of entanglement.
%so the mixedness of the reduced states can be considered as the entangledness of the state.
%so we can define the entangledness of the state as the mixedness of its reduced states.
%The more mixed the reduced state, the more entangled the bipartite pure state
It can be characterized for example by the von Neumann entropy  (\ref{eq:Neumann}) of the reduced state:
\begin{equation}
\label{eq:entanglement}
s(\psi)=S(\pi_1)\equiv H(\tpls{\eta}). 
\end{equation}
But, is this an entanglement measure in the sense of section \ref{subsec:QM.EntMeas.EntMeas}?
The answer is yes,
it has been shown in \cite{VidalEntMon,HorodeckiEntMeas} that
every unitary-invariant and concave function of the reduced state is non-increasing on average under LOCC~(\ref{eq:averagePure}),
which is the key property of the functions measuring entanglement.
%
%It has been shown in \cite{VidalEntMon,HorodeckiEntMeas} that
%if a function $F:\mathcal{D}(\mathcal{H})\to\field{R}$ is unitary-invariant and concave,
%then $f(\psi)=F(\pi_1)$ is non-increasing on average for pure states, i.e.~obey~(\ref{eq:averagePure}),
%which is the fundamental property of the functions measuring entanglement.
The von Neumann entropy~(\ref{eq:Neumann}),
the Tsallis entropies~(\ref{eq:qTsallis}) for all $q>0$,
and the R\'enyi entropies~(\ref{eq:qRenyi}) for all $0<q<1$
are known to be concave \cite{BengtssonZyczkowski}, and all of them are unitary-invariant,
so all of them can be used on the right-hand side of (\ref{eq:entanglement})
to get an entanglement measure.%
%%%%%%%%%%%%%%%%%%%%%%%%
\footnote{For the von Neumann entropy, this is called \emph{entanglement} function, and denoted by $E(\psi)$,
and for the R{\'e}nyi entropies, this is called \emph{generalized entanglement} function, and denoted by $E_q(\psi)$,
but we prefer to use this more general small letter--capital letter convention (\ref{eq:entanglement})
for functions on the states of the subsystems.}
%%%%%%%%%%%%%%%%%%%%%%%%
However, note that the von Neumann entropy of the reduced state (\ref{eq:entanglement}) is of particular importance
due to Schumacher's noiseless coding theorem, 
which is the quantum counterpart of Shannon's noiseless coding theorem of classical information theory \cite{NielsenChuang}.
For every generalized entropy 
we obviously have that $s(\psi)=0$ if and only if the state is separable,
that is, the (strong) discriminance property (\ref{eq:measdiscr.strong}) holds.

%*******************************************************************************
\subsection{Mixed states of bipartite systems}
\label{subsec:QM.EntMeas.2Mix}

For mixed states, 
for the quantification of entanglement
we can use the convex roof extension (\ref{eq:cnvroofext}) of pure state measures.
They have the operational meaning of the optimal mixing of the state
from pure states with respect to the given pure state measure.
Maybe this is the most plausible method of measuring the entanglement of mixed states.
Thanks to (\ref{eq:averageConvRoof}), the resulting function is also an entanglement measure,
moreover, it is entanglement monotone.
For example, the convex roof extended entanglement function (\ref{eq:entanglement}) is
\begin{equation}
\label{eq:entForm}
\cnvroof{s}(\varrho)=\min_{\sum_i p_i \pi_i=\varrho}  \sum_i p_i s(\pi_i),
%\min\Bigset{ \sum_jp_js(\psi_j)}{ \sum_jp_j\cket{\psi_j}\bra{\psi_j}=\varrho}
\end{equation}
which is called \emph{entanglement of formation}.
(Or \emph{generalized entanglement of formation} for the R{\'e}nyi entropies.)
%It is easy to check that its vanishing is necessary and sufficient condition of separability (\ref{eq:sep})
Since the  (strong) discriminance property (\ref{eq:measdiscr.strong}) holds for the local entropies,
it holds also for the convex roof extension of those (\ref{eq:convroof.discr}),
that is,
\begin{equation}
\label{eq:cnvRoofVanishSet}
\varrho\in\mathcal{D}_\text{sep}\qquad\Longleftrightarrow\qquad \cnvroof{s}(\varrho)=0.
\end{equation}
%Indeed, $\cnvroof{s}(\varrho)$ vanishes if and only if
%there exists a pure state $\sum_i p_i \pi_i$ decomposition of $\varrho$ for which the nonnegative $s(\pi_i)$ vanishes,
%if and only if the pure states $\pi_i$ are separable,
%which is just the (\ref{eq:sepPureDecomp}) definition of separability.

There are other entanglement measures, which do not arise as convex roof extension of pure state measures.
Such an entanglement measure is the negativity \cite{ZyczkowskietalVolSepStates,EisertPlenioConcNeg}.
It is related to the notion of partial transposition and the criterion of Peres \cite{PeresCrit}.
If the partial transposed density matrix has negative eigenvalue,
which implies entanglement (\ref{eq:critPeres}),
then its trace, which equals to $1$,
is less than the trace of its absolute value,%
%%%%%%%%%%%%%%%%%%%%%%%%
\footnote{The absolute value of a matrix $M$ is defined by the unique positive square root
of the positive matrix $M^\dagger M$, that is, $\abs{M}=\sqrt{M^\dagger M}$.
(If $M$ itself is positive then $\abs{M}=M$.)
The spectrum of $\abs{M}$, called the set of singular values of $M$,
consists of the
absolute values of the (generally complex) eigenvalues of $M$.}
%%%%%%%%%%%%%%%%%%%%%%%%
the latter is called trace-norm, $\norm{M}_{\tr}=\tr\sqrt{M^\dagger M}$.
It turns out that if we simply take the difference of these two traces
then we get an entanglement measure,
which is called \emph{negativity} 
\begin{equation}
\label{eq:neg}
  N(\varrho) = \norm{\varrho^{\transp_1}}_{\tr} - 1.
\end{equation}
The negativity is convex (\ref{eq:meas.conv}) and non-increasing on average (\ref{eq:meas.average})
hence entanglement monotone.
Moreover, what its greatest advantage is,
it is easy to calculate because there is no need of optimization, contrary to convex roof measures. 
Unfortunately, since the positivity of the partial transpose is only a necessary criterion of separability in general (\ref{eq:critPeres}),
there are entangled states of zero negativity, hence only weak discriminance (\ref{eq:measdiscr.weak}) holds for the negativity.
Because of (\ref{eq:critPeres}), however, the strong discriminance (\ref{eq:measdiscr.strong}) holds for qubit-qubit or qubit-qutrit systems.
It can easily be shown that the spectrum of the partial transposed density matrix, hence also the negativity,
is invariant under the action of the LU-group $\LieGrp{U}(\mathcal{H}_1)\times\LieGrp{U}(\mathcal{H}_2)$.

There are several other measures of entanglement \cite{BengtssonZyczkowski,Horodecki4}.
For the sake of completeness, we just make mention of some noteworthy ones, without further use.

For example there are the operational measures \emph{entanglement cost} and \emph{distillable entanglement}
\cite{BennettetalMixedStates,RainsDistEnt,PlenioVirmaniEntMeas},
which are in connection with entanglement manipulating LOCC protocols.
The distillable entanglement $E_\text{D}$ is in connection with the distillation of $k$ copies of Bell states out from the $m$ copies of a given state $\varrho$.
In the limit of $m\to\infty$, the ratio $r=k/m$ is characteristic of the state and the distillation protocol used.
What is characteristic of only the state $\varrho$ is the supremum of the $k/m$ ratios with respect to all distillation protocols,
which is the distillable entanglement
\begin{subequations}
\begin{equation}
\label{eq:EntDist}
E_\text{D}(\varrho) = 
\sup\Bigset{r}{\lim_{m\to\infty}
\Bigl( \inf_\Lambda\bignorm{\Lambda(\varrho^{\otimes m}) - (\cket{\text{B}}\bra{\text{B}})^{\otimes mr} }_{\tr} \Bigr)=0}.
\end{equation}
The entanglement cost $E_\text{C}$ is defined via the dual approach, 
that is, we want to obtain $m$ copies of a given state $\varrho$ by the expending of $k$ copies of Bell states.
Again, the ratio $r=k/m$ is characteristic of the state and the protocol used.
What is characteristic of only the state $\varrho$ is the infimum of the $k/m$ ratios with respect to all distillation protocols,
which is the entanglement cost
\begin{equation}
\label{eq:EntCost}
E_\text{C}(\varrho) = 
\inf\Bigset{r}{\lim_{m\to\infty}
\Bigl( \inf_\Lambda\bignorm{\Lambda\bigl((\cket{\text{B}}\bra{\text{B}})^{\otimes mr}\bigr) - \varrho^{\otimes m} }_{\tr} \Bigr)=0}.
\end{equation}
Finding optimal LOCC protocols for these purposes
makes the evaluation of the distillable entanglement and entanglement cost a very hard problem.
Both of the distillable entanglement and the entanglement cost are non-increasing on average under LOCC (\ref{eq:meas.average})
and normalized (\ref{eq:measnorm}).
Since there are bound-entangled states, 
only the weak discriminance (\ref{eq:measdiscr.weak}) holds for the distillalble entanglement,
while the discriminance of the entanglement cost is not known.

Another interesting and important measure is the \emph{squashed entanglement},
an additive entanglement monotone (\ref{eq:meas.average})-(\ref{eq:meas.conv}) 
with good asymptotic properties \cite{ChristandlWinterSquashed,BrandaoChristandlYardFaithfulSquashed,AlickiFannesContQMutInfSquashedAsympt},
\begin{equation}
\label{eq:squashedEnt}
E_\text{S}(\varrho_{12}) = 
\inf_{\varrho_{123}}\frac12\bigl( S(\varrho_{13}) + S(\varrho_{23}) - S(\varrho_3)- S(\varrho_{123})  \bigr),
\end{equation}
\end{subequations}
where the optimization takes place on all extended $\varrho_{123}$  states
from which the measured state can be reduced, that is, $\tr_3\varrho_{123}=\varrho_{12}$,
resulting in an especially hard optimization problem.
It is not known whether the discriminance property (\ref{eq:measdiscr.strong}) holds for this measure.

There are also geometrical measures of entanglement,
which are in connection with distances and distinguishability measures in the space of states \cite{BengtssonZyczkowski,Horodecki4}.

%*******************************************************************************
\subsection{State vectors of two qubits} 
\label{subsec:QM.EntMeas.2QBPure}

Now, consider the simplest composite system,
which is the system of two qubits, $\tpl{d}=(2,2)$.
A two-qubit state-vector $\cket{\psi}\in\mathcal{H}$ is expressed in the computational basis as
\begin{equation*}
\cket{\psi}=\sum_{i,j=0}^1\psi^{ij}\cket{i}\otimes\cket{j}.
\end{equation*}
Here we have only two Schmidt coefficients,
from which there is only one independent,
so all non-local properties are characterized by only one real parameter.
We use here another quantity, which is more convenient than the Schmidt coefficients
or even the LU-invariant $I_2(\psi)$ given in (\ref{eq:2canonPureLUinvs}).
The reduced states are mixed one-qubit states.
Their spectra can be expressed in terms of the concurrence,
as was seen in~(\ref{eq:conc2qubit}).
So we define the \emph{concurrence for a two-qubit pure state}
as the concurrence (\ref{eq:conc2}), (\ref{eq:qbconcsq}) of the reduced state
\begin{equation}
\label{eq:2QBpureConc}
c(\psi)=C(\tr_2\cket{\psi}\bra{\psi}).
\end{equation}
Now the~(\ref{eq:entanglement}) entanglement of $\cket{\psi}$ is
\begin{equation}
\label{eq:scalSc}
s(\psi)=\mathcal{S}\bigl(c(\psi)\bigr),
\end{equation}
where
\begin{equation}
\label{eq:calS}
\mathcal{S}(c)=h\left(\frac{1}{2}\Bigl(1+\sqrt{1-c^2}\Bigr)\right)
\end{equation}
with the binary entropy function $h(x)$ given in~(\ref{eq:binentr}).
Note that 
the entanglement $s(\psi)$ and the concurrence $c(\psi)$
are both good measures of entanglement in the sense of section~\ref{subsec:QM.EntMeas.EntMeas}.
Here we also see that for two qubits, they
are related by the monotone increasing function $\mathcal{S}(c)$.
%so the concurrence is also a good measure of entanglement.

Since we have that the concurrence-squared $C^2$ of a one-qubit mixed state is just four times its determinant,
see in (\ref{eq:conc2qubit}),
we have $c^2(\psi)=4\det(\tr_2\cket{\psi}\bra{\psi})$,
and $c(\psi)$ is just two times the usual determinant of $\psi$ regarded as a \emph{matrix}, that is,
%\footnote{
%Note that $\tr_2\cket{\psi}\bra{\psi}=(\psi^)^i_{\phantom{i}j}\cket{i}\bra{j}$}
\begin{equation}
\label{eq:2QBpureConcDet}
c(\psi)=2\abs{\det\psi}.
\end{equation}
The determinant can be expressed in terms of the antisymmetric tensor $\varepsilon$, see in (\ref{eq:epsilon}),
as
\begin{equation}
\label{eq:2QBpureConcEps}
\det\psi = \frac12 \varepsilon_{ii'}\varepsilon_{jj'}\psi^{ij}\psi^{i'j'},
\end{equation}
which shows that $\det\psi$ is invariant under local $\LieGrp{SL}(2,\field{C})$ transformations.
This is also a consequence of the (\ref{eq:epstraf}) transformation property of $\varepsilon$.
%as the matrix determinant has to.
Carrying out the sums, we get the expression for the determinant
\begin{equation*}
\det\psi = \psi^{00}\psi^{11}-\psi^{01}\psi^{10},
\end{equation*}
which shows, that $\det\psi$ is a permutation-invariant quantity,
because of the Schmidt decomposition.
On the other hand, the local concurrence $c(\psi)$ (as well as the local von Neumann entropy $s(\psi)$) is then
an $\LieGrp{U}(1)\times\LieGrp{SL}(2,\field{C})\times\LieGrp{SL}(2,\field{C})$-invariant.
The form (\ref{eq:2QBpureConcEps}) suggests that
$c(\psi)$ can also be written by the use of the spin flip (\ref{eq:sflipPsi}) as
\begin{equation}
\label{eq:2QBpureConcSF}
c(\psi)=\abs{\bracket{\tilde{\psi}\vert\psi}},
\end{equation}
where $\bra{\tilde{\psi}}
=\varepsilon\otimes\varepsilon\cket{\psi}$.

The $C$ concurrence (\ref{eq:conc2}) is normalized, 
so $0\leq c(\psi)\leq1$,
and the LOCC classes of pure states of this system
are labelled by this one continuous parameter.
On the other hand, 
$c(\psi)=0$ if and only if $\rk\psi=1$ (the state is separable),
and we have two SLOCC classes, the set of separable and entangled states,
$\rk\psi=2$ for that.
If we relax the normalization condition again, then we have the SLOCC classes
\begin{itemize}
\item $\mathcal{V}_\text{Null}$ (Class Null): The zero vector of $\mathcal{H}$.
\item $\mathcal{V}_{1|2}$ (Class $1|2$): These non-zero vectors are separable, which are of the form
$\cket{\psi_1}\otimes\cket{\psi_2}$.
\item $\mathcal{V}_{12}$ (Class $12$): All the other vectors.
\end{itemize}
Formally speaking,
these classes define disjoint, LGL-invariant subsets of $\mathcal{H}$,
and cover $\mathcal{H}$ entirely,
$\mathcal{H}=
\mathcal{V}_\text{Null}
\cup\mathcal{V}_{1|2}
\cup\mathcal{V}_{12}$.
Except $\mathcal{V}_\text{Null}$, these classes are not closed.

For any $\cket{\psi}\in\mathcal{H}$,
it can be determined to which class $\cket{\psi}$ belongs
by the vanishing of the norm
\begin{subequations}
\label{eq:pureLUinvs2}
\begin{equation}
\label{eq:pureLUinvs2.n}
n(\psi)= \norm{\psi}^2,
\end{equation}
and the local entropies
\begin{equation}
\label{eq:pureLUinvs2.sa}
c^2_a(\psi)= C^2(\pi_a) = 4 \det \pi_a.
\end{equation}
\end{subequations}
Here we use the concurrence-squared~(\ref{eq:conc2}),~(\ref{eq:conc2qubit})
although every entropy does the job,
since they vanish only for pure density matrices.
These are in general LU-invariant quantities,
(which is $\LieGrp{U}(2)^{\times2}$ in this case,)
moreover, $n$ is invariant under the larger group $\LieGrp{U}(4)$,
and $c^2$ under $\bigl[\LieGrp{U}(1)\times\LieGrp{SL}(2,\field{C})\bigr]^{\times2}
\isom\LieGrp{U}(1)\times\LieGrp{SL}(2,\field{C})^{\times2}$.
Then the SLOCC classes of pure two-qubit states
can be determined by the vanishing of these quantities
in the way which can be seen in table~\ref{tab:SLOCC2Pure}.
%%%%%%%%%%%%%%%%%%%%%%%%%%%%%%%%
\begin{table}
\begin{tabu}{X[c]||X[c]|X[c]}
\hline
Class                     & $n(\psi)$  & $c(\psi)$  \\
\hline
\hline
$\mathcal{V}_\text{Null}$ & $=0$       & $=0$          \\
\hline
$\mathcal{V}_{1|2}$       & $>0$       & $=0$          \\
\hline
$\mathcal{V}_{12}$        & $>0$       & $>0$          \\
\hline
\end{tabu}
\bigskip
\caption{SLOCC classes of two-qubit state vectors
identified by the vanishing of LU-invariants~(\ref{eq:pureLUinvs2}).}
\label{tab:SLOCC2Pure}
\end{table}
%%%%%%%%%%%%%%%%%%%%%%%%%%%%%%%%
Note that these quantities are entanglement-monotones~(\ref{eq:averagePure}):
$n$ trivially and $c^2$ by the reasoning after (\ref{eq:entanglement}).

%*******************************************************************************
\subsection{Mixed states of two qubits} 
\label{subsec:QM.EntMeas.2QBMix}

For $\varrho\in\mathcal{D}(\mathcal{H}_1\otimes\mathcal{H}_2)$ mixed states, 
a celebrated result is that the minimization in the formula of the entanglement of formation
can be carried out explicitly for qubits.
The main point is that the pure state concurrences $c(\psi_j)$ are the same
for the $\psi_j$s of the optimal decomposition,
so the minimization of $c$ gives the minimization of $s$
in the formula of (\ref{eq:scalSc}), that is,
\begin{equation}
\cnvroof{s}(\varrho)=\mathcal{S}\bigl(\cnvroof{c}(\varrho)\bigr).
\end{equation}
%with the monotone increasing function $\mathcal{S}$ given in (\ref{eq:calS}).
The minimization in the calculation of $\cnvroof{c}(\varrho)$ can be carried out explicitly
resulting in the so called \emph{Wootters concurrence}
\cite{HillWoottersConc,WoottersConc} 
\begin{equation}
\label{eq:WConc}
\cnvroof{c}(\varrho)=%\max\{0,
\bigl(\lambda^\downarrow_1-\lambda^\downarrow_2-\lambda^\downarrow_3-\lambda^\downarrow_4\bigr)^+,  %\},
\end{equation}
where $+$ in the superscript means the positive part%
%%%%%%%%%%%%%%%%%%%%%%%%
%\footnote{That is, $x^+=x$ if $x\geq0$, otherwise $x^+=0$.}
\footnote{That is, $x^+=\max\{0,x\}$.}
%%%%%%%%%%%%%%%%%%%%%%%%
and $\lambda^\downarrow_i$s 
are the decreasingly ordered eigenvalues 
of the positive matrix $\sqrt{\sqrt{\varrho}\tilde{\varrho}\sqrt{\varrho}}$,
written with the spin-flip $\tilde{\varrho}=(\varepsilon\otimes\varepsilon\varrho\varepsilon^\dagger\otimes\varepsilon^\dagger)^*$.
These eigenvalues are the same as the square root of the eigenvalues of the non-hermitian matrix $\varrho\tilde{\varrho}$.
The latter ones are more easy to calculate.
It can be illustrative to check that for pure states, 
(\ref{eq:WConc}) gives back the pure-state concurrence (\ref{eq:2QBpureConc}),
that is, $\cnvroof{c}(\cket{\psi}\bra{\psi})=\abs{\bracket{\tilde{\psi}\vert\psi}}=c(\psi)$, 
although this holds generally for convex roofs (\ref{eq:cnvroofpure}).
As was mentioned before,
the vanishing of $\cnvroof{s}(\varrho)$ (or, that of $\cnvroof{c}(\varrho)$, equivalently) 
is necessary and sufficient condition for separability (\ref{eq:convroof.discr}).
On the other hand, %since the invariance group is preserved by convex roof extension (\ref{eq:convroof.invgroup}),
it is easy to prove that the Wootters concurrence is invariant under the action of
 $\LieGrp{U}(1)\times\LieGrp{SL}(2,\field{C})\times\LieGrp{SL}(2,\field{C})$.

There is another exceptional property of two-qubit mixed states,
which is related to the partial transpose,
namely, the partially transposed two-qubit density matrix can have only one negative eigenvalue
\cite{SanperaetalLocalDescQInsep}.
For the negativity (\ref{eq:neg}) this gives the formula
\begin{equation}
\label{eq:neg2QB}
%  N(\varrho) = \max \left\{ 0, -2 \min\Spect\varrho^{t_2} \right\}.
  N(\varrho) = \bigl(-2 \min\Spect\varrho^{t_2} \bigr)^+,
\end{equation}
using again the positive part function as in the Wootters concurrence.
Note that in the two-qubit case
the positive partial transpose is necessary and sufficient criterion for separability (\ref{eq:critPeres6}),
so is the vanishing of the negativity.
In other words, the strong discriminance (\ref{eq:measdiscr.strong}) holds for the negativity for two-qubit systems.
On the other hand, negativity in this case has also a geometric meaning.
The noisy state
$1/(1+x)\varrho + x/(1+x) \frac12\Id\otimes\frac12\Id$
has positive partial transpose, therefore it is separable,
if and only if $2N(\varrho)\leq x$,
so $2N(\varrho)$ is the minimal relative weight of white noise
needed to wash out the entanglement in the two-qubit case \cite{VerstraeteMaxEnd2QBMix}.

%*******************************************************************************
\subsection{State vectors of three qubits}
\label{subsec:QM.EntMeas.3QBPure}

We have seen that the reason for the simplicity of the structure of entanglement of pure states for bipartite systems
was the existence of the Schmidt decomposition (\ref{eq:SchmidtDecomp}).
It is easy to see that pure states of multipartite systems ($n\geq3$) do not admit \emph{the usual form} of Schmidt decomposition in general.
If the state vector 
is in the straightforwardly generalized Schmidt form
$\cket{\psi}=\sum_{i=1}^{d_\text{min}}\sqrt{\eta_i}\cket{i}\otimes\cket{i}\otimes\dots\otimes\cket{i}$,
then the states of all composite subsystems are separable ones (\ref{eq:sepPureDecomp}),
although pure states with entangled bipartite subsystems can easily be constructed even in the case of three qubits.
Finding generalized Schmidt decompositions, that is, LU-canonical forms for systems of more-than-two subsystems
of arbitrary dimensional Hilbert spaces
is a difficult problem, which has not been carried out yet.
We note here that this problem is solved for the case of three qubits,
in which an LU-canonical form parametrized by six real parameters 
is obtained \cite{AcinetalGenSchmidt3QB,Acinetal3QBPureCanon},
but this particular form can not be generalized in a straightforward manner.

So, we turn to the simplest system which is not bipartite, the system of three qubits.
We have the Hilbert space $\mathcal{H}\equiv\mathcal{H}_{123}=\mathcal{H}_1\otimes\mathcal{H}_2\otimes\mathcal{H}_3$
with $\tpl{d}=(2,2,2)$ local dimensions.
%where, after the choice of an orthonormal basis $\{\cket{0},\cket{1}\}$ in each $\mathcal{H}_a$,
%$\mathcal{H}_a\isom\field{C}^2$.
Here we introduce a convention being very convenient for the tripartite case.
The letters $a$, $b$ and $c$ are variables
taking their values in the set of labels $L=\{1,2,3\}$.
When these variables appear in a formula,
they form a partition of $\{1,2,3\}$,
so they take always different values
and \emph{the formula is understood for all the different values of these variables automatically.}
Although, sometimes a formula is symmetric under the interchange of two such variables
in which case we keep only one of the identical formulas.

Let the three-qubit state vector $\cket{\psi}\in\mathcal{H}$ be expressed 
in the computational basis $\cket{ijk}=\cket{i}\otimes\cket{j}\otimes\cket{k}$ as
\begin{equation*}
\cket{\psi}=\sum_{i,j,k=0}^1\psi^{ijk}\cket{ijk}.
\end{equation*}
We also have the pure state $\pi=\cket{\psi}\bra{\psi}\in\mathcal{P}(\mathcal{H})$,
and in this tripartite system we have bipartite and singlepartite subsystems, 
for which we have the density matrices 
%$\pi_{bc}=\tr_a\bigl(\cket{\psi}\bra{\psi}\bigr)\in\mathcal{D}(\mathcal{H}_{bc})$ and
$\pi_{bc}=\tr_a\pi\in\mathcal{D}(\mathcal{H}_{bc})$ and
%$\pi_a=\tr_{bc}\bigl(\cket{\psi}\bra{\psi}\bigr)\in\mathcal{D}(\mathcal{H}_a)$.
$\pi_a=\tr_{bc}\pi\in\mathcal{D}(\mathcal{H}_a)$.
Now, how to characterize the entanglement in this system and in its subsystems?
We have, for example, the concurrence (\ref{eq:conc2}) of singlepartite subsystems
%\begin{equation}
%c_a(\psi)=C(\pi_a)
%\end{equation}
$C(\pi_a)$
measuring the entanglement of the subsystem $a$ with the rest of the system, which is $bc$ in this case.
On the other hand,
thanks to the Schmidt decomposition (\ref{eq:SchmidtDecompComp}),
all the reduced states of a three-qubit state are at most of rank two,
which makes the calculation
of the $\cnvroof{c}(\pi_{bc})$ Wootters concurrences (\ref{eq:WConc})
of the bipartite reduced states possible
in a closed form \cite{CKWThreetangle} (see in section \ref{subsec:ThreeQB.Pure.WConc}).
Moreover, these are bounded from above by the concurrence of the one-qubit subsystem
by the so called \emph{Coffmann-Kundu-Wootters inequality} as follows
\begin{equation}
\label{eq:CKWmonogamy}
 {\cnvroof{c}}^2(\pi_{ab})
+{\cnvroof{c}}^2(\pi_{ac})\leq
C^2(\pi_a).
\end{equation}
This means that there is a restriction on the entanglement of 
the subsystem $a$ with subsystems $b$ and $c$
by its entanglement with the subsystem $bc$,
which is called the \emph{monogamy}% of the concurrence \cite{CKWThreetangle}.
%%%%%%%%%%%%%%%%%%%%%%%%
\footnote{For a recent introduction to the monogamy,
 see \cite{KimGourSandersMonogamy}.}
%%%%%%%%%%%%%%%%%%%%%%%%
of the concurrence \cite{CKWThreetangle}.
For example, if $a$ is maximally entangled with $b$, 
then it can not be entangled (neither classically correlated) with $c$,
resembling the situation in a marriage, after which this relation of entanglement is named.
This is an entirely quantum feature,
there is no such restriction on correlations in classical systems.
This makes the quantum cryptography 
essentially different from its classical counterpart.
The generalization of the monogamy relation proved to be true for all $n$-qubit systems \cite{OsborneVerstraeteMonogamy},
which can be written as
\begin{equation}
\label{eq:monogamynQB}
\sum_{b\neq a} {\cnvroof{c}}^2(\pi_{ab}) \leq C^2(\pi_a).
\end{equation}
%for all $a=1,\dots,n$ subsystems.
However, it is known that the concurrence is not monogamous for subsystems of higher than two dimensions
\cite{OuNonMonogamous}.

It is interesting to find states which are somehow extremal in the sense of (\ref{eq:CKWmonogamy}).
It can be checked that for the W-state \cite{DurVidalCiracSLOCC3QB}
\begin{subequations}
\label{eq:WGHZ}
\begin{equation}
\label{eq:W}
\cket{\text{W}}=\frac{1}{\sqrt{3}}\bigl(\cket{100}+\cket{010}+\cket{001}\bigr)
\end{equation}
$C^2(\pi_a)=2(2/3)^2$, while ${\cnvroof{c}}^2(\pi_{ab})={\cnvroof{c}}^2(\pi_{ac})=(2/3)^2$,
hence the inequality (\ref{eq:CKWmonogamy}) is saturated,
meaning that
all the entanglement between the subsystems $a$ and $bc$ is shared in the $ab$ and $ac$ subsystems equally.
The other extremal case is that of the 
Greenberger-Horne-Zeilinger state \cite{GHZBeyondBell}
\begin{equation}
\label{eq:GHZ}
\cket{\text{GHZ}}=\frac{1}{\sqrt{2}}\bigl(\cket{000}+\cket{111}\bigr).
\end{equation}
\end{subequations}
This state is maximally entangled in the sense that its singlepartite subsystems are maximally mixed,
$C^2(\pi_a)=1$, 
while its bipartite subsystems are separable
${\cnvroof{c}}^2(\pi_{ab})={\cnvroof{c}}^2(\pi_{ac})=0$.
%(An ingenious illustration of these two kinds of tripartite entanglement can be seen in figure \ref{fig:Borromean}.)
Hence the difference between the two sides of inequality (\ref{eq:CKWmonogamy}) is maximal for $\cket{\text{GHZ}}$,
meaning that subsystem $a$ is entangled with $bc$,
but it is not entangled with $b$ or with $c$ individually.
This interesting distribution of entanglement is characterized by 
the difference between the two sides of the inequality~(\ref{eq:CKWmonogamy}),
which is called \emph{residual tangle}, or \emph{three-tangle} $\tau(\psi)$
\begin{equation}
\label{eq:CKW}
C^2(\pi_a)=
 {\cnvroof{c}}^2(\pi_{ab})
+{\cnvroof{c}}^2(\pi_{ac})+\tau(\psi).
\end{equation}
An important finding \cite{DurVidalCiracSLOCC3QB} is
that $\tau(\psi)$ is an entanglement monotone (\ref{eq:averagePure}),
so all terms in the above equality are measures of the amount of entanglement.
This means that in this three-qubit case
there are two kinds of entanglement, that is,
(bipartite) entanglement which is shared among pairs of qubits,
and (tripartite) entanglement which can not be seen in the two-qubit subsystems,
although it is present in the whole three-qubit system.

%%%%%%%%%%%%%%%%%%%%%%%%
%\begin{figure}
% \includegraphics{}
%\setlength{\unitlength}{0.6\textwidth}
%\begin{picture}(1,1)
%\end{picture}
% \caption{Graphical illustration of the entanglement in the tripartite-entangled states (\ref{eq:WGHZ}).
%W state has entangled bipartite subsystems, that is, any two rings are entangled after removing the third one,
%while GHZ state has separable bipartite subsystems, that is, any two rings are separable after removing the third one
%(Borromean rings).}
%\label{fig:Borromean}
%\end{figure}
%%%%%%%%%%%%%%%%%%%%%%%%

The explicit form of the three-tangle $\tau(\psi)$ is also noteworthy.
It is given by Cayley's $(2,2,2)$ hyperdeterminant $\Det\psi$ 
\cite{CayleyHDet,GelfandetalDiscriminants,CKWThreetangle} as
\begin{equation}
\label{eq:tau}
\tau(\psi)=4\abs{\Det\psi},
\end{equation}
where
\begin{equation}
\label{eq:HDet}
\Det\psi=-\frac12
\varepsilon_{ii'}\varepsilon_{jj'}
\varepsilon_{kk'}\varepsilon_{ll'}
\varepsilon_{mm'}\varepsilon_{nn'}
\psi^{ikl}\psi^{jk'l'}\psi^{i'mn}\psi^{j'm'n'}
\end{equation}
with $\varepsilon$ given in (\ref{eq:epsilon}).
This writing shows that $\Det\psi$ is invariant under local $\LieGrp{SL}(2,\field{C})$ transformations,
thanks to (\ref{eq:epstraf}).
Carrying out the sums, we get the expression for the hyperdeterminant 
{\setlength{\mathindent}{0.5\mathindent}
\begin{equation*}
\begin{split}
\Det\psi 
  &=\psi^{000}\psi^{111}\psi^{000}\psi^{111}+
     \psi^{110}\psi^{001}\psi^{110}\psi^{001}+
     \psi^{101}\psi^{010}\psi^{101}\psi^{010}+
     \psi^{011}\psi^{100}\psi^{011}\psi^{100}\\
 &-2\begin{aligned}[t]
    \bigl( \psi^{000}\psi^{111}\psi^{110}\psi^{001}+
     \psi^{000}\psi^{111}\psi^{101}\psi^{010}+
    &\psi^{000}\psi^{111}\psi^{011}\psi^{100}\\+
     \psi^{110}\psi^{001}\psi^{101}\psi^{010}+
    &\psi^{110}\psi^{001}\psi^{011}\psi^{100}\\+
    &\psi^{101}\psi^{010}\psi^{011}\psi^{100}\bigr)
\end{aligned}\\
&+4 \bigl(\psi^{000}\psi^{110}\psi^{101}\psi^{011}+
     \psi^{111}\psi^{001}\psi^{010}\psi^{100}\bigr),
\end{split}
\end{equation*} }\noindent
which shows that $\Det\psi$ is a permutation-invariant quantity.
In the calculations resulting in the formulas above,
index contractions by $\varepsilon$,
which was already used also in the two-qubit case,
appear at all times and in all places.
We will see in section \ref{sec:ThreeQB.Pure} that
these three-qubit results have a natural treatment in the terms of LSL-covariants.
On the other hand, the construction of
$c^2$ (\ref{eq:2QBpureConcDet}) and $\tau$ (\ref{eq:tau}),
given by the square of the usual $(2,2)$ determinant (\ref{eq:2QBpureConcEps}) and the $(2,2,2)$ hyperdeterminant (\ref{eq:HDet}),
can be generalized to $n$-qubit systems 
by formulating the index contractions with $\varepsilon$ \cite{WongChristensenPotMultipartEntMeas}.
Apart from these, the same structure appears in the invariant comb approach of multiqubit entanglement
\cite{OsterlohSiewertEntMonAntilin,OsterlohSiewertInvComb,EltschkaetalEntMon}.
%which are outside of the scope of this dissertation.

%This monogamy relation makes the very abstract notion of entanglement,
%and the quantification of that
%a bit more palpable:
%The monogamy relation 
%entanglement .............. resource

What can be said about the LOCC and SLOCC classifications of three-qubit pure states?
In \cite{Sudbery3qb},
the following set of algebraically independent LU-invariant homogeneous polynomials is given
for three-qubit state vectors,
\begin{subequations}
\label{eq:3QBcanonPureLUinvs}
\begin{align}
I_0(\psi) &= \tr \pi \equiv \norm{\psi}^2,\\
I_a(\psi) &= \tr \pi_a^2 ,\\
\label{eq:3QBcanonPureLUinvs.I4}
I_4(\psi) &=3\tr(\pi_b\otimes\pi_c)\pi_{bc} - \tr\pi_b^3 - \tr\pi_c^3,\\
I_5(\psi) &= \abs{\Det\psi}^2,
\end{align}
\end{subequations}
(the normalization is again relaxed).
These invariants are sufficient for the labelling of the LU-orbits,
that is, the LOCC classes of three-qubit pure states.
The structure of these invariants are more complicated than that of the invariants (\ref{eq:2canonPureLUinvs})
for two-qubit state vectors.
Here $I_4$ is the Kempe invariant \cite{Kempe3qb}, it is the same for all different $b,c\in\{1,2,3\}$ labels.
%which was a first development of
Note that the three invariants $I_a$ (together with $I_0$ in the unnormalized case), 
carrying all pieces of information about the local density matrices $\pi_a$,
are not sufficient for the characterization the LOCC classes.
Therefore there are states different only in the invariants $I_4$ and $I_5$,
that is, globally different states which are locally the same.
This is called hidden nonlocality,
and this is a useful resource of quantum cryptography \cite{Kempe3qb}.
%(Historically, $I_4$ was invented first \cite{Kempe3qb},)

While there are infinitely many LOCC classes labelled by six real parameters,
it is a celebrated result that there are SLOCC classes of finite number in this three-qubit case,
which is referred as ``three qubits can be entangled in two inequivalent ways'' \cite{DurVidalCiracSLOCC3QB}. 
More fully, there are $1+1+3+1+1$ three-qubit SLOCC classes,
that is, subsets invariant under LGL transformations:
\begin{itemize}
\item $\mathcal{V}_\text{Null}$ (Class Null): The zero-vector of $\mathcal{H}$.
\item $\mathcal{V}_{1|2|3}$ (Class $1|2|3$): These vectors are fully separable, which are of the form
$\cket{\psi_1}\otimes\cket{\psi_2}\otimes\cket{\psi_3}$.
\item $\mathcal{V}_{a|bc}$ (three biseparable Classes $a|bc$),
for example:
$\cket{\psi_1}\otimes\cket{\psi_{23}}\in\mathcal{V}_{1|23}$, where 
$\cket{\psi_{23}}$ is not separable.
\item $\mathcal{V}_\text{W}$ (Class W): 
This is the first class of tripartite entanglement,
when no subsystem can be separated from the others.
A representative element is the standard W state (\ref{eq:W})
\item $\mathcal{V}_\text{GHZ}$ (Class GHZ): 
This is the second class of tripartite entanglement,
the class of Greenberger-Horne-Zeilinger-type entanglement.
A representative element is the standard GHZ state (\ref{eq:GHZ}).
\end{itemize}
Formally speaking,
these classes define disjoint, LGL-invariant subsets of $\mathcal{H}$,
and cover $\mathcal{H}$ entirely,
$\mathcal{H}=
\mathcal{V}_\text{Null}
\cup\mathcal{V}_{1|2|3}
\cup\mathcal{V}_{1|23}
\cup\mathcal{V}_{2|13}
\cup\mathcal{V}_{3|12}
\cup\mathcal{V}_\text{W}
\cup\mathcal{V}_\text{GHZ}$.
Except $\mathcal{V}_\text{Null}$, these classes are not closed. 
%For the partial separarability issues, we define
%$\mathcal{V}_{123}=\mathcal{V}_\text{W}\cup\mathcal{V}_\text{GHZ}$.

%%%%%%%%%%%%%%%%%%%%%%%%%%%%%%%%
\begin{table}
\begin{tabu}{X[c]||X[c]|X[c]X[c]X[c]|X[c]}
\hline
Class                     & $n(\psi)$  & $c^2_1(\psi)$ & $c^2_2(\psi)$ & $c^2_3(\psi)$ & $\tau(\psi)$ \\
\hline
\hline
$\mathcal{V}_\text{Null}$ & $=0$       & $=0$        & $=0$        & $=0$        & $=0$         \\
\hline
$\mathcal{V}_{1|2|3}$     & $>0$       & $=0$        & $=0$        & $=0$        & $=0$         \\
\hline
$\mathcal{V}_{1|23}$      & $>0$       & $=0$        & $>0$        & $>0$        & $=0$         \\
$\mathcal{V}_{2|13}$      & $>0$       & $>0$        & $=0$        & $>0$        & $=0$         \\
$\mathcal{V}_{3|12}$      & $>0$       & $>0$        & $>0$        & $=0$        & $=0$         \\
\hline
$\mathcal{V}_\text{W}$    & $>0$       & $>0$        & $>0$        & $>0$        & $=0$         \\
$\mathcal{V}_\text{GHZ}$  & $>0$       & $>0$        & $>0$        & $>0$        & $>0$    \\ 
\hline
\end{tabu}
\bigskip
\caption{SLOCC classes of three-qubit state vectors
identified by the vanishing of LU-invariants~(\ref{eq:pureLUinvs3}).}
\label{tab:SLOCC3Pure}
\end{table}
%%%%%%%%%%%%%%%%%%%%%%%%%%%%%%%%

For any $\cket{\psi}\in\mathcal{H}$,
it can be determined to which class $\cket{\psi}$ belongs
by the vanishing of the following quantities:
the norm
\begin{subequations}
\label{eq:pureLUinvs3}
\begin{equation}
\label{eq:pureLUinvs3.n}
n(\psi)= \norm{\psi}^2,
\end{equation}
the local entropies
\begin{equation}
\label{eq:pureLUinvs3.sa}
c^2_a(\psi)= C^2(\pi_a) = 4 \det \pi_a,
\end{equation}
(here we use the concurrence-squared~(\ref{eq:conc2}),~(\ref{eq:conc2qubit})
although every entropy does the job,
since they vanish only for pure density matrices)
and the three-tangle (\ref{eq:tau})
\begin{equation}
\label{eq:pureLUinvs3.tau}
\tau(\psi)=4\abs{\Det\psi}.
\end{equation}
\end{subequations}
All of these quantities are LU-invariant ones,
(which is $\LieGrp{U}(2)^{\times3}$ in this case,)
moreover, $n$ is invariant under the larger group $\LieGrp{U}(8)$,
and $\tau$ under $\bigl[\LieGrp{U}(1)\times\LieGrp{SL}(2,\field{C})\bigr]^{\times3}
\isom\LieGrp{U}(1)\times\LieGrp{SL}(2,\field{C})^{\times3}$.
It follows from the invariance properties and other observations \cite{DurVidalCiracSLOCC3QB}
that the SLOCC classes of pure three-qubit states
can be determined by the vanishing of these quantities
in the way which can be seen in table~\ref{tab:SLOCC3Pure}.
Note that all of these quantities are entanglement-monotones~(\ref{eq:averagePure}):
$n$ trivially, $c_a$ by the reasoning after (\ref{eq:entanglement}),
and the entanglement-monotonicity of $\tau$ is proven in \cite{DurVidalCiracSLOCC3QB}.

We note here that in the tripartite case with local dimensions $\tpl{d}=(2,2,d_3)$ 
the SLOCC classification can also be carried out,
resulting in finite number of LSL-orbits, 
namely, eight and nine for $d_3=3$ and $d_3\geq4$, respectively \cite{Miyake22d}.

%*******************************************************************************
\subsection{Mixed states of three qubits}
\label{subsec:QM.EntMeas.3QBMix}

In~\cite{Acinetal3QBMixClass} Ac\'in et.~al.~have investigated
the classification of \emph{mixed} three-qubit states in connection with the pure state SLOCC classes.
They have shown that the tripartite classification scheme given in section \ref{subsec:QM.Ent.NPart}
can be naturally extended.
In their classification, Class 1 of fully entangled states
is divided into two subsets, 
namely the ones of GHZ and W-type entanglement, by the following definitions.
A state is of \emph{W-type} ($\mathcal{D}_{\text{W}}$) if it can be expressed as a mixture of projectors onto $2$-separable and Class W vectors
(therefore $\mathcal{D}_{\text{W}}$ is also a convex set)
and Class GHZ vector is required for a \emph{GHZ-type} mixed state ($\mathcal{D}_{\text{GHZ}}$).
Hence the following holds
\begin{equation}
\label{eq:incl}
\mathcal{D}_\text{$3$-sep}\subset
\mathcal{D}_\text{$2$-sep}\subset
\mathcal{D}_{\text{W}}\subset
\mathcal{D}_{\text{GHZ}}\equiv
\mathcal{D}_\text{$1$-sep}\equiv
\mathcal{D}.
\end{equation}
Let \emph{Class W} be the set $\mathcal{D}_{\text{W}}\setminus\mathcal{D}_{2-\text{sep}}$
and \emph{Class GHZ} be the set $\mathcal{D}_{\text{GHZ}}\setminus\mathcal{D}_{\text{W}}$,
so $\text{Class 1} = \text{Class W}\cup\text{Class GHZ}$.
%The convex-roof extension of $\tau$
%\begin{equation} 
%\label{eq:mixtau}
%\tau(\varrho)=\min\Bigset{\sum_jp_j\tau(\psi_j)}{\sum_jp_j\cket{\psi_j}\bra{\psi_j}=\varrho}
%\end{equation}
%is a good indicator for Class GHZ: $\tau(\varrho)\neq0$ exactly for Class GHZ.
% eleg az eredmenyes fejezetekbe

The key point leading to this classification
was that the new set $\mathcal{D}_{\text{W}}$ introduced above has to be closed.
A different set defined to be the set of states
which can be expressed as a mixture of projectors onto $2$-separable and Class GHZ vectors
would not be closed,
since $\tau(\psi)=0$ characterizes also the Class W vectors 
in the set of $\tau(\psi)\neq0$ Class GHZ vectors.

We note here that in the tripartite case with local dimensions $\tpl{d}=(2,2,d_3)$,
the eight and nine SLOCC classes of pure states (for $d_3=3$ and $d_3\geq4$, respectively)
give rise to a classification for mixed states  \cite{MiyakeVerstraete22dMixClass}
similar to the three-qubit case recalled here.

%*******************************************************************************
\subsection{State vectors of four qubits}
\label{subsec:QM.EntMeas.4QBPure}

In the case of four qubits
we have the Hilbert space $\mathcal{H}\equiv\mathcal{H}_{1234}=\mathcal{H}_1\otimes\mathcal{H}_2\otimes\mathcal{H}_3\otimes\mathcal{H}_4$
with $\tpl{d}=(2,2,2,2)$ local dimensions.
Let the four-qubit state vector $\cket{\psi}\in\mathcal{H}$ be expressed
in the computational basis $\cket{ijkl}=\cket{i}\otimes\cket{j}\otimes\cket{k}\otimes\cket{l}$ as
\begin{equation*}
\cket{\psi}=\sum_{i,j,k,l=0}^1\psi^{ijkl}\cket{ijkl}.
\end{equation*}
This case is much more complicated than that of three qubits,
even the SLOCC classes are labelled by continuous parameters.
%we have continuous many 

Let us start with the LSL-invariants, 
which are necessary to characterize the SLOCC classes.
It is known that there are four algebraically independent $\LieGrp{SL}(2,\field{C})^{\times 4}$
invariants \cite{LuqueThibonFourQubit,LevayFourQubitInvGeom} denoted by $H,L,M$ and $D$.
These are quadratic, quartic, quartic and sextic invariants 
of the coefficients ${\psi}^{ijkl}$ respectively.
The invariants $H$, $L$ and $M$ are given by the expressions
\begin{subequations}
\begin{equation}
\label{eq:H}
H(\psi)=\frac12\varepsilon_{ii'}\varepsilon_{jj'}\varepsilon_{kk'}\varepsilon_{ll'}
\psi^{ijkl}\psi^{i'j'k'l'}
=\frac12\bracket{\tilde{\psi}|\psi}
%\psi^{0000}\psi^{1111}-\psi^{0001}\psi^{1110}-\psi^{0010}\psi^{1101}+\psi^{0011}\psi^{1100}
% -\psi^{0100}\psi^{1011}+\psi^{0101}\psi^{1010}+\psi^{0110}\psi^{1001}-\psi^{0111}\psi^{1000},
\end{equation}
with the spin-flipped vector
$\bra{\tilde{\psi}}=\varepsilon\otimes\varepsilon\otimes\varepsilon\otimes\varepsilon\cket{\psi}$,
and
\begin{align}
\label{eq:L}
L(\psi)&=
\det\psi_{(12)(34)},\\
%\det_{(12)(34)}\psi,\\
%\det\psi^{(ij)(kl)}
%\det\begin{bmatrix}
%\psi^{0000}&\psi^{0001}&\psi^{0010}&\psi^{0011}\\
%\psi^{0100}&\psi^{0101}&\psi^{0110}&\psi^{0111}\\
%\psi^{1000}&\psi^{1001}&\psi^{1010}&\psi^{1011}\\
%\psi^{1100}&\psi^{1101}&\psi^{1110}&\psi^{1111}
%\end{bmatrix}=
%=\frac1{4!}\varepsilon_{(i^1j^1)(i^2j^2)(i^3j^3)(i^4j^4)}\varepsilon_{(k^1l^1)(k^2l^2)(k^3l^3)(k^4l^4)}
%\psi^{i^1j^1k^1l^1}
%\psi^{i^2j^2k^2l^2}
%\psi^{i^3j^3k^3l^3}
%\psi^{i^4j^4k^4l^4},
\label{eq:M}
M(\psi)&=
\det\psi_{(31)(24)},\\
%\det_{(31)(24)}\psi,\\
%\det\begin{bmatrix}
%\psi^{0000}&\psi^{0001}&\psi^{0100}&\psi^{0101}\\
%\psi^{1000}&\psi^{1001}&\psi^{1100}&\psi^{1101}\\
%\psi^{0010}&\psi^{0011}&\psi^{0110}&\psi^{0111}\\
%\psi^{1010}&\psi^{1011}&\psi^{1110}&\psi^{1111}
%\end{bmatrix},
\label{eq:N}
N(\psi)&=
\det\psi_{(14)(23)},
%\det_{(14)(23)}\psi,
%\det\begin{bmatrix}
%\psi^{0000}&\psi^{0010}&\psi^{0100}&\psi^{0110}\\
%\psi^{0001}&\psi^{0011}&\psi^{0101}&\psi^{0111}\\
%\psi^{1000}&\psi^{1010}&\psi^{1100}&\psi^{1110}\\
%\psi^{1001}&\psi^{1011}&\psi^{1101}&\psi^{1111}
%\end{bmatrix},
\end{align}
where $\psi_{(ab)(cd)}$ is a matrix of two indices,
which indices are composed of the binary indices running on subsystems%
%%%%%%%%%%%%%%%%%%%%%%%%
\footnote{That is, for example for $\psi^{IJ}_{(14)(23)}$
the capital indices $I=0,1,2,3$ are equivalent to $il=00,01,10,11$ in $\psi^{ijkl}$ respectively,
and so on.}
%%%%%%%%%%%%%%%%%%%%%%%%
$ab$ and $cd$,
and then $\det$ is the usual matrix determinant for these $4\times4$ matrices.
The three quantities $L$, $N$ and $M$ above are not linearly independent
because the $L(\psi)+M(\psi)+N(\psi)=0$ equation holds.
However, any two of them are linearly independent.
The sextic generator is
\begin{equation}
D(\psi)=\det B,
\end{equation}
\end{subequations}
given by the $3\times3$ matrix $B$,
which is the coefficient-matrix of a bi-quadratic form
composed of the $\psi^{ijkl}$ coefficients as
\begin{equation*}
%b(\ve{x},\ve{t})= 
\frac12\varepsilon_{jj'}\varepsilon_{kk'} (\psi^{ijkl} x_it_l)  (\psi^{i'j'k'l'} x_{i'}t_{l'})=
%b(\ve{x},\ve{t})=
\begin{bmatrix}
x_0^2,x_0x_1,x_1^2
\end{bmatrix}
B
\begin{bmatrix}
t_0^2\\
t_0t_1\\
t_1^2
\end{bmatrix}.
\end{equation*}

In the four-qubit case there are continuously many SLOCC classes.
However, there is a classification which concerns not the LSL orbits,
but the orbit-types, 
based on the construction of an LSL canonical form.
In this classification, it turns out that
``four qubits can be entangled in nine different ways'',
up to permutations of the subsystems
\cite{VerstraeteetalSLOCC4QB,ChterentalDjokovicSLOCC4QB}.
Here the different orbit-types are parametrized by complex parameters,
which are in connection with the invariants above.
However, the same value of the invariants is only necessary but not sufficient condition
for two states to belong to the same orbit \cite{ChterentalDjokovicSLOCC4QB}.
A different approach is given in \cite{LiLiSLOCC}.

%*******************************************************************************
\section{Summary}
\label{sec:QM.Sum}

We have seen that the structure of entanglement in multipartite systems is rather complex.
We have seen these structures in the examples of some basic systems composed of small number of qubits.
(There are also a small number of other explicit results for few-qutrit systems in the literature.) %\cite{}.)
All of these results are coming from ad hoc constructions
based on mathematical coincidences
making it possible to obtain compact useful and manageable formulas.
Unfortunately, these can hardly be generalized, or can not be generalized at all,
simply because the structure itself is very complicated.

In the following chapters we present our results, which are of two kinds.
In chapters \ref{chap:Ferm}, \ref{chap:SepCrit} and \ref{chap:ThreeQB}
we show such ad hoc results of particular quantum systems,
while in chapters \ref{chap:Deg6} and \ref{chap:PartSep}
we show some general results 
working for arbitrary number of subsystems of arbitrary dimensional Hilbert spaces,
with detailed elaboration of the tripartite case.

\chapter{Two-qubit mixed states with fermionic purifications}
\label{chap:Ferm}

% eloszo problemafelvetes
As we have seen in section~\ref{subsec:QM.Ent.2Part},
the most plausible method for the quantification of entanglement of mixed states contains an implicit step
which is the minimization over the decompositions in the convex roof construction.
This problem is solved for the case of two-qubit states,
where the result of the minimization can be written in an explicit form,
resulting in the celebrated formula of the Wotters concurrence (\ref{eq:WConc}).
However, this formula is still implicit in another sense:
matrix diagonalization is needed to obtain a formula given by the parameters of the state.
On the other hand, 
matrix diagonalization is needed to obtain the negativity (\ref{eq:neg2QB}) too.

% mit csinaltunk
In this chapter, 
we investigate a special $12$-parameter family of two-qubit density matrices \cite{LevayNagyPipekTwoFermions},
for which the Wootters concurrence and the negativity can be expressed in a closed formula
given by the parameters of the state,
and we give a detailed explicit characterization
of the special four-qubit purification of these states.

The material of this chapter covers thesis statement~\ref{statement:ferm}
(page \pageref{statement:ferm}).
\begin{organization}
\item[\ref{sec:Ferm.Rho}]
we present the parametrized family of density matrices we wish to study.
Using suitable local unitary transformations we transform this family to a canonical form.
\item[\ref{sec:Ferm.Meas}]
based on these results 
we calculate the Wootters concurrence 
and the negativity  
(sections~\ref{subsec:Ferm.Meas.Conc} and~\ref{subsec:Ferm.Meas.Neg}). 
We give a formula for the upper and lower bounds of negativity 
for a given concurrence (section~\ref{subsec:Ferm.Meas.ConcNeg}).
The mixeness is also calculated (section~\ref{subsec:Ferm.Meas.Pur}).
\item[\ref{sec:Ferm.Four}]
we consider the special four-qubit purifications of this state.
We analyze the structure of the above quantities
together with further ones characterizing four-qubit entanglement
and discuss how they are related to each other.
In particular we prove that the relevant entanglement measures
associated with the four-qubit state 
satisfy the generalized Coffman-Kundu-Wootters inequality of distributed entanglement. % \cite{OsborneVerstraeteMonogamy}.
For the residual tangle we derive an explicit formula,
containing two from the four algebraically independent four-qubit invariants.
\item[\ref{sec:Ferm.Concl}]
we give a summary and some remarks.
\end{organization}

%******************************************************************************

\section{The density matrix}
\label{sec:Ferm.Rho}

The formula for the Wootters concurrence of two-qubit density matrices, given in (\ref{eq:WConc}),
is written with the eigenvalues of $\sqrt{\sqrt{\varrho}\tilde{\varrho}\sqrt{\varrho}}$,
which are hard to express using the matrix elements of $\varrho$.
However, one can impose some conditions on $\varrho$,
which enables the Wootters concurrence to be calcualted in a closed form.
For example, a two-qubit density matrix reduced from a pure three-qubit state
is at the most of rank two,
which enables the explicit calculation of the Wootters concurrence in
terms of the amplitudes of the original three-qubit state (section \ref{subsec:ThreeQB.Pure.WConc}).
This could be generalized considering two-qubit mixed states reduced from four-qubit pure states.
However, every two-qubit density matrix can be reduced from a pure four-qubit state, 
hence as an extra constraint we impose
an antisymmetry condition on the amplitudes of
state vector 
\begin{subequations}
\label{eq:antipsi}
\begin{equation}
\cket{\psi}=\sum_{i,j,k,l=0}^1\psi^{ijkl}\cket{ijkl}
\in\mathcal{H}=\mathcal{H}_1\otimes\mathcal{H}_2\otimes\mathcal{H}_3\otimes\mathcal{H}_4,
\end{equation}
as
\begin{equation}
\label{eq:anti}
\psi^{ijkl}=-\psi^{klij},
\end{equation}
\end{subequations}
that is, we impose antisymmetry in the first and second \emph{pairs} of indices.

An alternative (and more physical) way is the one of imposing such
constraints on the original Hilbert space $\mathcal{H}\isom\field{C}^{16}$
which renders to have a tensor product structure on one of its six
dimensional subspaces $\mathcal{H}$ of the form
\begin{equation}
\mathcal{H}_{4\wedge4}=(\field{C}^2\otimes \field{C}^2)\wedge (\field{C}^2\otimes \field{C}^2) \subset \mathcal{H},
 \label{eq:tps}
\end{equation}
where $\wedge$ refers to antisymmetrization. 
As we know quantum tensor product structures are observable-induced, % \cite{ZanardiVirtualQSubsys,ZanardiLidarLloydTPS}, 
hence in order to specify this system with a tensor product structure of equation (\ref{eq:tps}) 
we have to specify the experimentally accessible interactions and
measurements that account for the admissible operations we can perform on the system. 
For example we can realize the system 
as a pair of fermions with four single particle states where 
a part of the admissible operations are local unitary transformations of the form
\begin{equation}
\cket{\psi} \qquad\longmapsto\qquad (U\otimes V)\otimes (U\otimes V)\cket{\psi},\qquad 
U,V\in \LieGrp{U}(2). %,\qquad \cket{\psi}\in \mathcal{H}. 
\label{eq:local}
\end{equation}
Taken together with equation (\ref{eq:anti}) this transformation rule 
clearly indicates that the $12$ and $34$ subsystems 
form two \emph{indistinguishable} subsystems of fermionic type.
In this sense, the reduced density matrices arising from fermionic states
that are elements of the tensor product structure as shown by equation (\ref{eq:tps})
are called density matrices with \emph{fermionic purifications}.

Let us parametrize the $6$ amplitudes of this normalized four qubit state
$\cket{\psi}$ with the antisymmetry property of equation (\ref{eq:anti}) as
\begin{subequations}
\begin{equation}
\psi^{ijkl}=\frac{1}{2}\bigl(\varepsilon^{ik}A^{jl}+B^{ik}\varepsilon^{jl}\bigr),
\label{eq:4para}
\end{equation}
where $\varepsilon$ is given in (\ref{eq:epsilon}),
and $A$ and $B$ are \emph{symmetric} matrices%
%%%%%%%%%%%%%%%%%%%%%%%
\footnote{
Note that, for example,
$A\in \mathcal{H}_2\otimes\mathcal{H}_4\isom\Lin(\mathcal{H}_4^*\to\mathcal{H}_2)$,
on the other hand, we regard 
$\cc{\varepsilon}\in \mathcal{H}_2\otimes\mathcal{H}_4\isom\Lin(\mathcal{H}_4^*\to\mathcal{H}_2)$,
(the $\varepsilon^{jl}$s are the coefficients of that)
and
$\cc{\sigma_{1,2,3}}\in \mathcal{H}_4^*\otimes\mathcal{H}_4\isom \Lin(\mathcal{H}_4^*\to\mathcal{H}_4^*)$.
Simillarly on the Hilbert spaces $\mathcal{H}_1$ and $\mathcal{H}_3$.}
%%%%%%%%%%%%%%%%%%%%%%%
of the form
\begin{equation}
 \label{eq:szimm}
A = \begin{bmatrix}
z_1-iz_2&-z_3\\
-z_3&-z_1-iz_2
\end{bmatrix}= 
\cc{\varepsilon}(\ve{z}\cc{\boldsymbol{\sigma}}),\qquad
B = \begin{bmatrix}
w_1-iw_2&-w_3\\
-w_3&-w_1-iw_2
\end{bmatrix}=
\cc{\varepsilon}(\ve{w}\cc{\boldsymbol{\sigma}}),
\end{equation}
where $\ve{z},\ve{w}\in \field{C}^3$
and the notation of (\ref{eq:qubit}) is used.
It is straightforward to check that the normalization condition
of the state $\cket{\psi}$ takes the form
\begin{equation}
\label{eq:psinorm}
 \norm{\psi}^2 = \norm{\ve{w}}^2 + \norm{\ve{z}}^2 = 1
\end{equation}
\end{subequations}
The density matrices we wish to study are arising as reduced ones of $\pi=\cket{\psi}\bra{\psi}$ of the form
\begin{equation*}
\varrho := \pi_{12}=\tr_{34}\cket{\psi}\bra{\psi}. %\label{eq:densi}
\end{equation*}
Notice that since the $12$ and $34$ subsystems are
by definition \emph{indistinguishable} we also have
$\varrho=\pi_{12}=\pi_{34}$.

A calculation of the partial trace yields the following explicit form for $\varrho$
\begin{subequations}
\label{eq:rho}
\begin{equation}
\varrho=\frac{1}{4}\bigl( \Id\otimes\Id + \Lambda \bigr),
\end{equation}
where $\Lambda$ is the traceless matrix
\begin{align}
\label{eq:Lambda}
 \Lambda &= \sr{x}\otimes\Id + \Id\otimes\sr{y}
+\sr{w}\otimes\src{z} +\src{w}\otimes\sr{z},\\
\label{eq:xy}
 \ve{x} &= i\ve{w}\times\cc{\ve{w}}, \qquad
 \ve{y} = i\ve{z}\times\cc{\ve{z}}.
\end{align}
\end{subequations}
Notice that $\ve{x}, \ve{y} \in \field{R}^3$, and
$\ve{x}\ve{w}=\ve{x}\cc{\ve{w}}=\ve{y}\ve{z}=\ve{y}\cc{\ve{z}}=0$.
Due to this, and the identities
\begin{equation}
\label{eq:xynorm}
 \norm{\ve{x}}^2 = \norm{\ve{w}}^4 - \abs{\ve{w}^2}^2,\qquad
 \norm{\ve{y}}^2 = \norm{\ve{z}}^4 - \abs{\ve{z}^2}^2,
\end{equation}
it can be checked that $\Lambda$ satisfies the identity
\begin{equation}
 \label{eq:Lambda2}
 \Lambda^2 = \bigl(1-\eta^2\bigr) \Id\otimes\Id,
\end{equation}
where
\begin{equation}
\label{eq:ceta}
 \eta \equiv\abs{ \ve{w}^2-\ve{z}^2 }, \qquad
 0 \leq \eta \leq 1.
\end{equation}
Notice that the quantity $\eta$ is just the Schliemann-measure of entanglement 
for two-fermion systems with $4$ single particle states  
\cite{SchliemannetalTwoFermions,LevayNagyPipekTwoFermions}.
Indeed the density matrix $\varrho$ (with a somewhat different parametrization) 
can alternatively be obtained as a reduced one 
arising from such fermionic systems 
after a convenient global $\LieGrp{U}(4)$ (see in \cite{LevayNagyPipekTwoFermions}), 
and a local $\LieGrp{U}(2)\times\LieGrp{U}(2)$ 
transformation of the form $\Id\otimes {\sigma}_2$.

Now by employing suitable local unitary transformations we would
like to obtain a canonical form for $\varrho$. 
According to equation (\ref{eq:local}) 
the transformations operating on subsystems $12$ or equivalently $34$ 
are of the form $U\otimes V\in \LieGrp{U}(2)\times \LieGrp{U}(2)$.

As a first step let us consider the unitary transformation
\begin{equation*}
%\label{eq:tafU}
 U(\ve{u}, \alpha)= e^{-i\frac{\alpha}{2}\sr{u}}=
 \cos\left(\alpha/2\right)\Id - i\sin\left(\alpha/2\right)\sr{u},
\end{equation*}
which is a spin-$1/2$ representation of an $\LieGrp{SU}(2)$ rotation 
around the axis $\ve{u} \in S^2\subset \field{R}^3$,
($\norm{\ve{u}}^2 = 1$) 
with an angle $\alpha$.
A particular rotation from $\ve{x}$ to $\ve{x}'$ ($\ve{x}'\neq -\ve{x}$) can be written in the form
$U(\ve{u}, \alpha)\sr{x} U(\ve{u}, \alpha)^\dagger = \sr{x'}$
with the parameters 
$\ve{u} = (\ve{x}\times \ve{x}')/(\norm{\ve{x}\times \ve{x}'})$
and
$\cos\alpha = \ve{x}\ve{x}'/\ve{x}\ve{x}$
leading to the transformation operator
\begin{equation*}
%\label{eq:specU}
 U(\ve{u}, \alpha) = \frac{1}{\sqrt{2\norm{\ve{x}}^2(\norm{\ve{x}}^2 + \ve{x}'\ve{x})}}
\left(\norm{\ve{x}}^2 \Id + (\sr{x})(\sr{x'})\right).\\
\end{equation*}
Employing this, we can rotate the vector $\ve{x}$ and $\ve{y}$ to the direction of the coordinate axis $z$ as 
$U_\ve{x}\sr{x}U_\ve{x}^\dagger = r\sigma_3$,
$V_\ve{y}\sr{y}V_\ve{y}^\dagger = s\sigma_3$.
In this case the matrices are of the form
\begin{subequations}
\begin{equation}
\label{eq:UxVy}
U_\ve{x}:=\frac{1}{\sqrt{2r(r+x_3)}}\bigl(r\Id + \sigma_3(\sr{x})\bigr),\qquad
V_\ve{y}:=\frac{1}{\sqrt{2s(s+y_3)}}\bigl(s\Id + \sigma_3(\sr{y})\bigr),
\end{equation}
with the parameters
\begin{equation}
\label{eq:rs}
 r = \norm{\ve{x}},\qquad
 s = \norm{\ve{y}}.
\end{equation}
\end{subequations}
Moreover, it can be checked that due to the
special form of $\ve{x}$ and $\ve{y}$, the
transformations above rotate the third components of $\ve{w}$ and $\ve{z}$ into zero, 
%$U^\dagger_\ve{x}\sr{w}U_\ve{x} = \sr{w'}$, 
%$V^\dagger_\ve{y}\sr{z}V_\ve{y} = \sr{z'}$
%and
\begin{subequations}
\begin{align}
\label{eq:wuj}
U_\ve{x}\sr{w}U_\ve{x}^\dagger &= \sr{w'},&\qquad
\ve{w}' = \ve{w} - \frac{\ve{w}\ve{x}'}{r^2+\ve{x}\ve{x}'}(\ve{x} + \ve{x}')
 &= \begin{bmatrix}
w_1-\frac{x_1}{r+x_3}w_3\\
w_2-\frac{x_2}{r+x_3}w_3\\
0
\end{bmatrix},\\
\label{eq:zuj}
V_\ve{y}\sr{z}V_\ve{y}^\dagger &= \sr{z'},&\qquad
 \ve{z}' = \ve{z} - \frac{\ve{z}\ve{y}'}{s^2+\ve{y}\ve{y}'}(\ve{y} + \ve{y}')
 &= \begin{bmatrix}
z_1-\frac{y_1}{s+y_3}z_3\\
z_2-\frac{y_2}{s+y_3}z_3\\
0
\end{bmatrix}.
\end{align}
\end{subequations}
%(This effect is like a stereographical projection.)

Every $U \in \LieGrp{U}(2)$ unitary transformation 
acting on an arbitrary $\ve{a}\in \field{C}^3$ 
as $U \sr{\ve{a}} U^\dagger = \sr{\ve{a'}} $
preserves $\ve{a}^2$ and $\norm{\ve{a}}^2$,
%since $\ve{a}^2 = -\det(\sr{\ve{a}})$, 
since $2\ve{a}^2 = \tr(\sr{\ve{a}})(\sr{\ve{a}})$, 
and $2\norm{\ve{a}}^2 = \tr(\sr{\ve{a}})^\dagger(\sr{\ve{a}})$.
Hence
\begin{equation}
\label{eq:wzinvar}
\ve{w}'^2 = \ve{w}^2,\qquad
\ve{z}'^2 = \ve{z}^2,\qquad
\norm{\ve{w}'}^2 = \norm{\ve{w}}^2,\qquad
\norm{\ve{z}'}^2 = \norm{\ve{z}}^2,\qquad
\eta' = \eta
\end{equation}
are invariant under local $\LieGrp{U}(2)\times \LieGrp{U}(2)$ transformations.%
%%%%%%%%%%%%%%%%%%%%%%%%
\footnote{The entanglement measure  $\eta$ is also invariant 
under the larger group of $\LieGrp{U}(4)$ transformations.}
%%%%%%%%%%%%%%%%%%%%%%%%

Now by employing the local $\LieGrp{U}(2)\times \LieGrp{U}(2)$ transformations
$U_\ve{x}\otimes V_\ve{y}$, the density matrix can be cast to
the form,
\begin{subequations}
\begin{equation}
\label{eq:rhotraf}
\varrho' =
\left(U_\ve{x}\otimes V_\ve{y}\right) \varrho \left(U_\ve{x}^\dagger\otimes V_\ve{y}^\dagger\right)
= \frac{1}{4}\left(\Id + \Lambda' \right),
\end{equation}
where
\begin{equation}
\label{eq:Lamtraf}
 \Lambda' = r\sigma_3\otimes\Id + \Id\otimes s\sigma_3
+\sr{w'}\otimes\src{z'} +\src{w'}\otimes\sr{z'}
\end{equation}
has the special form%
%%%%%%%%%%%%%%%%%%%%%%%%
\footnote{For better readability, zeroes in matrices are often denoted with dots.}
%%%%%%%%%%%%%%%%%%%%%%%%
\begin{equation}
\label{eq:Xshape}
 \Lambda' = \begin{bmatrix}
\alpha_3 & \cdot & \cdot & \alpha_1-i\alpha_2\\
\cdot & \beta_3 & \beta_1-i\beta_2 & \cdot\\
\cdot & \beta_1+i\beta_2 & -\beta_3 & \cdot\\
\alpha_1+i\alpha_2 & \cdot & \cdot & -\alpha_3
\end{bmatrix}
\end{equation}
with the quantities defined as
\begin{align}
\label{eq:ab}
\vs{\alpha} &= \begin{bmatrix}
 \xi_1 - \xi_2 \\
 \zeta_1 + \zeta_2 \\
 r + s
\end{bmatrix}\in\field{R}^3,& \quad
\vs{\beta}  &= \begin{bmatrix}
 \xi_1 + \xi_2 \\
 \zeta_1 - \zeta_2 \\
 r - s
\end{bmatrix}\in\field{R}^3,\\
\label{eq:abe}
\xi_1 &= w'_1 \cc{z'}_1 + \cc{w'}_1 z'_1,& \quad
\zeta_1 &= w'_2 \cc{z'}_1 + \cc{w'}_2 z'_1, \\
\label{eq:abv}
\xi_2 &= w'_2 \cc{z'}_2 + \cc{w'}_2 z'_2,& \quad
\zeta_2 &= w'_1 \cc{z'}_2 + \cc{w'}_1 z'_2.
\end{align}
\end{subequations}
Thanks to the special X-shape of $\Lambda'$,
we can regard $\varrho'$ as the direct sum of two $2\times 2$ blocks
$1/4(\Id+\sr{\vs{\alpha}})$ and $1/4(\Id+\sr{\vs{\beta}})$.
Having obtained this canonical form of the reduced density matrix
$\varrho$, now we turn to the calculation of the corresponding
entanglement measures.

%**********************************************************************************************************
\section{Measures of entanglement for the density matrix}
\label{sec:Ferm.Meas}

In this section we calculate the Wootters concurrence (\ref{eq:WConc})
and the negativity (\ref{eq:neg2QB})
of this density matrix $\varrho$.
Both of these measures are invariant under the local unitary transformations
which we have used to obtain the canonical form $\varrho'$,
therefore we can calculate them for the canonical form.
Moreover, the steps leading from $\varrho'$ to the measures
leave the X-shape of $\varrho'$ invariant,
so we have to calculate eigenvalues of matrices which are of size $2\times2$ only.
This keeps the symbolical calculations in a relatively easy way.

\subsection{Wootters concurrence}
\label{subsec:Ferm.Meas.Conc}

Let us start with the Wootters concurrence (\ref{eq:WConc})
of the density matrix $\varrho$.
As we have mention in section \ref{subsec:QM.EntMeas.2QBMix},
the the Wootters concurrence is invariant under
$\LieGrp{U}(1)\times\LieGrp{SL}(2,\field{C})\times \LieGrp{SL}(2, \field{C})$, % invariant \cite{OsterlohSiewertEntMonAntilin},
so we can use the canonical form $\varrho'$
we have obtained in the previous section via using the transformation 
$U_\ve{x}\otimes V_\ve{y}\in \LieGrp{SU}(2)\times \LieGrp{SU}(2)
\subset \LieGrp{U}(1)\times \LieGrp{SL}(2, \field{C})\times \LieGrp{SL}(2, \field{C})$ 
for its calculation.

The matrix $\varrho'\tilde{\varrho'}$ has the same X-shape as $\varrho'$,
with the blocks
$1/16(\tilde{\alpha}_0\Id+\sr{\tilde{\vs{\alpha}}})$ and
$1/16(\tilde{\beta}_0\Id+\sr{\tilde{\vs{\beta}}})$
where $\tilde{\alpha}_\mu$ and $\tilde{\beta}_\nu$ with $\mu,\nu=0,1,2,3$
are quadratic in $\vs{\alpha}$ and $\vs{\beta}$:
\begin{equation}
\label{eq:abtraff}
\tilde{\alpha}_\mu =
\begin{bmatrix}
1+\alpha_1^2+\alpha_2^2-\alpha_3^2\\
2\alpha_1 -i2\alpha_2\alpha_3\\
2\alpha_2 +i2\alpha_3\alpha_1\\
0
\end{bmatrix} \in \field{C}^4,\qquad
\tilde{\beta}_\nu =
\begin{bmatrix}
1+\beta_1^2+\beta_2^2-\beta_3^2\\
2\beta_1 -i2\beta_2\beta_3\\
2\beta_2 +i2\beta_3\beta_1\\
0
\end{bmatrix} \in \field{C}^4.
\end{equation}
The eigenvalues of the blocks
$1/16\bigl(\tilde{\alpha}_0\Id+\sr{\tilde{\vs{\alpha}}}\bigr)$ and
$1/16\bigl(\tilde{\beta}_0\Id+\sr{\tilde{\vs{\beta}}}\bigr)$ are
$1/16\Bigl(\tilde{\alpha}_0 \pm \sqrt{ \tilde{\vs{\alpha}}^2}\Bigr)$ and
$1/16\Bigl(\tilde{\beta}_0 \pm \sqrt{ \tilde{\vs{\beta}}^2}\Bigr)$,
respectively. 
Now, we can express these with the help of the quantities 
$\vs{\alpha}$, $\vs{\beta}$ of (\ref{eq:abtraff}) 
and get the eigenvalues of $\varrho'\tilde{\varrho}'$ in the form
\begin{equation*}
 \Eigv \bigl( \varrho'\tilde{\varrho}'\bigr) =
\Biggl\{
\frac{1}{16}\left(\sqrt{\alpha_1^2+\alpha_2^2} \pm \sqrt{1-\alpha_3^2}\right)^2,
\frac{1}{16}\left(\sqrt{\beta_1^2+\beta_2^2} \pm \sqrt{1-\beta_3^2}\right)^2
\Biggr\}.
\end{equation*}
Now, using equations (\ref{eq:ab})-(\ref{eq:abv}), 
we have to express these as functions of the original quantities 
$\ve{z}^2$, $\ve{w}^2$, $\norm{\ve{z}}^2$ and $\norm{\ve{w}}^2$. 
Straightforward calculation shows that
\begin{subequations}
\label{eq:jajj}
\begin{align}
\label{eq:jajj.1}
\alpha_1^2+\alpha_2^2 &=
2\norm{\ve{w}'}^2\norm{\ve{z}'}^2 +
\ve{w}'^2{\cc{\ve{z}'}}^2 + {\cc{\ve{w}'}}^2\ve{z}'^2-2rs,\\
\label{eq:jajj.2}
\beta_1^2+\beta_2^2 &=
2\norm{\ve{w}'}^2\norm{\ve{z}'}^2 +
\ve{w}'^2{\cc{\ve{z}'}}^2 + {\cc{\ve{w}'}}^2\ve{z}'^2+2rs,\\
\label{eq:jajj.3}
1-\alpha_3^2 &=
2\norm{\ve{w}'}^2\norm{\ve{z}'}^2 +
\ve{w}'^2{\cc{\ve{w}'}}^2 + \ve{z}'^2{\cc{\ve{z}'}}^2-2rs,\\
\label{eq:jajj.4}
1-\beta_3^2 &=
2\norm{\ve{w}'}^2\norm{\ve{z}'}^2 +
\ve{w}'^2{\cc{\ve{w}'}}^2 + \ve{z}'^2{\cc{\ve{z}'}}^2+2rs.
\end{align}
\end{subequations}
The formulas above are expressed in terms of quantities invariant
under the transformation yielding the canonical form (\ref{eq:wzinvar})
so we can simply omit the primes.
Hence by using equation (\ref{eq:xynorm}) and~(\ref{eq:ceta}) we can establish that
\begin{equation*}
%\begin{split}
\alpha_1^2+\alpha_2^2 = 1-\alpha_3^2 - \eta^2,\qquad 
\beta_1^2 +\beta_2^2  = 1-\beta_3^2  - \eta^2.
%\end{split}
\end{equation*}
For further use, we define the quantities
\begin{equation}
\label{eq:gammapm}
 \gamma_+ := r+s \equiv \alpha_3, \qquad
 \gamma_- := r-s \equiv \beta_3.
\end{equation}
With these, 
the spectrum of $\sqrt{\sqrt{\varrho}\tilde{\varrho}\sqrt{\varrho}}$,
which is the same as
the square root of the eigenvalues of $\varrho\tilde{\varrho}$, is
\begin{equation*}
%\begin{split}
\Spect\sqrt{\sqrt{\varrho}\tilde{\varrho}\sqrt{\varrho}} = \Biggl\{
\frac{1}{4}\left( \sqrt{1-\gamma_+^2} \pm \sqrt{1-\gamma_+^2-\eta^2} \right),
\frac{1}{4}\left( \sqrt{1-\gamma_-^2} \pm
\sqrt{1-\gamma_-^2-\eta^2} \right) \Biggr\}.
%\end{split}
\end{equation*}
The greatest one of these is 
$\lambda^\downarrow_1 = 1/4\left(\sqrt{1-\gamma_-^2} + \sqrt{1-\gamma_-^2-\eta^2} \right)$,
and, after subtracting the others from it,
we get finally the nice formula for the Wootters concurrence (\ref{eq:WConc})
\begin{equation}
 \label{eq:conc}
 \cnvroof{c}(\varrho) = %\max\left\{ 0,
\frac{1}{2}\left( \sqrt{1-\eta^2-\gamma_-^2}- \sqrt{1-\gamma_+^2}\right)^+ %\right\}
\end{equation}
with the quantities defined in equations (\ref{eq:xynorm}),~(\ref{eq:ceta}),
(\ref{eq:rs}) and (\ref{eq:gammapm}) containing our
basic parameters $\mathbf{ w}$ and $\mathbf{ z}$ of $\varrho$.
One can easily check by the definitions (\ref{eq:gammapm})
that 
\begin{equation}
\label{eq:KleinEnt}
%\text{$\varrho$ entangled}\qquad\Longleftrightarrow\qquad \eta^2 < 4rs. %\equiv \gamma_+^2-\gamma_-^2 \geq0.
\varrho\in\mathcal{D}_\text{sep}\qquad\Longleftrightarrow\qquad \eta^2 \geq 4rs. %\equiv \gamma_+^2-\gamma_-^2 \geq0.
\end{equation}
Hence the surface dividing the entangled and separable states
in the space of these density matrices
is given by the equation
%\begin{equation}
$\eta^2 = 4rs$,
%\end{equation}
which is a special deformation of the $\eta = 0$ Klein-quadric \cite{LevayNagyPipekTwoFermions}.
This can also be seen from the (\ref{eq:negat}) formula of negativity,
see in the next subsection.

%**********************************************************************************************************
\subsection{Negativity}
%\label{sec_Neg}
\label{subsec:Ferm.Meas.Neg}

The next entanglement measure which we are able to calculate for $\varrho$ is the negativity (\ref{eq:neg2QB}), 
which is defined in this case 
by the smallest eigenvalue of the partially transposed density matrix 
(section~\ref{subsec:QM.EntMeas.2QBMix}).
Since the eigenvalues of the density matrix together with those of its partial transpose
are invariant under $\LieGrp{U}(2)\times\LieGrp{U}(2)$ transformations,
we can use again the
%$\LieGrp{SU}(2)\times \LieGrp{SU}(2) \subset \LieGrp{U}(2)\times\LieGrp{U}(2)$-transformed 
canonical form $\varrho'$ of equation (\ref{eq:rhotraf}).

Consider the $\varrho'^{\transp_2}$ partial transpose of $\varrho'$ with respect to the second subsystem.
This operation results in the following transformation of $\Lambda'$ of (\ref{eq:Lamtraf})
\begin{equation*}
\Lambda'^{\transp_2} = r\sigma_3\otimes\Id + \Id\otimes s\sigma_3 +\sr{w'}\otimes(\src{z'})^\transp +\src{w'}\otimes(\sr{z'})^\transp,
\end{equation*}
meaning that only $z'_2$ changes to $-z'_2$, thanks to (\ref{eq:Pauli.mx}).
By virtue of this, retaining the (\ref{eq:abe}) and~(\ref{eq:abv}) definitions of
$\xi_1$, $\xi_2$, $\zeta_1$, $\zeta_2$, and redefining
$\vs{\alpha},\vs{\beta}\in \field{R}^3$ of equation (\ref{eq:ab}) as
\begin{equation}
\label{eq:abre}
\vs{\alpha} = \begin{bmatrix}
 \xi_1 + \xi_2 \\
 \zeta_1 - \zeta_2 \\
 r + s
\end{bmatrix}\in\field{R}^3, \qquad\qquad\qquad
\vs{\beta}  = \begin{bmatrix}
 \xi_1 - \xi_2 \\
 \zeta_1 + \zeta_2 \\
 r - s
\end{bmatrix}\in\field{R}^3,
\end{equation}
the calculation proceeds as in section~\ref{subsec:Ferm.Meas.Conc}.
%%%%%%%%%%%%%%%%%%%%%%%%%%%%%%%%%%%%%%%%%%%%%%%%
%With this, we can regard $4\varrho'^{\transp_2}$ as the direct sum of the blocks
%$(\Id+\sr{\vs{\alpha}})$ and $(\Id+\sr{\vs{\beta}})$ again,
%which has the eigenvalues
%$1 \pm \norm{\vs{\alpha}}$ and
%$1 \pm \norm{\vs{\beta}}$, respectively.
%It can be calculated similarly to (\ref{eq:jajj}):
%\begin{equation}
%\begin{split}
%\label{eq:jajjj}
%\alpha_1^2+\alpha_2^2 &=
%2\norm{\ve{w}'}^2\norm{\ve{z}'}^2 +
%{\ve{w}'}^2{\cc{\ve{z}'}}^2 + {\cc{\ve{w}'}}^2{\ve{z}'}^2+2rs,\\
%\beta_1^2+\beta_2^2 &=
%2\norm{\ve{w}'}^2\norm{\ve{z}'}^2 +
%{\ve{w}'}^2{\cc{\ve{z}'}}^2 + {\cc{\ve{w}'}}^2{\ve{z}'}^2-2rs,\\
%1-\alpha_3^2 &=
%2\norm{\ve{w}'}^2\norm{\ve{z}'}^2 +
%{\ve{w}'}^2{\cc{\ve{w}'}}^2 + {\ve{z}'}^2{\cc{\ve{z}'}}^2-2rs,\\
%1-\beta_3^2 &=
%2\norm{\ve{w}'}^2\norm{\ve{z}'}^2 +
%{\ve{w}'}^2{\cc{\ve{w}'}}^2 + {\ve{z}'}^2{\cc{\ve{z}'}}^2+2rs.
%\end{split}
%\end{equation}
%(the difference is only in the sign of the last terms.)
%Discard the primes, we get:
%%%%%%%%%%%%%%%%%%%%%%%%%%%%%%%%%%%%%%%%%%%%%%%%%%%%%
Namely, we can regard $\varrho'^{\transp_2}$ as the direct sum of the blocks
$1/4(\Id+\sr{\vs{\alpha}})$ and $1/4(\Id+\sr{\vs{\beta}})$ again,
and the spectrum of $\varrho^{\transp_2}$ is the union of the spectra of them
\begin{equation*}
\Spect\bigl({\varrho'}^{\transp_2}\bigr) =
\Spect\bigl(\varrho^{\transp_2}\bigr) =
\left\{ \frac{1}{4}\bigl(1 \pm \norm{\vs{\alpha}}\bigr), \frac{1}{4}\bigl(1 \pm \norm{\vs{\beta}}\bigr) \right\}.
\end{equation*}
Then straightforward calculation shows that%
%%%%%%%%%%%%%%%%%%%%%%%%
\footnote{With the redefined $\vs{\alpha}$ and $\vs{\beta}$ (\ref{eq:abre}),
only the sign of the term $2rs$ flips in (\ref{eq:jajj.1}) and in (\ref{eq:jajj.2}).
Then the primes can be dropped again, leading to these expressions.} 
%%%%%%%%%%%%%%%%%%%%%%%%
\begin{equation*}
\norm{\vs{\alpha}}^2= 1 - \eta^2 + 4rs,\qquad
\norm{\vs{\beta}}^2 = 1 - \eta^2 - 4rs.
\end{equation*}
The former one leads to the smaller eigenvalue with the minus sign.
It follows from the definition (\ref{eq:gammapm})
that $\gamma_+^2 - \gamma_-^2=4rs$,
hence the negativity of $\varrho$ is
\begin{equation}
 \label{eq:negat}
 N(\varrho) = %\max \left\{ 0,
    \frac{1}{2} \left( \sqrt{ 1 - \eta^2 -\gamma_-^2 +\gamma_+^2} -1 \right)^+,  %\right\},
\end{equation}
with the usual quantities given in equations (\ref{eq:ceta}) and~(\ref{eq:rs}).

%**********************************************************************************************************
\subsection{Comparsion of Wootters concurrence and negativity}
%\label{sec_ConcNeg}
\label{subsec:Ferm.Meas.ConcNeg}
For a two-qubit density matrix
we can write the following inequalities between the Wootters concurrence and the negativity
\begin{equation}
\label{eq:Vineq}
 \sqrt{ (1-\cnvroof{c})^2 + {\cnvroof{c}}^2 } - (1-\cnvroof{c}) 
\leq N \leq \cnvroof{c},
\end{equation}
which are known from the papers of Verstraete et.~al.~\cite{VerstraeteetalNegConc,AudenaertetalNegConc}.
Our special case with fermionic correlations may give extra restrictions between 
these quantities, so we can pose the question,
whether we can replace the inequalities in (\ref{eq:Vineq}) by stronger ones.
Indeed, we give stronger bounds on the negativity, 
which are illustrated in figure~\ref{fig:ejjha}.
%%%%%%%%%%%%%%%%%%%%%%%%%%%%%%%%
\begin{figure}[t]
% \centering
 \setlength{\unitlength}{0.001418440\textwidth}% this 1/(423)*0.
 \begin{picture}(423,420)
  \put(0,0){\includegraphics[width=0.6\textwidth]{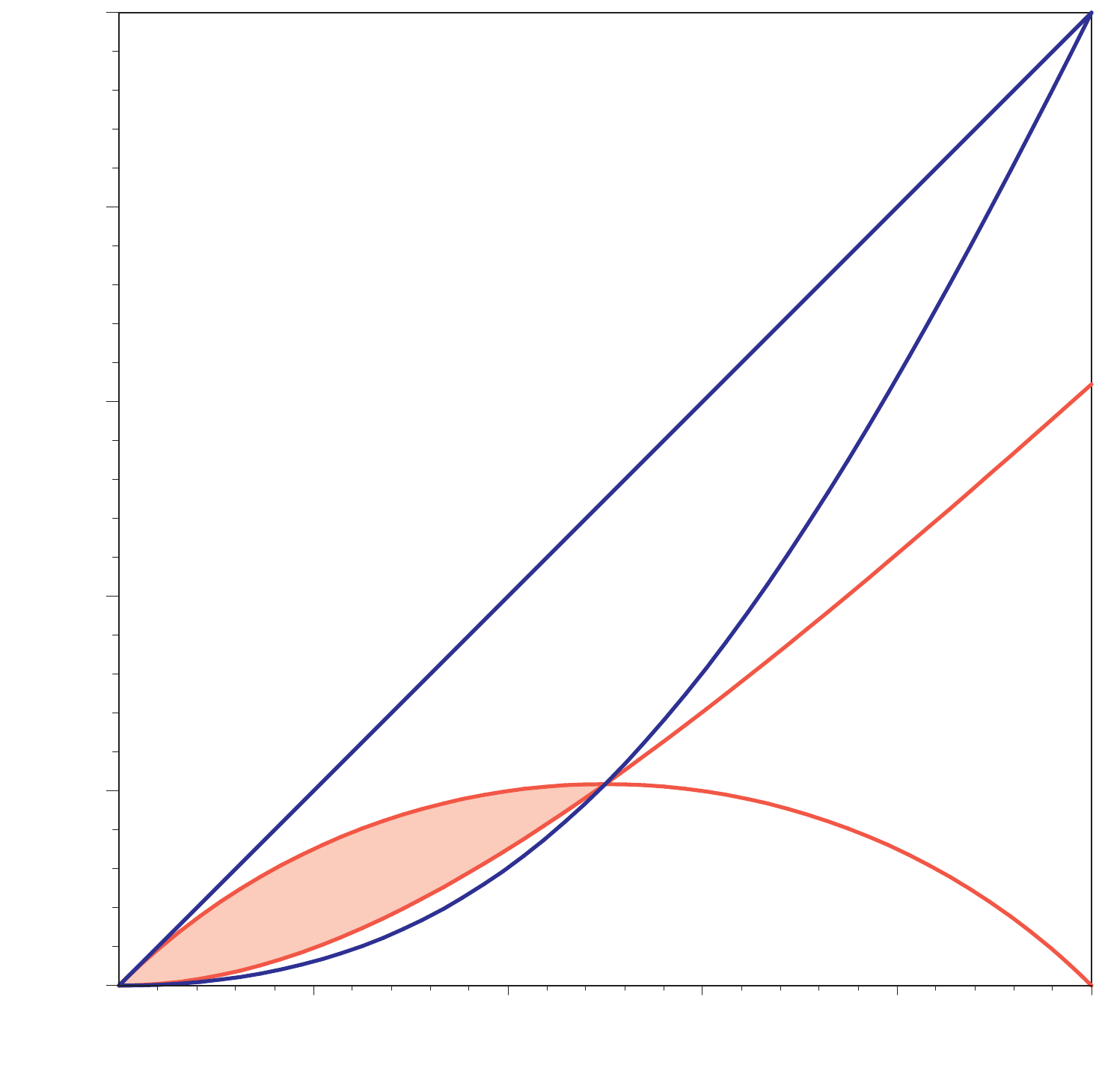}}
  \put(340,280){\makebox(0,0)[r]{\strut{}(\ref{eq:Vineq})}}
  \put(200,125){\makebox(0,0)[r]{\strut{}(\ref{eq:upperb})}}
  \put(145,070){\makebox(0,0)[r]{\strut{}(\ref{eq:lowerb})}}
  \put(10,226){\makebox(0,0)[r]{\strut{}$N(\varrho)$}}
  \put(253,6){\makebox(0,0)[r]{\strut{}$\cnvroof{c}(\varrho)$}}
{\small
  \put(40,34){\makebox(0,0)[r]{\strut{}$0$}}
  \put(130,27){\makebox(0,0)[r]{\strut{}$0.2$}}
  \put(204,27){\makebox(0,0)[r]{\strut{}$0.4$}}
  \put(280,27){\makebox(0,0)[r]{\strut{}$0.6$}}
  \put(354,27){\makebox(0,0)[r]{\strut{}$0.8$}}
  \put(425,27){\makebox(0,0)[r]{\strut{}$1$}}
  
  \put(38,114){\makebox(0,0)[r]{\strut{}$0.2$}}
  \put(38,190){\makebox(0,0)[r]{\strut{}$0.4$}}
  \put(38,264){\makebox(0,0)[r]{\strut{}$0.6$}}
  \put(38,338){\makebox(0,0)[r]{\strut{}$0.8$}}
  \put(40,412){\makebox(0,0)[r]{\strut{}$1$}}
}
 \end{picture}
 \caption{Range of values of negativity for a given Wootters concurrence.
The blue lines represent the bounds (\ref{eq:Vineq}) hold for all two-qubit density matrices,
while the red ones represent the upper (\ref{eq:upperb}) and lower (\ref{eq:lowerb}) bounds
hold for the studied family of density matrices (\ref{eq:rho}).}
 \label{fig:ejjha}
\end{figure}
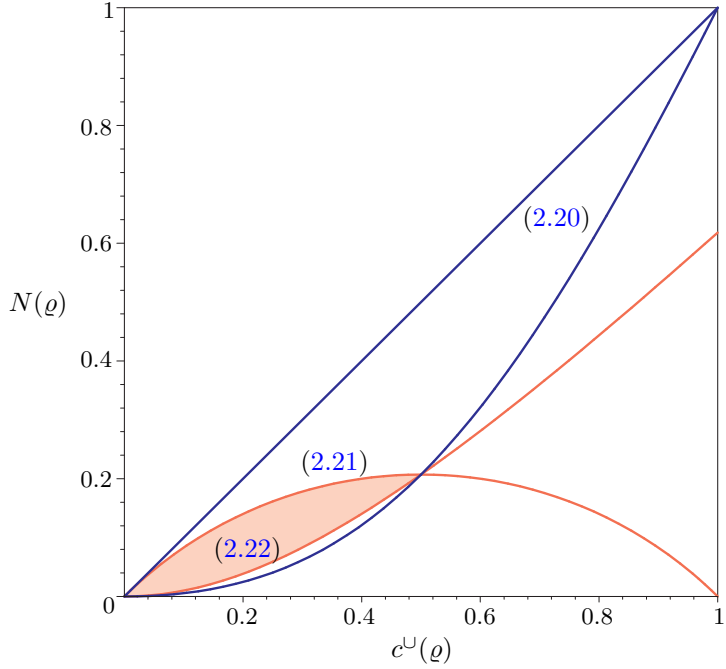
%%%%%%%%%%%%%%%%%%%%%%%%%%%%%%%%

First, we give a stronger upper bound for the negativity than the one in (\ref{eq:Vineq}).
This upper bound is the following:
\begin{equation}
\label{eq:upperb}
 N(\varrho) \leq \frac{1}{2} \left( \sqrt{ 2 - \bigl(1-2\cnvroof{c}(\varrho)\bigr)^2 } -1 \right),
\end{equation}
see the red line in figure~\ref{fig:ejjha}.
To see this inequality, 
insert equations (\ref{eq:conc}) and~(\ref{eq:negat}) into~(\ref{eq:upperb}),
then, after some algebra, rearrange the terms,
and factorize the sum:
\begin{equation*}
\begin{split}
0
&\leq
-2(1-\eta^2 -\gamma_-^2)
+2\sqrt{1-\eta^2-\gamma_-^2} 
-2\sqrt{1-\gamma_+^2}
+2\sqrt{1-\eta^2-\gamma_-^2}\sqrt{1-\gamma_+^2}\\
&=2\left( \sqrt{1-\eta^2-\gamma_-^2} - \sqrt{1-\gamma_+^2 } \right)
\left( 1 - \sqrt{1-\eta^2-\gamma_-^2} \right).
\end{split}
\end{equation*}
The second parenthesis is obviously nonnegative.
For entangled states $\cnvroof{c}(\varrho)>0$,
and the first parenthesis is proportional to that.
To see that the (\ref{eq:upperb}) upper bound is the tightest,
consider the special case when $\ve{w} = \ve{z}$.
Such states realize the boundary,
so the inequality in (\ref{eq:upperb}) is saturated for such states. %turns to equality.
Indeed, in this case $\eta = 0$, $r=s$, $\gamma_+= 2r$, $\gamma_-= 0$,
leading to %for entangled states,
$\cnvroof{c} = 1/2\bigl(1-\sqrt{1-4r^2}\bigr)$
and
$N = 1/2\bigl(\sqrt{1+4r^2}-1\bigr)$.
These depend only on $r$, wich can be expressed from $\cnvroof{c}$,
thus we can express the negativity of these states with their Wootters concurrence,
and get back the curve of the upper bound.
On the other hand, these states realize the whole boundary,
since $0\leq r=s \leq \norm{w}^2 =\norm{z}^2 = 1/2$,
so these states can arise for all allowed values of Wootters concurrence.
%
%We can also construct a one-parameter family of states 
%having maximal negativity for given concurrence, hence realizing the boundary.
%Let us use an anzatz for $\ve{w} = \ve{z}$ of the form (\ref{eq:wuj}) and~(\ref{eq:zuj}) as follows
%\begin{equation}
%\label{eq:upperstates}
%\ve{w}'=\ve{z}'=\frac{1}{2} \begin{bmatrix} 1\\\ee^{i\varphi}\\0 \end{bmatrix},\qquad 0\leq\varphi\leq\frac\pi2.
%\end{equation}
%For this we have $0\leq r=s=1/2\sin\varphi\leq1/2$, 
%leading to $\cnvroof{c} = 1/2(1-\sqrt{1-\sin^2\varphi})$
%and $N = 1/2(\sqrt{1+\sin^2\varphi}-1)$.
%So $0\leq\varphi\leq\pi/2$ is a parameter running through the whole curve of the upper bound.

Second, we give a stronger lower bound for the negativity than the one in (\ref{eq:Vineq}).
For this, note that the term $1-\eta^2-\gamma_-^2$ appears 
in the Wootters concurrence (\ref{eq:conc})
and also in the negativity (\ref{eq:negat}).
Expressing it from the formula of the Wootters concurrence (\ref{eq:conc}) %we get
%$1-\eta^2-\gamma_-^2=\Bigl(2\cnvroof{c}(\varrho)+\sqrt{1-\gamma_+^2}\Bigr)^2$,
%which can be inserted 
and inserting it into the formula of the negativity (\ref{eq:negat})
%leading to
leads to
\begin{equation*}
 N(\varrho) = \frac{1}{2} \left(\sqrt{1+(2\cnvroof{c})^2+4\cnvroof{c}\sqrt{1-\gamma_+^2}} -1 \right).
\end{equation*}
%This is the negativity expressed with the concurrence, but this depends on the parameter $\gamma_+^2$.
Since $0\leq\gamma_+^2\equiv(r+s)^2\leq1$, 
as can be concluded from (\ref{eq:gammapm}), (\ref{eq:psinorm}) and (\ref{eq:xynorm}),
we can obtain a lower and an upper bound for the negativity
for $\gamma_+^2=1$ and $0$:
%\begin{equation*}
$ 1/2 \bigl(\sqrt{1+(2\cnvroof{c})^2} -1 \bigr)
 \leq  N(\varrho) \leq
 1/2 \bigr(\sqrt{1+(2\cnvroof{c})^2+4\cnvroof{c}} -1 \bigr)=\cnvroof{c}$.
%\end{equation*}
The upper bound is weaker than the one we have in (\ref{eq:upperb}),
actually this is the upper bound in (\ref{eq:Vineq}) valid for general two-qubit density matrices.
But the lower bound
\begin{equation}
\label{eq:lowerb}
 \frac{1}{2} \left(\sqrt{1+(2\cnvroof{c})^2} -1 \right) \leq  N(\varrho)
\end{equation}
is stronger in the $0\leq \cnvroof{c}\leq1/2$ interval
that the one in (\ref{eq:Vineq}),
see the red line in figure~\ref{fig:ejjha}.
To see that the (\ref{eq:lowerb}) lower bound is the tightest,
we need that $\gamma_+^2=1$ can be realized independently of the Wootters concurrence.
From  (\ref{eq:gammapm}), (\ref{eq:psinorm}) and (\ref{eq:xynorm}),
we have that $\gamma_+^2=1$ if and only if when $\ve{w}^2 = 0$ and $\ve{z}^2=0$.
%Such states realize the boundary,
%so the inequality in (\ref{eq:lowerb}) is saturated for such states. %turns to equality.
%Indeed, 
In this case $\eta = 0$, $r=\norm{\ve{w}}^2$, $s=\norm{\ve{z}}^2$,
and $\gamma_-=\norm{\ve{w}}^2-\norm{\ve{z}}^2$,
leading to
$\cnvroof{c} = 1/2\sqrt{1-\gamma_-^2}$
and
$N = 1/2\Bigl(\sqrt{2+\gamma_-^2}-1\Bigr)$.
These states realize the whole boundary,
since $0\leq \gamma_-^2 = \bigl(2\norm{\ve{w}}^2-1\bigr)^2 \leq 1$,
so these states can arise for all allowed values of the Wootters concurrence.
%
%We can also construct a one-parameter family of states having minimal negativity for given concurrence, hence realizing the boundary.
%Let us use an anzatz for $\ve{w}^2 = 0$, $\ve{z}^2=0$ of the form (\ref{eq:wuj}) and~(\ref{eq:zuj}) as follows
%\begin{equation}
%\label{eq:lowerstates}
%\ve{w}'=\frac{\cos\vartheta}{\sqrt2} \begin{bmatrix} 1\\i\\0 \end{bmatrix},\qquad 
%\ve{z}'=\frac{\sin\vartheta}{\sqrt2} \begin{bmatrix} 1\\i\\0 \end{bmatrix},\qquad 
%0\leq\vartheta\leq\frac\pi4.
%\end{equation}
%For this we have $0\leq \gamma_-^2=\cos^22\vartheta \leq1$,
%leading to $\cnvroof{c} = 1/2\sqrt{1-\cos^22\vartheta}$
%and $N = 1/2(\sqrt{1+\sin^22\vartheta}-1)$.
%So $0\leq\vartheta\leq\pi/4$ is a parameter running through the whole curve of the lower bound.

Summarizing we have
\begin{equation}
\label{eq:ineq}
\begin{split}
\sqrt{ (1-\cnvroof{c})^2 + (\cnvroof{c})^2 } - (1-\cnvroof{c}) 
&\leq \frac{1}{2} \left(\sqrt{1+(2\cnvroof{c})^2} -1 \right)\\
&\leq N
\leq \frac{1}{2} \left( \sqrt{ 2 - ( 1-2\cnvroof{c})^2 } -1 \right)
\leq \cnvroof{c}
\end{split}
\end{equation}
where the first and last bounds on the negativity hold 
for general two-qubit mixed states (\ref{eq:Vineq}),
the stronger second (\ref{eq:lowerb}) and third~(\ref{eq:upperb}) bounds hold
for the special family of two-qubit mixed states given in (\ref{eq:rho}).
For these states $0\leq\cnvroof{c}\leq1/2$,
the first inequality holds only in this interval.

As an illustration, we can obtain examples of states of the canonical form, which are realizing the boundaries,
that is, have maximal or minimal negativity for a given Wootters concurrence.
For the upper bound (\ref{eq:upperb}) we need states for which $\ve{w}' = \ve{z}'$
for all $0\leq r=s\leq 1/2$ parameter values.
This can be achieved by varying the relative phase of $w_1'$ and $w_2'$ (having the same magnitude)
and fixing the lengths $\norm{\ve{w}'}^2=\norm{\ve{z}'}^2=1/2$.
For the lower bound (\ref{eq:lowerb}) we need states for which $\ve{w}^2 = \ve{z}^2 = 0$ 
for all $0\leq \gamma_-^2 \leq 1$ parameter values.
This can be achieved by varying the relative lengths $\norm{\ve{w}'}^2$ and $\norm{\ve{z}'}^2$
and fixing the relative phase of $w_1'$ and $w_2'$ and of $z_1'$ and $z_2'$.
%On the other hand, we can merge together the two one-parameter families of states given above,
%and get a two-parameter family of states realizing states for every given concurrence and negativity:
Summarizing, we introduce the states depending on two parameters $\vartheta$ and $\varphi$ as follows
\begin{equation}
\label{eq:2param}
\ve{w}'=\frac{\cos\vartheta}{\sqrt2} \begin{bmatrix} 1\\\ee^{i\varphi}\\0 \end{bmatrix},\qquad 
\ve{z}'=\frac{\sin\vartheta}{\sqrt2} \begin{bmatrix} 1\\\ee^{i\varphi}\\0 \end{bmatrix},\qquad 
0\leq2\vartheta,\varphi\leq\frac\pi2,
\end{equation}
leading to the density matrix
\begin{equation*}
\varrho'= \frac14\begin{bmatrix}
1+\sin\varphi & \cdot & \cdot & -i\sin2\vartheta\cos\varphi\\
\cdot & 1+\cos2\vartheta  & 1  & \cdot \\
\cdot & 1  & 1-\cos2\vartheta  & \cdot \\
i\sin2\vartheta\cos\varphi & \cdot & \cdot & 1-\sin\varphi
\end{bmatrix}.
\end{equation*}
For these states $\ve{w}' = \ve{z}'$ if and only if $\vartheta=\pi/4$,
then they have maximal negativity for a given Wootters concurrence (\ref{eq:upperb}),
on the other hand,
$\ve{w}^2 = \ve{z}^2 = 0$ if and only if $\varphi=\pi/2$,
then they have minimal negativity for a given Wootters concurrence (\ref{eq:lowerb}).
In both cases, the free parameter runs through the whole boundary.
In general we have
$r=\cos^2\vartheta\sin\varphi$,
$s=\sin^2\vartheta\sin\varphi$,
$\gamma_+=\sin\varphi$,
$\gamma_-=\cos2\vartheta\sin\varphi$ and
$\eta=\cos2\vartheta\cos\varphi$, leading to
\begin{align*}
%\cnvroof{c}(\varrho') &= \max \left\{ 0,\frac12\left(\sin2\vartheta - \cos\varphi\right) \right\},\\
% N(\varrho')          &= \max \left\{ 0,\frac12\left(\sqrt{\sin^22\vartheta +\sin^2\varphi} -1 \right) \right\},
\cnvroof{c}(\varrho') &= \frac12\bigl(\sin2\vartheta - \cos\varphi\bigr)^+,\\
 N(\varrho')          &= \frac12\left(\sqrt{\sin^22\vartheta +\sin^2\varphi} -1 \right)^+,
\end{align*}
as can be seen in figure~\ref{fig:ejjha2}.
%%%%%%%%%%%%%%%%%%%%%%%%%%%%%%%%
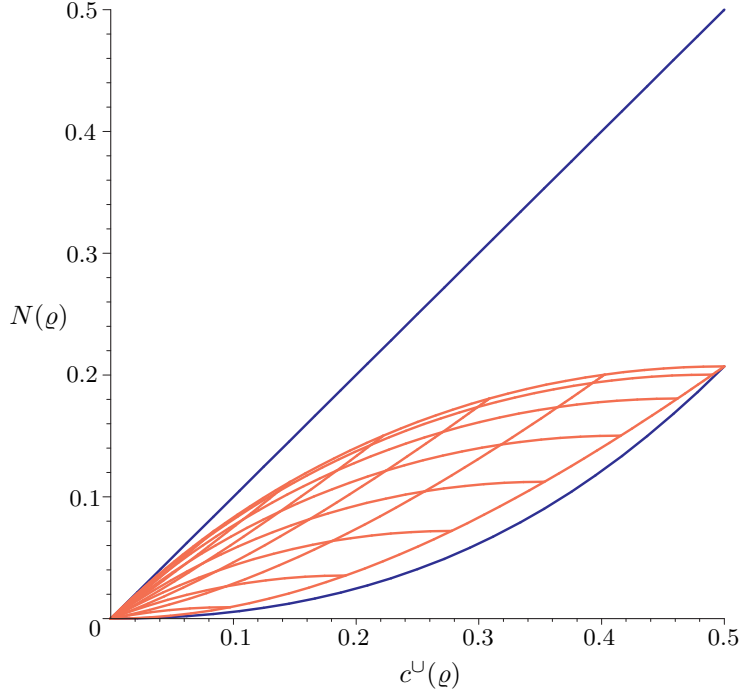
\begin{figure}[t]
% \centering
 \setlength{\unitlength}{0.001466992\textwidth}% this 1/(409)*0.6
 \begin{picture}(409,420)
  \put(0,0){\includegraphics[width=0.6\textwidth]{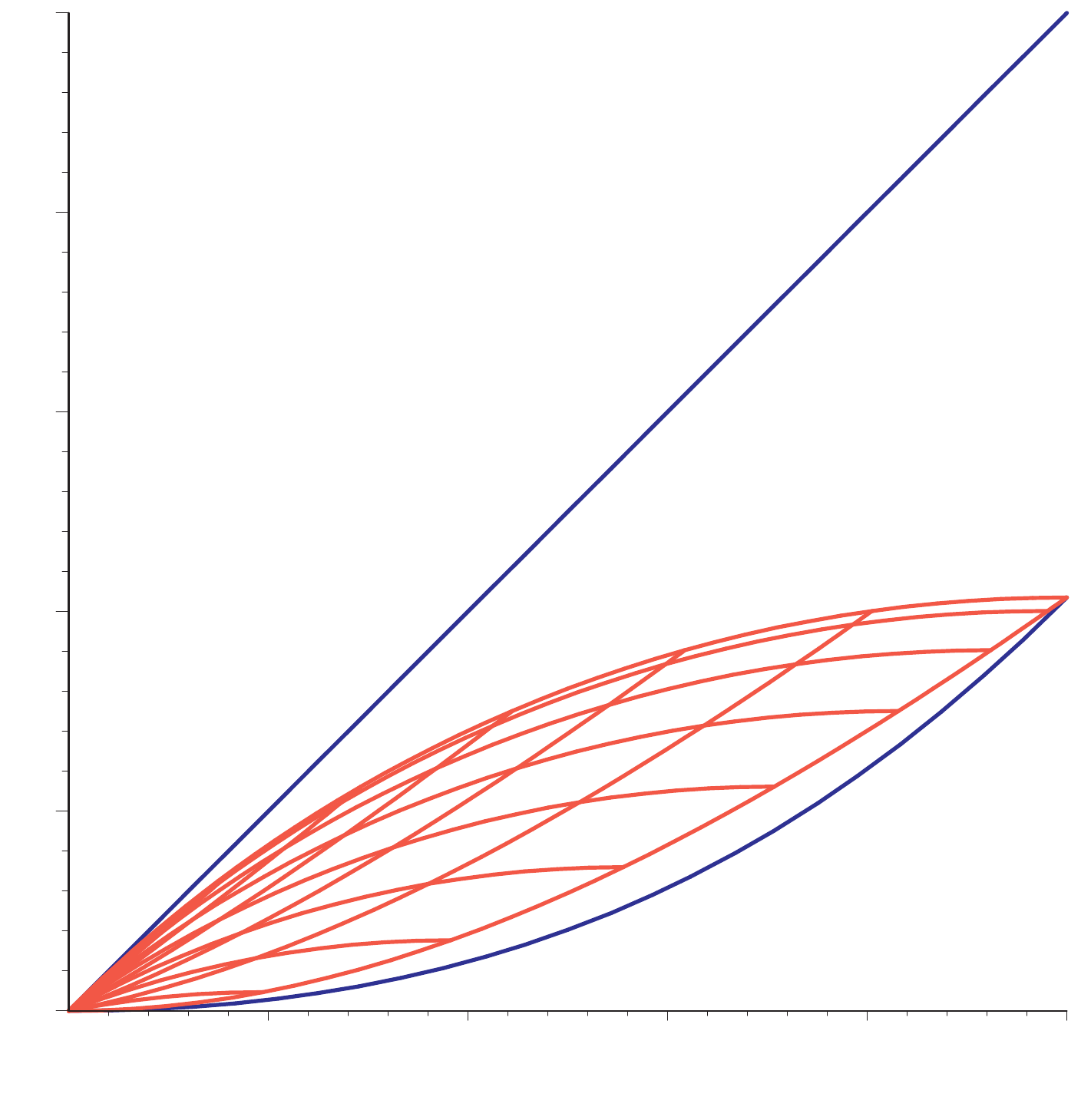}}
  \put(0,226){\makebox(0,0)[r]{\strut{}$N(\varrho)$}}
  \put(240,6){\makebox(0,0)[r]{\strut{}$\cnvroof{c}(\varrho)$}}
{\small
  \put(20,34){\makebox(0,0)[r]{\strut{}$0$}}
  \put(112,28){\makebox(0,0)[r]{\strut{}$0.1$}}
  \put(184,28){\makebox(0,0)[r]{\strut{}$0.2$}}
  \put(260,28){\makebox(0,0)[r]{\strut{}$0.3$}}
  \put(334,28){\makebox(0,0)[r]{\strut{}$0.4$}}
  \put(410,28){\makebox(0,0)[r]{\strut{}$0.5$}}

  \put(20,114){\makebox(0,0)[r]{\strut{}$0.1$}}
  \put(18,190){\makebox(0,0)[r]{\strut{}$0.2$}}
  \put(18,264){\makebox(0,0)[r]{\strut{}$0.3$}}
  \put(17,339){\makebox(0,0)[r]{\strut{}$0.4$}}
  \put(18,414){\makebox(0,0)[r]{\strut{}$0.5$}}
}
 \end{picture}
 \caption{Range of values of negativity for a given concurrence.
The lines with constant $\vartheta$ and $\varphi$ of the two-parameter state (\ref{eq:2param})
are also drawn.}
 \label{fig:ejjha2}
\end{figure}
%%%%%%%%%%%%%%%%%%%%%%%%%%%%%%%%

In (\ref{eq:2param}) we gave a two-parameter submanifold of the states given in (\ref{eq:rho}).
The maximally entangled state is at the parameter value $\vartheta=\pi/4$, $\varphi=\pi/2$,
having $\cnvroof{c}_\text{max} = 1/2$ and $N_\text{max} = \bigl(\sqrt{2} - 1\bigr)/2$.
Now, we obtain all states of the canonical form with maximal entanglement, not only this one lying in this submanifold.
%It can be seen 
%by calculating the intersection of the corresponding curves of (\ref{eq:ineq})
%that for maximally entangled states 
%$\cnvroof{c}(\varrho_\text{max}) = 1/2$ and
%$N(\varrho_\text{max}) = (\sqrt{2} - 1)/2$.
From the (\ref{eq:conc}) formula of Wootters concurrence one can see that
\begin{align*}
\cnvroof{c} = \cnvroof{c}_{\text{max}} = \frac{1}{2} \qquad\Longleftrightarrow\qquad
&\Bigl( \eta^2=0  \quad\text{and}\quad  \gamma_-^2=0  \quad\text{and}\quad  \gamma_+^2=1 \Bigr)\\
\Longleftrightarrow\qquad
&\Bigl( \ve{w}^2=\ve{z}^2  \quad\text{and}\quad  r=s  \quad\text{and}\quad  4r^2=1 \Bigr)\\
\Longleftrightarrow\qquad
&\Bigl( \norm{\ve{w}}^2=\norm{\ve{z}}^2 = \frac{1}{2} \quad\text{and} \quad \ve{w}^2=\ve{z}^2 = 0 \Bigr),
\end{align*}
using (\ref{eq:psinorm}),~(\ref{eq:ceta}),~(\ref{eq:xynorm}) and~(\ref{eq:gammapm}).
%Since the (\ref{eq:rhotraf}) transformation on $\varrho$
%preserves the quantities appearing here, 
%we can easily calculate the (\ref{eq:rhotraf}) canonical form
%of the maximally entangled state $\varrho'$.
%To find all parameters satisfying this, 
%let us choose an ansatz of the form (\ref{eq:wuj}) and~(\ref{eq:zuj}) 
One can check that the only parameters of the canonical form satisfying this
are of the form
\begin{equation*}
\ve{w}_\text{max}'=\frac{1}{2} e^{i\delta_1}  \begin{bmatrix} 1\\i\\0 \end{bmatrix},
\qquad
\ve{z}_\text{max}'=\frac{1}{2} e^{i\delta_2}  \begin{bmatrix} 1\\i\\0 \end{bmatrix},
\end{equation*}
leading to the density matrix
\begin{equation}
\varrho'_\text{max} = \frac{1}{4} \begin{bmatrix}
 \frac{3}{2}  & \cdot & \cdot & \cdot\\
  \cdot &  1 & \ee^{i\delta}  & \cdot\\
  \cdot & \ee^{-i\delta} & 1 & \cdot\\
  \cdot & \cdot & \cdot & \frac{1}{2}
  \end{bmatrix}.
\end{equation}
Here $\delta = \delta_1-\delta_2 $ is the only parameter characterizing this
maximally entangled density matrix 
of the canonical form (\ref{eq:rhotraf})
for the family of states given in (\ref{eq:rho}).

%**********************************************************************************************************
\subsection{Mixedness}
%\label{sec_Pur}
\label{subsec:Ferm.Meas.Pur}
Tdegree of mixedness of a density matrix
can be characterized, for example, by
the purity (\ref{eq:purity}),
the participiation ratio (\ref{eq:partratio})
and the concurrence-squared (\ref{eq:conc2}), among other quantites.
As we have seen in section \ref{subsec:QM.EntMeas.2Pure},
the latter one, calculated for $\varrho=\pi_{12}$, measures the entanglement of the pure state $\cket{\psi}$
under the $12|34$ split, if we \emph{do not} consider this state as a fermionic state,
but rather a special four-qubit one.
(The fermionic entanglement is a different story \cite{SchliemannetalTwoFermions,GhirardiMarinattoWeberEntCompQSys,GhirardiMarinattoEntIndist}.)
For our $\varrho$, thanks to the special property (\ref{eq:Lambda2}) of $\Lambda$,
these quantities can easily be calculated, leading to
\begin{subequations}
\begin{align}
 P(\varrho) &= \frac{1}{4} (2-\eta^2),&\qquad
 \frac{1}{4} &\leq P(\varrho) \leq \frac{1}{2},\\
 R(\varrho) &= \frac{4}{2-\eta^2},&\qquad
 2 &\leq R(\varrho) \leq 4,\\
\label{eq:conc12|34}
 C^2(\varrho) &= \frac{1}{3}(2+\eta^2),&\qquad
 \frac{2}{3} &\leq C^2(\varrho) \leq 1,
\end{align}
\end{subequations}
by virtue of equation (\ref{eq:ceta}).
The concurrence-squared $C^2(\varrho)$ never vanishes,
which means that the four-qubit state vector $\cket{\psi}$ with the antisymmetry property (\ref{eq:anti})
is never separable under the $12|34$ split,
which is a manifestation of its fermionic nature.
On the other hand, its entanglement (\ref{eq:conc12|34}) increases with $\eta$,
while the entanglement $\cnvroof{c}(\varrho)$ inside the $12$ (or equivalently $34$) subsystem (\ref{eq:conc})
decreases with that.
Although this may be interesting, 
but this can not lead to a usual monogamy relation (\ref{eq:monogamynQB}) at this point,
(that is, involving Wootters concurrences,)
since $C^2(\varrho)$ is related to a split involving $d=4$ qudits.
Even if there would be some general relations involving $C^2(\varrho)$ 
on multipartite entanglement in the fashion of monogamy, 
there are too many quantities which can not be calculated in a closed form at this point,
namely, the entanglement inside the tripartite subsystems with respect to bipartite and tripartite splits.
Anyway, the usual monogamy relation for four qubits (\ref{eq:monogamynQB})
can be studied, as is done in the next section.

%**********************************************************************************************************
\section{Relating different measures of entanglement}
\label{sec:Ferm.Four}

In this section we would like to discuss the relations of entanglement monogamy 
for multiqubit systems (\ref{eq:monogamynQB})
applied to the four-qubit pure state with the antisymmetry property (\ref{eq:anti}).
During this, we calculate the remaining quantities we need,
and see how all these quantities are related to the
important ones characterizing four-qubit pure state entanglement.
The latter ones are given in section \ref{subsec:QM.EntMeas.4QBPure}.

\subsection{One-qubit subsystems}
\label{sec:Ferm.Four.1qb}

First of all let us notice that the
\begin{subequations}
\label{eq:red1part}
\begin{align}
%\label{eq:red1}
\pi_1=\tr_{234}\pi = \tr_2 \varrho = \frac{1}{2}(\Id+\sr{x}),\\ 
%\label{eq:red2}
\pi_2=\tr_{134}\pi = \tr_1 \varrho = \frac{1}{2}(\Id+\sr{y})
\end{align}
reduced density matrices describe the entanglement
properties of subsystems $1$ and $2$ to the rest of the system
described by the four-qubit state $\pi=\cket{\psi}\bra{\psi}$.
Because of the antisymmetry (\ref{eq:anti}) of $\cket{\psi}$,
we also have
\begin{equation}
\pi_3=\pi_1,\qquad \pi_4=\pi_2.
\end{equation}
\end{subequations}
We can characterize
%It is well-known that the measures describing 
how much these subsystems are entangled with the rest 
with the concurrence-squares (\ref{eq:conc2}) of the states of the subsystems,
\begin{subequations}
\label{eq:conc1part}
\begin{align}
C^2(\pi_1)&=4\det \pi_1 =1-r^2,&\qquad 
C^2(\pi_2)&=4\det \pi_2 =1-s^2,
\intertext{and}
C^2(\pi_3)&=C^2(\pi_1),&\qquad
C^2(\pi_4)&=C^2(\pi_2).
\end{align}
\end{subequations}
The concurrence-squared (\ref{eq:conc2}) ranges from $0$ to $1$ in general,
and in our case the $0\leq C^2(\pi_a)\leq1$ bounds can be saturated
due to (\ref{eq:xynorm}).

\subsection{Other two-qubit subsystems}
\label{sec:Ferm.Four.2qb}
Now turn to the bipartite subsystems.
We already know the Wootters concurrence of the states $\pi_{12}=\pi_{34}=\varrho$
from (\ref{eq:conc}).
A straightforward calculation of the bipartite density matrices 
$\pi_{14}$ and $\pi_{23}$ shows
that they again have the form of equation (\ref{eq:rho})
with the sign of $\ve{w}$ is changed in the first case 
and the vectors $\ve{w}$ and $\ve{z}$ are exchanged in the second.
Since these transformations do not change the value of the Wootters concurrence, we have
\begin{equation}
\label{eq:2concurrences}
{\cnvroof{c}}^2(\pi_{12})={\cnvroof{c}}^2(\pi_{14})={\cnvroof{c}}^2(\pi_{23})={\cnvroof{c}}^2(\pi_{34}).
\end{equation}
Now the only two-qubit density matrices we have not discussed yet are the ones
$\pi_{13}$ and $\pi_{24}$.
Their forms are%
%%%%%%%%%%%%%%%%%%%%%%%%
\footnote{Note our convention: $B_{ik}=\cc{(B^{ik})}$, $A_{jl}=\cc{(A^{jl})}$.}
%%%%%%%%%%%%%%%%%%%%%%%%
\begin{subequations}
\label{eq:ujdensity}
\begin{align}
(\pi_{13})^{ik}_{\phantom{ik}i'k'}
&=\frac{1}{2}\left(\norm{\ve{z}}^2\varepsilon^{ik}\varepsilon_{i'k'}
+ B^{ik}B_{i'k'}\right),\\
(\pi_{24})^{jl}_{\phantom{jl}j'l'}
&=\frac{1}{2}\left(\norm{\ve{w}}^2\varepsilon^{jl}\varepsilon_{j'l'}
+ A^{jl}A_{j'l'}\right).
\end{align}
\end{subequations}
Recall now that the (\ref{eq:local}) transformation property of  the (\ref{eq:antipsi}) four-qubit state 
gives rise to the corresponding ones for the reduced density matrices
\begin{align*}
\pi_{13}\qquad\longmapsto\qquad (U\otimes U)\pi_{13}(U^\dagger\otimes U^\dagger),\\
\pi_{24}\qquad\longmapsto\qquad (V\otimes V)\pi_{24}(V^\dagger\otimes V^\dagger).
%\label{eq:transro}
\end{align*}
For $U,V\in \LieGrp{SU}(2)$, due to (\ref{eq:epstraf})  
%$V\varepsilon V^t=U\varepsilon U^t=\varepsilon$, hence
the tensors occurring in (\ref{eq:ujdensity}) transform as
\begin{align*}
\varepsilon\qquad&\longmapsto\qquad\varepsilon,\\
A\qquad&\longmapsto\qquad VAV^\transp,\\
B\qquad&\longmapsto\qquad UBU^\transp.
%\label{eq:trnsf}
\end{align*}
Using the (\ref{eq:szimm}) definition of $A$ we have for example
\begin{equation*}
VAV^\transp=V\varepsilon^*(\ve{z}\boldsymbol{\sigma}^*)V^\transp
=\varepsilon^*V^*(\ve{z}\boldsymbol{\sigma}^*)V^\transp
=\varepsilon^*(V(\cc{\ve{z}}\boldsymbol{\sigma})V^\dagger)^*
=\varepsilon^*(\cc{\ve{z}'}{\boldsymbol{\sigma}})^*
=\varepsilon^*(\ve{z}'\boldsymbol{\sigma}^*),
\end{equation*}
where by choosing $V\equiv V_\ve{y}$ of (\ref{eq:UxVy})
we get the (\ref{eq:zuj}) form for $\ve{z}'$.
Finally these manipulations yield for $\pi_{24}$ the canonical form,
which is of X-shape again,
\begin{equation}
\pi_{24}=\frac{1}{2}\begin{bmatrix}
\kappa_0+\kappa_3 & \cdot & \cdot & \kappa_1-i{\kappa}_2 \\
\cdot & \norm{\ve{w}'}^2 &-\norm{\ve{w}'}^2 & \cdot \\
\cdot &-\norm{\ve{w}'}^2 & \norm{\ve{w}'}^2 & \cdot \\
\kappa_1+i\kappa_2 & \cdot & \cdot & \kappa_0-\kappa_3
\end{bmatrix},
\label{eq:24matrix}
\end{equation}
where
\begin{equation*}
\kappa_\mu = 
\begin{bmatrix}
\norm{\ve{z}'}^2 \\ %=\norm{\ve{z}}^2\\
\abs{z_2'}^2-\abs{z_1'}^2\\
-2\Re(z_1'\cc{{z_2'}})\\
-2\Im(z_1'\cc{{z_2'}})
\end{bmatrix}.
%\label{eq:hopfmap}
\end{equation*}
Notice that
%\begin{equation}
$\kappa_0^2={\kappa}_1^2+{\kappa}_2^2+{\kappa}_3^2=\norm{\ve{z}'}^4=\norm{\ve{z}}^4$,
%\end{equation}
hence the eigenvalues of $\pi_{24}$ are 
\begin{equation*}
\Spect\pi_{24}=\{\norm{\ve{w}}^2,\norm{\ve{z}}^2,0,0\},
\end{equation*}
that is, this mixed state is of rank two.
The structure of $\pi_{13}$ is similar with the roles of $\ve{w}$ and $\ve{z}$ exchanged.
Following the same steps as in section~\ref{subsec:Ferm.Meas.Conc},
we get for the corresponding squared Wootters concurrences the following expressions
\begin{equation}
\label{eq:ujconc}
{\cnvroof{c}}^2(\pi_{13})=\left( \norm{\ve{z}}^2-\abs{\ve{w}^2}\right)^2,\qquad
{\cnvroof{c}}^2(\pi_{24})=\left( \norm{\ve{w}}^2-\abs{\ve{z}^2}\right)^2.
\end{equation}
With these, we have all the quantities we need 
to write the monogamy relations for this four-qubit system.

\subsection{Four-qubit invariants}
\label{sec:Ferm.Four.FourInv}
Before doing this,
let us now understand the meaning of the invariant $\eta$ from the
four-qubit point of view. 
In section~\ref{subsec:QM.EntMeas.4QBPure},
the independent $\LieGrp{SL}(2,\field{C})^{\times4}$-invariant homogeneous polynomials
were listed.
These polynomials are sufficient for the characterization of four-qubit entanglement in some sense,
and they show how these fermionic state vectors are embedded in the whole Hilbert-space.
A straightforward calculation shows that for the (\ref{eq:antipsi}) four-qubit state we
have $M=D=0$, however,
\begin{equation*}
H=-\frac{1}{2}\bigl(\ve{z}^2+\ve{w}^2\bigr),\qquad 
L=\frac{1}{16}\bigl(\ve{z}^2-\ve{w}^2\bigr)^2,
%\label{eq:HL}
\end{equation*}
hence, due to (\ref{eq:ceta})
\begin{equation*}
\abs{L} =\frac{1}{16}\eta^2.
%\label{eq:ujeta}
\end{equation*}
For convenience we also introduce the quantity
\begin{equation}
\label{eq:csigma}
\sigma:=\bigabs{\ve{w}^2+\ve{z}^2}=2\abs{H}.
\end{equation}
Hence $\eta=\abs{\ve{w}^2-\ve{z}^2}$ and $\sigma=\abs{\ve{w}^2+\ve{z}^2}$
are related to the only nonvanishing four qubit invariants $L$ and $H$.

\subsection{Monogamy of entanglement}
\label{sec:Ferm.Four.Monogamy}
Using the  (\ref{eq:ceta}) and  (\ref{eq:csigma}) definitions of $\eta$ and $\sigma$ and (\ref{eq:xynorm}), 
one can check that
the Wootters concurrences of the $13$ and $24$ subsystems, 
given in (\ref{eq:ujconc}), can be written as
\begin{subequations}
\label{eq:gazos}
\begin{align}
{\cnvroof{c}}^2(\pi_{13})&=s^2+\frac{1}{2}\left(\eta^2+\sigma^2\right)-2\norm{\ve{z}}^2\abs{\ve{w}^2},\\
{\cnvroof{c}}^2(\pi_{24})&=r^2+\frac{1}{2}\left(\eta^2+\sigma^2\right)-2\norm{\ve{w}}^2\abs{\ve{z}^2}.
\end{align}
\end{subequations}
%Hence we have the inequality
%\begin{equation}
%{\cnvroof{c}}^2(\pi_{13})+{\cnvroof{c}}^2(\pi_{24})\leq s^2+r^2+{\eta}^2+{\sigma}^2.
%\end{equation}
%Moreover, 
Since ${\cnvroof{c}}^2(\pi_{12})={\cnvroof{c}}^2(\pi_{14})$,
after taking the square of (\ref{eq:conc}) we get
\begin{equation}
{\cnvroof{c}}^2(\pi_{12})+{\cnvroof{c}}^2(\pi_{14})
= 1-r^2-s^2-\frac{1}{2}\eta^2-\sqrt{\bigl(1-\eta^2-{\gamma}_-^2\bigr)\bigl(1-\gamma_+^2\bigr)}.
\end{equation}
Combining this result with equations (\ref{eq:gazos}) and (\ref{eq:conc1part}) we obtain
the equations relating measures of entanglement which are involved in monogamy
\begin{subequations}
\begin{align}
{\cnvroof{c}}^2(\pi_{12})+{\cnvroof{c}}^2(\pi_{13})+{\cnvroof{c}}^2(\pi_{14})+\Sigma_1&=C^2(\pi_1),\\
{\cnvroof{c}}^2(\pi_{12})+{\cnvroof{c}}^2(\pi_{23})+{\cnvroof{c}}^2(\pi_{24})+\Sigma_2&=C^2(\pi_2),
\end{align}
\end{subequations}
where
\begin{subequations}
\label{eq:Sigmas}
\begin{align}
\Sigma_1&=2\norm{\ve{z}}^2\abs{\ve{w}^2}+\sqrt{\left(\frac{1}{2}\sigma^2+p_+\right)\left(\frac{1}{2}\sigma^2+p_-\right)}-\frac{1}{2}\sigma^2,\\
\Sigma_2&=2\norm{\ve{w}}^2\abs{\ve{z}^2}+\sqrt{\left(\frac{1}{2}\sigma^2+p_+\right)\left(\frac{1}{2}\sigma^2+p_-\right)}-\frac{1}{2}\sigma^2,
\end{align}
with
\begin{equation}
p_\pm=2\norm{\ve{z}}^2\norm{\ve{w}}^2\pm \frac{1}{2}\bigl(4rs-\eta^2\bigr).
\end{equation}
\end{subequations}
Notice that by virtue of (\ref{eq:xynorm}) $p_-$ is nonnegative.
Moreover, according to (\ref{eq:KleinEnt}), 
for nonseparable states $\pi_{12},\pi_{14},\pi_{34},\pi_{23}$
we have nonzero Wootters concurrence 
(see in (\ref{eq:2concurrences}) and (\ref{eq:conc}))
hence $\eta^2<4rs$ hence $p_+$ is also nonnegative.
In this case the \emph{residual tangles} $\Sigma_1$ and $\Sigma_2$ as defined by equations (\ref{eq:Sigmas})
are positive as they should be by virtue of (\ref{eq:monogamynQB}),
hence the generalized monogamy inequalities hold
%Coffman-Kundu-Wootters inequalities 
%of distributed entanglement hold
\begin{subequations}
\begin{align}
{\cnvroof{c}}^2(\pi_{12})+{\cnvroof{c}}^2(\pi_{13})+{\cnvroof{c}}^2(\pi_{14})&\leq C^2(\pi_1), \\
{\cnvroof{c}}^2(\pi_{12})+{\cnvroof{c}}^2(\pi_{23})+{\cnvroof{c}}^2(\pi_{24})&\leq C^2(\pi_2).
\end{align}
\end{subequations}
%For separable states we have 
%$\mathcal{C}_{12}=\mathcal{C}_{14}=\mathcal{C}_{34}=\mathcal{C}_{23}=0$
For separable $\pi_{12}$, $\pi_{14}$, $\pi_{34}$, $\pi_{23}$ states
the corresponding Wootters concurrences are zero,
and a calculation shows that the inequalities above in the form 
${\cnvroof{c}}^2(\pi_{13})\leq C^2(\pi_1)$ and
${\cnvroof{c}}^2(\pi_{24})\leq C^2(\pi_2)$ still hold with residual tangles
\begin{subequations}
\begin{align}
\label{eq:Sigmas2}
\Sigma_1&=2\norm{\ve{z}}^2\bigl(\abs{\ve{w}^2} +\norm{\ve{w}}^2\bigr),\\
\Sigma_2&=2\norm{\ve{w}}^2\bigl(\abs{\ve{z}^2} +\norm{\ve{z}}^2\bigr).
\end{align}
\end{subequations}

Equations (\ref{eq:Sigmas}) and~(\ref{eq:Sigmas2}) show the structure of the residual tangle.
Unlike in the well-known three-qubit case 
these quantities among others contain two invariants $\eta$ and $\sigma$ characterizing \emph{four-qubit} correlations.
The role of $\sigma$ 
%(which for a general four-qubit state is a permutation-invariant)
(which is in connection with $H$, which is permutation-invariant also for general four-qubit states)
is to be compared with the similar role the permutation invariant three-tangle 
$\tau$ plays within the three-qubit context, see in section \ref{subsec:QM.EntMeas.3QBPure}.
An important difference to the three-qubit case is 
that the residual tangles $\Sigma_1$ and $\Sigma_2$ 
seem to be lacking the important entanglement monotone property (\ref{eq:averagePure}).
However, according to a conjecture  \cite{HaromkinaiFourQubits},
the sum $\Sigma_1+\Sigma_2$ could be an entanglement monotone.
We hope that our explicit form will help to settle this issue at least for
our special four-qubit state of equations (\ref{eq:anti}).

%*********************************************************************************************************
\section{Summary and remarks}
\label{sec:Ferm.Concl}

In this chapter we have investigated the structure of 
a $12$ parameter family of two-qubit density matrices 
with fermionic purifications. 
Our starting point was a four-qubit state vector
with a special antisymmetry constraint imposed on its amplitudes (\ref{eq:anti}),
then the density matrices are the bipartite-reduced ones with respect to the $12$ (or $34$) subsystem.
We obtained an explicit form for these bipartite reduced states %$\pi_{12}$ 
in terms of the $6$ independent complex amplitudes $\ve{w}$ and $\ve{z}$ of the four-qubit states.
Employing local unitary transformations we derived the canonical form for this state. %$\pi_{12}$.
This form enabled an explicit calculation for different entanglement measures,
namely the Wootters concurrence and the negativity.
The bounds of the negativity in the Wootters concurrence are also calculated,
and turned out to be much stronger than those for general two-qubit states. 
The quantities occurring in these formulas (and some additional ones) 
are subject to monogamy relations of distributed entanglement 
similar to the ones showing up in the Coffman-Kundu-Wootters relations for three-qubits.
They are characterizing the entanglement trade off between different subsystems.
We have invariants $\eta$ and $\sigma$ describing the intrinsically four-partite correlations,
entanglement measures (Wootters concurrences) keeping track the mixed state entanglement 
of the bipartite subsystems embedded in the four-qubit one,
and the singlepartite concurrences
measuring how much singlepartite subsystems are entangled individually to the rest.
We derived explicit formulas displaying how these important quantities are related.

\begin{remarks}
\item The number of real parameters describing an (unnormalized) two-qubit state is $16$.
Here we treat $12$ of this $16$. 
To our knowledge, there are no explicit results in the literature for such a high number of parameters.
\item The issues of generalized monogamy are of fundamental importance.
First, monogamy is a property of the \emph{measures} we use for quantifying entanglement
rather than that of the entanglement itself. 
For example, the concurrence is monogamous only for qubits but not for subsystems of arbitrary dimensions,
as was mentioned in section \ref{subsec:QM.EntMeas.3QBPure}.
Conversely, however, we can suppose that entanglement itself is monogamous in general,
whatewer it means,
and then we demand that a proper entanglement measure should be monogamous.
What can be known at this time is that there exists a measure which is monogamous for all
Hilbert space dimensions, which is the squashed entanglement (\ref{eq:squashedEnt}) \cite{KoashiWinterMonogamySquashed}.
Unfortunately, it is extremely hard to calculate, even numerically.
The four-qubit state we considered in this chapter is of special form (\ref{eq:anti}),
so it would be interesting 
as to whether the squashed entanglements of subsystems
can be calculated for that.
%
%Moreover, 
%and also for this state considered as, for example, $\tpl{d}=(4,2,2)$ qudit-qubit-qubit state.
\item Note the structure of these claculations.
The original parameters of the state were $\ve{w}$ and $\ve{z}$.
Then, in the following, every important quantity 
(e.g., $\eta$, $\sigma$, $\gamma_\pm$, $r$, $s$) 
were expressed 
with $\ve{w}^2$, $\ve{z}^2$, $\norm{\ve{w}}^2$ and $\norm{\ve{z}}^2$,
which were invariant under the LU transformation resulting in the canonical form.
These important quantities determined the quantities related to entanglement
${\cnvroof{c}}^2(\pi_{ab})$, $C^2(\pi_{ab})$, $C^2(\pi_a)$ and $\Sigma_{1,2}$,
therefore they are of importance in themselves. 
\item On the other hand, in \cite{Makhlin2qbMixed}
the LU-orbit structure of two-qubit mixed states is completely characterized
by a set of LU-invariant homogeneous polynomials,
which  means that two such states are LU-equivalent 
if and only if all these invariants take the same value for them.
These invariants come from algebraic considerations
and they carry more or less geometrical meaning,
but they say nothing about entanglement.
The whole set consists of $18$ invariants,
given by the coefficients of the states in the $\{\sigma_\mu\otimes\sigma_\nu\}$ basis,
so they are easy to evaluate for the state of fermionic purification we considered (\ref{eq:rho}).
It turns out that most of them vanish, except five ones %$I_2$, $I_3$, $I_4$, $I_7$ and $I_{14}$,
which can be expressed with the important quantities above:
%$I_4=r^2$, $I_7=s^2$, $I_{14}=2I_4I_7$,
%$I_2=1-\eta^2-I_4-I_7$,
%$I_3=I_2^2-I_{14}$
$I_4=r^2$, $I_7=s^2$, $I_{14}=2r^2s^2$,
$I_2=1-\eta^2-r^2-s^2$ and
$I_3=(1-\eta^2-r^2-s^2)^2 - 2r^2s^2$.
Through these,
one might assign some meaning to these invariants $I_{\dots}$ of \cite{Makhlin2qbMixed} 
in the terms of entanglement, that is, using the quantities
${\cnvroof{c}}^2(\pi_{ab})$, $C^2(\pi_{ab})$, $C^2(\pi_a)$ and $\Sigma_{1,2}$.
However, note that such results are obtained only for density matrices of the form (\ref{eq:rho}),
so it can happen that these are valid only for a zero-measured subset of the states.
But these still can be useful, in a converse manner:
\item Our considered states of fermionic purification (\ref{eq:rho})
can serve to be a good tool (toy-model) for testing ideas concerning two-qubit entanglement
(a possible one from which was mentioned in the previous item): 
If some of them do not work for this special case, (for which everything is easy to calculate,)
then those must be discarded.
%\item An important property 
%of the four-qubit state with the antisymmetry property is
%that all its parameters $\ve{w}$ and $\ve{z}$ 
%appears in the reduced state of the $12$  (or $34$) subsystem.
%This does not mean that these parameters
%
%are non-local in the sense ????
%
%of the $12$ parameter family of two-qubit density matrices
%
%we studied is that all the $12$ parameters defining the four-qubit state vector
%appeare also in the
%
%
%this suggest that 
%
\end{remarks}

\chapter{\texorpdfstring{All degree $6$ local unitary invariants of multipartite systems}{All degree 6 local unitary invariants of multipartite systems}}
\label{chap:Deg6}

In the previous chapter, we have seen some examples for the occurence of LU-invariant quantities
for the quantification of entanglement.
This was natural, since, as we have seen in section \ref{subsec:QM.Ent.2Part},
the notion of entanglement of a composite quantum system
is invariant under unitary transformations on the subsystems.
Hence everything that can be said about the entanglement of a composite system
can also be said in the terms of LU-invariants.
In this sense, 
the investigation of LU-invariants is a
natural way of studying quantum entanglement.
 
%Recently, important contributions were given for this topic,
%more precisely,
%for the topic of
%LU-invariant polynomials for pure quantum states
%by Hero et.~al.~\cite{HWLUA,HWWLUA} and Vrana~\cite{PetiLUA1,PetiLUA23}.
%
Important recent developments in this direction are
the general results of Hero et.~al.~\cite{HWLUA,HWWLUA} and Vrana~\cite{PetiLUA1,PetiLUA23}
on LU-invariant polynomials for pure quantum states.
In~\cite{PetiLUA23}, 
it has been pointed out that the \emph{inverse limit} (in the local dimensions) of algebras of LU-invariant polynomials
of finite dimensional $n$-partite quantum systems is \emph{free,} 
and an \emph{algebraically independent generating set} for that has been given.
This approach using the inverse limit construction is different from the usual,
when the LU orbit structure is investigated first, for given local dimensions, and 
then invariants separating the orbits are being searched for 
\cite{RainsInv,LindenPopescuOnMultipartEnt,LindenetalNonlocalParamsMultipartDensMatrices,GrasslLU,Sudbery3qb,Makhlin2qbMixed,Acinetal3QBPureCanon}.
The structure of algebras of LU-invariant polynomials for given local dimensions is very complicated,
the inverse limit of these~\cite{HWLUA,HWWLUA,PetiLUA23}, however, 
has a remarkably simple structure: it is free~\cite{PetiLUA23},
and an algebraically independent generating set can be given for that.
Moreover, from the results for \emph{pure} states, 
one can also obtain algebraically independent LU-invariant polynomials for \emph{mixed} states~\cite{PetiLUA23}.

In this chapter we give illustrations for these general results 
on LU-invariant polynomials.
In particular, 
we write out explicitly the \emph{linearly independent basis} of the inverse limit of algebras
and single out the members of the \emph{algebraically independent generating set} from them
in the first three graded subspaces of the algebra. 
We give these polynomials in nice index-free formulas for \emph{arbitrary number of subsystems.}

The material of this chapter covers thesis statement~\ref{statement:deg6}
(page \pageref{statement:deg6}).
\begin{organization}
\item[\ref{sec:deg6.luinvs}]
we introduce the general writings of an LU-invariant polynomial
and preclude the appearance of identical ones in a less abstract way than was done originally in~\cite{HWWLUA,PetiLUA23}.
We discuss the cases of pure and mixed quantum states (in section \ref{subsec:deg6.luinvs.pure} and \ref{subsec:deg6.luinvs.mixed}).
\item[\ref{sec:deg6.graphsops}]
following~\cite{HWWLUA,PetiLUA23}, we introduce graphs for the LU-invariant polynomials (in section \ref{sec:deg6.graphsops.invargraphs}).
Then we learn to read off matrix operations
(such as partial trace, matrix product, tensorial product or partial transpose)
from graphs (in section \ref{sec:deg6.graphsops.mxopgraphs}).
If this can be done for a whole graph of an LU-invariant polynomial, 
then we can write an index-free formula for that by these operations.
\item[\ref{sec:deg6.pureinv}]
we give these index-free formulas for pure state invariants of degree $2$, $4$ and $6$ 
(in section \ref{sec:deg6.pureinv.1}, \ref{sec:deg6.pureinv.2}, and \ref{sec:deg6.pureinv.3}).
Using graphs, these formulas can be given for arbitrary number of subsystems.
\item[\ref{sec:deg6.mixinv}]
we discuss the connection of pure and mixed quantum states from another point of view
and we show the formulas for mixed states (in section \ref{sec:deg6.mixinv.1}, \ref{sec:deg6.mixinv.2}, and \ref{sec:deg6.mixinv.3}).
\item[\ref{sec:deg6.alg}]
we give an algorithm for the construction of the 
labelling of different invariant polynomials of degree $6$.
(For degree $2$ and $4$, this task is trivial.)
\item[\ref{sec:deg6.summary}]
we give a summary and some remarks.
\end{organization}

%*******************************************************************************
%*******************************************************************************
\section{Local unitary invariant polynomials}
\label{sec:deg6.luinvs}

Let us start with the general way of writing of LU-invariant homogeneous polynomials.
As usual, let $\mathcal{H}=\mathcal{H}_1\otimes\dots\otimes\mathcal{H}_n$ be
the Hilbert space of a $n$-partite composite system
of local dimensions $\tpl{d}=(d_1,\dots,d_n)$.
First we consider LU-invariant homogeneous polynomials for pure states,
in which case we give the polynomials in terms of the state vector $\cket{\psi}$ 
rather than in terms of pure states $\pi=\cket{\psi}\bra{\psi}$.
We do this because 
this allows us to handle the labelling of polynomials being different in general,
although it turns out that all of such polynomials can be written also in the terms of pure states.

%*******************************************************************************
\subsection{Invariants for pure states}
\label{subsec:deg6.luinvs.pure}
%As before,
%let $\mathcal{H}=\mathcal{H}_1\otimes\dots\otimes\mathcal{H}_n$ be
%the Hilbert space of a $n$-partite composite system,
%where $\dim \mathcal{H}_j=d_j$
%and $\tpl{d}$ denotes the $n$-tuple of these local dimensions: $\tpl{d}=(d_1,\dots,d_n)$.
A state vector can be written as
\begin{equation*}
\cket{\psi}=\psi^{i_1\dots i_n}\cket{i_1\dots i_n} \in \mathcal{H},
\end{equation*}
where $\cket{i_j}\in\mathcal{H}_j$ for $i_j=1,\dots,d_j$ is an orthonormal basis for all $j=1,\dots,n$,
and the summation over $i_j=1,\dots,d_j$ is understood.
As usual in the topic of quantum invariants, the norm of $\psi$ does not have to be fixed.

It is well-known (see e.g.~in~\cite{Sudbery3qb}) that the way to get local unitary invariant polynomials is the following.
We write down the term%
%%%%%%%%%%%%%%%%%%%%%%%%
\footnote{Note our conventions: $\psi_{i_1\dots i_n}=\cc{(\psi^{i_1\dots i_n})}$.}
%%%%%%%%%%%%%%%%%%%%%%%%
$\bigl(\psi^{i_1\dots i_n}\psi_{i'_1\dots i'_n}\bigr)$ $m$ times
(with different indices)
and contract all primed indices with unprimed indices on the same $\mathcal{H}_j$.
A polynomial obtained in this way is of degree $2m$,
that is, degree $m$ in the coefficients and also in their complex conjugates.
This is the only case in which unitary invariants can arise~\cite{PetiLUA23},
so it is convenient to use this natural gradation,
and to call this polynomial of \emph{grade} $m$.
(In the case of mixed states the grade coincides with the degree in the matrix-elements of the density matrix.)
The possible index-contractions on an $\mathcal{H}_j$ are encoded by the elements of $\DscrGrp{S}_m$,
the group of permutations of $m$ letters.
A $\sigma_j\in \DscrGrp{S}_m$ tells us that the primed index of the $l$th term is contracted with the unprimed index of the $\sigma_j(l)$th term,
so there is an index-contraction scheme for all $n$-tuples of permutations 
$\tpls{\sigma}=(\sigma_1,\dots,\sigma_n) \in \DscrGrp{S}_m^n$, written as
\begin{equation}
\label{eq:purinv0}
%f_{(\sigma_1,\dots,\sigma_n)}(\psi)=
f_{\tpls{\sigma}}(\psi)=
\psi^{i_1^1\dots i_n^1}\cdots
\psi^{i_1^m\dots i_n^m}
\psi_{i_1^{\sigma_1(1)}\dots i_n^{\sigma_{n}(1)}}\cdots
\psi_{i_1^{\sigma_1(m)}\dots i_n^{\sigma_{n}(m)}},
\end{equation}
where the summation
over $i_j^l=1,\dots,d_j$ for all $j=1,\dots,n$ and $l=1,\dots,m$ is understood.%
%%%%%%%%%%%%%%%%%%%%%%%%
\footnote{The lower labels of $i$s refer to the subsystems
and the upper ones refer to the different index-contractions.}
%%%%%%%%%%%%%%%%%%%%%%%%

However, different $n$-tuples of permutations can give rise to the same polynomial.
We have the terms $\bigl(\psi^{i_1^l\dots}\psi_{i_1^{\sigma_1(l)}\dots}\bigr)$ $m$ times,
\begin{equation*}
 \Bigl(\psi^{i_1^1\dots}\psi_{i_1^{\sigma_1(1)}\dots}\Bigr)
 \Bigl(\psi^{i_1^2\dots}\psi_{i_1^{\sigma_1(2)}\dots}\Bigr)\dots
 \Bigl(\psi^{i_1^m\dots}\psi_{i_1^{\sigma_1(m)}\dots}\Bigr), 
\end{equation*}
but it makes no difference if we permute the $\psi^{i_1^l\dots}$s or $\psi_{i_1^{\sigma_1(l)}\dots}$s 
among these terms, since, being scalar variables, they commute.
This is equivalent to the relabelling of the indices (in the upper labels),
which can be formulated by the permutations $\alpha,\beta\in \DscrGrp{S}_m$
encoding the permutations of $\psi_{i_1^{\sigma_1(l)}\dots}$s and $\psi^{i_1^l\dots}$s, respectively, as
\begin{align*}
& \Bigl(\psi^{i_1^1\dots}\psi_{i_1^{\sigma_1(1)}\dots}\Bigr)
  \Bigl(\psi^{i_1^2\dots}\psi_{i_1^{\sigma_1(2)}\dots}\Bigr)\dots
  \Bigl(\psi^{i_1^m\dots}\psi_{i_1^{\sigma_1(m)}\dots}\Bigr)\\
&=\Bigl(\psi^{i_1^{\beta(1)}\dots}\psi_{i_1^{\alpha\sigma_1(1)}\dots}\Bigr)
  \Bigl(\psi^{i_1^{\beta(2)}\dots}\psi_{i_1^{\alpha\sigma_1(2)}\dots}\Bigr)\dots
  \Bigl(\psi^{i_1^{\beta(m)}\dots}\psi_{i_1^{\alpha\sigma_1(m)}\dots}\Bigr)\\
&=\Bigl(\psi^{i_1^1\dots}\psi_{i_1^{\alpha\sigma_1\beta^{-1}(1)}\dots}\Bigr)
  \Bigl(\psi^{i_1^2\dots}\psi_{i_1^{\alpha\sigma_1\beta^{-1}(2)}\dots}\Bigr)\dots
  \Bigl(\psi^{i_1^m\dots}\psi_{i_1^{\alpha\sigma_1\beta^{-1}(m)}\dots}\Bigr).
\end{align*}
(Here we have written out only the indices on $\mathcal{H}_1$ to get shorter expressions,
but, obviously, the same $\alpha$ and $\beta$ work on every index running on every $\mathcal{H}_j$.)
Therefore we have
\begin{equation}
f_{(\sigma_1,\dots,\sigma_n)}(\psi)=
f_{(\alpha\sigma_1\beta^{-1},\dots,\alpha\sigma_n\beta^{-1})}(\psi),
\end{equation}
giving rise to an equivalence relation on $\DscrGrp{S}_m^n$:
\begin{equation}
%(\sigma_1,\dots,\sigma_n)\sim (\sigma'_1,\dots,\sigma'_n)$ if and only if $\exists \alpha,\beta\in \DscrGrp{S}_m: \sigma'_j=\alpha\sigma_j\beta^{-1},  j=1,\dots,n.
\tpls{\sigma}\sim\tpls{\sigma}' \qquad\defn\qquad 
\exists \alpha,\beta\in \DscrGrp{S}_m:\; \sigma'_j=\alpha\sigma_j\beta^{-1},\; j=1,\dots,n
\end{equation}
and the equivalence classes are denoted by
\begin{equation*}
[\tpls{\sigma}]_\sim=[\sigma_1,\dots,\sigma_n]_\sim
=\{(\alpha\sigma_1\beta^{-1},\dots,\alpha\sigma_n\beta^{-1})\mid \alpha,\beta\in \DscrGrp{S}_m\}.
\end{equation*}
The set of these equivalence classes is the double-cosets of $\DscrGrp{S}_m^n$ by the diagonal action,
denoted by $\Delta \backslash \DscrGrp{S}_m^n /\Delta$,
 where the subgroup $\Delta=\{(\delta,\dots,\delta)\mid \delta\in \DscrGrp{S}_m\}\subseteq \DscrGrp{S}_m^n$.

Thus, the ambiguity arising from the commutativity of the $m$ terms $\psi^{\dots}$ and $\psi_{\dots}$ 
in (\ref{eq:purinv0}) has been handled
by the labelling of the polynomials by the elements of $\Delta \backslash \DscrGrp{S}_m^n /\Delta$.
As a next step, it would be desirable to get one representing element for every equivalence class.
Unfortunately, this can not be done generally, (i.e., for an arbitrary $m$,)
but we can make the equivalence classes smaller by throwing off some of their elements in a general way.
Every equivalence class has elements having the identity permutation $e$ in the last position.
Indeed, we have $\alpha\sigma_n\beta^{-1}=e$ 
in $(\alpha\sigma_1\beta^{-1},\dots,\alpha\sigma_n\beta^{-1})$ if we set $\alpha=\beta\sigma_n^{-1}$, since
\begin{align*}
&\bigl(\alpha\sigma_1\beta^{-1},\dots,\alpha\sigma_{n-1}\beta^{-1},\alpha\sigma_n\beta^{-1}\bigr)\sim\\
&\qquad
 \bigl(\beta\sigma_n^{-1}\sigma_1\beta^{-1},\dots,\beta\sigma_n^{-1}\sigma_{n-1}\beta^{-1},e\bigr)=
 \bigl(\beta\sigma'_1\beta^{-1},\dots,\beta\sigma'_{n-1}\beta^{-1},e\bigr),
\end{align*}
which is actually an orbit of $\DscrGrp{S}_m^{n-1}\times\{e\}$ under the action of simultaneous conjugation.
So it is useful to define another equivalence relation on $\DscrGrp{S}_m^n$:
\begin{equation}
\tpls{\sigma}\approx\tpls{\sigma}' \qquad\defn\qquad 
\exists \beta\in \DscrGrp{S}_m:\; \sigma'_j=\beta\sigma_j\beta^{-1},\; j=1,\dots,n
\end{equation}
and the equivalence classes are denoted by
\begin{equation*}
[\tpls{\sigma}]_\approx=[\sigma_1,\dots,\sigma_n]_\approx
=\{(\beta\sigma_1\beta^{-1},\dots,\beta\sigma_n\beta^{-1})\mid \beta\in \DscrGrp{S}_m\}.
\end{equation*}
The set of these equivalence classes is denoted by $\DscrGrp{S}_m^n/\DscrGrp{S}_m$.
This equivalence is defined on $\DscrGrp{S}_m^{n-1}$ in the same way.
$\DscrGrp{S}_m^{n-1}$ can be injected into $\DscrGrp{S}_m^n$ by
$\imath:\DscrGrp{S}_m^{n-1}\hookrightarrow \DscrGrp{S}_m^n$ as $\imath(\sigma_1,\dots,\sigma_{n-1})=(\sigma_1,\dots,\sigma_{n-1},e)$,
which is compatible with the equivalence $\approx$, but not with $\sim$.
Note that 
\begin{equation*}
\tpls{\sigma}\approx\tpls{\sigma}'  \qquad\Longrightarrow\qquad \tpls{\sigma}\sim\tpls{\sigma}',
\end{equation*}
therefore a $\sim$-equivalence class is the union of disjoint $\approx$-equivalence classes
\begin{equation}
\label{eq:classdecomp}
[\tpls{\sigma}]_\sim=[\tpls{\sigma}^{(1)}]_\approx\cup[\tpls{\sigma}^{(2)}]_\approx\cup\dots
\end{equation}
The elements of $[\tpls{\sigma}]_\sim$ which have $\sigma_n=e$ form one of the $\approx$-equivalence classes of the right-hand side.
This $\approx$-equivalence class (element of $\DscrGrp{S}_m^{n-1}/\DscrGrp{S}_m$) is also suitable for the labelling of the polynomials
instead of the original $\sim$--equivalence class (element of $\Delta \backslash \DscrGrp{S}_m^n /\Delta$).

The meaning of the choice $\alpha\sigma_n\beta^{-1}=e$ is that
the indices on $\mathcal{H}_n$ are contracted \emph{inside}
every term $(\psi^{i_1^l\dots i_n^l}\psi_{i_1^{\sigma_1(l)}\dots i_n^l})$.
This ``couples together'' the pairs of $\psi^{\dots}$ and $\psi_{\dots}$.
The simultaneous conjugation means the permutation of the $m$ terms $(\psi^{\dots}\psi_{\dots})$,
which is the remaining ambiguity arising from the commutativity of these terms.
Note, that we have singled out the last Hilbert space $\mathcal{H}_n$ in this construction.
In the general aspects, it makes no difference which Hilbert space is singled out,
but as we write the pure-state invariants using matrix operations,
it can happen---and usually it will happen---that this freedom manifests itself
in the different writings of the same pure state invariant.

Summing up,
for a composite system of $n$ subsystems,
the LU-invariant polynomial given by $[\sigma_1,\dots,\sigma_{n-1}]_\approx \in \DscrGrp{S}_m^{n-1}/\DscrGrp{S}_m$ is
\begin{equation}
\label{eq:purinv}
f_{[\sigma_1,\dots,\sigma_{n-1}]_\approx}(\psi)=
\psi^{i_1^1\dots i_n^1}\cdots
\psi^{i_1^m\dots i_n^m}
\psi_{i_1^{\sigma_1(1)}\dots i_{n-1}^{\sigma_{n-1}(1)} i_n^1}\cdots
\psi_{i_1^{\sigma_1(m)}\dots i_{n-1}^{\sigma_{n-1}(m)} i_n^m }.
\end{equation}
By the use of the $\DscrGrp{S}_m^{n-1}/\DscrGrp{S}_m$ labelling,
we have got rid of the formal equivalence of polynomials arising from the commutativity of the terms,
and have got a set of LU-invariant polynomials for the elements of the set $\DscrGrp{S}_m^{n-1}/\DscrGrp{S}_m$.
Can it happen that different elements of $\DscrGrp{S}_m^{n-1}/\DscrGrp{S}_m$ give the same polynomial?
Are there linear dependencies among these polynomials?
It is not known in general,
but sometimes there is more to be known:
(\ref{eq:purinv}) gives a linearly independent basis in each $m$ graded subspace of the inverse limit of the algebras~\cite{HWWLUA}.
Moreover,---as the main result of~\cite{PetiLUA23} states,---% 
an \emph{algebraically independent generating set} is formed by
the polynomials given in (\ref{eq:purinv}) 
for which the defining $n-1$ permutations \emph{together} act transitively on the set of $m$ labels.
For the algebras of given local dimensions $\tpl{d}=(d_1,\dots,d_n)$, 
the above polynomials form a basis as long as $m\leq d_j$ (for all $j$),
otherwise they become linearly dependent.
The algebraic independency also fails if we restrict ourselves to given local dimensions.
(The algebra of LU-invariant polynomials is usually not even free for given local dimensions.)

%*******************************************************************************
\subsection{Invariants for mixed states}
\label{subsec:deg6.luinvs.mixed}

Now consider a mixed quantum state of the $n$-partite composite system.
This state is given by the density operator, acting on $\mathcal{H}$,
written as
\begin{equation*}
\varrho=\varrho^{i_1\dots i_n}_{\phantom{i_1\dots i_n}i'_1\dots i'_n}\cket{i_1\dots i_n}\bra{i'_1\dots i'_n}
\in\mathcal{D}(\mathcal{H}).
\end{equation*}
This density operator, by definition, a positive semidefinite self adjoint operator,
but, as usual in the topic of quantum invariants, the trace of $\varrho$ does not have to be fixed.

The general form of an LU-invariant polynomial is given by a simillar index-contraction scheme,
 encoded by $\tpls{\sigma}=(\sigma_1,\dots,\sigma_n)\in \DscrGrp{S}_m^n $, as in the case of pure states:
\begin{equation}
\label{eq:mixinv0}
%f_{(\sigma_1,\dots,\sigma_n)}(\varrho)=
f_{\tpls{\sigma}}(\varrho)=
\varrho^{i_1^1\dots i_n^1}_{\phantom{i_1^1\dots i_n^1}i_1^{\sigma_1(1)}\dots i_n^{\sigma_{n}(1)}}\cdots
\varrho^{i_1^m\dots i_n^m}_{\phantom{i_1^m\dots i_n^m}i_1^{\sigma_1(m)}\dots i_n^{\sigma_{n}(m)}},
\end{equation}
where the summation
over $i_j^l=1,\dots,d_j$ for all $j=1,\dots,n$ and $l=1,\dots,m$ is understood again.
(We denote the pure and the mixed state invariants with the same symbol,
the distinction between them is their arguments, which are vectors and matrices, respectively.)

Here we can carry out a similar construction as in the case of pure states,
with one difference.
Namely, the building blocks of the polynomials 
are the $(\varrho^{i_1\dots  i_n}_{\phantom{i_1\dots  i_n}i'_1\dots i'_n})$ matrix-elements of the density operator
instead of the former $(\psi^{i_1\dots i_n}\psi_{i'_1\dots i'_n})$s.
Hence there is no step corresponding to the ``double coset'' construction,
because we can not move the ``two parts'' of $\varrho$ independently as has been done in the case of $\psi^{\dots}\psi_{\dots}$,
since in general $\varrho$ is not of rank one.
This means that we can not relabel the primed and unprimed indices independently.
The possible relabelling is given by a single $\beta\in \DscrGrp{S}_m$, as
\begin{align*}
& \Bigl(\varrho^{i_1^1\dots}_{\phantom{i_1^1\dots}i_1^{\sigma_1(1)}\dots}\Bigr)
  \Bigl(\varrho^{i_1^2\dots}_{\phantom{i_1^2\dots}i_1^{\sigma_1(2)}\dots}\Bigr)\dots
  \Bigl(\varrho^{i_1^m\dots}_{\phantom{i_1^m\dots}i_1^{\sigma_1(m)}\dots}\Bigr)\\
&=\Bigl(\varrho^{i_1^{\beta(1)}\dots}_{\phantom{i_1^{\beta(1)}\dots}i_1^{\beta\sigma_1(1)}\dots}\Bigr)
  \Bigl(\varrho^{i_1^{\beta(2)}\dots}_{\phantom{i_1^{\beta(2)}\dots}i_1^{\beta\sigma_1(2)}\dots}\Bigr)\dots
  \Bigl(\varrho^{i_1^{\beta(m)}\dots}_{\phantom{i_1^{\beta(m)}\dots}i_1^{\beta\sigma_1(m)}\dots}\Bigr)\\
&=\Bigl(\varrho^{i_1^1\dots}_{\phantom{i_1^1\dots}i_1^{\beta\sigma_1\beta^{-1}(1)}\dots}\Bigr)
  \Bigl(\varrho^{i_1^2\dots}_{\phantom{i_1^2\dots}i_1^{\beta\sigma_1\beta^{-1}(2)}\dots}\Bigr)\dots
  \Bigl(\varrho^{i_1^m\dots}_{\phantom{i_1^m\dots}i_1^{\beta\sigma_1\beta^{-1}(m)}\dots}\Bigr).
\end{align*}
Therefore we have
\begin{equation}
f_{(\sigma_1,\dots,\sigma_n)}(\varrho)=
f_{(\beta\sigma_1\beta^{-1},\dots,\beta\sigma_n\beta^{-1})}(\varrho),
\end{equation}
that is, the elements of the orbits in $\DscrGrp{S}_m^n$ under the action of simultaneous conjugation
give the same polynomial.
Let these orbits be denoted by
$[\sigma_1,\dots,\sigma_n]_\approx$,
as before, and the LU-invariant polynomial given by this is
\begin{equation}
\label{eq:mixinv}
f_{[\sigma_1,\dots,\sigma_n]_\approx}(\varrho)=
%\sum_{\substack{i_1^1,\dots,i_n^1\\ \vdots\\i_1^m,\dots,i_n^m}}
\varrho^{i_1^1\dots i_n^1}_{\phantom{i_1^1\dots i_n^1}i_1^{\sigma_1(1)}\dots i_n^{\sigma_n(1)}}\cdots
\varrho^{i_1^m\dots i_n^m}_{\phantom{i_1^m\dots i_n^m}i_1^{\sigma_1(m)}\dots i_n^{\sigma_n(m)}}.
\end{equation}

The independence of these follows from the independence of the pure state invariants.
%when that is the case for the latter ones.
This is because we can obtain the independent mixed state invariants
of the system with local dimensions $\tpl{d}=(d_1,\dots,d_n)$,
if we add a large enough $\mathcal{H}_{n+1}$ Hilbert space,
and calculate the invariants (\ref{eq:purinv}) for a state vector $\cket{\phi}\in \mathcal{H}\otimes\mathcal{H}_{n+1}$.
(See \cite{PetiLUA23} for the abstract construction.)
Since in (\ref{eq:purinv}) we have not permuted the last (this time $n+1$th) indices, 
we can read off the invarians for $\varrho=\tr_{n+1}\cket{\phi}\bra{\phi}$ from (\ref{eq:purinv}).
(If $\dim\mathcal{H}_{n+1}\geq\prod_{j=1}^n\dim\mathcal{H}_j$, then $\varrho$ can be of full rank,
and we can get all $\varrho$ acting on $\mathcal{H}$ in this way.)
%******
Note that if we simply substitute $\varrho$ by a pure state $\cket{\psi}\bra{\psi}$ in (\ref{eq:mixinv}),
then we do not get a linearly independent set of $n$-partite pure state invariants
for all the labels $[\sigma_1,\dots,\sigma_n]_\approx\in \DscrGrp{S}_m^n/\DscrGrp{S}_m$.
However, if we restrict this for the case when $\sigma_n=e$, then we get back the
linearly independent set of pure state invariants 
from the linearly independent set of mixed state ones of a $n$-partite system,
\begin{equation}
\label{eq:purmix}
 f_{[\sigma_1,\dots,\sigma_{n-1}]_\approx}(\psi)
=      f_{[\sigma_1,\dots,\sigma_{n-1},e]_\approx}\bigl(      \cket{\psi}\bra{\psi}\bigr)
\equiv f_{[\sigma_1,\dots,\sigma_{n-1}  ]_\approx}\bigl(\tr_n \cket{\psi}\bra{\psi}\bigr).
\end{equation}

%*******************************************************************************
%*******************************************************************************
\section{Graphs and matrix operations}
\label{sec:deg6.graphsops}

Having obtained the generally different LU-invariant homogeneous polynomials,
we get the labels of those as equivalence-classes of tuples of permutations,
which encode the structure of the polynomials.
In this section we see that the natural treatment of these, 
hence that of the invariants themselves,
are given in the terms of unlabelled graphs.

%*******************************************************************************
\subsection{Graphs of invariants}
\label{sec:deg6.graphsops.invargraphs}
The index-contraction scheme of the LU-invariant polynomials given in the previous section
can be made more expressive by the use of graphs~\cite{HWWLUA,PetiLUA23}.
Let us start with pure state invariants.
For a grade $m$ invariant, given by $\tpls{\sigma}=(\sigma_1,\dots,\sigma_n)\in \DscrGrp{S}_m^n$,
one can draw a graph with $m$ vertices 
with the labels $l=1,\dots,m$.
These vertices represent 
the $m$ terms $(\psi^{i_1^l\dots i_n^l}\psi_{i_1^{\sigma_1(l)}\dots i_n^{\sigma_n(l)}})$
in (\ref{eq:purinv0}).
The edges of the graph are directed and coloured with $n$ different colours.
The edges of the $j$th colour encode
the index contractions on the $j$th Hilbert space
given by the permutation $\sigma_j$ of $(\sigma_1,\dots,\sigma_n)$:
for every $l=1,\dots,m$ there is an edge with head and tail on the $l$th and $\sigma_j(l)$th vertex, respectively,
meaning a contracted $j$th index of the $l$th $\psi^{\dots}$ and the $\sigma_j(l)$th $\psi_{\dots}$.

Some elements of $\DscrGrp{S}_m^n$ give rise to the same polynomials, as was elaborated in the previous section.
How can we tell that story in the language of graphs?
Following the previous section, for a given $\tpls{\sigma}\in \DscrGrp{S}_m^n$,
the elements of the $\sim$-equivalence class $[\tpls{\sigma}]_\sim \in \Delta\backslash \DscrGrp{S}_m^n/\Delta$ give the same invariant.
First we set $\sigma_n=e$,
which means that we select the graphs having a loop of colour $n$ on every vertex.
For a given $\sim$-equivalence class, there are still many graphs of that kind,
and they are given by the elements of the corresponding $\approx$-equivalence class.
How are these graphs related to each other?
A simultaneous conjugation by a $\beta\in \DscrGrp{S}_m$ means the relabelling of the vertices,
that is, the relabelling of the indices (in the upper label) 
of the terms $\bigl(\psi^{i_1^l\dots i_n^l}\psi_{i_1^{\sigma_1(l)}\dots i_n^{\sigma_n(l)}}\bigr)$.
So the elements of a $\approx$-equivalence class give the same graph with all the possible labellings,
and the $\approx$-equivalence class itself gives an \emph{unlabelled} graph.

Since the elements of $\DscrGrp{S}_m^n$ related by simultaneous conjugation give rise to the same unlabelled graph,
the decomposition in (\ref{eq:classdecomp}) shows that there may exist
many unlabelled graphs (many $\approx$-classes) giving rise to the same polynomial defined by a given $\sim$-class.
For example, 
we can set $\sigma_j=e$ for a $j\neq n$,
which results graphs where the edges of colour $j\neq n$ form loops on every vertex.
On the other hand, there may be $\approx$-classes in the given $\sim$-class which does not contain $e$.
All of these graphs give the same polynomial, but it can happen that
some of them can be formulated using matrix operations (in different ways for different graphs)
and some of them not.
(It turns out (see in next section) that every polynomial can be formulated using matrix operations up to $m=3$.)

The case of mixed states is simpler because there are no $\sim$-classes involved.
For a grade $m$ invariant given by $\tpls{\sigma}\in \DscrGrp{S}_m^n$,
the vertices represent the terms $\bigl(\varrho^{i_1^l\dots i_n^l}_{\phantom{i_1^l\dots i_n^l}i_1^{\sigma_1(l)}\dots i_n^{\sigma_n(l)}}\bigl)$,
and only the polynomials given by the elements of the $\approx$-equivalence class $[\tpls{\sigma}]_\approx\in \DscrGrp{S}_m^n/\DscrGrp{S}_m$
are the same by the commutativity of these terms.
This means that we simply omit the labelling of the vertices of the graph given by $\tpls{\sigma}$.

%*******************************************************************************
\subsection{Graphs of matrix operations}
\label{sec:deg6.graphsops.mxopgraphs}

The 
$\bigl(\psi^{i_1\dots i_n}\psi_{i'_1\dots i'_n}\bigr)$ and $\bigl(\varrho^{i_1\dots i_n}_{\phantom{i_1\dots i_n}i'_1\dots i'_n}\bigr)$
building blocks
of the polynomials
can be regarded as matrices 
with row and column multiindices being the unprimed and primed ones, respectively.
So we expect that some of the invariant polynomials can be written using only matrix operations,
such as partial trace, matrix products, tensorial products or partial transpose.
How can we read off matrix operations from the graphs corresponding to the invariant polynomials?
This is a difficult question in general and, as we will see, not all graphs can be encoded using matrix operations.
It is more instructive to look at the graphs coming from the matrix operations first,
and then to search for these elementary subgraphs in a general graph coming from a polynomial given by an element of $\DscrGrp{S}_m^n$.

Let us see some matrix operations and their graphs.
The matrix multiplication means 
contraction of the column indices of the first matrix with the row indices of the second matrix,
the trace means contraction of the column indices with the row indices
and the partial transposition means the swap of the given row and column indices.
First consider only the indices belonging to the Hilbert space of only the first subsystem, that is, we have edges of only one colour.
For a general matrix $M$, which is represented by the vertices of the graphs, 
the multiplicaton by itself gives the edge from one vertex to another,
the $r$th power $M^r$ is a chain of edges (without loops),
and the trace of it closes this chain into a loop. (See in the first row of figure~\ref{fig:mxopgraphs}.)
Now let us take into account indices belonging to the second subsystem.
Then $\tr M^r$ is the same loop as before but with doubled edges,
while the partial transposition $\tr(M^{\transp_2})^r$ reverses the loop of the corresponding colour
(second row of figure~\ref{fig:mxopgraphs}).
The partial traces in $\tr \tr_1 (M^r) \tr_1 (M^s)$  make smaller loops on a subsystem
(third row of figure~\ref{fig:mxopgraphs}).
There is a little trick, which is proved to be very useful later: 
$\tr (\tr_1 M^2) (\tr_1 M)=\tr M^2(\Id_1\otimes \tr_1 M)$.
In the language of graphs we just bend the corresponding edge next to the vertex representing $\tr_1 M$,
and we draw a circle on it, representing the identity matrix, which is just contract indices
(last row of figure~\ref{fig:mxopgraphs}).
If the graph is the union of disjoint graphs, then the corresponding polynomial is factorizable, 
since the summations corresponding to the disjoint pieces can be carried out independently.
This almost trivial situation is getting more complicated,
if we take into account the indices of all the subsystems, that is, the edges of all colours.
Examples are shown in the next section.

\begin{figure}[ht]
\includegraphics{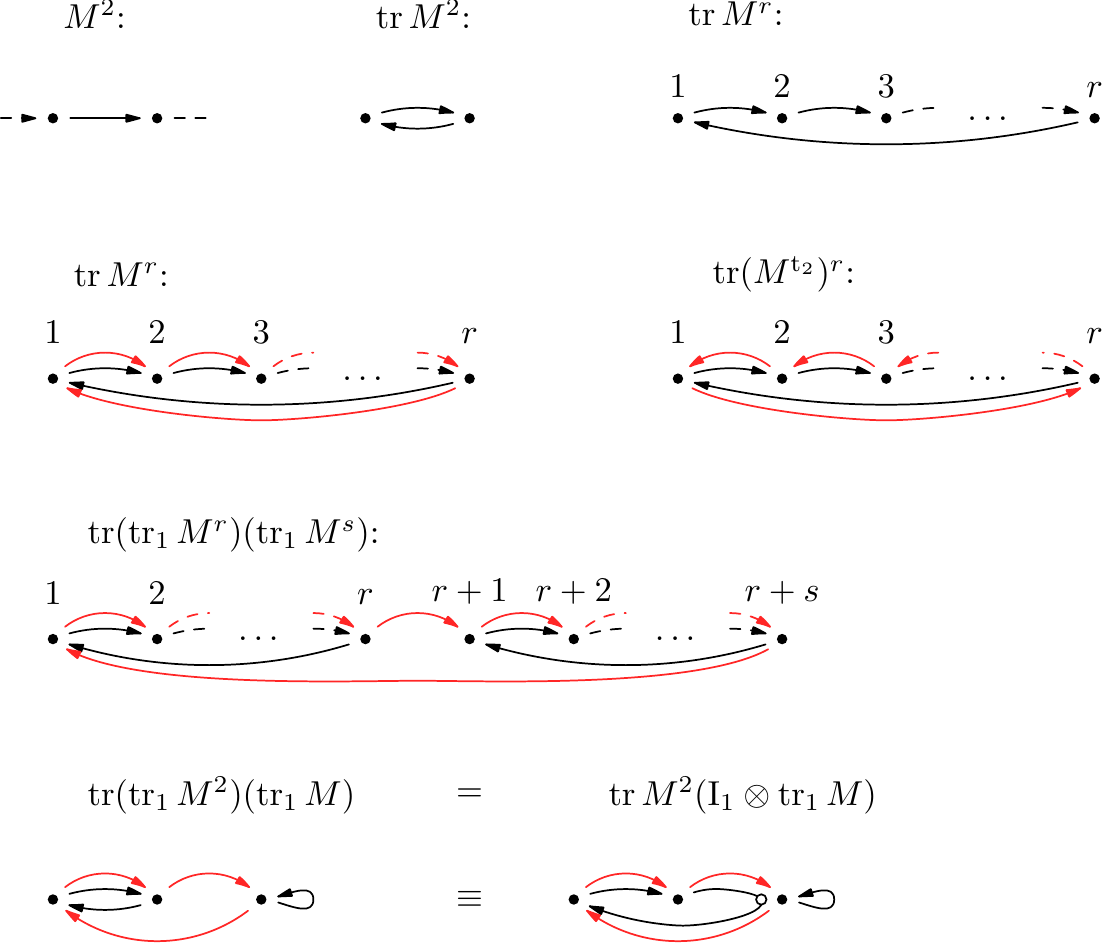}\centering
\caption{Elementary matrix operations represented by graphs.
In the first row: $n=1$, there is only one colour of edges representing the index contractions.
In the other rows: $n=2$, two different colours of edges correspond to the index contractions on the two Hilbert spaces,
black and red on $\mathcal{H}_1$ and $\mathcal{H}_2$, respectively.}\label{fig:mxopgraphs}
\end{figure}

%*******************************************************************************
%*******************************************************************************
\section{Pure state invariants}
\label{sec:deg6.pureinv}

In the following, we illustrate how a pure state LU-invariant polynomial 
(encoded by $[\tpls{\sigma}]_\sim\in\Delta \backslash \DscrGrp{S}_m^n /\Delta$)
is given by different unlabelled graphs (encoded by $[\tpls{\sigma}]_\approx\in \DscrGrp{S}_m^n/\DscrGrp{S}_m$).
(While an unlabelled graph 
is given by different labelled graphs (index-contraction scheme, encoded by $\tpls{\sigma}\in \DscrGrp{S}_m^n$).)
On the other hand,
different unlabelled graphs give rise to different writings by matrix operations of the same polynomial.
The polynomials are labelled here by the elements of $\DscrGrp{S}_m^{n-1}/\DscrGrp{S}_m$ 
instead of the elements of $\Delta \backslash \DscrGrp{S}_m^n /\Delta$,
which give special unlabelled graphs having loops of colour $n$ on every vertex.
(Sometimes, e.g., in \cite{HWWLUA}, these loops are omitted, and only the first $n-1$-coloured edges are drawn.)
For a permutation $n-1$-tuple $\tpls{\sigma}\in \DscrGrp{S}_m^{n-1}$,
$[\tpls{\sigma}]_\approx\in \DscrGrp{S}_m^{n-1}/\DscrGrp{S}_m$,
and we can write for the corresponding invariant
$[\imath(\tpls{\sigma})]_\sim
=[\imath(\tpls{\sigma})]_\approx\cup
[\tpls{\sigma}^{(2)}]_\approx\cup
[\tpls{\sigma}^{(3)}]_\approx\cup\dots$,
where $\tpls{\sigma}^{(2)},\tpls{\sigma}^{(3)},\dots\in \DscrGrp{S}_m^n$ are
representing elements of $\approx$-classes giving different graphs for the same invariant.

Let us see how these technics work.
As a warm-up, we show for all $n$ the trivial case of $m=1$ and the almost trivial case of $m=2$.
This is followed by the case of $m=3$, which is more interesting
because of the non-Abelian structure of $\DscrGrp{S}_3$.
This is done for all $n$ too.
For $\cket{\psi}\in\mathcal{H}$,
as we have seen, everything can be formulated using the
rank-one density matrix $\pi\equiv\pi_{12\dots n}=\cket{\psi}\bra{\psi}$.
As usual, we denote the reduced density matrices with the label of subsystems
which are not traced out, 
for example $\pi_{2\dots n}=\tr_1 \pi_{12\dots n}$, and so on.

%*******************************************************************************
\subsection{\texorpdfstring{Invariant polynomials of grade $m=1$ (degree $2$)}{Invariant polynomials of grade m=1 (degree 2)}}
\label{sec:deg6.pureinv.1}
For $m=1$, we have the trivial $\DscrGrp{S}_1=\{e\}$,
and for all $n$ number of subsystems $[e,e,\dots,e]_\sim=[e,e,\dots,e]_\approx$,
so $\Delta \backslash \DscrGrp{S}_1^n /\Delta\isom \DscrGrp{S}_1^n/\DscrGrp{S}_1\isom \DscrGrp{S}_1$,
meaning only one kind of graphs, having only one vertex.
Every edge---of $n$ different colors for the $n$ subsystems---starts and ends here (figure~\ref{fig:m1k}).
These graphs mean a simple trace, which is the only possible index contraction.
The label of the polynomial is the only one element of $\DscrGrp{S}_1^{n-1}/\DscrGrp{S}_1\isom \DscrGrp{S}_1$, and
\begin{equation}
\label{eq:purinv1k}
f_{[e,\dots,e]_\approx}(\psi) =  \tr\pi_{12\dots n} =\Vert\psi\Vert^2.
\end{equation}
\begin{figure}%[ht]
\includegraphics{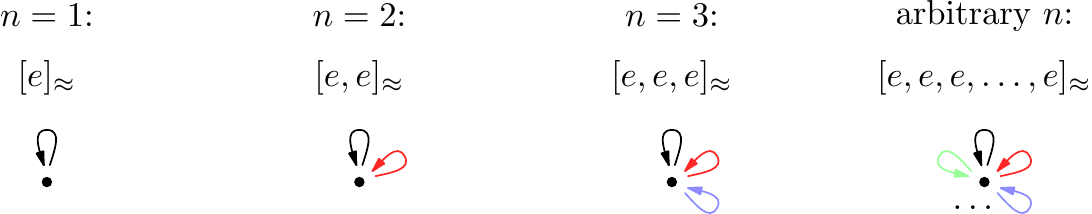}\centering
\caption{Graphs corresponding to the $m=1$ invariant polynomials.
 Black, red, blue and green edges
 represent index-contractions on the first, second, third and last Hilbert spaces, respectively.}\label{fig:m1k}
\end{figure}

%*******************************************************************************
\subsection{\texorpdfstring{Invariant polynomials of grade $m=2$ (degree $4$)}{Invariant polynomials of grade m=2 (degree 4)}}
\label{sec:deg6.pureinv.2}
For $m=2$, we have%
%%%%%%%%%%%%%%%%%%%%%%%%
\footnote{Denote the permutations with $e=(1)(2)$ and $t=(12)$.} 
%%%%%%%%%%%%%%%%%%%%%%%%
$\DscrGrp{S}_2=\{e,t\}$
with the conjugacy-classes $[e]$ and $[t]$,
so the labels of the polynomials are $\DscrGrp{S}_2^{n-1}/\DscrGrp{S}_2\isom \DscrGrp{S}_2^{n-1}$ for all $n$.
On the other hand, $[\tpls{\sigma}]_\sim=[\tpls{\sigma}]_\approx\cup [\overline{\tpls{\sigma}}]_\approx $,
(where $\overline{\tpls{\sigma}}_i= \overline{ {\sigma}_i}$, and $\overline{t}=e$, $\overline{e}=t$)
so there are two kinds of graphs for every polynomial.

For singlepartite system ($n=1$, $\pi\equiv\pi_1$), 
the only polynomial is given by
\begin{equation*}
 {}[e]_\sim=[e]_\approx\cup[t]_\approx.
\end{equation*}
From its graphs (figure~\ref{fig:m2k12}) we have
\begin{equation*}
 f_{[]_\approx}(\psi) = (\tr\pi_1)^2 = \tr\pi_1^2 = \Vert\psi\Vert^4.
\end{equation*}

For bipartite system ($n=2$, $\pi\equiv\pi_{12}$),
there are two linearly independent polynomials. These are given by
\begin{align*}
 {}[e,e]_\sim &=[e,e]_\approx\cup[t,t]_\approx,\\
 {}[t,e]_\sim &=[t,e]_\approx\cup[e,t]_\approx.
\end{align*}
From their graphs (figure~\ref{fig:m2k12}) we have
\begin{align*}
 f_{[e]_\approx}(\psi) &= (\tr\pi_{12})^2 = \tr\pi_{12}^2 = \Vert\psi\Vert^4,\\
 f_{[t]_\approx}(\psi) &= \tr\pi_1^2 = \tr\pi_2^2.
\end{align*}

For tripartite system ($n=3$, $\pi\equiv\pi_{123}$),
there are four linearly independent polynomials. These are given by
\begin{align*}
 {}[e,e,e]_\sim &=[e,e,e]_\approx\cup[t,t,t]_\approx,\\
 {}[e,t,e]_\sim &=[e,t,e]_\approx\cup[t,e,t]_\approx,\\
 {}[t,e,e]_\sim &=[t,e,e]_\approx\cup[e,t,t]_\approx,\\
 {}[t,t,e]_\sim &=[t,t,e]_\approx\cup[e,e,t]_\approx.
\end{align*}
From their graphs we have
\begin{align*}
 f_{[e,e]_\approx}(\psi) &= (\tr\pi_{123})^2 = \tr\pi_{123}^2 = \Vert\psi\Vert^4,\\
 f_{[e,t]_\approx}(\psi) &= \tr\pi_2^2 = \tr\pi_{13}^2,\\
 f_{[t,e]_\approx}(\psi) &= \tr\pi_1^2 = \tr\pi_{23}^2,\\
 f_{[t,t]_\approx}(\psi) &= \tr\pi_{12}^2 = \tr\pi_3^2.
\end{align*}

\begin{figure}%[ht]
\includegraphics{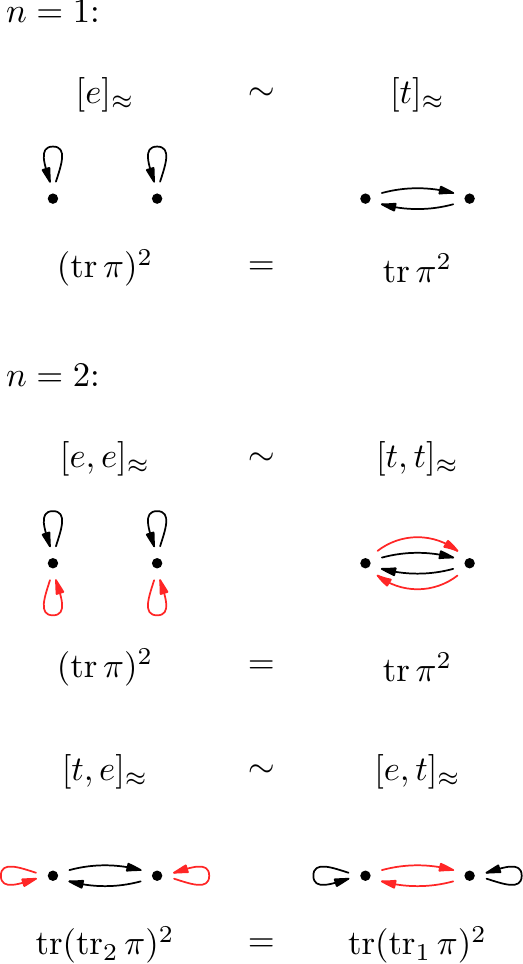}\centering
\caption{Graphs corresponding to the $m=2$ invariant polynomials for $n=1$ and $2$.
 Black and red edges
 represent index-contractions on the first and second Hilbert spaces, respectively.}\label{fig:m2k12}
\end{figure}

The construction of these formulas can easily be generalized to arbitrary number of subsystems.
For this, take a look at the graph on the left of the last line of figure~\ref{fig:m2k12}.
This time, let the 
red lines 
represent the index-contractions on \emph{all} Hilbert spaces
on which $\sigma_j=e$,
and the black lines
represent the index-contractions on \emph{all} Hilbert spaces
on which $\sigma_j=t$.
Thus, we can read off the matrix operations for arbitrary $n$.
The other way of writing the polynomial can be reached by the interchange of the roles of the
black and red lines.
So, for arbitrary number of subsystems ($\pi\equiv\pi_{12\dots n}$)
for the polynomials for $[\tpls{\sigma}]_\sim=[\tpls{\sigma}]_\approx\cup [\overline{\tpls{\sigma}}]_\approx$,
we have
\begin{equation}
\label{eq:purinv2k}
 f_{[\sigma_1,\dots,\sigma_{n-1}]_\approx}(\psi)
= \tr(\tr_{\{n\}\cup\{j\mid \sigma_j=e\}}\pi)^2
= \tr(\tr_{\{  j\mid \sigma_j=t\}}\pi)^2.
%\tr\pi_{\{j\mid \sigma_j=t\}}^2 
%\equiv\tr\pi_{\{n,j\mid \sigma_j=e\}}^2
\end{equation}
(This was used for qubits in~\cite{Sudbery3qb}.)
The number of these, which is the dimension of the grade $m=2$ subspace of the inverse limit of the algebras, is $2^{n-1}$.
The set of algebraically independent generators
contains all the $m=2$ polynomials from (\ref{eq:purinv2k}), 
except the ones for which there are only $e$s in $[\sigma_1,\dots,\sigma_{n-1}]_\approx$ labelling the polynomial.
(This is the only way for the permutations not to act transitively on the set of $m=2$ labels.)
The number of these is $2^{n-1}-1$.

%*******************************************************************************
\subsection{\texorpdfstring{Invariant polynomials of grade $m=3$ (degree $6$)}{Invariant polynomials of grade m=3 (degree 6)}}
\label{sec:deg6.pureinv.3}
For $m=3$, we have%
%%%%%%%%%%%%%%%%%%%%%%%%
\footnote{Denote the permutations with $e=(1)(2)(3)$, $t=(12)(3)$ and $s=(123)$.}
%%%%%%%%%%%%%%%%%%%%%%%%
$\DscrGrp{S}_3=\{e,s,s^2,t,ts,ts^2\}$
with the conjugacy-classes $[e]=\{e\}$, $[s]=\{s,s^2\}$ and $[t]=\{t,ts,ts^2\}$.
This time, we have no simple general rule for the splitting of a $\sim$-class to $\approx$-classes.

For singlepartite system ($n=1$, $\pi\equiv\pi_1$), 
the only polynomial is biven by
\begin{equation*}
 {}[e]_\sim=[e]_\approx\cup[t]_\approx\cup[s]_\approx.
\end{equation*}
From its graphs (figure~\ref{fig:m3k12}) we have
\begin{equation*}
 f_{[]_\approx}(\psi) = (\tr\pi_1)^3 = \tr\pi_1^2\tr\pi_1 = \tr\pi_1^3 = \Vert\psi\Vert^6.
\end{equation*}

For bipartite system ($n=2$, $\pi\equiv\pi_{12}$), 
there are three linearly independent polynomials. These are given by
\begin{align*}
 {}[e,e]_\sim &=[e,e]_\approx\cup[t,t]_\approx\cup[s,s]_\approx,\\
 {}[t,e]_\sim &=[t,e]_\approx\cup[e,t]_\approx\cup[s,t]_\approx\cup[t,s]_\approx,\\
 {}[s,e]_\sim &=[s,e]_\approx\cup[e,s]_\approx\cup[s,s^2]_\approx\cup[t,ts]_\approx
\end{align*}
(so $\DscrGrp{S}_3^1/\DscrGrp{S}_3\isom \{[e],[s],[t]\}$).
From their graphs (figure~\ref{fig:m3k12}) we have
\begin{align*}
 f_{[e]_\approx}(\psi) &= (\tr\pi_{12})^3 = \tr\pi_{12}^2\tr\pi_{12} = \tr\pi_{12}^3 = \Vert\psi\Vert^6,\\
 f_{[t]_\approx}(\psi) &= (\tr\pi_a)^2\tr\pi_{12} = \tr(\tr_a\pi_{12}^2) \pi_b, \\
 f_{[s]_\approx}(\psi) &= \tr\pi_a^3 = \tr(\pi_{12}^{\transp_a})^3 = \tr(\pi_1\otimes\pi_2)\pi_{12}
\end{align*}
for all $a,b\in \{1,2\}$, $a\neq b$.
\begin{figure}
\includegraphics{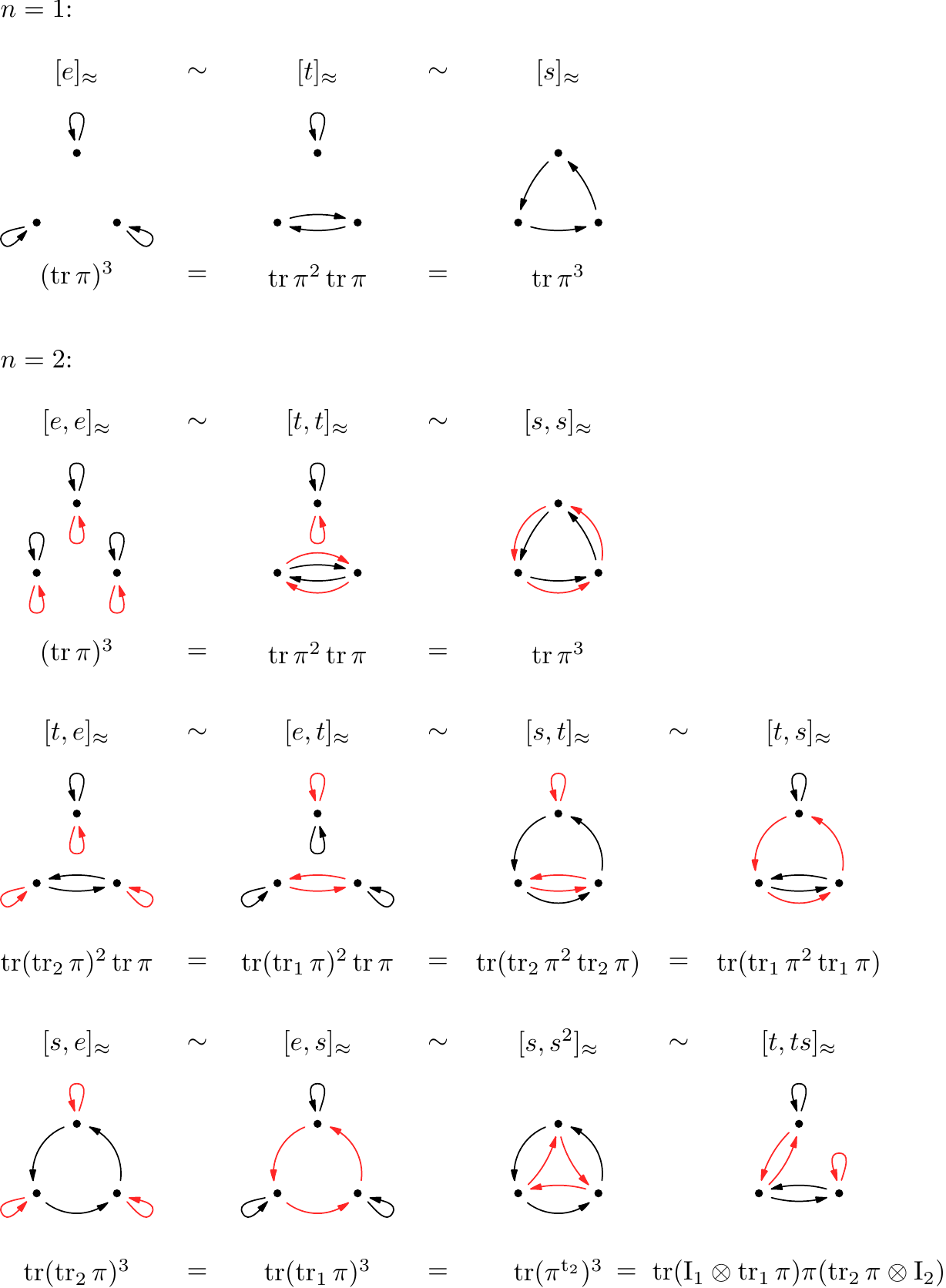}\centering
\caption{Graphs corresponding to the $m=3$ invariant polynomials for $n=1$ and $2$.
 Black and red edges
 represent index-contractions on the first and second Hilbert spaces, respectively.
 The formulas of the polynomials given by matrix operations, which can be read off from the graphs, are also written out.
 For the last one, we have used the trick in the last line of figure~\ref{fig:mxopgraphs} twice.}\label{fig:m3k12}
\end{figure}

For tripartite system ($n=3$, $\pi\equiv\pi_{123}$), 
it turns out that there are eleven linearly independent polynomials. These are given by
\begin{align*}
 {}[e,e,e]_\sim &=[e,e,e]_\approx\cup[t,t,t]_\approx\cup[s,s,s]_\approx,\\
 {}[e,t,e]_\sim &=[e,t,e]_\approx\cup[t,e,t]_\approx\cup[s,t,s]_\approx\cup[t,s,t]_\approx,\\
 {}[t,e,e]_\sim &=[t,e,e]_\approx\cup[e,t,t]_\approx\cup[t,s,s]_\approx\cup[s,t,t]_\approx,\\
 {}[t,t,e]_\sim &=[t,t,e]_\approx\cup[e,e,t]_\approx\cup[s,s,t]_\approx\cup[t,t,s]_\approx,\\
 {}[e,s,e]_\sim &=[e,s,e]_\approx\cup[s,e,s]_\approx\cup[s,s^2,s]_\approx\cup[t,ts,t]_\approx,\\
 {}[s,e,e]_\sim &=[s,e,e]_\approx\cup[e,s,s]_\approx\cup[s^2,s,s]_\approx\cup[ts,t,t]_\approx,\\
 {}[s,s,e]_\sim &=[s,s,e]_\approx\cup[e,e,s]_\approx\cup[s,s,s^2]_\approx\cup[t,t,ts]_\approx,\\
 {}[s,s^2,e]_\sim &=[s,s^2,e]_\approx\cup[s,e,s^2]_\approx\cup[e,s,s^2]_\approx\cup[t,ts,ts^2]_\approx,\\
 {}[t,s,e]_\sim &=[t,s,e]_\approx\cup[t,e,s]_\approx\cup[e,t,ts]_\approx\cup[t,s,s^2]_\approx\cup[s,t,ts]_\approx\cup[s,t,ts^2]_\approx,\\
 {}[s,t,e]_\sim &=[s,t,e]_\approx\cup[e,t,s]_\approx\cup[t,e,ts]_\approx\cup[s,t,s^2]_\approx\cup[t,s,ts]_\approx\cup[t,s,ts^2]_\approx,\\
 {}[t,ts,e]_\sim &=[t,ts,e]_\approx\cup[e,s,t]_\approx\cup[s,e,t]_\approx\cup[s,s^2,t]_\approx\cup[t,ts,s]_\approx\cup[t,ts^2,s]_\approx.
\end{align*}
Here we do not write out all the 49 formulas for the graphs coming from the $\approx$-classes above,
we just show two or three of them for every polynomial.
\begin{align*}
 f_{[e,e]_\approx}(\psi) &= \Vert\psi\Vert^6 = (\tr\pi_{123})^3,\\
 f_{[e,t]_\approx}(\psi) &=  \tr\pi_{123} \tr\pi_2^2 = \tr\pi_{123} \tr\pi_{13}^2,\\
 f_{[t,e]_\approx}(\psi) &=  \tr\pi_{123} \tr\pi_1^2 = \tr\pi_{123} \tr\pi_{23}^2,\\
 f_{[t,t]_\approx}(\psi) &=  \tr\pi_{123} \tr\pi_{12}^2 = \tr\pi_{123} \tr\pi_3^2,\\
 f_{[e,s]_\approx}(\psi) &= \tr\pi_2^3 = \tr\pi_{13}^3,\\
 f_{[s,e]_\approx}(\psi) &= \tr\pi_1^3 = \tr\pi_{23}^3,\\
 f_{[s,s]_\approx}(\psi) &= \tr\pi_{12}^3 = \tr\pi_3^3,\\
 f_{[t,s]_\approx}(\psi) &= \tr(\Id_1\otimes\pi_2)\pi_{12}^2 = \tr(\Id_1\otimes\pi_3)\pi_{13}^2
                 = \tr(\pi_2\otimes\pi_3)\pi_{23}, \\
 f_{[s,t]_\approx}(\psi) &= \tr(\pi_1\otimes\Id_2)\pi_{12}^2 = \tr(\Id_2\otimes\pi_3)\pi_{23}^2
                 =  \tr(\pi_1\otimes\pi_3)\pi_{13},\\
 f_{[t,ts]_\approx}(\psi) &=\tr(\pi_1\otimes\pi_2)\pi_{12}
                  = \tr(\pi_1\otimes\Id_3)\pi_{13}^2 = \tr(\pi_2\otimes\Id_3)\pi_{23}^2,\\
 f_{[s,s^2]_\approx}(\psi) &= \tr(\pi_{ab}^{\transp_a})^3 \qquad \text{for all distinct $a,b\in\{1,2,3\}$}.
\end{align*}
The last one of these is the Kempe-invariant,
which has arisen in the context of hidden nonlocality (section \ref{subsec:QM.EntMeas.3QBPure}).
Kempe has defined this for $\tpl{d}=(2,2,2)$, i.e.,~three qubits,
but it can be written by her definition%
%%%%%%%%%%%%%%%%%%%%%%%%
\footnote{See (17) in~\cite{Kempe3qb}.} 
%%%%%%%%%%%%%%%%%%%%%%%%
for all $\tpl{d}=(d_1,d_2,d_3)$ three-qudit systems.
It is observed by Sudbery~\cite{Sudbery3qb} that for three qubits the Kempe-invariant
can be expressed as
\begin{equation*}
\begin{split}
f_{[s,s^2]_\approx}(\psi) 
&=3f_{[t,s]_\approx}(\psi)-f_{[e,s]_\approx}(\psi)-f_{[s,s]_\approx}(\psi)\\
&=3f_{[s,t]_\approx}(\psi)-f_{[s,e]_\approx}(\psi)-f_{[s,s]_\approx}(\psi)\\
&=3f_{[t,ts]_\approx}(\psi)-f_{[e,s]_\approx}(\psi)-f_{[s,e]_\approx}(\psi), \qquad \text{for $\tpl{d}=(d_1,d_2,d_3)$}.
\end{split}
\end{equation*}
(This form was given in (\ref{eq:3QBcanonPureLUinvs.I4}).)
However, this is only for qubits: 
If $m\leq d_j$ for all $j$ (so at least qutrits)
then the 11 polynomials listed above are linearly independent.
Another important three-qubit permutation- and LU-invariant polynomial of degree $6$,
which has arisen in twistor-geometric~\cite{PeterGeom3QBEnt} 
and Freudenthal~\cite{BorstenetalFreudenthal3QBEnt} approach of three-qubit entanglement,
is the norm square of the Freudenthal-dual of $\psi$.
It can be written as
\begin{equation*}
6\Vert T(\psi,\psi,\psi)\Vert^2
=4f_{[s,s^2]_\approx}(\psi) +5f_{[e,e]_\approx}(\psi) 
- 3f_{[e,t]_\approx}(\psi)-3f_{[t,e]_\approx}(\psi)-3f_{[t,t]_\approx}(\psi).
\end{equation*}
Note, that this expression is not unique, 
since these $f_{[\sigma_1,\sigma_2]_\approx}$ polynomials are not linearly independent in the case of three qubits.
We will return to this quantity in section \ref{subsec:ThreeQB.Pure.NewInvs}.
 
It is not obvious, but the construction of the grade $m=3$ polynomials
can be generalized to arbitrary number of subsystems.
To do this, consider an invariant given by $\tpls{\sigma}\in \DscrGrp{S}_3^n$, 
where all $\sigma_j\in\{t,ts,ts^2\}$.
This can be seen in figure~\ref{fig:m3k},
\begin{figure}
\includegraphics{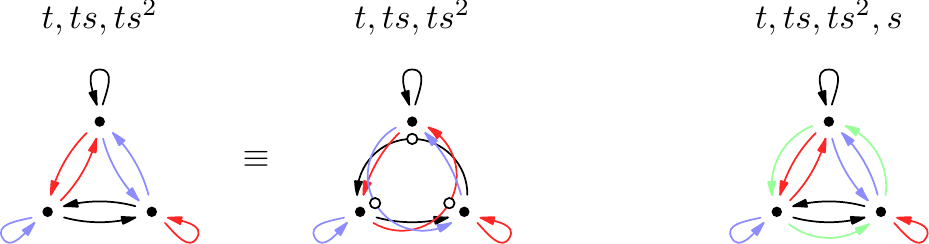}\centering
\caption{Graphs corresponding to the $m=3$ invariant polynomials.
 Black, red, blue and green edges
 represent index-contractions on the Hilbert spaces on which $\sigma_j=t$, $ts$, $ts^2$ and $s$ respectively.
 For the first graph, we show how the trick in the last line of figure~\ref{fig:mxopgraphs} was used three times.}\label{fig:m3k}
\end{figure}
with the evident redefinitions of the meaning of the colours:
let the
black, red and blue edges
represent the index-contractions on \emph{all} Hilbert spaces
on which $\sigma_j=t$, $ts$ and $ts^2$, respectively.
Using the trick in the last line of figure~\ref{fig:mxopgraphs} three times, we have, that
the polynomial is given by
the trace of the product of the three factors, which are
$\Id_{\{j\mid\sigma_j=t\}}\otimes\tr_{\{j\mid\sigma_j=t\}}\pi$,
and other two with $ts$ and $ts^2$.
Note, that the order of these are arbitrary,
since it is related to the relabelling of the vertices of the graph,
and the edge-configuration is invariant for that.
However, if there are some subsystems on which $\sigma_j=s$,
that fixes the order of these terms up to \emph{cyclic} permutation.
It turns out that we have to use the reverse ordering in the product,
the terms $\Id_{\{j\mid\sigma_j=\tau\}}\otimes\tr_{\{j\mid\sigma_j=\tau\}}\pi$
with $\tau=ts^2$ first, then with $\tau=ts$ and then with $\tau=t$,
since they have the fixed points $1$, $2$ and $3$, respectively.
On the subsystems on which $\sigma_j=s$,
the indices are intact (partial traces act only on other subsystems), 
they are contracted in the appropriate way. 
If there are some subsystems on which $\sigma_j=s^2$,
then we use $\pi^{\transp_{\{j\mid\sigma_j=s^2\}}}$ instead of $\pi$ to reduce the situation to the known case.
Similarly, if there are some subsystems on which $\sigma_j=e$,
then we use $\tr_{\{j\mid\sigma_j=e\}}\pi$ instead of $\pi$.
Summing up, 
for arbitrary number of subsystems ($\pi\equiv\pi_{12\dots n}$)
we have the following formula for the $m=3$ polynomials:
\begin{equation}
\label{eq:purinv3k}
f_{[\sigma_1,\dots,\sigma_{n-1}]_\approx}(\psi)=
\;\tr
\prod_{\tau=ts^2,ts,t}
%\prod_{\tau\in[t]}
\Bigl(
\Id_{\{j\mid \sigma_j=\tau\}} 
\otimes
\tr_{\{n\}\cup\{j\mid \sigma_j\in\{\tau,e\}\}}\pi^{\transp_{\{j\mid \sigma_j=s^2\}}}
\Bigr),
\end{equation}
where the $\prod$ product symbol means non-commutative product,
in the order of its subscript.
It can be instructive to check that
this gives back the formulas for the special cases $n=1$, $2$ and $3$.

%*******************************************************************************
%*******************************************************************************
\section{Mixed-state invariants}
\label{sec:deg6.mixinv}

In section~\ref{sec:deg6.luinvs}, we considered the mixed state invariants of $n$ subsystems
as pure state invariants of $n+1$ subsystems.
By considering the graphs of the invariants, given in section~\ref{sec:deg6.graphsops},
we can clarify this from another point of view.

An invariant can be given by $\tpls{\sigma}\in \DscrGrp{S}_m^n$,
this encodes an index-contraction for the matrices of the operators
$\pi=\cket{\psi}\bra{\psi}$ or $\varrho$ for pure or mixed states, respectively.
If $\tpls{\sigma} \not\approx \tpls{\sigma}'$ for $\tpls{\sigma},\tpls{\sigma}'\in \DscrGrp{S}_m^n$,
then they give rise to different polynomials for mixed states,
while it can happen that $\tpls{\sigma} \sim \tpls{\sigma}'$,
so they give rise to the same polynomial for pure states.
In this case, the unlabelled graphs given by $[\tpls{\sigma}]_\approx$ and $[\tpls{\sigma}']_\approx\neq[\tpls{\sigma}]_\approx$
are related to each other
by the independent permutation of the heads and tails of the edges,
while the corresponding operation is
the independent permutation of the coefficients $\psi^{\dots}$ and $\psi_{\dots}$ in (\ref{eq:purinv0}).
This operation is not allowed for mixed states.
In section~\ref{sec:deg6.pureinv},
we have given the decompositions of $\sim$-equivalence classes into $\approx$-equivalence classes
for some grade $m$ and for some $n$ numbers of subsystems,
leading to the different writings of the same polynomial.
For mixed states of $n$ subsystems, these polynomials are not the same anymore.
This offers us a different point of view, which seems to be more natural.
Let us consider the pure state invariants as the special cases of the mixed state invariants
instead of considering the mixed state invariants as pure state invariants of a bigger system.
We have the mixed state formula (\ref{eq:mixinv}) for the set of invariants, 
if we substitute a pure state $\cket{\psi}\bra{\psi}$ into them, then some of them will coincide,
but we can keep this in hand.

For the sake of completeness, we show the mixed state polynomials (\ref{eq:mixinv}) below.
Comparing these formulas with the ones for pure states,
one can see 
how the $n$-partite mixed state invariants are related to the $n+1$-partite pure state ones 
(\ref{eq:purmix}),
or,
how the different $n$-partite mixed state invariants coincide with the same $n$-partite pure state invariants.

Let $\varrho\equiv\varrho_{12\dots n}$ be a density matrix on $\mathcal{H}$.
The polynomials are labelled by $[\tpls{\sigma}]_\approx\in \DscrGrp{S}_m^n/\DscrGrp{S}_m$.

\subsection{\texorpdfstring{Invariant polynomials of grade $m=1$ (degree $1$)}{Invariant polynomials of grade m=1 (degree 1)}}
\label{sec:deg6.mixinv.1}
For $m=1$ we have for all $n$
\begin{equation}
\label{eq:mixinv1k}
f_{[e,\dots,e]_\approx}(\varrho) = \tr\varrho_{12\dots n}.
\end{equation}

\subsection{\texorpdfstring{Invariant polynomials of grade $m=2$ (degree $2$)}{Invariant polynomials of grade m=2 (degree 2)}}
\label{sec:deg6.mixinv.2}
For $m=2$, 
for singlepartite system ($n=1$, $\varrho\equiv\varrho_1$), we have
\begin{align*}
 f_{[e]_\approx}(\varrho) &= (\tr\varrho_1)^2,\\
 f_{[t]_\approx}(\varrho) &= \tr\varrho_{1}^2,
\end{align*}
(for a $d_1=2$ one-qubit system, the determinant is an element of this subspace,
$2\det \varrho = f_{[e]_\approx}(\varrho)-f_{[t]_\approx}(\varrho)$)
for bipartite system ($n=2$, $\varrho\equiv\varrho_{12}$), we have
\begin{align*}
 f_{[e,e]_\approx}(\varrho) &= (\tr\varrho_{12})^2,\\
 f_{[e,t]_\approx}(\varrho) &= \tr\varrho_2^2,\\
 f_{[t,e]_\approx}(\varrho) &= \tr\varrho_1^2,\\
 f_{[t,t]_\approx}(\varrho) &= \tr\varrho_{12}^2,
\end{align*}
for tripartite system ($n=3$, $\varrho\equiv\varrho_{123}$), we have
\begin{align*}
 f_{[e,e,e]_\approx}(\varrho) &= (\tr\varrho_{123})^2,\\
 f_{[e,e,t]_\approx}(\varrho) &= \tr\varrho_3^2,\\
 f_{[e,t,e]_\approx}(\varrho) &= \tr\varrho_2^2,\\
 f_{[e,t,t]_\approx}(\varrho) &= \tr\varrho_{23}^2,\\
 f_{[t,e,e]_\approx}(\varrho) &= \tr\varrho_1^2,\\
 f_{[t,e,t]_\approx}(\varrho) &= \tr\varrho_{13}^2,\\
 f_{[t,t,e]_\approx}(\varrho) &= \tr\varrho_{12}^2,\\
 f_{[t,t,t]_\approx}(\varrho) &= \tr\varrho_{123}^2,
\end{align*}
and for arbitrary number of subsystems ($\varrho\equiv\varrho_{12\dots n}$) we have
\begin{equation}
\label{eq:mixinv2k}
 f_{[\sigma_1,\dots,\sigma_n]_\approx}(\varrho)
=\tr(\tr_{\{j\mid \sigma_j=e\}}\varrho)^2.
\end{equation}
The number of these, which is the dimension of the grade $m=2$ subspace of the inverse limit of algebras, is $2^n$.
The set of algebraically independent generators
contains all the $m=2$ polynomials from (\ref{eq:mixinv2k}),
except the ones for which there are only $e$s in $[\sigma_1,\dots,\sigma_n]_\approx$ labelling the polynomial.
The number of these is $2^n-1$.

\subsection{\texorpdfstring{Invariant polynomials of grade $m=3$ (degree $3$)}{Invariant polynomials of grade m=3 (degree 3)}}
\label{sec:deg6.mixinv.3}
For $m=3$,
for singlepartite system ($n=1$, $\varrho\equiv\varrho_1$), we have
\begin{align*}
 f_{[e]_\approx}(\varrho) &= (\tr\varrho_1)^3,\\
 f_{[t]_\approx}(\varrho) &= \tr\varrho_1 \tr\varrho_1^2,\\
 f_{[s]_\approx}(\varrho) &= \tr\varrho_1^3,
\end{align*}
(for a $d_1=3$ one-qutrit system, the determinant is an element of this subspace,
$6\det \varrho = f_{[e]_\approx}(\varrho)-3f_{[t]_\approx}(\varrho)+2f_{[s]_\approx}(\varrho)$)
for bipartite system ($n=2$, $\varrho\equiv\varrho_{12}$), we have
\begin{align*}
 f_{[e,e]_\approx}(\varrho) &= (\tr\varrho_{12})^3,\\
 f_{[e,t]_\approx}(\varrho) &= \tr\varrho_{12} \tr\varrho_2^2,\\
 f_{[t,e]_\approx}(\varrho) &= \tr\varrho_{12} \tr\varrho_1^2,\\
 f_{[t,t]_\approx}(\varrho) &= \tr\varrho_{12} \tr\varrho_{12}^2,\\
 f_{[e,s]_\approx}(\varrho) &= \tr\varrho_2^3,\\
 f_{[s,e]_\approx}(\varrho) &= \tr\varrho_1^3,\\
 f_{[s,s]_\approx}(\varrho) &= \tr\varrho_{12}^3,\\
 f_{[t,s]_\approx}(\varrho) &= \tr(\Id_1\otimes\varrho_2)\varrho_{12}^2,\\
 f_{[s,t]_\approx}(\varrho) &= \tr(\varrho_1\otimes\Id_2)\varrho_{12}^2,\\
 f_{[t,ts]_\approx}(\varrho)&= \tr(\varrho_1\otimes\varrho_2)\varrho_{12},\\
 f_{[s,s^2]_\approx}(\varrho) &= \tr(\varrho_{12}^{\transp_1})^3 = \tr(\varrho_{12}^{\transp_2})^3,
\end{align*}
and for arbitrary number of subsystems ($\varrho\equiv\varrho_{12\dots n}$) we have
\begin{equation}
\label{eq:mixinv3k}
f_{[\sigma_1,\dots,\sigma_n]_\approx}(\varrho)=
\;\tr
%\prod_{\tau\in[t]}
\prod_{\tau=ts^2,ts,t}
\Bigl(
\Id_{\{j\mid \sigma_j=\tau\}}  \otimes
\tr_{\{j\mid \sigma_j\in\{\tau,e\}\}}\varrho^{\transp_{\{j\mid \sigma_j=s^2\}}}
\Bigr),
\end{equation}
where the $\prod$ product symbol means non-commutative product,
in the order of its subscript.
This gives back the formulas for the special cases $n=1$ and $2$.

%*******************************************************************************
%*******************************************************************************
\section{\texorpdfstring{Algorithm for $\DscrGrp{S}_3^r/\DscrGrp{S}_3$}{Algorithm for S3r/S3}}
\label{sec:deg6.alg}
The formula in (\ref{eq:purinv3k}) gives the grade $m=3$ invariant polynomials for 
a $(\sigma_1,\dots,\sigma_{n-1}) \in \DscrGrp{S}_3^{n-1}$ $n-1$-tuple of permutations,
but the linearly independent ones are labelled by $[\sigma_1,\dots,\sigma_{n-1}]_\approx \in \DscrGrp{S}_3^{n-1}/\DscrGrp{S}_3$.
Since the group-structure of $\DscrGrp{S}_3$ is not too complicated,
we can give an algorithm
to construct exactly one representative element 
$\tpls{\sigma}$ for all orbits $[\tpls{\sigma}]_\approx$,
i.e., to construct the elements of $\DscrGrp{S}_3^r/\DscrGrp{S}_3$.
The choice $r=n-1$ and $r=n$ gives the labels for pure and mixed state invariants, respectively.

Again, $\DscrGrp{S}_3=\{e,s,s^2,t,ts,ts^2\}$, $s=(123)$, $t=(12)(3)$,
and its conjugacy classes are $[e]=\{e\}$, $[s]=\{s,s^2\}$, $[t]=\{t,ts,ts^2\}$.
First, we write the conjugation table, that is, for $\beta,\gamma\in \DscrGrp{S}_3$,
\begin{equation*}
\beta\gamma\beta^{-1}:\qquad
\begin{array}{l||l|ll|lll}
\beta & e   & s   & s^2 & t & ts & ts^2\\
\hline
\hline
e   & e  & s   & s^2 & t   & ts  & ts^2\\
s   & e  & s   & s^2 & ts  & ts^2& t   \\
s^2 & e  & s   & s^2 & ts^2& t   & ts  \\
t   & e  & s^2 & s   & t   & ts^2& ts  \\
ts  & e  & s^2 & s   & ts^2& ts  & t   \\
ts^2& e  & s^2 & s   & ts  & t   & ts^2\\
\end{array}
\end{equation*}

For every position of the list $(\sigma_1,\dots,\sigma_r) \in \DscrGrp{S}_3^r$
there is a conjugacy class $[\sigma_j]$ of $\DscrGrp{S}_3$, which remains unchanged under simultaneous conjugation.
For a given $(\sigma_1,\dots,\sigma_r)$, we would like to single out one representative element
in the orbit of simultaneous conjugation.
To do this, we examine the orbits.
\begin{itemize}
\item If $\sigma_j=e$ for all $j$, we have the trivial orbit of length $1$.
Besides this case, we do not have to deal with the positions in which $e$s occur,
since $\sigma_j=e$ remains unchanged under simultaneous conjugation.
\item If $\sigma_j\in[s]$ (besides $e$) for all $j$, 
then we can choose the element of $[\tpls{\sigma}]_\approx$ which has $s$ in the first position in which an element of $[s]$ occurs.
These orbits are of length $2$.
\item If $\sigma_j\in[t]$ (besides $e$) for all $j$, then we have two kinds of orbits.
If $\sigma_j$ is the same for all $j$ for which $\sigma_j\in[t]$,
then we can choose the element which has $t$ in the first position in which an element of $[t]$ occurs.
These orbits are of length $3$.
On the other hand, 
if there are at least two different $\sigma_j$s for which $\sigma_j\in[t]$,
then it is not enough to fix only one position.
It can be checked by the conjugation table above
that for the ordered pairs of different elements of $[t]$
there exists \emph{exactly one} permutation which brings them into $(t,ts)$ by simultaneous conjugation.
So we can uniquely choose the elements which have
$t$ in the first position in which an element of $[t]$ occurs and
$ts$ in the first position in which a different element of $[t]$ occurs.
These orbits are of length $6$.
\item If both $[s]$ and $[t]$ occur (besides $e$), 
then we have to fix two positions again.
It can be checked from the conjugation table above
that for every pair given by the elements of $[s]\times[t]$
there exists \emph{exactly one} permutation which brings it into $(s,t)$ by simultaneous conjugation.
So we can uniquely choose the element which has
$s$ in the first position in which an element of $[s]$ occurs and
$t$ in the first position in which an element of $[t]$ occurs.
These orbits are of length $6$.
\end{itemize}

With the help of the observations above, we can formulate the following algorithm
generating $\DscrGrp{S}_3^r/\DscrGrp{S}_3$, that is, the labels of the polynomials.
\begin{enumerate}
\item For every position of the list $(\sigma_1,\dots,\sigma_r) \in \DscrGrp{S}_3^r$,
assign one of the conjugacy-classes of $\DscrGrp{S}_3$.
Do this in all possible ways, and apply the following steps for all of them.
\item Write $e$ into all positions to which $[e]$ has been assigned.
\item Take the first of the positions to which $[s]$ has been assigned, and write $s$ there. 
To the others of such positions, write either $s$ or $s^2$ in all possible ways.
\item If there is no position with $[s]$, then
take the first of the positions to which $[t]$ has been assigned, and write $t$ there.
To the following of such positions, write either $t$ or $ts$ in all possible ways,
but after the occurrence of the first $ts$, write either $t$, $ts$ or $ts^2$ in all possible ways.
On the other hand,
if there is at least one position with $[s]$, then
take the first of the positions to which $[t]$ has been assigned, and write $t$ there. 
To the others of such positions, write either $t$, $ts$ or $ts^2$ in all possible ways.
\end{enumerate}

What is the number of the labels obtained in this way?
This could be found by the use of some combinatorics, but we do not have to follow that way.
If the local dimensions $3\leq d_j$, 
then the elements of $\DscrGrp{S}_3^{n-1}/\DscrGrp{S}_3$ label the 
linearly independent grade $m=3$ invariants,
and their number, the dimension of the grade $m=3$ subspace of the inverse limit of the algebras
is given in~\cite{HWLUA,PetiLUA23}.
For $m=3$ pure state invariants, this is
\begin{equation*}
\abs{\DscrGrp{S}_3^{n-1}/\DscrGrp{S}_3} = 6^{n-2}+3^{n-2}+2^{n-2},
\end{equation*}
while for mixed state invariants, this is
\begin{equation*}
\abs{\DscrGrp{S}_3^n/\DscrGrp{S}_3} = 6^{n-1}+3^{n-1}+2^{n-1}
\end{equation*}
(see also in~\cite{oeisA074528}).
One can easily check that 
the set of algebraically independent generators
contains all the $m=3$ polynomials from (\ref{eq:purinv3k}) or~(\ref{eq:mixinv3k}),
except the ones for which, using the labelling algorithm above,
there are only $e$s and $t$s in $[\tpls{\sigma}]_\approx$ labelling the polynomial.
(This is the way for the permutations not to act transitively on the set of $m=3$ labels.)
The number of these is 
$6^{n-2}+3^{n-2}+2^{n-2}-2^{n-1}=6^{n-2}+3^{n-2}-2^{n-2}$ for pure states, and
$6^{n-1}+3^{n-1}-2^{n-1}$ for mixed states.

%*******************************************************************************
%*******************************************************************************
\section{Summary and remarks}
\label{sec:deg6.summary}
In this chapter we have written out explicitly
the LU-invariant polynomials for pure and mixed states,
given in (\ref{eq:purinv}) and~(\ref{eq:mixinv}), for grades $m=1,2,3$.
This was done for arbitrary number of subsystems of arbitrary dimensions.
New results are given by the nice compact formulas of grade $m=3$ invariants for pure (\ref{eq:purinv3k}) and mixed (\ref{eq:mixinv3k}) quantum states,
and in the algorithm generating the different equivalence classes of permutation $n$-tuples of $\DscrGrp{S}_3$
under simultaneous conjugation, given in section~\ref{sec:deg6.alg}.
The latter is necessary to eliminate identical polynomials.
Connections between pure and mixed state invariant polynomials have been illustrated as well.
These results are obtained by the use of graphs corresponding to the polynomials \cite{HWWLUA,PetiLUA23}.

\begin{remarks}
%%%%%%%%%%%%%%%%%%%%%%%%
\item The key point and new feature in this topic,
which is not our result, is the independency \cite{HWWLUA,PetiLUA23}. 
The polynomials in (\ref{eq:purinv}) and~(\ref{eq:mixinv})
give the \emph{linearly independent} basis of the $m$ graded subspace of the algebra of LU-invariant polynomials
if $m\leq d_j$ for all $j$ \cite{HWWLUA},
and some of them (the ones for which the defining permutations together act transitively)
become an \emph{algebraically independent} generating set 
in the inverse limit of algebras, that is, if $d_j\to\infty$ for all $j$ \cite{PetiLUA23}.
This independency result shows the power of the elegant approach using the inverse limit of the algebras of LU-invariant polynomials.
%%%%%%%%%%%%%%%%%%%%%%%%
\item However, for a given $\tpl{d}=(d_1,\dots,d_n)$ system,
it seems to be usual~\cite{LindenPopescuOnMultipartEnt,LindenetalNonlocalParamsMultipartDensMatrices}
that it is not enough to use only the polynomials of maximal degree $2m$, 
where $m\leq d_j$ for all $j$, for the separation of the LU-orbits.
(According to the relatively simple case of $\tpl{d}=(2,2,2)$ three qubits,
where it is known~\cite{Sudbery3qb,Acinetal3QBPureCanon},
that we need an $m=3$, an $m=4$ and an $m=6$ invariant polynomial beyond the $m\leq2$ ones, 
namely the Kempe invariant, the three-tangle, and the Grassl-invariant respectively.)
If $m\nleq d_j$ for a $j$,
the generators given in (\ref{eq:purinv}) and~(\ref{eq:mixinv}) will not be linearly independent,
and the algebraic relations between them exhibit a complicated structure.
%%%%%%%%%%%%%%%%%%%%%%%%
\item But maybe this structure is not too complicated if we find the right point of view.
At this time we have some unpublished results
about linearly independent polynomials for given $\tpl{d}=(d_1,\dots,d_n)$ local dimensions.
%%%%%%%%%%%%%%%%%%%%%%%%
\item Note, that
(first) the same degree of the pure state invariants in the coefficients and in their complex conjugate,
(second) the much simpler labelling of the mixed state invariants than that of the pure ones,
(third) considering the pure state invariants as the special cases of the mixed ones,
seem to stress that 
the density matrices are the natural objects in the topic of unitary invariants
instead of the state vectors.
This approach is widely supported by the whole machinery of quantum physics,
where the elements of the lattice of the subspaces of the Hilbert space
are often regarded to be more fundamental than the elements of the Hilbert space themselves (chapter \ref{chap:QM}).
%%%%%%%%%%%%%%%%%%%%%%%%
\item The illustrating polynomials given in this chapter
could have been written in a convenient index-free form using 
partial trace, matrix product, tensorial product and partial transpose
for grade $m=1$, $2$ and $3$.
However, we note that
it can happen that a grade $m\geq4$ invariant polynomial
can not be written by using these operations only.
At this time, we can not formulate general necessary and sufficient conditions for this,
but we can give an enlightening example.
For the use of matrix operations, we have to write down the matrices one after the other,
this fixes the order of the vertices in some sense.
The partial traces form loops of edges.
If we can find an ordering of the vertices (up to cyclic permutations),
in which these loops of every colours contain only adjacent points with respect to this ordering,
then the matrix-operations can be written for the entire polynomial.
This situation can be seen in the third row of figure~\ref{fig:mxopgraphs}.
After some drawing, one can check that 
there is no such an ordering of the vertices for the graph in figure~\ref{fig:imp},
which seems to be the most simple exapmle for such a situation.
\begin{figure}\centering
\includegraphics{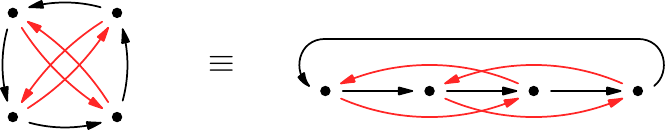}
\caption{An example for an $m=4$, $n=2$ mixed state LU-invariant polynomial 
which can not be written by the considered simple matrix operations.}\label{fig:imp}
\end{figure}
%%%%%%%%%%%%%%%%%%%%%%%%
%\item There is also a successive way to obtain the polynomials
%
%%%%%%%%%%%%%%%%%%%%%%%%
\item Note that these polynomials, 
although they give us a basis in the terms of which the entanglement can be described,
but they still do not measure entanglement in the sense of section \ref{subsec:QM.EntMeas.EntMeas}.
A very important research direction is 
finding such linear-combinations (or other, more complicated functions) of these polynomials 
for which (\ref{eq:averagePure}) or (\ref{eq:meas}) hold,
which are therefore proper entanglement measures.
%%%%%%%%%%%%%%%%%%%%%%%%
\item Since the convex roof extensions of polynomials can be known to be semi-algebraic functions \cite{PetiPriv,ChenDjokovicSemialg},
if we have a pure state LU-invariant polynomial (expressed in the basis of pure state polynomials (\ref{eq:purinv}))
which is entanglement monotone (\ref{eq:averagePure}),
then its convex-roof extension is an entanglement monotone (\ref{eq:averageConvRoof}) LU-invariant semi-algebraic function,
which can be expressed in the terms of the basis of mixed state polynomials (\ref{eq:mixinv}).
A very important research direction would be the understanding of convex-roof extension 
in terms of these LU-invariant polynomials, if that is possible.
%%%%%%%%%%%%%%%%%%%%%%%%
\item The toy-model in the previous chapter could give us hints of this.
For the two-qubit mixed states of fermionic purification (\ref{eq:rho}) we have that
\begin{align*}
\eta^2&= 2(\tr\varrho)^2-4\tr\varrho^2 = 2f_{[e,e]_\approx}(\varrho) -4 f_{[t,t]_\approx}(\varrho),\\
r^2&=2\tr\varrho_1^2 - (\tr\varrho_1)^2 = 2f_{[t,e]_\approx}(\varrho) - f_{[e,e]_\approx}(\varrho),\\
s^2&=2\tr\varrho_2^2 - (\tr\varrho_2)^2 = 2f_{[e,t]_\approx}(\varrho) - f_{[e,e]_\approx}(\varrho),
%\gamma_\pm^2&= 2f_{[e,t]_\approx}(\varrho) + 2f_{[t,e]_\approx}(\varrho) - 2f_{[e,e]_\approx}(\varrho)
%\pm\sqrt{(2f_{[t,e]_\approx}(\varrho) - f_{[e,e]_\approx}(\varrho))(2f_{[e,t]_\approx}(\varrho) - f_{[e,e]_\approx}(\varrho))}
\end{align*}
so we can express the Wootters concurrence (\ref{eq:conc}) (which is itself a convex roof extension (\ref{eq:WConc}))
and also the negativity (\ref{eq:neg})
with the basis elements (\ref{eq:mixinv}).
Of course, again, 
we can be sure only in that the formulas obtained in this way hold in the zero-measure subset of states defined by (\ref{eq:rho}).
\end{remarks}

\chapter{Separability criteria for mixed three-qubit states}
\label{chap:SepCrit}

% eloszo problemafelvetes
In chapter \ref{chap:Ferm} we investigated some measures of entanglement
for a special family of two-qubit mixed states.
A good entanglement measure vanishes for separable states,
in this sense it provides a necessary criterion of separability.
This idea will return in the next chapter,
but now we discuss some other methods for the decision of separability,
which are easier to calculate
and work also in such cases 
in which we do not even have the possibility of evaluating a measure,
or we do not have any measures at all.
%or we do not even have measures.

A mixed state is separable if it can be decomposed to an ensemble of separable pure states,
(section~\ref{subsec:QM.Ent.2Part}).
Such a decomposition is not unique, 
and it is difficult to decide 
whether for a given density operator such a decomposition exists at all.
One can make some observations for separable pure states
which can be extended to mixed states with the help of convex calculus.
The \emph{separability criteria} obtained in this way are \emph{necessary but not sufficient} ones.
(Or equivalently sufficient but not necessary \emph{criteria of entanglement.})
A widely known example is the partial transposition criterion of Peres (section~\ref{subsec:QM.Ent.2Part}).
On the other hand one can construct \emph{necessary and sufficient} criteria 
using sophisticated mathematical methods,
dealing, for example, with witness operators, positive but not completely positive maps %\cite{HorodeckiPosMapWitness},
or semidefinite programming (sections \ref{subsec:QM.Ent.2Part}, \ref{subsec:QM.Ent.NPart}). %\cite{Dohertycrit1,Dohertycrit2,Dohertycrit3}.
Unfortunately these criteria are difficult to use for general density matrices
due to the rapidly growing complexity of the problem,
and only the necessary but not sufficient criteria are used in practice.

% mit csinalunk
In this chapter, we give a comprehensive survey 
of the necessary but not sufficient criteria of separability of mixed quantum states.
We review and compare the criteria known from the literature
and give a case study of a special two-parameter class of three-qubit density matrices.
The form of these density matrices is simple enough
to calculate explicitly the set of states for which these criteria hold.

% thesis statement 
The material of this chapter covers thesis statement \ref{statement:sepcrit}
(page \pageref{statement:sepcrit}).
% organization
\begin{organization}
\item[\ref{sec:SepCrit.Rho}]
we introduce the parametrized permutation-invariant family of \emph{three-qubit} density matrices, our concern,
and make some observations 
about the separability class structure of \emph{permutation-invariant} three-qubit mixed states.
After having set the stage, in the next sections we investigate some criteria for separability classes.
\item[\ref{sec:SepCrit.2Part}]
we consider our quantum-state as a $\tpl{d}=(2,4)$ qubit-qudit system
and we recall and use some \emph{bipartite} separability criteria,
namely the \emph{majorization} and the \emph{entropy criteria},
which are related to the notion of mixedness of the subsystems
(sections~\ref{subsec:SepCrit.2Part.Maj} and~\ref{subsec:SepCrit.2Part.Entr}),
the \emph{partial transposition} and the \emph{reduction criteria},
which are particular cases of the positive map criterion
(sections~\ref{subsec:SepCrit.2Part.PPT} and~\ref{subsec:SepCrit.2Part.Red}),
and the \emph{reshuffling criterion}, which 
in addition to the partial transposition criterion
is the other one of the two independent permutation criteria for bipartite systems 
(section~\ref{subsec:SepCrit.2Part.Resh}).
\item[\ref{sec:SepCrit.3Part}]
we consider our quantum-state as a proper $\tpl{d}=(2,2,2)$ three-qubit system
and investigate some \emph{tripartite} criteria for separability classes.
We recall the \emph{permutation criteria} for permutation-invariant three-qubit case
giving rise to \emph{another reshuffling criterion} (section~\ref{subsec:SepCrit.3Part.Perm}).
Then we use \emph{quadratic Bell inequalities}
%Then we use some criteria using the \emph{expectation value of local spin-observables}
(section~\ref{subsec:SepCrit.3Part.Spin}),
and other criteria using \emph{swap operators} 
(section~\ref{subsec:SepCrit.3Part.Hub})
and \emph{explicit expressions of matrix elements}
(section~\ref{subsec:SepCrit.3Part.Matrix}).
The latter makes it possible to determine a set of entangled states of positive partial transpose.
\item[\ref{sec:SepCrit.Class1}] 
we investigate some other aspects of entanglement.
First, the \emph{W and GHZ classes of fully entangled mixed states} (section~\ref{subsec:SepCrit.Class1.WGHZ}),
then the entanglement of two-qubit subsystems
by the calculation of the \emph{Wootters concurrence} (section~\ref{subsec:SepCrit.Class1.Woott}). 
\item[\ref{sec:SepCrit.Concl}]
%we sum up the results and make some remarks on them.
%summary and some remarks are left.
%conclusions and some comments are left.
%summary and some notes are left.
%we give a summary, some remarks, and open questions.
we give a summary and some remarks.
\end{organization}

%*****************************************************************************
\section{A symmetric family of mixed three-qubit states}
%\label{sec:Rho}
\label{sec:SepCrit.Rho}

In this chapter we investigate three-qubit states,
so we have the Hilbert space $\mathcal{H}=\mathcal{H}_1\otimes\mathcal{H}_2\otimes\mathcal{H}_3$
of local dimensions $\tpl{d}=(2,2,2)$.
Let $\varrho\in\mathcal{D}(\mathcal{H})$ be the mixture of
the GHZ state (\ref{eq:GHZ}), 
the W state (\ref{eq:W}) and
the white noise (\ref{eq:whitenoise}),
\begin{equation}
\label{eq:lo}
\varrho=d\frac{1}{8}\Id\otimes\Id\otimes\Id
+g\cket{\text{GHZ}}\bra{\text{GHZ}}
+w\cket{\text{W}}\bra{\text{W}},
\end{equation}
where $0\leq d,g,w\leq1$ real numbers are the weights
characterizing the mixture, that is, $d+g+w=1$.
In the following sections we plot the subsets of states for which the separability criteria hold on the $g$-$w$-plane,
that is, we project the probability-simplex onto the $d=0$ plane.
A point on this plane determines the third coordinate as $d=1-g-w$.
Sometimes it is convenient to use the rescaled parameters
$\td=d/8$, $\tg=g/2$, $\tw=w/3$.

The \emph{GHZ-W mixture} ($d=0$ line) is well studied.
The three tangle (\ref{eq:tau}) %\cite{CKWThreetangle}
with its convex roofs, %\cite{UhlmannFidelityConcurrence},
the Wootters concurrences (\ref{eq:WConc}) %\cite{HillWoottersConc,WoottersConc},
the singlepartite concurrences (\ref{eq:pureLUinvs3.sa}) with their convex roofs,
and the mixed-state extension of CKW inequality (\ref{eq:CKWmonogamy}) %\cite{CKWThreetangle}
were given for this mixture in the paper of Lohmayer et.~al.~\cite{MixedThreetangle}.
These results give an upper bound for values of these quantities on the whole simplex defined in equation (\ref{eq:lo}),
since if $\cnvroof{f}(\varrho_{g,w})=\min\sum_ip_if(\psi_i)$ where the minimum is taken over all decomposition
$\sum_ip_i\cket{\psi_i}\bra{\psi_i}=\varrho_{g,w}$,
and $\varrho_{d,g,w}=d\Id\otimes\Id\otimes\Id/8+(1-d)\varrho_{g/(1-d),w/(1-d)}$,
and $f(\psi)=0$ on product states,
then $\cnvroof{f}(\varrho_{d,g,w})\leq(1-d)\cnvroof{f}(\varrho_{g/(1-d),w/(1-d)})$.

The white noise can be regarded in some sense
as the ``center'' of the set of density matrices.
Mixing a state with white noise
is the way to investigate the effect of environmental decoherence~\cite{GuhneTothEntDet}.
A noisy state is usually of full rank,
so methods for density matrices of low rank
(like range criterion \cite{HorodeckiRangeCrit},
or finding optimal decompositions with respect to some pure-state measures)
usually fail for such states.

On the other hand, there are exact results for the \emph{GHZ-white noise mixture} ($w=0$ line).
In~\cite{DurCiracTarrach3QBMixSep} D\"ur, Cirac and Tarrach,
using their results about a special class of GHZ-diagonal states,
have shown
that $\varrho$ is fully separable \emph{if and only if} $0\leq g\leq1/5$.
Moreover, it follows from their observations that
if the state is separable under a bipartition then it is fully separable,
so Class 2.8 is empty for these states (section \ref{subsec:QM.Ent.NPart}).
In~\cite{GuhneSevinckCrit} G\"uhne and Seevinck
gives \emph{necessary and sufficient} condition of tripartite entanglement for GHZ-diagonal states,
which contain the noisy GHZ state:
for $1/5< g\leq3/7$ the state is $2$-separable, yet inseparable under bipartitions, that is, in Class 2.1,
and for $3/7< g\leq1$ the state is fully entangled, that is, in Class 1.
Unfortunately there are no such results for other subsets of the simplex given in equation (\ref{eq:lo}).

The noisy GHZ-W mixture given in equation (\ref{eq:lo}) is clearly a \emph{permutation invariant} one,
hence the reduced density matrices of $\varrho$ are all of the same form,
$\varrho_{12}=\varrho_{23}=\varrho_{31}$
and
$\varrho_1=\varrho_2=\varrho_3$.
The explicit forms of these matrices are given in equations (\ref{eq:MxR23}) and~(\ref{eq:MxR1}).
What can we say about the tripartite separability-classes given in section~\ref{subsec:QM.Ent.NPart}
for permutation-invariant three-qubit states \emph{in general}?
Clearly, if a permutation-invariant state is in $\mathcal{D}_{\alpha_2}$ for \emph{a particular} $\alpha_2$,
then it is in $\mathcal{D}_{\alpha_2}$ for \emph{every} $\alpha_2$.
So permutation-invariant states can not be in Classes 2.2-2.7,
we have to investigate separability criteria only for Class 2.1, Class 2.8 and Class 3 (figure~\ref{fig:3partpi}).
%%%%%%%%%%%%%%%%%%%%%%%%%%%%%%%%
\begin{figure}
 \includegraphics{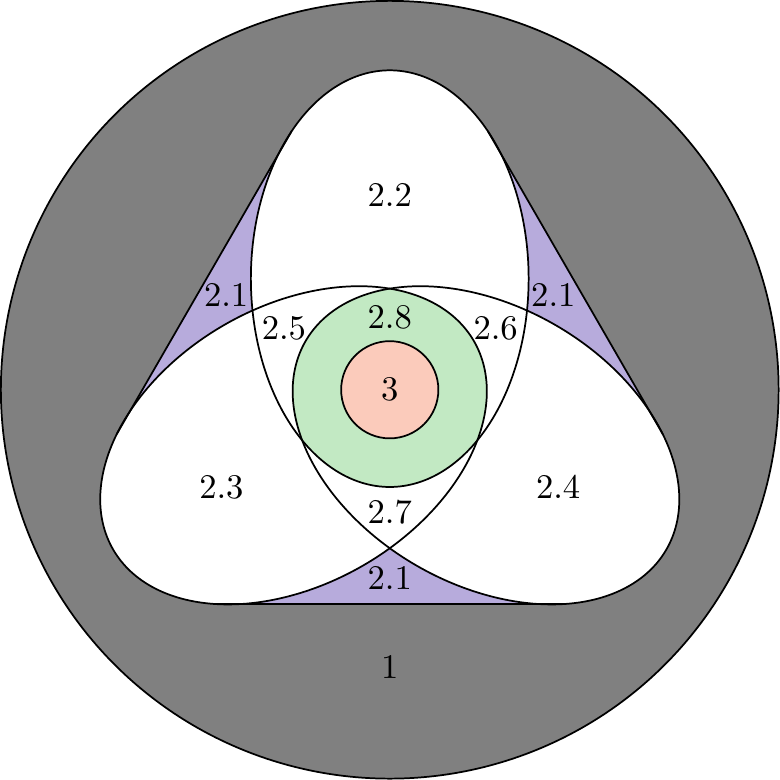}
 \caption{Separability classes for three subsystems.
The tinted subsets of the diagram can contain permutation-invariant states.}
\label{fig:3partpi}
\end{figure}
%%%%%%%%%%%%%%%%%%%%%%%%%%%%%%%%
(Note that the $2$-separability of a permutation-invariant tripartite state
does not mean that there exists a decomposition containing only permutation-invariant members
(bosonic separability problem).
Moreover, it turns out that if the latter holds then the state must be the white noise,%
%%%%%%%%%%%%%%%%%%%%%%%%
\footnote{To see this, write out a member of the decomposition
with the help of the $\sigma_i$ Pauli-matrices and $x_i$, $y_j$ real coefficients
as
$\varrho_1\otimes\varrho_{23}
=1/2(\Id+\sum_ix_i\sigma_i)\otimes1/4(\Id\otimes\Id+\sum_jy_j\sigma_j\otimes\sigma_j)
=1/8(\Id\otimes\Id\otimes\Id+\sum_ix_i\sigma_i\otimes\Id\otimes\Id+\sum_jy_j\Id\otimes\sigma_j\otimes\sigma_j
+\sum_{ij}x_iy_j\sigma_i\otimes\sigma_j\otimes\sigma_j)$
which can be permutation-invariant if and only if $x_i=0$, $y_j=0$.} 
%%%%%%%%%%%%%%%%%%%%%%%%
so for permutation-invariant tripartite states in Classes 2.1 and 2.8,
there does not exist a decomposition containing \emph{only} permutation-invariant members.)

The remaining question is whether the nonempty classes could contain permutation-invariant states in general.
Class 1 and Class 3 are clearly nonempty for permutation-invariant states,
and for Classes 2.1 and 2.8 we show explicit examples.
For Class 2.1, let us consider
the maximally bipartite entangled Bell-state (\ref{eq:B}).
The uniform mixture of the rank one projectors to the subspaces
$\cket{0}_1\otimes\cket{\text{B}}_{23}$,
$\cket{0}_2\otimes\cket{\text{B}}_{31}$ and
$\cket{0}_3\otimes\cket{\text{B}}_{12}$
gives a state
which is by construction a permutation-invariant $2$-separable one, having the matrix%
%%%%%%%%%%%%%%%%%%%%%%%%
\footnote{For better readability, zeroes in matrices are often denoted with dots.}
%%%%%%%%%%%%%%%%%%%%%%%%
\begin{subequations}
\begin{equation}
\frac16\begin{bmatrix}
 3   &\cdot&\cdot& 1   &\cdot& 1   & 1   &\cdot\\
\cdot&\cdot&\cdot&\cdot&\cdot&\cdot&\cdot&\cdot\\
\cdot&\cdot&\cdot&\cdot&\cdot&\cdot&\cdot&\cdot\\
 1   &\cdot&\cdot& 1   &\cdot&\cdot&\cdot&\cdot\\
\cdot&\cdot&\cdot&\cdot&\cdot&\cdot&\cdot&\cdot\\
 1   &\cdot&\cdot&\cdot&\cdot& 1   &\cdot&\cdot\\
 1   &\cdot&\cdot&\cdot&\cdot&\cdot& 1   &\cdot\\
\cdot&\cdot&\cdot&\cdot&\cdot&\cdot&\cdot&\cdot
\end{bmatrix}\;\in\text{Class 2.1}.
\end{equation}
It can be easily checked that its partial transpose is not positive,
(see in (\ref{eq:critPeres}), 
this is enough to be checked for only one subsystem because of the permutation-invariance of the state)
so it is not $\alpha_2$-separable for any $2$-partite split,
hence it is in Class 2.1.
An example for a permutation-invariant state in Class 2.8
is given in equation (14) of \cite{Acinetal3QBMixClass} with $0<a=b=1/c$, having the matrix
\begin{equation} 
\frac1{2+3\left(a+\frac1a\right)}\begin{bmatrix}
 1   &\cdot&\cdot&\cdot&\cdot&\cdot&\cdot& 1   \\
\cdot& a   &\cdot&\cdot&\cdot&\cdot&\cdot&\cdot\\
\cdot&\cdot& a   &\cdot&\cdot&\cdot&\cdot&\cdot\\
\cdot&\cdot&\cdot&\frac1a&\cdot&\cdot&\cdot&\cdot\\
\cdot&\cdot&\cdot&\cdot& a   &\cdot&\cdot&\cdot\\
\cdot&\cdot&\cdot&\cdot&\cdot&\frac1a&\cdot&\cdot\\
\cdot&\cdot&\cdot&\cdot&\cdot&\cdot&\frac1a&\cdot\\
 1   &\cdot&\cdot&\cdot&\cdot&\cdot&\cdot& 1   \\
\end{bmatrix}\;\in\text{Class 2.8 if and only if $a\neq1$}.
\end{equation}
\end{subequations}
This state is entangled if and only if $a\neq1$,
and it is in $\mathcal{D}_{\alpha_2}$ for all $\alpha_2$ \cite{Acinetal3QBMixClass}.

If we consider the $\tpl{d}=(2,2,2)$ three-qubit system
as a $\tpl{d}=(2,4)$ qubit-qudit system then some well-known and easy-to-use \emph{bipartite} separability criteria
give rise to separability criteria for $\bigcap_{\alpha_2}\mathcal{D}_{\alpha_2}$,
hence for the union of Class 2.8 and Class 3.
(This one is also a convex set since it is the intersection of convex ones.)
First we investigate these criteria.

%******************************************************************************
\section{Bipartite separability criteria}
%\label{sec:Twopart}
\label{sec:SepCrit.2Part}

In this section we consider our system as a $\tpl{d}=(2,4)$ qubit-qudit system
(with Hilbert-spaces $\mathcal{H}_A=\mathcal{H}_1$ and $\mathcal{H}_B=\mathcal{H}_{23}$)
and investigate some criteria of $1|23$-separability
which means the union of Classes 2.8 and 3.
To do this,
we will need the spectra of the density matrix $\varrho$ given in equation (\ref{eq:lo}) and of its marginals,
having the matrices
\begin{subequations}
\begin{align}
\label{eq:MxR}
\varrho=&\begin{bmatrix}
 \td+\tg & \cdot   & \cdot   & \cdot   & \cdot   & \cdot   & \cdot   & \tg     \\
 \cdot   & \td+\tw & \tw     & \cdot   & \tw     & \cdot   & \cdot   & \cdot   \\
 \cdot   & \tw     & \td+\tw & \cdot   & \tw     & \cdot   & \cdot   & \cdot   \\
 \cdot   & \cdot   & \cdot   & \td     & \cdot   & \cdot   & \cdot   & \cdot   \\
 \cdot   & \tw     & \tw     & \cdot   & \td+\tw & \cdot   & \cdot   & \cdot   \\
 \cdot   & \cdot   & \cdot   & \cdot   & \cdot   & \td     & \cdot   & \cdot   \\
 \cdot   & \cdot   & \cdot   & \cdot   & \cdot   & \cdot   & \td     & \cdot   \\
 \tg     & \cdot   & \cdot   & \cdot   & \cdot   & \cdot   & \cdot   & \td+\tg 
\end{bmatrix},\\
\label{eq:MxR23}
\varrho_{23}=&\begin{bmatrix}
2\td+\tg+\tw & \cdot   & \cdot   & \cdot   \\
 \cdot   &2\td+\tw & \tw     & \cdot   \\
 \cdot   & \tw     &2\td+\tw & \cdot   \\
 \cdot   & \cdot   & \cdot   &2\td+\tg  
\end{bmatrix},\\
\label{eq:MxR1}
\varrho_1=&\begin{bmatrix}
4\td+\tg+2\tw & \cdot   \\
 \cdot   &4\td+\tg+\tw   \\
\end{bmatrix}.
\end{align}
\end{subequations}
Due to the special structure of $\varrho$,
finding the eigenvalues of these matrices is not a difficult task.
It turns out that all of the eigenvalues are linear in the parameters $g$ and $w$,
and 
\begin{subequations}
\begin{align}
\label{eq:spectRho}
\Spect\varrho&=
\begin{aligned}[t]\bigl\{
\;\td+2\tg &=(3+21g-3w)/24,\\
\td+3\tw   &=(3-3g+21w)/24,\\
\td        &=(3-3g-3w)/24\quad\text{(6 times)}\bigr\},
\end{aligned}\\
\label{eq:spectRho23}
\Spect\varrho_{23}&=
\begin{aligned}[t]\bigl\{
\qquad 
2\td+2\tw    &=(6-6g+10w)/24,\\
2\td+\tg+\tw &=(6+6g+2w)/24,\\
2\td+\tg     &=(6+6g-6w)/24,\\
2\td         &=(6-6g-6w)/24 \quad \bigr\},
\end{aligned}\\
\label{eq:spectRho1}
\Spect\varrho_{1}&=
\begin{aligned}[t]\bigl\{\;
4\td+\tg+2\tw &=(12+4w)/24\\
4\td+\tg+\tw  &=(12-4w)/24 \;\bigr\}.
\end{aligned}
\end{align}
\end{subequations}
Here and in the following, we give expressions with and without $\tilde\;$.
This is because the expressions with the quantities $\td,\tg,\tw$ are much simpler 
and still expressive
as they refer to the original mixing weights,
on the other hand we plot in the $g,w$ coordinates.

%******************************************************************************
\subsection{Majorization criterion}
%\label{sec:Maj}
\label{subsec:SepCrit.2Part.Maj}

Now we turn to the \emph{majorization criterion} for bipartite systems.
It has been found by Nielssen and Kempe \cite{NielssenKempeMaj}, and it states that
\emph{for a separable state the whole system is more disordered than any of its subsystems,} that is,
\begin{equation}
\label{eq:critMaj}
\varrho\quad\text{separable}\qquad\Longrightarrow\qquad 
\varrho\preceq\varrho_A \quad\text{and}\quad \varrho\preceq\varrho_B,
\end{equation}
where the comparsion of disorderness is given by majorization, see in section~\ref{subsec:QM.QuantSys.Mixedness}.
The right-hand side~of (\ref{eq:critMaj}) can also be true for entangled states,
but if it does not hold then the state must be entangled.

Let us see what the majorization criterion states about the noisy GHZ-W mixture $\varrho$ given by equation (\ref{eq:lo}).
We can write out the right-hand side~of (\ref{eq:critMaj}) explicitly
using the spectra given by equations (\ref{eq:spectRho})-(\ref{eq:spectRho1}),
then we have to decide whether the inequalities in (\ref{eq:major}) hold.
For this we have to order the eigenvalues of the density matrices in non-increasing order.
These orderings depend on the ranges of the parameters
and it turns out that we have to distinguish between four cases.
These cases are as follows:
$0\leq g\leq2/3w$,
$2/3w\leq g\leq w$,
$w\leq g\leq4/3w$ and
$4/3w\leq g\leq1$.
It also turns out that in all of these cases
every inequality of (\ref{eq:major}) holds except three ones.
These are as follows:
\begin{equation} 
\label{eq:tablicsku}
\begin{split} 
\varrho&\in\bigcap_{\alpha_2}\mathcal{D}_{\alpha_2}\qquad\Longrightarrow\qquad\\
&\left\{\begin{array}{c|cc|c} 
\text{case} & \text{(i)} & \text{(ii)} & \text{(iii)} \\ 
\hline 
0\leq g\leq2/3w & w\leq3/11-3/11g & w\leq 1-3g      & w\leq9/17+3/17g\\ 
2/3w\leq g\leq w& w\leq3/19+3/19g & w\leq 1-3g      & w\leq9/17+3/17g\\ 
w\leq g\leq4/3w & 3g-3/5\leq w    & w\leq 1-3g      & 3g-9/7\leq w \\ 
4/3w\leq g\leq1 & 3g-3/5\leq w    & w\leq3/11-3/11g & 3g-9/7\leq w
\end{array} \right.
\end{split}
\end{equation}
where in columns (i) and (ii) are the first two inequalities of (\ref{eq:major})
(that is, for $k=1,2$)
written on $\varrho\preceq\varrho_{23}$,
and in column (iii)  is the first inequality of (\ref{eq:major})
written on $\varrho\preceq\varrho_1$ in all of the four cases.
We can make the inequalities of (\ref{eq:tablicsku}) expressive with the help of figure~\ref{fig:Maj}.
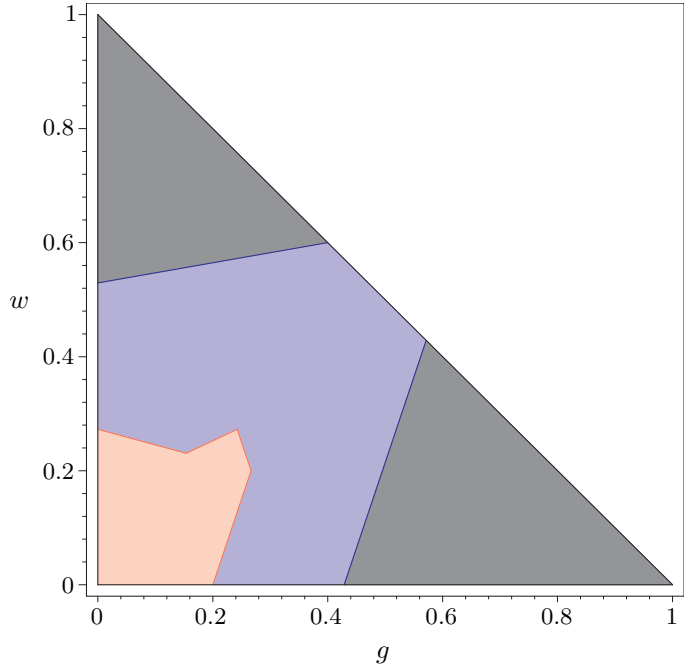
\begin{figure}
 \setlength{\unitlength}{0.001428571\textwidth}% this 1/(420)*0.6 = 0.001428571
 \begin{picture}(420,415)
  \put(0,0){\includegraphics[width=0.6\textwidth]{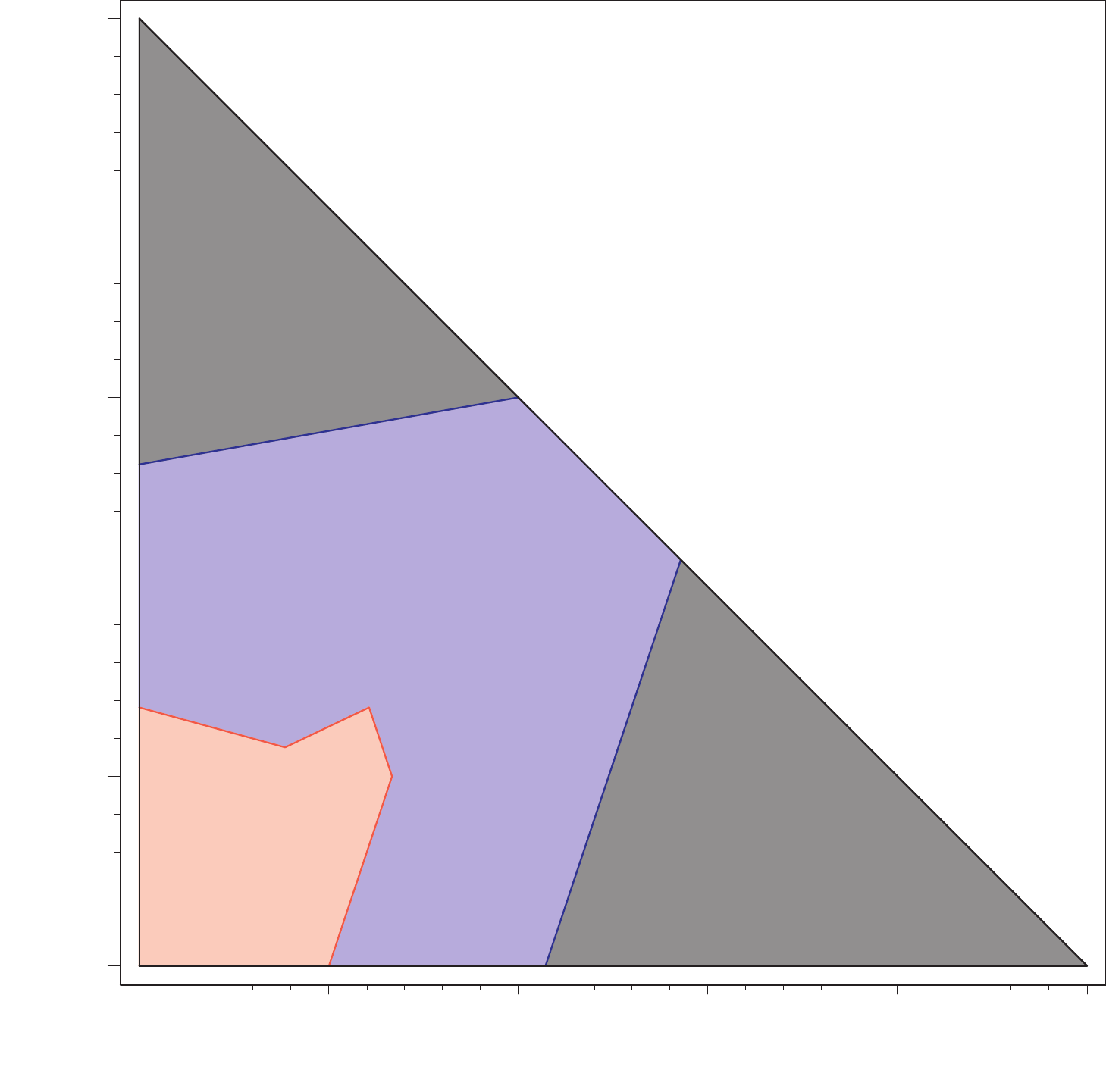}}
  \put(10,226){\makebox(0,0)[r]{\strut{}$w$}}
  \put(235,6){\makebox(0,0)[r]{\strut{}$g$}}
{\small
  \put(56,27){\makebox(0,0)[r]{\strut{}$0$}}
  \put(133,27){\makebox(0,0)[r]{\strut{}$0.2$}}
  \put(205,27){\makebox(0,0)[r]{\strut{}$0.4$}}
  \put(277,27){\makebox(0,0)[r]{\strut{}$0.6$}}
  \put(349,27){\makebox(0,0)[r]{\strut{}$0.8$}}
  \put(415,27){\makebox(0,0)[r]{\strut{}$1$}}
  \put(38,46){\makebox(0,0)[r]{\strut{}$0$}}
  \put(38,118){\makebox(0,0)[r]{\strut{}$0.2$}}
  \put(38,190){\makebox(0,0)[r]{\strut{}$0.4$}}
  \put(38,262){\makebox(0,0)[r]{\strut{}$0.6$}}
  \put(38,334){\makebox(0,0)[r]{\strut{}$0.8$}}
  \put(40,406){\makebox(0,0)[r]{\strut{}$1$}}
}
 \end{picture}
 \caption{Majorization criterion for the state (\ref{eq:lo}) on the $g$-$w$-plane.
(Red domain: $\varrho\preceq\varrho_1$ and $\varrho\preceq\varrho_{23}$,
blue domain: $\varrho\preceq\varrho_1$ and $\varrho\npreceq\varrho_{23}$,
grey domains: $\varrho\npreceq\varrho_1$ and $\varrho\npreceq\varrho_{23}$.)}
 \label{fig:Maj}
\end{figure}
It can be seen that in our case $\varrho\preceq\varrho_{23}$ implies $\varrho\preceq\varrho_1$,
so the bigger subsystem (the trace map on smaller subsystem) gives the stronger condition.
(This is not true in general. One can find a permutation invariant three-qubit state
where $\varrho\preceq\varrho_1$ and $\varrho\preceq\varrho_{23}$ can hold independently.)

The right-hand side~of (\ref{eq:critMaj}) holds for states of parameters in the red domain of figure~\ref{fig:Maj},
so it contains Classes 2.8 and 3.
On the other hand, states of parameters in the blue or grey domain are in Classes 2.1 or 1,
but there can also be such states in the red domain.
Moreover, \emph{the union of Classes 2.8 and 3 is a convex set inside the red domain.}
In the following we consider some other criteria in order to decrease the area of the red domain.
In this way we can identify more states to be in Classes 2.1 or 1.
But before this, we can make an interesting observation here.
One can check that
\emph{for the GHZ-white noise mixture} ($w=0$ line) the majorization criterion
$\varrho\preceq\varrho_1$ and $\varrho\preceq\varrho_{23}$
is necessary and sufficient for full-separability,
moreover, the criterion
$\varrho\preceq\varrho_1$ and $\varrho\npreceq\varrho_{23}$
is necessary and sufficient for Class 2.1,
and the criterion
$\varrho\npreceq\varrho_1$ and $\varrho\npreceq\varrho_{23}$
is necessary and sufficient for Class 1.
(See section~\ref{sec:SepCrit.Rho} for summary of known exact results on the GHZ-white noise mixture.)
Hence the condition of tripartite entanglement
is the violation of both majorization of (\ref{eq:critMaj}) for the GHZ-white noise mixture.

%******************************************************************************
\subsection{Entropy criterion}
%\label{sec:Entr}
\label{subsec:SepCrit.2Part.Entr}

Now we can turn to the \emph{entropy criterion} for bipartite density matrices
\cite{Horodecki2EntCrit,Horodecki3EntCrit,TerhalEntCrit,VollbrechtWolfEntCrit}.
This is an entropy-based reformulation of the statement
``for a separable state the whole system is more disordered than its subsystems'', that is,
\begin{equation}
\label{eq:critEntr}
\varrho\quad\text{separable}\qquad\Longrightarrow\qquad 
S_q(\varrho)\geq S_q(\varrho_A) \quad\text{and}\quad S_q(\varrho)\geq S_q(\varrho_B).
\end{equation}
The right-hand side~of (\ref{eq:critEntr}) can also be true for entangled states,
but if it does not hold then the state must be entangled.
Here, we can use both R{\'e}nyi (\ref{eq:qRenyi}) and Tsallis~(\ref{eq:qTsallis}) entropies,
this is why we dropped the superscript of $S_q$,
since both kinds of generalized entropies lead to the same condition if $q\neq1$:
\begin{equation*}
\varrho\quad\text{separable}\qquad\Longrightarrow\qquad 
\tr \varrho^q \geq \tr \varrho_A^q \quad\text{and}\quad \tr \varrho^q \geq \tr \varrho_B^q.
\end{equation*}
Note that for $2\leq q$ integer parameters, the condition is expressed 
in the terms of the basic LU-invariant homogeneous polynomials given in the previous section.
The entropy criterion follows from the majorization criterion,
since the generalized entropies are \emph{Schur concave} functions on the set of probability distributions
(section \ref{subsec:QM.QuantSys.Mixedness}),
but, historically, entropic criteria for some particular $q$ parameters were proved first.
Therefore \emph{the entropy criterion can not be stronger than the majorization criterion.}
In the following we illustrate this with the state $\varrho$ given in equation (\ref{eq:lo})
for some particular choice of $q$.

The rank of $\varrho$, $\varrho_{23}$ and $\varrho_1$
can be determined easily due to the simple form of the spectra in equations (\ref{eq:spectRho})-(\ref{eq:spectRho1}).
Hence the entropy criterion for Hartley entropy (\ref{eq:qHartley}) can be readily examined.
First, $\rk\varrho=8$ if and only if $d\neq0$.
The right-hand side~of (\ref{eq:critEntr}) holds for these states.
It is true for all states that $\rk\varrho_1=2$.
On the $w=1-g$ line ($d=0$) we have to make distinction between the pure and mixed cases.
If $g=1$ (pure GHZ state) or $w=1$ (pure W state) then $\rk\varrho_{23}=2$ and $\rk\varrho=1$,
hence for this case $S_0(\varrho)\ngeq S_0(\varrho_1)$ and $S_0(\varrho)\ngeq S_0(\varrho_{23})$.
For the nontrivial mixtures of GHZ and W states $\rk\varrho_{23}=3$ and $\rk\varrho=2$
hence $S_0(\varrho)\geq S_0(\varrho_1)$ but $S_0(\varrho)\ngeq S_0(\varrho_{23})$.
So we can establish that \emph{the entropy criterion in the limit $q\to0$
(quantum-Hartley entropy)
is too weak, it identifies only the GHZ-W mixture to be entangled in the simplex.}

\begin{figure}
 \setlength{\unitlength}{0.001428571\textwidth}% this 1/(420)*0.6 = 0.001428571
 \begin{picture}(420,415)
  \put(0,0){\includegraphics[width=0.6\textwidth]{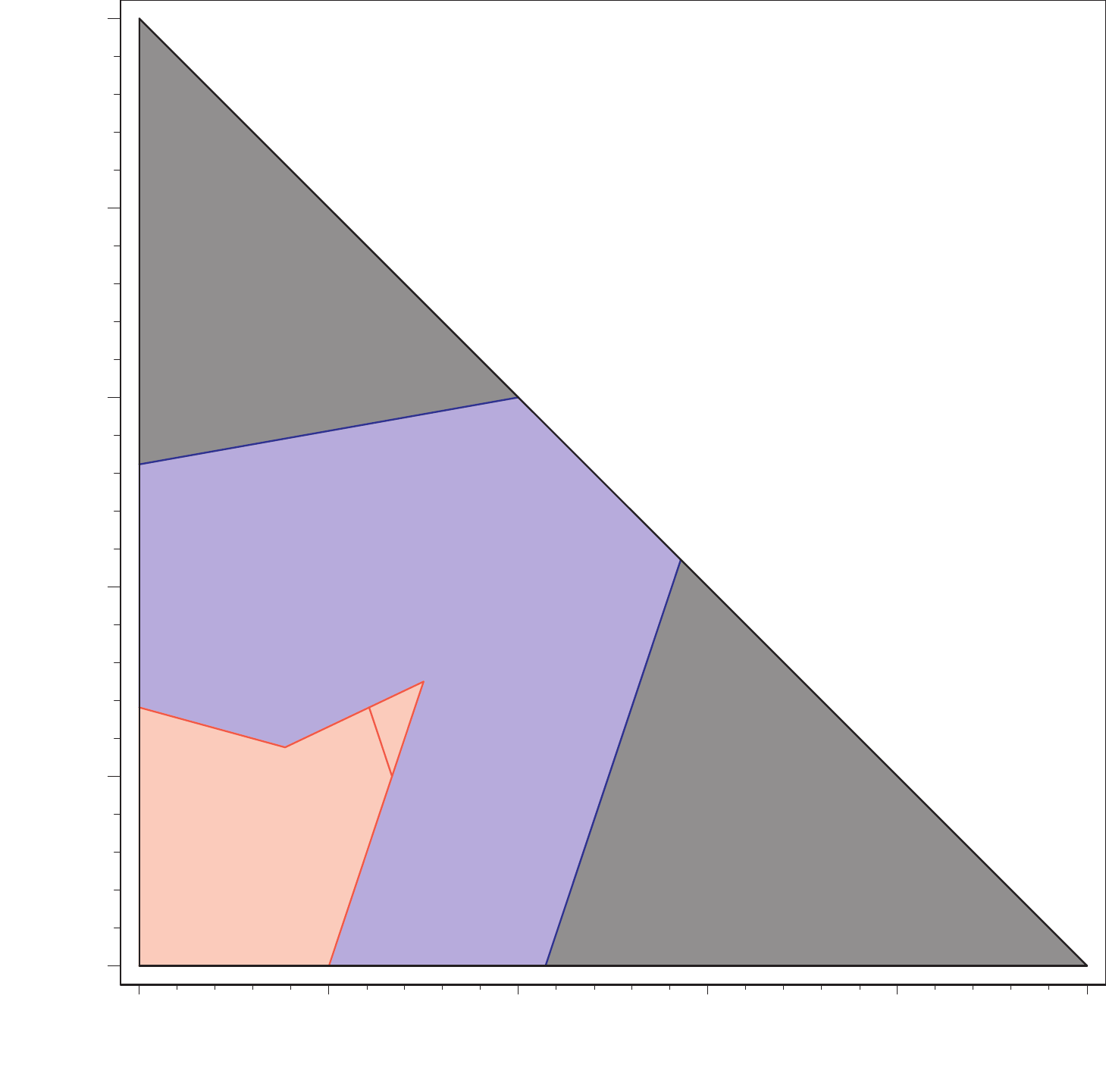}}
  \put(10,226){\makebox(0,0)[r]{\strut{}$w$}}
  \put(235,6){\makebox(0,0)[r]{\strut{}$g$}}
{\small
  \put(56,27){\makebox(0,0)[r]{\strut{}$0$}}
  \put(133,27){\makebox(0,0)[r]{\strut{}$0.2$}}
  \put(205,27){\makebox(0,0)[r]{\strut{}$0.4$}}
  \put(277,27){\makebox(0,0)[r]{\strut{}$0.6$}}
  \put(349,27){\makebox(0,0)[r]{\strut{}$0.8$}}
  \put(415,27){\makebox(0,0)[r]{\strut{}$1$}}
  \put(38,46){\makebox(0,0)[r]{\strut{}$0$}}
  \put(38,118){\makebox(0,0)[r]{\strut{}$0.2$}}
  \put(38,190){\makebox(0,0)[r]{\strut{}$0.4$}}
  \put(38,262){\makebox(0,0)[r]{\strut{}$0.6$}}
  \put(38,334){\makebox(0,0)[r]{\strut{}$0.8$}}
  \put(40,406){\makebox(0,0)[r]{\strut{}$1$}}
}
 \end{picture}
 \caption{Entropy criterion in the limit $q\to\infty$ for the state (\ref{eq:lo}) on the $g$-$w$-plane.
(Red domain: $S_\infty(\varrho) \geq S_\infty(\varrho_1)$ and $S_\infty(\varrho) \geq S_\infty(\varrho_{23})$,
blue domain:       $S_\infty(\varrho) \geq S_\infty(\varrho_1)$ and $S_\infty(\varrho)\ngeq S_\infty(\varrho_{23})$,
grey domains:      $S_\infty(\varrho)\ngeq S_\infty(\varrho_1)$ and $S_\infty(\varrho)\ngeq S_\infty(\varrho_{23})$.)}
 \label{fig:Sinf}
\end{figure}
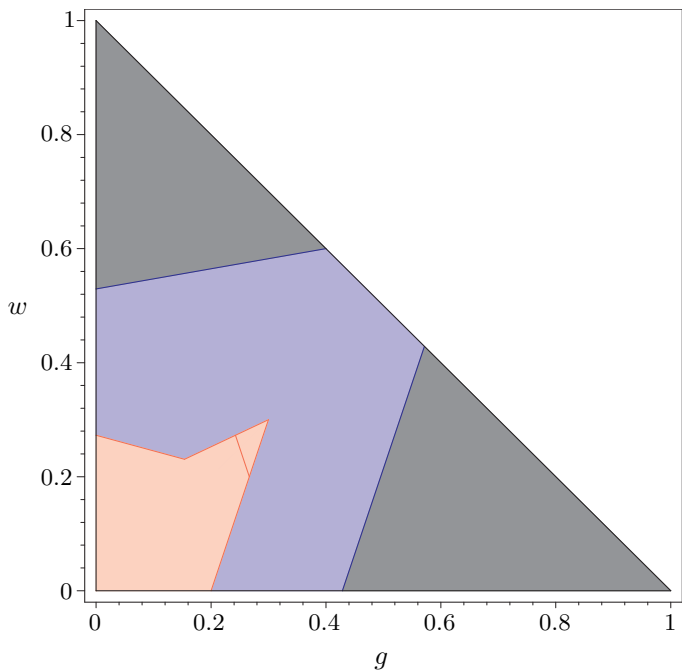

Consider now the entropy criterion in the $q\to\infty$ limit.
This can easily be done because the inequalities of the right-hand side~of (\ref{eq:critEntr})
are the same as the ones in the (i)th and (iii)th column of (\ref{eq:tablicsku}),
which are written on the maximal eigenvalues.
Hence in this case we have fewer restrictions,
and one can see in figure~\ref{fig:Sinf} that
the right-hand side~of (\ref{eq:critEntr}) indeed holds for more states than the right-hand side~of~(\ref{eq:critMaj})
in the case of the majorization criterion.
Hence \emph{the entropy criterion in the $q\to\infty$ limit
(quantum-Chebyshev entropy)
identifies a little bit fewer state to be entangled
than the majorization criterion in our case.}

Increasing $q$ from $0$ to $\infty$ one can see in figure~\ref{fig:S}
how the borderlines of the domains of the entropy criterion
shrink to the ones in figure~\ref{fig:Sinf}.
It is not true in general
that if $H_q(\tpl{p})\leq H_q(\tpl{q})$ and $q\geq q'$ then $H_{q'}(\tpl{p})\leq H_{q'}(\tpl{q})$.
For these particular spectra it seems like that the domains of smaller $q$s would contain the domains of larger $q$s,
but for the large values of $q$ one can see that this is not true (inset in figure \ref{fig:S}).
However, no line can cross the border of the domain of majorization criterion,
since the entropy criterion can not be stronger than the majorization criterion.

%\begin{figure}
% \includegraphics[width=0.6\textwidth]{fig_sepcritS}
%\end{figure}

\begin{figure}
 \setlength{\unitlength}{0.001428571\textwidth}% this 1/(420)*0.6 = 0.001428571
 \begin{picture}(420,415)
  \put(0,0){\includegraphics[width=0.6\textwidth]{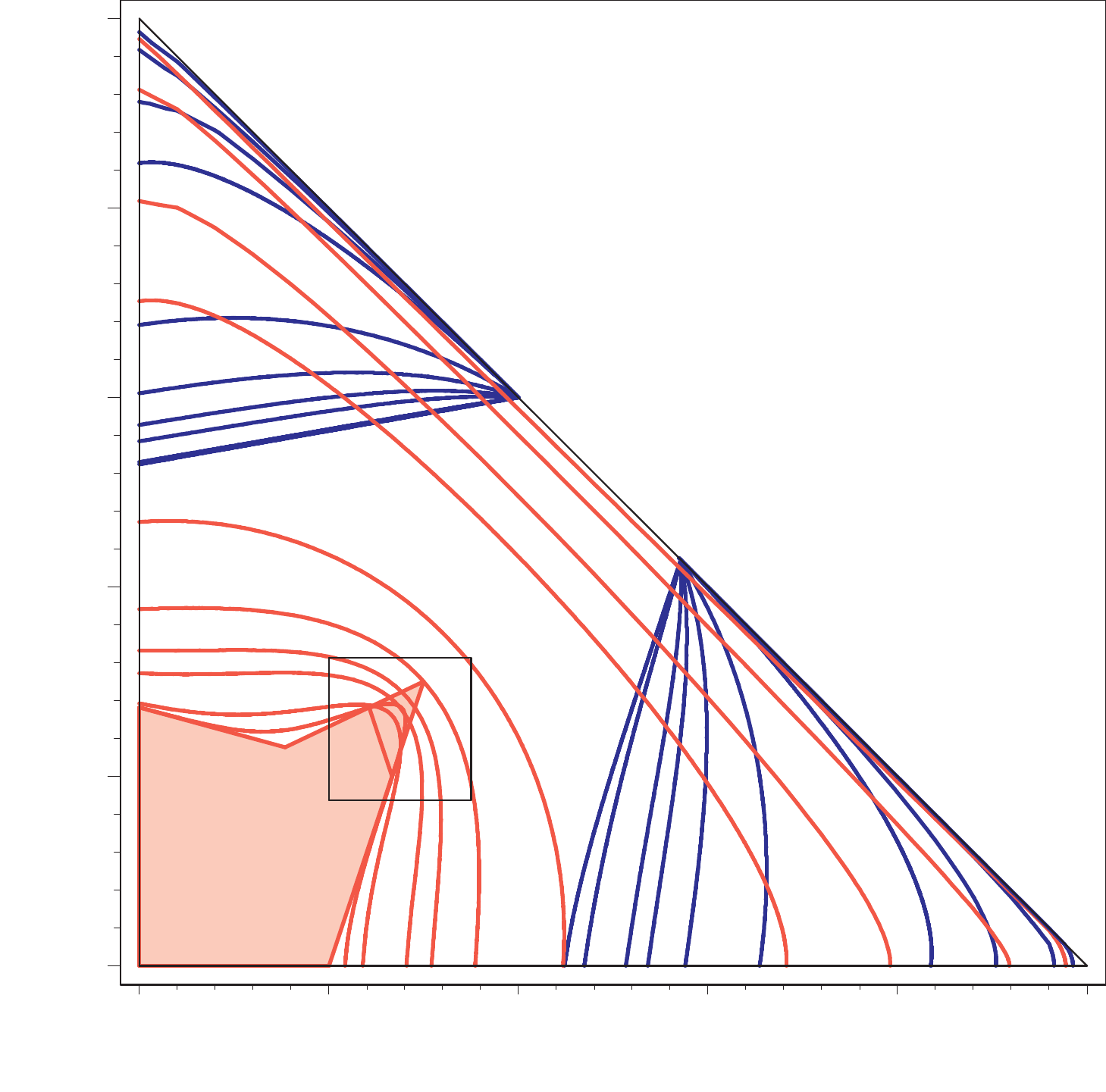}}
  \put(206,208){\includegraphics[width=0.3\textwidth]{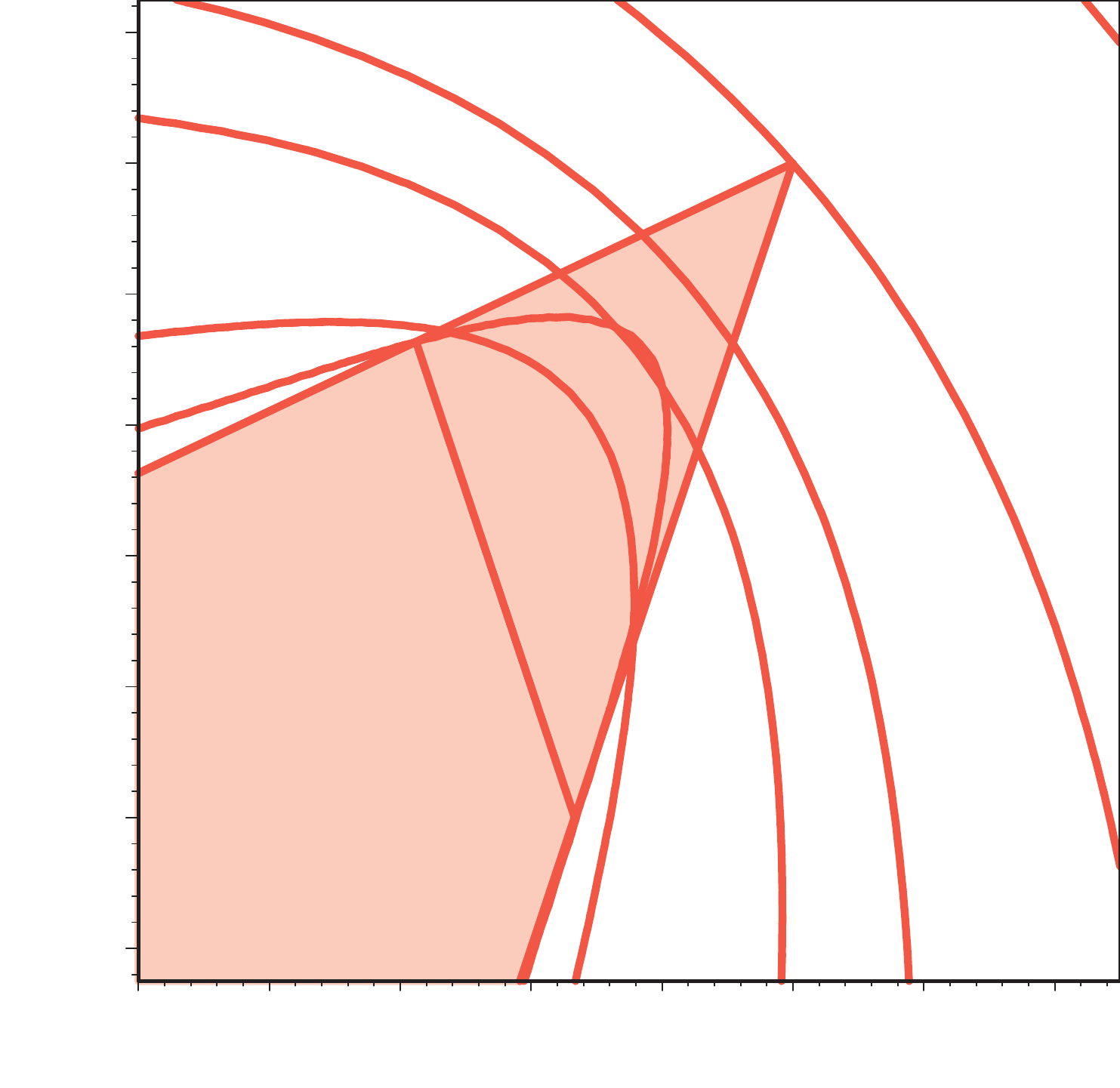}}
  \put(10,226){\makebox(0,0)[r]{\strut{}$w$}}
  \put(235,6){\makebox(0,0)[r]{\strut{}$g$}}
{\small
  \put(56,27){\makebox(0,0)[r]{\strut{}$0$}}
  \put(133,27){\makebox(0,0)[r]{\strut{}$0.2$}}
  \put(205,27){\makebox(0,0)[r]{\strut{}$0.4$}}
  \put(277,27){\makebox(0,0)[r]{\strut{}$0.6$}}
  \put(349,27){\makebox(0,0)[r]{\strut{}$0.8$}}
  \put(415,27){\makebox(0,0)[r]{\strut{}$1$}}
  \put(38,46){\makebox(0,0)[r]{\strut{}$0$}}
  \put(38,118){\makebox(0,0)[r]{\strut{}$0.2$}}
  \put(38,190){\makebox(0,0)[r]{\strut{}$0.4$}}
  \put(38,262){\makebox(0,0)[r]{\strut{}$0.6$}}
  \put(38,334){\makebox(0,0)[r]{\strut{}$0.8$}}
  \put(40,406){\makebox(0,0)[r]{\strut{}$1$}} }
{\scriptsize
  \put(249,219){\makebox(0,0)[r]{\strut{}$0.2$}}
  \put(292,219){\makebox(0,0)[r]{\strut{}$0.24$}}
  \put(340,219){\makebox(0,0)[r]{\strut{}$0.28$}}
  \put(390,219){\makebox(0,0)[r]{\strut{}$0.32$}}
  \put(228,257){\makebox(0,0)[r]{\strut{}$0.2$}}
  \put(226,306){\makebox(0,0)[r]{\strut{}$0.24$}}
  \put(227,355){\makebox(0,0)[r]{\strut{}$0.28$}}
  \put(227,404){\makebox(0,0)[r]{\strut{}$0.32$}}
}
 \end{picture}
 \caption{Entropy criterion for
$q=1/4, 2/4, 3/4, 1, 2, 3, 4, 5, 10, 20$ for the state (\ref{eq:lo}) on the $g$-$w$-plane.
(Red  curves: the border of the domain $S_q(\varrho)\geq S_q(\varrho_1)$ and $S_q(\varrho) \geq S_q(\varrho_{23})$,
blue curves: the border of the domain $S_q(\varrho)\geq S_q(\varrho_1)$ and $S_q(\varrho)\ngeq S_q(\varrho_{23})$.
Red domain: copied from figure~\ref{fig:Sinf} of the $q\to\infty$ case.)}
 \label{fig:S}
\end{figure}
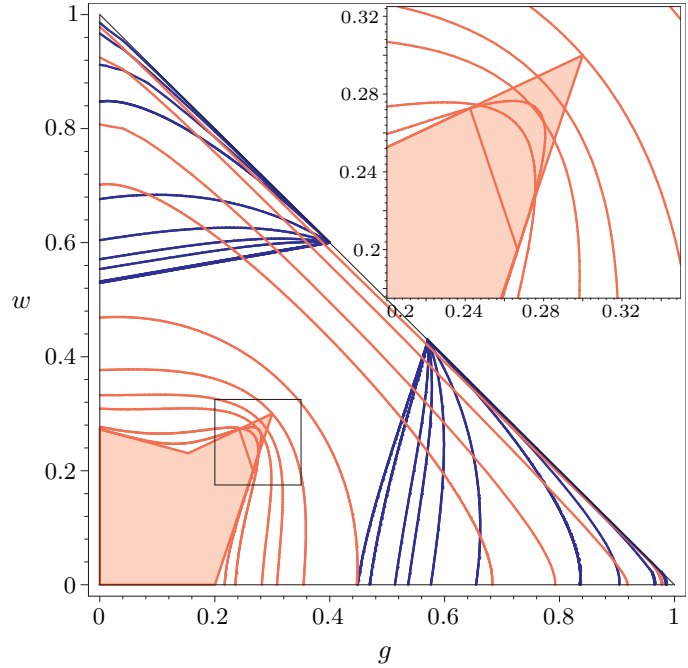

Although the entropic criteria maybe plausible and motivated physically,
since those are related to the mixedness and entropies,
but note that spectral properties are not sufficient for the detection of entanglement.
An example for this can be given by the following two density matrices:
\begin{equation*}
\frac13\begin{bmatrix}
 1    & \cdot & \cdot & \cdot \\
\cdot &  1    &  1    & \cdot \\
\cdot &  1    &  1    & \cdot \\
\cdot & \cdot & \cdot & \cdot
\end{bmatrix},\qquad \qquad
\frac13\begin{bmatrix}
 2    & \cdot & \cdot & \cdot \\
\cdot & \cdot & \cdot & \cdot \\
\cdot & \cdot & \cdot & \cdot \\
\cdot & \cdot & \cdot &  1   
\end{bmatrix}.
\end{equation*}
Both of them have the same spectrum, 
and the states of the subsystems have the same spectrum as well,
so entropy and majorization criteria give the same for both of them.
But the first one is the state of the bipartite subsystem of the W state (\ref{eq:W}), 
and it is entangled,
while the second one is diagonal, hence separable.

%******************************************************************************
\subsection{Partial transposition criterion}
%\label{sec:PPT}
\label{subsec:SepCrit.2Part.PPT}

We recall now the \emph{partial transposition criterion} of Peres % (\ref{eq:critPeres}) of Peres
(section~\ref{subsec:QM.Ent.2Part}).
It considers the positivity of the partially transposed bipartite density matrix,
and for our case when $d_A=2$ and $d_B=4$ it is
\begin{equation}
\label{eq:critPPT}
\varrho\quad\text{separable}\qquad\Longrightarrow\qquad \varrho^{\transp_A} \geq 0.
\end{equation}
The partial transposition criterion is the consequence of the \emph{positive maps criterion} \cite{HorodeckiPosMapWitness}.
It states that
\begin{equation}
\label{eq:critPosMaps}
\varrho\quad\text{separable}\qquad\Longleftrightarrow\qquad
\bigl(\Phi\otimes\Id\bigr)(\varrho) \geq 0\quad
\text{for all positive $\Phi\in\Lin\bigl(\Lin(\mathcal{H}_1)\bigr)$}.
\end{equation}
This is a necessary and sufficient criterion, but we are not able to check it for all $\Phi$.
But we can consider a particular class of positive maps to obtain necessary but not sufficient criteria.
For example for $\Phi(\omega)=\omega^\transp$ we get back the partial transposition criterion.%
%%%%%%%%%%%%%%%%%%%%%%%%
\footnote{Note that for completely positive maps $\bigl(\Phi\otimes\Id\bigr)(\varrho) \geq 0$
holds by definition, so it is enough to consider only positive but not completely positive maps.}
%%%%%%%%%%%%%%%%%%%%%%%%

Let us apply the partial transposition criterion to the state $\varrho$ of equation (\ref{eq:lo}),
which results the matrix
\begin{equation}
\label{eq:MxRT1}
\varrho^{\transp_1}=\begin{bmatrix}
 \td+\tg & \cdot   & \cdot   & \cdot   & \cdot   & \tw     & \tw     & \cdot   \\
 \cdot   & \td+\tw & \tw     & \cdot   & \cdot   & \cdot   & \cdot   & \cdot   \\
 \cdot   & \tw     & \td+\tw & \cdot   & \cdot   & \cdot   & \cdot   & \cdot   \\
 \cdot   & \cdot   & \cdot   & \td     & \tg     & \cdot   & \cdot   & \cdot   \\
 \cdot   & \cdot   & \cdot   & \tg     & \td+\tw & \cdot   & \cdot   & \cdot   \\
 \tw     & \cdot   & \cdot   & \cdot   & \cdot   & \td     & \cdot   & \cdot   \\
 \tw     & \cdot   & \cdot   & \cdot   & \cdot   & \cdot   & \td     & \cdot   \\
 \cdot   & \cdot   & \cdot   & \cdot   & \cdot   & \cdot   & \cdot   & \td+\tg 
\end{bmatrix}.
\end{equation}
The spectrum of $\varrho^{\transp_1}$ can easily be calculated due to its block-structure, leading to
\begin{equation}
\label{eq:spectRhoPT}
\Spect\varrho^{\transp_1}=
\begin{aligned}[t]\bigl\{\;
\td+\tw/2\pm\sqrt{4\tg^2+\tw^2}/2 &=(3-3g+w\pm4\sqrt{9g^2+w^2})/24,\\ 
\td+\tg/2\pm\sqrt{\tg^2+8\tw^2}/2 &=(3+3g-3w\pm2\sqrt{9g^2+32w^2})/24,\\
\td+\tg &=(3+9g-3w)/24,\\
\td+2\tw &=(3-3g+13w)/24,\\
\td &=(3-3g-3w)/24\quad\text{(2 times)}\;\bigr\}.
\end{aligned}
\end{equation}
Only the lower-sign version of the first two pairs of eigenvalues can be less than zero
hence we get two inequalities for the positivity of $\varrho^{\transp_1}$:
\begin{subequations}
\begin{align}
\varrho\in\bigcap_{\alpha_2}\mathcal{D}_{\alpha_2}\qquad\Longrightarrow\qquad 
\label{eq:critPPT1}
&\left\{
\begin{aligned}
0&\leq \td^2+\td\tw-\tg^2\\
0&\leq -135g^2- 15w^2- 6gw-18g+ 6w+9,
\end{aligned}
\right.\\
\label{eq:critPPT2}
&\left\{
\begin{aligned}
0&\leq \td^2+\td\tg-2\tw^2\\
0&\leq - 27g^2-119w^2-18gw+18g-18w+9.
\end{aligned}
\right.
\end{align}
\end{subequations}
Each inequality of these holds inside an ellipse.
These ellipses intersect nontrivially
and in the intersection the right-hand side~of (\ref{eq:critPPT}) holds.
(Red curves in figure~\ref{fig:PPTRed}.)
The parameter values $g=1/5$ and $w=(24\sqrt2-9)/119=0.209589\dots$ are
the bounds for the union of Class 2.8 and 3
for the GHZ-white noise ($w=0$) and the W-white noise ($g=0$) mixtures respectively.

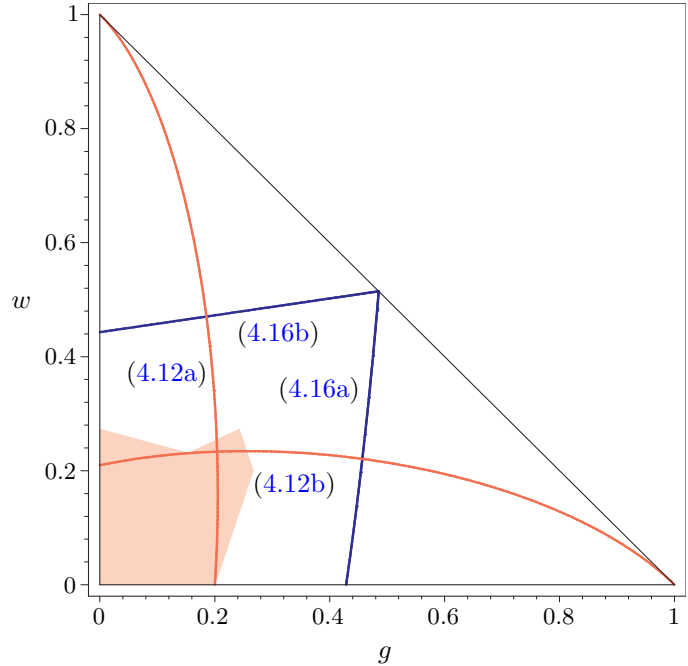
\begin{figure}
 \setlength{\unitlength}{0.001428571\textwidth}% this 1/(420)*0.6 = 0.001428571
 \begin{picture}(420,415)
  \put(0,0){\includegraphics[width=0.6\textwidth]{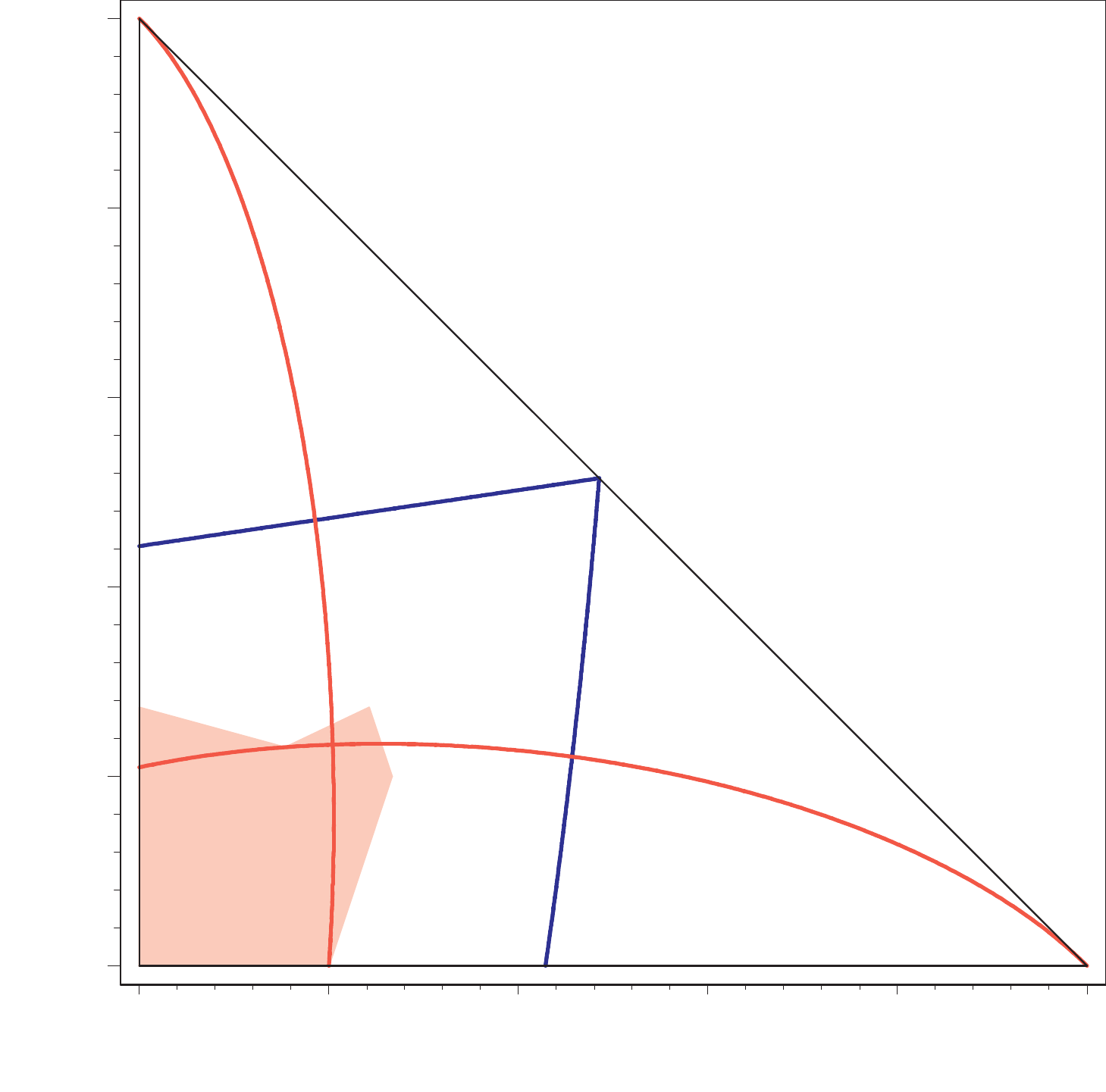}}
  \put(10,226){\makebox(0,0)[r]{\strut{}$w$}}
  \put(235,6){\makebox(0,0)[r]{\strut{}$g$}}
{\small
  \put(56,27){\makebox(0,0)[r]{\strut{}$0$}}
  \put(133,27){\makebox(0,0)[r]{\strut{}$0.2$}}
  \put(205,27){\makebox(0,0)[r]{\strut{}$0.4$}}
  \put(277,27){\makebox(0,0)[r]{\strut{}$0.6$}}
  \put(349,27){\makebox(0,0)[r]{\strut{}$0.8$}}
  \put(415,27){\makebox(0,0)[r]{\strut{}$1$}}
  \put(38,46){\makebox(0,0)[r]{\strut{}$0$}}
  \put(38,118){\makebox(0,0)[r]{\strut{}$0.2$}}
  \put(38,190){\makebox(0,0)[r]{\strut{}$0.4$}}
  \put(38,262){\makebox(0,0)[r]{\strut{}$0.6$}}
  \put(38,334){\makebox(0,0)[r]{\strut{}$0.8$}}
  \put(40,406){\makebox(0,0)[r]{\strut{}$1$}}
}
  \put(120,180){\makebox(0,0)[r]{\strut{}(\ref{eq:critPPT1})}}
  \put(200,110){\makebox(0,0)[r]{\strut{}(\ref{eq:critPPT2})}}
  \put(215,170){\makebox(0,0)[r]{\strut{}(\ref{eq:critRed3})}}
  \put(190,205){\makebox(0,0)[r]{\strut{}(\ref{eq:critRed4})}}
 \end{picture}
 \caption{Partial transposition and reduction criteria for the state (\ref{eq:lo}) on the $g$-$w$-plane.
(Inequalities (\ref{eq:critPPT1}) and~(\ref{eq:critPPT2}) of partial transposition criterion hold inside the intersection of the red ellipses.
Blue curves: the borders of domains inside the additional inequalities (\ref{eq:critRed3}) and~(\ref{eq:critRed4}) of reduction criterion hold.
Red domain: copied from figure~\ref{fig:Maj} of majorization criterion.)}
 \label{fig:PPTRed}
\end{figure}

The partial transposition criterion states that
if a state is in Classes 2.8 or 3 then its parameters are inside the intersection of the ellipses,
but there can also be states of Classes 2.1 or 1 in this domain.
On the other hand the states must be in Classes 2.1 or 1 for parameters outside.
\emph{The inequalities of (\ref{eq:critPPT1}) and~(\ref{eq:critPPT2}) are strong in detection of GHZ and W states, respectively.}
In figure~\ref{fig:PPTRed} we have also plotted the corresponding domain of the majorization criterion.
(One can check that the only intersection-points of the borderlines
of the corresponding domains of the two criteria
are $(g=2/13,w=3/13)$ and $(g=1/5,w=0)$.
This criterion is also a necessary and sufficient one for the
full separability of the $w=0$ GHZ-white noise mixture.)
It can be seen that
\emph{the partial transposition criterion
gives stronger condition than the majorization criterion,}
it identifies more state to be in Classes 2.1 or 1.
Hence the majorization criterion can not identify entangled states of positive partial transpose (PPTES)
on the simplex defined in equation (\ref{eq:lo}).

%******************************************************************************
\subsection{Reduction criterion}
%\label{sec:Red}
\label{subsec:SepCrit.2Part.Red}

The next one of the examined criteria is the \emph{reduction criterion}~\cite{HorodeckiRedCrit,CerfRedCrit}.
It states that
\begin{equation}
\label{eq:critRed}
\varrho\quad\text{separable}\qquad\Longrightarrow\qquad 
\varrho_A\otimes\Id_B-\varrho\geq0 \quad\text{and}\quad \Id_A\otimes\varrho_B-\varrho\geq0.
\end{equation}
This is the consequence of the positive maps criterion (\ref{eq:critPosMaps})
for the particular positive map $\Phi(\omega)=(\tr\omega)\Id - \omega$.
The importance of this criterion is that
its violation is sufficient criterion of \emph{distillability}~\cite{HorodeckiRedCrit}.
It is known that
\emph{the reduction criterion can not be stronger than the partial transposition criterion
and they are equivalent for qubit-qudit systems}~\cite{HorodeckiRedCrit}.
Since our state $\varrho$ defined in equation (\ref{eq:lo}) is the permutation invariant one of three qubits
considered as a $\tpl{d}=(2,4)$ qubit-qudit system,
the equivalence of these two criteria means that
some kind of pure state entanglement between $1$ and $23$ can be distilled out
from every state of non-positive partial transpose.
In other words \emph{in the simplex defined by equation (\ref{eq:lo})
there are no bound entangled $\tpl{d}=(2,4)$ states of non-positive partial transpose,}
while the entangled states of positive partial transpose are all bound entangled,
which is a general result \cite{Horodecki3BoundEnt}.
% amit desztillalok, az nem biztos, hogy perm.inv.

We can illustrate the equivalence of the partial transposition and reduction criteria.
To do this we have to examine the positivity of the matrices
$\Id_1\otimes\varrho_{23}-\varrho$ and
$\varrho_1\otimes\Id_{23}-\varrho$, which are of the form
\begin{subequations}
\begin{equation}
\label{eq:MxRRI123}
\Id_1\otimes\varrho_{23}-\varrho=\begin{bmatrix}
 \td+\tw & \cdot   & \cdot   & \cdot   & \cdot   & \cdot   & \cdot   &-\tg     \\
 \cdot   & \td     & \cdot   & \cdot   &-\tw     & \cdot   & \cdot   & \cdot   \\
 \cdot   & \cdot   & \td     & \cdot   &-\tw     & \cdot   & \cdot   & \cdot   \\
 \cdot   & \cdot   & \cdot   & \td+\tg & \cdot   & \cdot   & \cdot   & \cdot   \\
 \cdot   &-\tw     &-\tw     & \cdot   & \td+\tg & \cdot   & \cdot   & \cdot   \\
 \cdot   & \cdot   & \cdot   & \cdot   & \cdot   & \td+\tw & \tw     & \cdot   \\
 \cdot   & \cdot   & \cdot   & \cdot   & \cdot   & \tw     & \td+\tw & \cdot   \\
-\tg     & \cdot   & \cdot   & \cdot   & \cdot   & \cdot   & \cdot   & \td 
\end{bmatrix},
\end{equation}
{\setlength{\mathindent}{0pt}
\begin{multline}
\label{eq:MxRR1I23}
\varrho_1\otimes\Id_{23}-\varrho=\\
\begin{bmatrix}
3\td+2\tw    & \cdot       & \cdot       & \cdot       & \cdot       & \cdot       & \cdot       &-\tg     \\
 \cdot       &3\td+\tg+\tw &-\tw         & \cdot       &-\tw         & \cdot       & \cdot       & \cdot   \\
 \cdot       &-\tw         &3\td+\tg+\tw & \cdot       &-\tw         & \cdot       & \cdot       & \cdot   \\
 \cdot       & \cdot       & \cdot       &3\td+\tg+2\tw& \cdot       & \cdot       & \cdot       & \cdot   \\
 \cdot       &-\tw         &-\tw         & \cdot       &3\td+\tg     & \cdot       & \cdot       & \cdot   \\
 \cdot       & \cdot       & \cdot       & \cdot       & \cdot       &3\td+\tg+\tw & \cdot       & \cdot   \\
 \cdot       & \cdot       & \cdot       & \cdot       & \cdot       & \cdot       &3\td+\tg+\tw & \cdot   \\
-\tg         & \cdot       & \cdot       & \cdot       & \cdot       & \cdot       & \cdot       &3\td+\tw
\end{bmatrix}.
\end{multline}}\noindent
\end{subequations}
Since $(\tr\omega)\Id-\omega = (\varepsilon\omega\varepsilon^\dagger)^\transp$ for $2\times2$ matrices
%(with the Pauli matrix $\sigma_2=\left[\begin{smallmatrix}0&-i\\i&0\end{smallmatrix}\right]$)
it turns out that
\begin{subequations}
\begin{align}
\label{eq:spectRhoI123}
\Spect(\Id_1\otimes\varrho_{23}-\varrho)&=\Spect\varrho^{\transp_1},
\intertext{while}
\label{eq:spectRho1I23}
\Spect(\varrho_1\otimes\Id_{23}-\varrho)&= 
\begin{aligned}[t]\bigl\{
3\td+3\tw/2\pm\sqrt{4\tg^2+\tw^2}/2 &=(9-9g+3w\pm4\sqrt{9g^2+w^2})/24, \\
3\td+\tg\pm\sqrt{2}\tw &=(9+3g-9w\pm8\sqrt{2}w)/24, \\
3\td+\tg+2\tw &=(9+3g+7w)/24\quad\text{(2 times)}, \\
3\td+\tg+\tw &=(9+3g-w)/24\quad\text{(2 times)}\;\bigr\}.
\end{aligned}
\end{align}
\end{subequations}
For $\Id_1\otimes\varrho_{23}-\varrho\geq0$ we have the same conditions as in equations (\ref{eq:critPPT1})-(\ref{eq:critPPT2})
of the partial transposition criterion.
The additional inequalities arise from the lower-sign version of the first two eigenvalues of $\varrho_1\otimes\Id_{23}-\varrho$,
leading to the criteria
\begin{subequations}
\begin{align}
\varrho\in\bigcap_{\alpha_2}\mathcal{D}_{\alpha_2}\qquad\Longrightarrow\qquad 
\label{eq:critRed3}
&\left\{
\begin{aligned}
0&\leq 9\td^2-\tg^2+9\td\tw+2\tw^2\\
0&\leq -63g^2- 7w^2- 54gw-162g+ 54w+81,
\end{aligned}
\right.\\
\label{eq:critRed4}
&\left\{
\begin{aligned}
0&\leq 3\td+\tg-\sqrt{2}\tw\\
0&\leq 3g-(9+8\sqrt{2})w+9.
\end{aligned}
\right.
\end{align}
\end{subequations}
The first one of them is true outside a hyperbola,
the second one is true under a line.~(Blue curves in figure~\ref{fig:PPTRed}.)

It can be seen that the last two inequalities (\ref{eq:critRed3})-(\ref{eq:critRed4}) do not restrict the
ones in equations (\ref{eq:critPPT1})-(\ref{eq:critPPT2}),
as it has to be,
and because of (\ref{eq:spectRhoI123})
\emph{the reduction criterion and the partial transposition criterion
hold for the same states of the GHZ-W-white noise mixture.}
Here we get the stronger condition for the map $\Phi(\omega)=(\tr\omega)\Id - \omega$
acting on the smaller subsystem.
We can also observe that
the inequalities of (\ref{eq:critRed3}) and~(\ref{eq:critRed4}) are good in detection of GHZ and W state respectively,
but not so good as the ones of partial transposition criterion.
However, one can check that on the $w=0$ GHZ-white noise mixture
the reduction criterion
$\Id_1\otimes\varrho_{23}-\varrho\geq0$ and $\varrho_1\otimes\Id_{23}-\varrho\geq0$
is necessary and sufficient for full-separability,
the criterion
$\Id_1\otimes\varrho_{23}-\varrho\ngeq0$ and $\varrho_1\otimes\Id_{23}-\varrho\geq0$
is necessary and sufficient for Class 2.1,
and the criterion
$\Id_1\otimes\varrho_{23}-\varrho\ngeq0$ and $\varrho_1\otimes\Id_{23}-\varrho\ngeq0$
is necessary and sufficient for Class 1
in the same fashion as in the majorization criterion of section~\ref{subsec:SepCrit.2Part.Maj}.

%******************************************************************************
\subsection{Reshuffling criterion}
%\label{sec:Resh}
\label{subsec:SepCrit.2Part.Resh}

The \emph{reshuffling criterion} is independent of the partial transposition criterion,
so it can detect entangled states of positive partial transpose.
It states that
\begin{equation}
\label{eq:critReshuff}
\varrho\quad\text{separable}\qquad\Longrightarrow\qquad 
\norm{R(\varrho)}_{\tr} \leq 1,
\end{equation}
where the \emph{trace-norm} is $\norm{M}_{\tr}=\tr\sqrt{M^\dagger M}$,
and the \emph{reshuffling map} $R$ is defined on matrix elements as
$[R(\varrho)]^{i\phantom{i'}j\phantom{j'}}_{\phantom{i}i'\phantom{j}j'}=\varrho^{ij}_{\phantom{ij}i'j'}$.

The four nonzero singular values of the $4\times16$ reshuffled density matrix
\setcounter{MaxMatrixCols}{16}
\begin{equation}
\label{eq:MxRR24}
R(\varrho)=\begin{bmatrix}
 \td+\tg & \cdot   & \cdot   & \cdot   & \cdot   & \td+\tw & \tw     & \cdot   & \cdot   & \tw     & \td+\tw & \cdot   & \cdot   & \cdot   & \cdot   & \td     \\
 \cdot   & \cdot   & \cdot   & \tg     & \tw     & \cdot   & \cdot   & \cdot   & \tw     & \cdot   & \cdot   & \cdot   & \cdot   & \cdot   & \cdot   & \cdot   \\
 \cdot   & \tw     & \tw     & \cdot   & \cdot   & \cdot   & \cdot   & \cdot   & \cdot   & \cdot   & \cdot   & \cdot   & \tg     & \cdot   & \cdot   & \cdot   \\
 \td+\tw & \cdot   & \cdot   & \cdot   & \cdot   & \td     & \cdot   & \cdot   & \cdot   & \cdot   & \td     & \cdot   & \cdot   & \cdot   & \cdot   & \td+\tg 
\end{bmatrix},
\end{equation}
that is, the square root of the nonnegative eigenvalues of $R(\varrho)^\dagger R(\varrho)$ are
\begin{equation}
\Spect\sqrt{R(\varrho)R(\varrho)^\dagger} =\left\{
\frac12\sqrt{p_1\pm2\sqrt{p_2}},
\sqrt{\tg^2+2\tw^2},
\sqrt{\tg^2+2\tw^2}
\right\},
\end{equation}
where
\begin{equation*}
\begin{split}
p_1=& 16\td^2 +4\tg^2 +10\tw^2 +8\td\tg +12\td\tw,\\
p_2=& 64\td^4 +9\tw^4 +64\td^3\tg +96\td^3\tw +12\td\tw^3 \\
& +16\td^2\tg^2 +40\td^2\tw^2 +4\tg^2\tw^2 %\\
 +80\td^2\tg\tw +16\td\tg^2\tw +24\td\tg\tw^2.
\end{split}
\end{equation*}
The sum of them is less or equal than $1$ inside a curve of high degree
which can be seen in figure~\ref{fig:Resh} (red curve).
States of Classes 2.8 and 3 must be inside this curve,
states outside this curve must belong to Classes 2.1 or 1,
but one can see that \emph{this criterion does not restrict the partial transposition criterion,}
it does not detect PPTESs
in the GHZ-W-white noise mixture (\ref{eq:lo}).

\begin{figure}
 \setlength{\unitlength}{0.001428571\textwidth}% this 1/(420)*0.6 = 0.001428571
 \begin{picture}(420,415)
  \put(0,0){\includegraphics[width=0.6\textwidth]{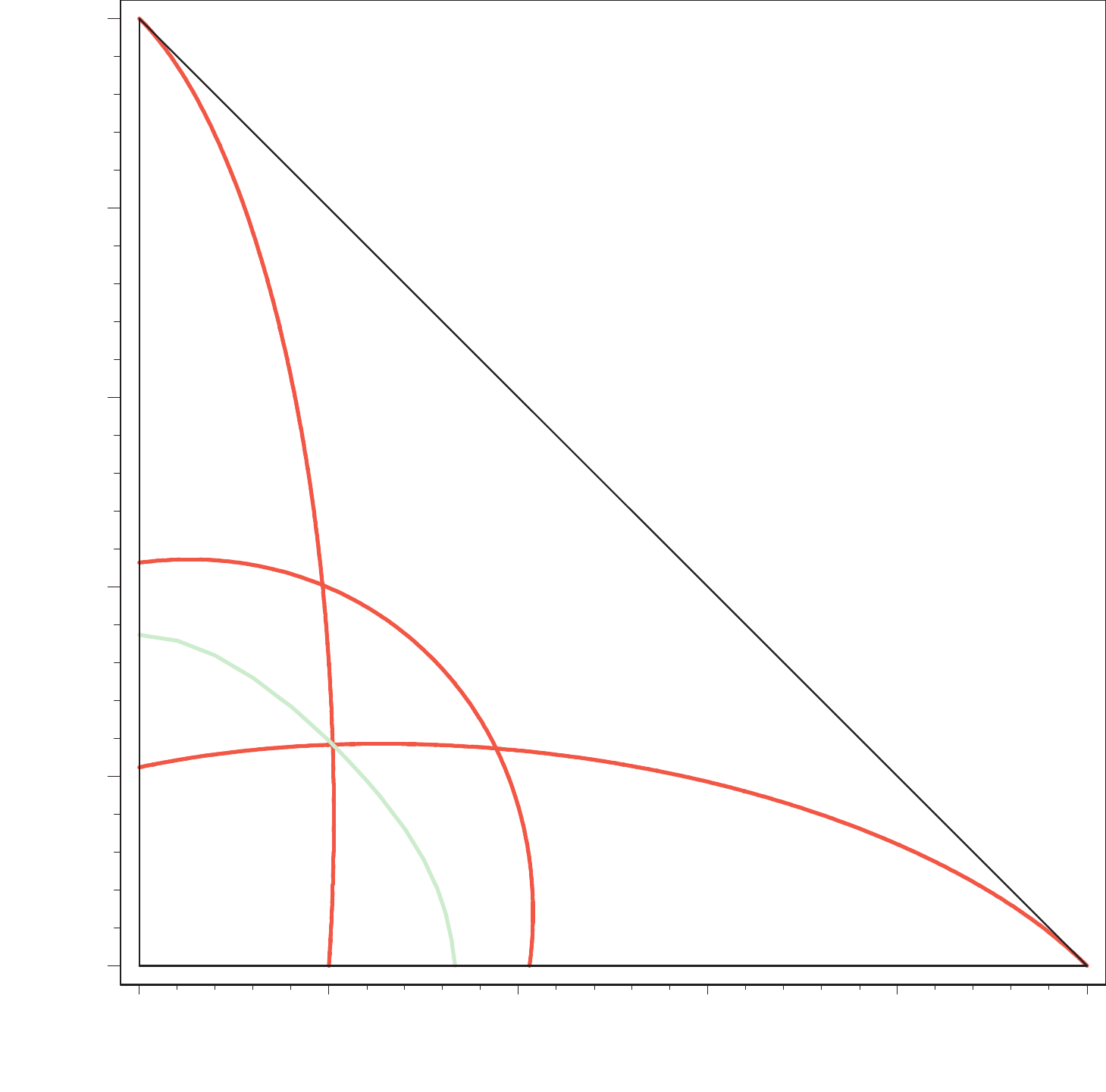}}
  \put(10,226){\makebox(0,0)[r]{\strut{}$w$}}
  \put(235,6){\makebox(0,0)[r]{\strut{}$g$}}
{\small
  \put(56,27){\makebox(0,0)[r]{\strut{}$0$}}
  \put(133,27){\makebox(0,0)[r]{\strut{}$0.2$}}
  \put(205,27){\makebox(0,0)[r]{\strut{}$0.4$}}
  \put(277,27){\makebox(0,0)[r]{\strut{}$0.6$}}
  \put(349,27){\makebox(0,0)[r]{\strut{}$0.8$}}
  \put(415,27){\makebox(0,0)[r]{\strut{}$1$}}
  \put(38,46){\makebox(0,0)[r]{\strut{}$0$}}
  \put(38,118){\makebox(0,0)[r]{\strut{}$0.2$}}
  \put(38,190){\makebox(0,0)[r]{\strut{}$0.4$}}
  \put(38,262){\makebox(0,0)[r]{\strut{}$0.6$}}
  \put(38,334){\makebox(0,0)[r]{\strut{}$0.8$}}
  \put(40,406){\makebox(0,0)[r]{\strut{}$1$}}
}
  \put(110,260){\makebox(0,0)[r]{\strut{}(\ref{eq:critPPT1})}}
  \put(280,100){\makebox(0,0)[r]{\strut{}(\ref{eq:critPPT2})}}
 \end{picture}
 \caption{Reshuffling criteria for the state (\ref{eq:lo}) on the $g$-$w$-plane.
(Red  curve: reshuffling criterion for $\tpl{d}=(2,4)$ system as in section~\ref{subsec:SepCrit.2Part.Resh},
green curve: reshuffling criterion for $\tpl{d}=(2,2,2)$ system as in section~\ref{subsec:SepCrit.3Part.Perm}.
We have also copied the borderlines of the domains in which (\ref{eq:critPPT1}) and~(\ref{eq:critPPT2}) of the partial transposition criterion hold
from figure~\ref{fig:PPTRed}.
The inequalities hold on the side of the curves containing the origin.)}
 \label{fig:Resh}
\end{figure}
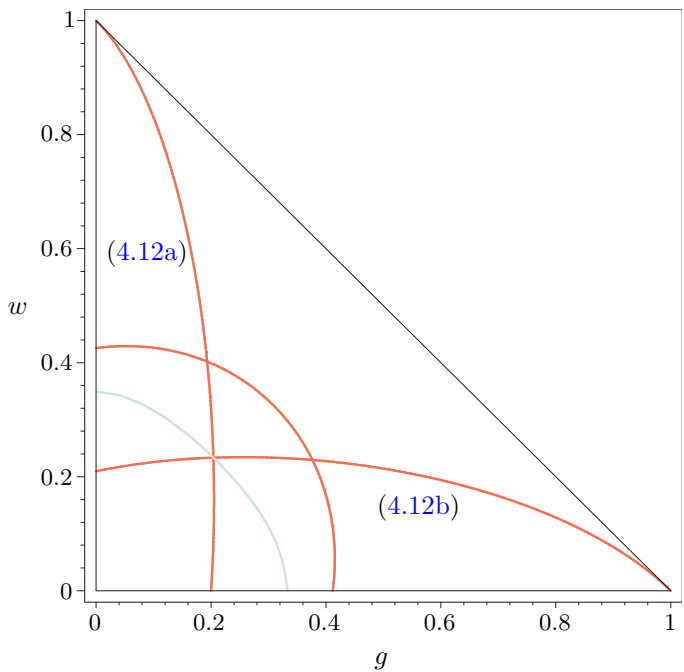

%******************************************************************************
\section{Tripartite separability criteria}
%\label{sec:Threepart}
\label{sec:SepCrit.3Part}

In this section we consider the state given in equation (\ref{eq:lo})
as the state of a proper $\tpl{d}=(2,2,2)$ three-qubit system
and investigate some general $3$-qubit $k$-separability criteria.

%******************************************************************************
\subsection{Permutation criterion}
%\label{sec:Perm}
\label{subsec:SepCrit.3Part.Perm}

First consider the \emph{permutation criterion} in general, which is given in \cite{HorodeckiPerm}.
Note that the reshuffling and the partial transpose of a density matrix are
nothing else than the permutations of the local matrix indices.
Moreover, since
the trace norm is the sum of the absolute values of the eigenvalues for hermitian matrices
and the trace is invariant under partial transposition
it turns out that $\varrho^{\transp_1}\geq0$ if and only if $\norm{\varrho^{\transp_1}}_{\tr}=1$.
So the partial transposition criterion (\ref{eq:critPPT})
and the reshuffling criterion (\ref{eq:critReshuff}) can be formulated in the same fashion.
Moreover, this can be done for $n$ subsystems in a general way~\cite{HorodeckiPerm}.

Let $\sigma\in \DscrGrp{S}_{2n}$ be a permutation of the $2n$ matrix indices
and let $\Phi_\sigma$ the map realizing this index permutation.
On elementary tensors %$\varphi_1\otimes\varphi_2\otimes\dots\otimes\varphi_{2n}\in\mathcal{H}_1\otimes\mathcal{H}_2\otimes\dots\otimes\mathcal{H}_{2n}$
 it permutes the factors
$\Phi_\sigma(\varphi_1\otimes\varphi_2\otimes\dots\otimes\varphi_{2n})
=\varphi_{\sigma(1)}\otimes\varphi_{\sigma(2)}\otimes\dots\otimes\varphi_{\sigma(2n)}$,
where $\varphi_i$ is an element of $\mathcal{H}_j$ or $\mathcal{H}_j^*$,
this is why we have not used the bra-ket notations.
If we apply this to a density operator 
\begin{equation*}
\varrho=\sum \varrho^{i_1i_2\dots i_n}_{\phantom{i_1i_2\dots i_n}i_{n+1}i_{n+2}\dots i_{2n}}
\cket{i_1i_2\dots i_n}\bra{i_{n+1}i_{n+2}\dots i_{2n}},
\end{equation*}
which is regarded as an element of
$\Lin(\mathcal{H}_1\otimes\mathcal{H}_2\otimes\dots\otimes\mathcal{H}_n)\isom
\mathcal{H}_1\otimes\mathcal{H}_2\otimes\dots\otimes\mathcal{H}_n\otimes
\mathcal{H}_1^*\otimes\mathcal{H}_2^*\otimes\dots\otimes\mathcal{H}_n^*$,
then the resulting operator is not a linear transformation of a particular Hilbert space in general.%
%%%%%%%%%%%%%%%%%%%%%%%%
\footnote{It is not unique, hence important to specify, how the tensor factors give rise to the linear operator,
because the permutation criterion is formulated via the use of an operator norm
of operators mapping between different Hilbert spaces.
Here we adopt the convention that if we have a tensor of $2n$ factors,
then it represents a linear map from the dual of the second $n$ factors to the first $n$ factors.}
%%%%%%%%%%%%%%%%%%%%%%%%
For example, for the permutation $\sigma=(35)\in \DscrGrp{S}_6$ in the $n=3$ case,
the reshuffled state $\Phi_{(35)}(\varrho)\in
\mathcal{H}_1\otimes\mathcal{H}_2\otimes\mathcal{H}_2^*\otimes
\mathcal{H}_1^*\otimes\mathcal{H}_3\otimes\mathcal{H}_3^*$,
so it is regarded as an element of 
$\Lin(\mathcal{H}_1\otimes\mathcal{H}_3^*\otimes\mathcal{H}_3\to\mathcal{H}_1\otimes\mathcal{H}_2\otimes\mathcal{H}_2^*)$.
Because of these,
in the \emph{general} writing of the reshuffling of a density operator
\begin{equation*}
[\Phi_\sigma(\varrho)]_{i_{\sigma(1)}i_{\sigma(2)}\dots i_{\sigma(n)};i_{\sigma(n+1)}i_{\sigma(n+2)}\dots i_{\sigma(2n)}}
=\varrho_{i_1i_2\dots i_n,i_{n+1}i_{n+2}\dots i_{2n}},
\end{equation*}
we can not distinguish between upper and lower indices independently from the defining permutation $\sigma$,
although we do that for \emph{particular} $\sigma$ permutations.

Now the \emph{permutation criterion} states that
\begin{equation}
\label{eq:critPerm}
\varrho\quad\text{fully separable}\qquad\Longrightarrow\qquad 
\norm{\Phi_\sigma(\varrho)}_{\tr} \leq 1,\quad \forall \sigma\in \DscrGrp{S}_{2n}.
\end{equation}
(This is easy to prove if we note that the trace-norm is sub-additive, hence convex, 
so it is enough to prove $\norm{\Phi_\sigma(\pi)}_{\tr} \leq 1$ for the fully separable pure state $\pi$.
For these we can use that $\norm{M\otimes M'}_{\tr}\leq \norm{M}_{\tr}\norm{M'}_{\tr}$,
and $\norm{\varphi\otimes\varphi'}_{\tr}=1$ for normalized vectors.)
The permutation criterion gives $\abs{S_{2n}}=(2n)!$ criteria
but not all of them are inequivalent.
It is known \cite{HorodeckiPerm} that
\emph{for two subsystems, every criteria given by the permutation criterion turn out to be equivalent
either the partial transposition criterion or the reshuffling criterion,}
which were used in the previous section.
In~\cite{ClarissePerm}, Clarisse has shown that there are only six inequivalent criteria
in the case of three subsystems, which are the
three singlepartite transpositions ($\Phi_{(14)}$, $\Phi_{(25)}$ and $\Phi_{(36)}$) 
and three bipartite reshufflings ($\Phi_{(35)}$, $\Phi_{(34)}$ and $\Phi_{(24)}$).
For our \emph{permutation-invariant} three-qubit system
all the singlepartite transpositions give the same condition,
which we have already investigated in section~\ref{subsec:SepCrit.2Part.PPT}.
On the other hand, all the bipartite reshufflings give another condition,
which is a new one.

So let $R'=\Phi_{(35)}$ the map implementing the reshuffling of the $2$ and $3$ subsystems,
that is,
 $[R'(\varrho)]^{ij\phantom{j'i'}k}_{\phantom{ij}j'i'\phantom{k}k'}=\varrho^{ijk}_{\phantom{ijk}i'j'k'}$,
resulting in the matrix
\begin{equation}
\label{eq:MxRR222}
R'(\varrho)=\begin{bmatrix}
 \td+\tg & \cdot   & \cdot   & \td+\tw & \cdot   & \cdot   & \tw     & \cdot   \\
 \cdot   & \cdot   & \tw     & \cdot   & \cdot   & \tg     & \cdot   & \cdot   \\
 \cdot   & \tw     & \cdot   & \cdot   & \tw     & \cdot   & \cdot   & \cdot   \\
 \td+\tw & \cdot   & \cdot   & \td     & \cdot   & \cdot   & \cdot   & \cdot   \\
 \cdot   & \tw     & \cdot   & \cdot   & \td+\tw & \cdot   & \cdot   & \td     \\
 \tw     & \cdot   & \cdot   & \cdot   & \cdot   & \cdot   & \cdot   & \cdot   \\
 \cdot   & \cdot   & \tg     & \cdot   & \cdot   & \cdot   & \cdot   & \cdot   \\
 \cdot   & \cdot   & \cdot   & \cdot   & \td     & \cdot   & \cdot   & \td+\tg 
\end{bmatrix}.\\
\end{equation}
With this, we have to calculate the eigenvalues of the matrix
$R'(\varrho)^\dagger R'(\varrho)$ for the two-parameter state $\varrho$ given in equation (\ref{eq:lo}).
This $8\times8$ matrix can be transformed
by simultaneous row-column permutation
into blockdiagonal form
consisting of three blocks of the sizes $3\times3$, $3\times3$ and $2\times2$.
However, the forms of the $g,w$-depending eigenvalues of the $3\times3$ blocks are still too complicated,
so we only plot the border of the domain in which the criterion (\ref{eq:critPerm}) holds
(green curve in figure~\ref{fig:Resh}).

The condition $\norm{R'(\varrho)}_{\tr} \leq 1$ holds inside the green curve
in figure~\ref{fig:Resh}.
This figure suggests that
%One can check that 
this reshuffling does not give stricter condition for full separability
than the partial transposition criterion, hence it can not identify PPTESs.
However, we can not be sure in this
 due to the difficult computation of $\norm{R'(\varrho)}_{\tr}$.
Fully separable states must be enclosed by the curves belonging to (\ref{eq:critPPT1})-(\ref{eq:critPPT2}) of partial transposition criterion,
states outside this domain must belong to Classes 2.8, 2.1 or 1.
However, in section~\ref{subsec:SepCrit.2Part.PPT} the partial transposition criterion has yielded condition for Classes 2.8 and 3,
so we can conclude that states outside this domain must belong to Classes 2.1 or 1,
the Class 2.8 is completely restricted into this domain.

%Ezt meg kellene nezni.

%******************************************************************************
%\subsection{Criteria on spin-observables}
\subsection{Quadratic Bell inequalities}
\label{subsec:SepCrit.3Part.Spin}

In~\cite{SeevinckUffinkMixSep} Seevinck and Uffink introduced a systematic way
to obtain necessary criteria of separability for all the separability-classes of an $n$-qubit system,
based on the quadratic Bell inequalities of two-qubit systems (section \ref{subsec:QM.Ent.2Part}).
Their new criteria generalize some previously known multipartite criteria, %(see references in~\cite{SeevinckUffinkMixSep})
such as Laskowski-{\.Z}ukowski criterion (necessary for $k$-separability) \cite{LaskowskiZukowskiGenBell},
Mermin-type separability inequalities (necessary for $k$-separability) 
\cite{UffinkQuadBellMultipartEnt,GisinBechmannBellineqNqubits,SeevinckSvetlichnyBellPartSep,NagataetalSepTestsMultipart,CollinsetalBellNpart,RoyMultipartSepExpStrongerLHVM},
Fidelity-criterion (necessary for $2$-separability) \cite{SeevinckUffinkTripart,ZhengetalFidelity}
(which is also known as projection based witness \cite{TothGuhneEntDetStab})
and D{\"u}r-Cirac depolarization criterion (necessary for $\alpha_k$-separability) \cite{DurCiracTarrach3QBMixSep,DurCiracTarrachBMixSep}.
We consider the three-qubit case and get criteria for Class 2.1, Class 2.8 and Class 3
given in section~\ref{subsec:QM.Ent.NPart}.

The method of Seevinck and Uffink is formulated in a recursive way 
in terms of \emph{three orthogonal spin-observables} on each subsystem,
$(X^{(1)},Y^{(1)},Z^{(1)})$.
%$(X_2^{(1)},Y_2^{(1)},Z_2^{(1)},I_2^{(1)})$,
%$(X_3^{(1)},Y_3^{(1)},Z_3^{(1)},I_3^{(1)})$.
Here the superscript $(1)$ denotes that these are single-qubit operators.
Let $I^{(1)}$ denote the $2\times2$ identity operator, $I^{(1)}=\Id$.
From the $(X^{(1)},Y^{(1)},Z^{(1)},I^{(1)})$ one-qubit observables acting on the subsystems $2$ and $3$
one can form
two sets of two-qubit observables % acting on the $2$ and $3$ subsystems:
$(X_x^{(2)},Y_x^{(2)},Z_x^{(2)},I_x^{(2)})$.
Here the superscript $(2)$ denotes that these are two-qubit operators
and $x=0,1$ refers to the two sets, which are
\begin{align}
\label{eq:spin2}
X_0^{(2)}&=\frac12\left(X^{(1)}\otimes X^{(1)} - Y^{(1)}\otimes Y^{(1)}\right),&\quad
X_1^{(2)}&=\frac12\left(X^{(1)}\otimes X^{(1)} + Y^{(1)}\otimes Y^{(1)}\right),\notag\\
Y_0^{(2)}&=\frac12\left(Y^{(1)}\otimes X^{(1)} + X^{(1)}\otimes Y^{(1)}\right),&\quad
Y_1^{(2)}&=\frac12\left(Y^{(1)}\otimes X^{(1)} - X^{(1)}\otimes Y^{(1)}\right),\notag\\
Z_0^{(2)}&=\frac12\left(Z^{(1)}\otimes I^{(1)} + I^{(1)}\otimes Z^{(1)}\right),&\quad
Z_1^{(2)}&=\frac12\left(Z^{(1)}\otimes I^{(1)} - I^{(1)}\otimes Z^{(1)}\right),\notag\\
I_0^{(2)}&=\frac12\left(I^{(1)}\otimes I^{(1)} + Z^{(1)}\otimes Z^{(1)}\right),&\quad
I_1^{(2)}&=\frac12\left(I^{(1)}\otimes I^{(1)} - Z^{(1)}\otimes Z^{(1)}\right).
\end{align}
(Note that $I_x^{(2)}$s are \emph{not} identity operators.)
From these two-qubit observables
and the one-qubit ones acting on subsystem $1$
one can form
four sets of three-qubit observables acting on the full system
$(X_x^{(3)},Y_x^{(3)},Z_x^{(3)},I_x^{(3)})$.
Here the superscript $(3)$ denotes that these are three-qubit operators
and $x=0,1,2,3$ refers to the four sets, which are
\begin{align}
\label{eq:spin3}
X_y^{(3)}    &=\frac12\left(X^{(1)}\otimes X^{(2)}_{y/2} - Y^{(1)}\otimes Y^{(2)}_{y/2}\right),&\quad
X_{y+1}^{(3)}&=\frac12\left(X^{(1)}\otimes X^{(2)}_{y/2} + Y^{(1)}\otimes Y^{(2)}_{y/2}\right),\notag\\
Y_y^{(3)}    &=\frac12\left(Y^{(1)}\otimes X^{(2)}_{y/2} + X^{(1)}\otimes Y^{(2)}_{y/2}\right),&\quad
Y_{y+1}^{(3)}&=\frac12\left(Y^{(1)}\otimes X^{(2)}_{y/2} - X^{(1)}\otimes Y^{(2)}_{y/2}\right),\notag\\
Z_y^{(3)}    &=\frac12\left(Z^{(1)}\otimes I^{(2)}_{y/2} + I^{(1)}\otimes Z^{(2)}_{y/2}\right),&\quad
Z_{y+1}^{(3)}&=\frac12\left(Z^{(1)}\otimes I^{(2)}_{y/2} - I^{(1)}\otimes Z^{(2)}_{y/2}\right),\notag\\
I_y^{(3)}    &=\frac12\left(I^{(1)}\otimes I^{(2)}_{y/2} + Z^{(1)}\otimes Z^{(2)}_{y/2}\right),&\quad
I_{y+1}^{(3)}&=\frac12\left(I^{(1)}\otimes I^{(2)}_{y/2} - Z^{(1)}\otimes Z^{(2)}_{y/2}\right),
\end{align}
for $y=0,2$.
(Again, $I_x^{(3)}$s are \emph{not} identity operators.)

Now for particular $\alpha_2$, investigating some relations among the expectation-values of these operators
with respect to the state $\varrho$,
one can get some nontrivial inequalities
valid for all $\varrho\in\mathcal{D}_{\alpha_2}$.
From these, one can form inequalities valid for a given separability class of section~\ref{subsec:QM.Ent.NPart}.
Here we recall these criteria for the classes we need to deal with
\cite{SeevinckUffinkMixSep} 
\begin{equation}
\label{eq:critSU2}
\varrho\in\mathcal{D}_\text{$2$-sep}\qquad\Longrightarrow\qquad
\sqrt{\bracket{X_x^{(3)}}^2+\bracket{Y_x^{(3)}}^2}\leq\sum_{y\neq x}\sqrt{\bracket{I_y^{(3)}}^2-\bracket{Z_y^{(3)}}^2}
\end{equation}
%for $x=0,1,2,3$,
\begin{equation}
\label{eq:critSU283}
%\varrho\in\mathcal{D}_{1-23}\cap\mathcal{D}_{2-31}\cap\mathcal{D}_{3-12}
\varrho\in\bigcap_{\alpha_2}\mathcal{D}_{\alpha_2} \qquad\Longrightarrow\qquad
\max_x\Bigl\{\bracket{X_x^{(3)}}^2+\bracket{Y_x^{(3)}}^2\Bigr\}
\leq \min_x\Bigl\{\bracket{I_x^{(3)}}^2-\bracket{Z_x^{(3)}}^2\Bigr\}\leq1/4
\end{equation}
and
\begin{equation}
\label{eq:critSU3}
\varrho\in\mathcal{D}_\text{$3$-sep}\qquad\Longrightarrow\qquad
\max_x\Bigl\{\bracket{X_x^{(3)}}^2+\bracket{Y_x^{(3)}}^2\Bigr\}
\leq\min_x\Bigl\{\bracket{I_x^{(3)}}^2-\bracket{Z_x^{(3)}}^2\Bigr\}\leq1/16,
\end{equation}
all of them are given for $x=0,1,2,3$.
One has to do optimization of the local spin observables $(X^{(1)},Y^{(1)},Z^{(1)})$
to get violation of the respective inequality for a given state.

In the following we will consider some special
measurement-settings when the observables
$(X^{(1)},Y^{(1)},Z^{(1)})$ are the same for each subsystem.
Writing out explicitly $(X_x^{(3)},Y_x^{(3)},Z_x^{(3)},I_x^{(3)})$,
one can see that \emph{for a permutation-invariant state}
the squares of the expectation values are the same for $x=1,2,3$, that is,
$\bracket{X_1^{(3)}}^2=\bracket{X_2^{(3)}}^2= \bracket{X_3^{(3)}}^2 $,
and the same for $Y_x^{(3)}$s, $Z_x^{(3)}$s and $I_x^{(3)}$s.
Hence we have to consider merely the $x=0,1$ indices.

First consider the usual Pauli matrices (\ref{eq:Pauli.mx})
\begin{align*}
&\text{Setting I:}& \quad (X^{(1)},Y^{(1)},Z^{(1)})=(\sigma_1,\sigma_2,\sigma_3)\quad\text{for each subsystem}.
\end{align*}
The inequalities (\ref{eq:critSU2})-(\ref{eq:critSU3}) can be written
as relatively simple expressions in the matrix elements \cite{SeevinckUffinkMixSep}:
\begin{equation}
\label{eq:critSU2.1g}
\begin{split}
\varrho\in&\mathcal{D}_\text{$2$-sep}\qquad\Longrightarrow\\
&\begin{aligned}[t]
\abs{\varrho^{000}_{\phantom{000}111}}&\leq 
 \sqrt{\varrho^{110}_{\phantom{110}110}\varrho^{001}_{\phantom{001}001}}
+\sqrt{\varrho^{101}_{\phantom{101}101}\varrho^{010}_{\phantom{010}010}}
+\sqrt{\varrho^{011}_{\phantom{011}011}\varrho^{100}_{\phantom{100}100}},\\
\abs{\varrho^{110}_{\phantom{110}001}}&\leq
 \sqrt{\varrho^{000}_{\phantom{000}000}\varrho^{111}_{\phantom{111}111}}
+\sqrt{\varrho^{101}_{\phantom{101}101}\varrho^{010}_{\phantom{010}010}}
+\sqrt{\varrho^{011}_{\phantom{011}011}\varrho^{100}_{\phantom{100}100}},\\
\abs{\varrho^{101}_{\phantom{101}010}}&\leq
 \sqrt{\varrho^{110}_{\phantom{110}110}\varrho^{001}_{\phantom{001}001}}
+\sqrt{\varrho^{000}_{\phantom{000}000}\varrho^{111}_{\phantom{111}111}}
+\sqrt{\varrho^{011}_{\phantom{011}011}\varrho^{100}_{\phantom{100}100}},\\
\abs{\varrho^{011}_{\phantom{011}100}}&\leq
 \sqrt{\varrho^{110}_{\phantom{110}110}\varrho^{001}_{\phantom{001}001}}
+\sqrt{\varrho^{101}_{\phantom{101}101}\varrho^{010}_{\phantom{010}010}}
+\sqrt{\varrho^{000}_{\phantom{000}000}\varrho^{111}_{\phantom{111}111}},
\end{aligned}
\end{split}
\end{equation}
\begin{equation}
\label{eq:critSU283.1g}
\begin{split}
\varrho\in&\bigcap_{\alpha_2}\mathcal{D}_{\alpha_2}\qquad\Longrightarrow\\
&\begin{aligned}
\max&\left\{
\abs{\varrho^{000}_{\phantom{000}111}}^2,
\abs{\varrho^{110}_{\phantom{110}001}}^2,
\abs{\varrho^{101}_{\phantom{101}010}}^2,
\abs{\varrho^{011}_{\phantom{011}100}}^2\right\}\\
\leq&\min\left\{
\varrho^{000}_{\phantom{000}000}\varrho^{111}_{\phantom{111}111},
\varrho^{110}_{\phantom{110}110}\varrho^{001}_{\phantom{001}001},
\varrho^{101}_{\phantom{101}101}\varrho^{010}_{\phantom{010}010},
\varrho^{011}_{\phantom{011}011}\varrho^{100}_{\phantom{100}100}\right\}
\leq1/16
\end{aligned}
\end{split}
\end{equation}
and
\begin{equation}
\label{eq:critSU3.1g}
\begin{split}
\varrho\in&\mathcal{D}_\text{$3$-sep}\qquad\Longrightarrow\\
&\begin{aligned}
\max&\left\{
\abs{\varrho^{000}_{\phantom{000}111}}^2,
\abs{\varrho^{110}_{\phantom{110}001}}^2,
\abs{\varrho^{101}_{\phantom{101}010}}^2,
\abs{\varrho^{011}_{\phantom{011}100}}^2\right\}\\
\leq&\min\left\{
\varrho^{000}_{\phantom{000}000}\varrho^{111}_{\phantom{111}111},
\varrho^{110}_{\phantom{110}110}\varrho^{001}_{\phantom{001}001},
\varrho^{101}_{\phantom{101}101}\varrho^{010}_{\phantom{010}010},
\varrho^{011}_{\phantom{011}011}\varrho^{100}_{\phantom{100}100}\right\}
\leq1/64.
\end{aligned}
\end{split}
\end{equation}
Let us consider another two special measurement settings, which are
\begin{align*}
&\text{Setting II:}&  \quad (X^{(1)},Y^{(1)},Z^{(1)})=(\sigma_2,\sigma_3,\sigma_1)\quad\text{for each subsystem},\\
&\text{Setting III:}& \quad (X^{(1)},Y^{(1)},Z^{(1)})=(\sigma_3,\sigma_1,\sigma_2)\quad\text{for each subsystem}.
\end{align*}
The inequalities of (\ref{eq:critSU2})-(\ref{eq:critSU3}) written for these two settings
are much more complicated expressions in symbolic matrix elements
than the ones in (\ref{eq:critSU2.1g})-(\ref{eq:critSU3.1g}).
But for the state $\varrho$ given in equation (\ref{eq:lo})
it is not too difficoult to write out these inequalities explicitly.
It turns out that for each of these three settings
the $x=1$ inequality of (\ref{eq:critSU2}),
the second inequality of (\ref{eq:critSU283}) and
the second inequality of (\ref{eq:critSU3})
hold for all the parameter values of the simplex.
Because of this, the criteria hold for Class 3
are not stricter than the ones for the union of Class 2.8 and Class 3.
The remaining inequalities for the three measurement settings are as follows:
%%%%%%%%%%%%%%%%%%%%%%%%%%%%%
\begin{figure}
 \setlength{\unitlength}{0.001428571\textwidth}% this 1/(420)*0.6 = 0.001428571
 \begin{picture}(420,415)
  \put(0,0){\includegraphics[width=0.6\textwidth]{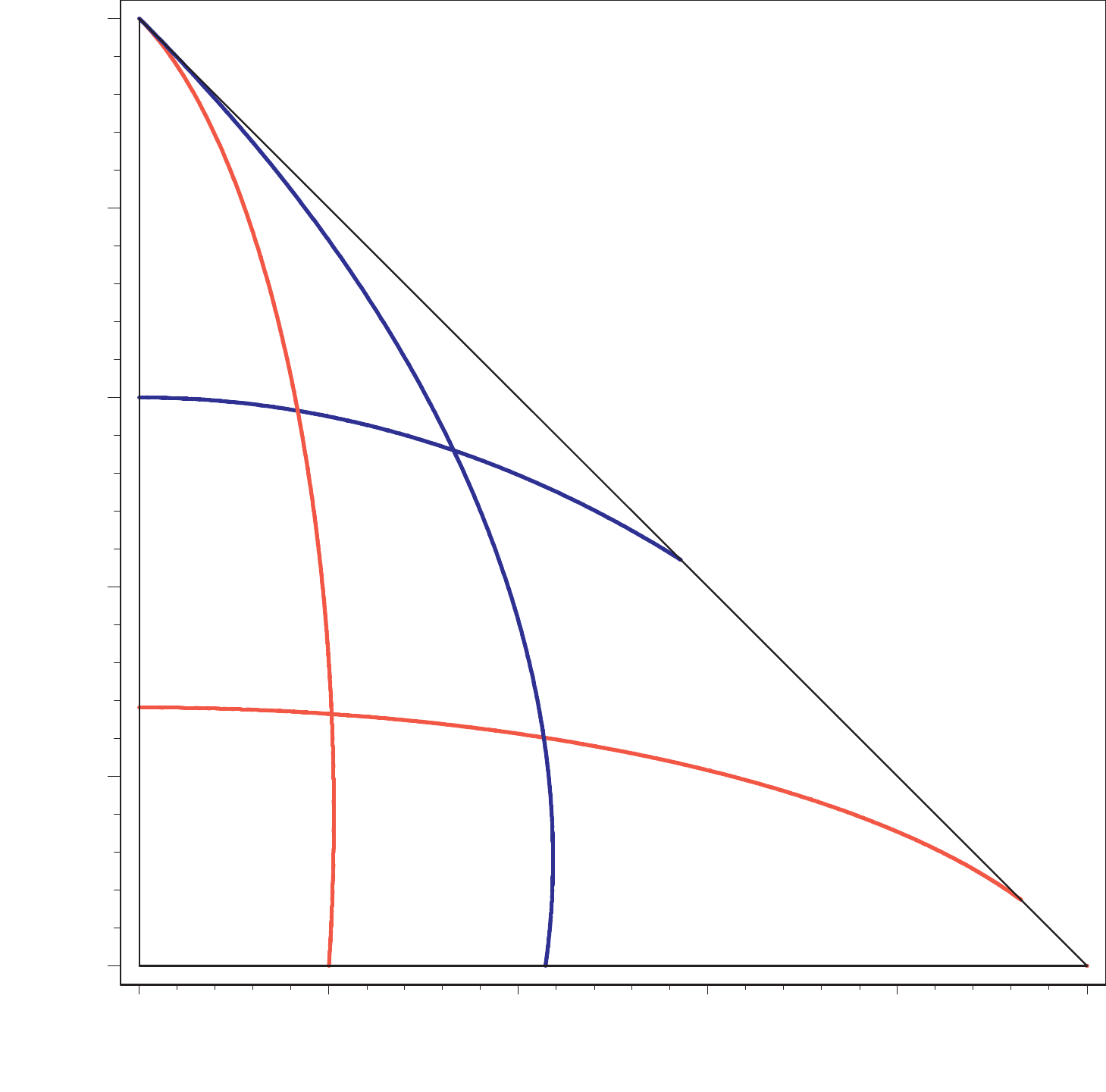}}
  \put(10,226){\makebox(0,0)[r]{\strut{}$w$}}
  \put(235,6){\makebox(0,0)[r]{\strut{}$g$}}
{\small
  \put(56,27){\makebox(0,0)[r]{\strut{}$0$}}
  \put(133,27){\makebox(0,0)[r]{\strut{}$0.2$}}
  \put(205,27){\makebox(0,0)[r]{\strut{}$0.4$}}
  \put(277,27){\makebox(0,0)[r]{\strut{}$0.6$}}
  \put(349,27){\makebox(0,0)[r]{\strut{}$0.8$}}
  \put(415,27){\makebox(0,0)[r]{\strut{}$1$}}
  \put(38,46){\makebox(0,0)[r]{\strut{}$0$}}
  \put(38,118){\makebox(0,0)[r]{\strut{}$0.2$}}
  \put(38,190){\makebox(0,0)[r]{\strut{}$0.4$}}
  \put(38,262){\makebox(0,0)[r]{\strut{}$0.6$}}
  \put(38,334){\makebox(0,0)[r]{\strut{}$0.8$}}
  \put(40,406){\makebox(0,0)[r]{\strut{}$1$}}
}
  \put(120,180){\makebox(0,0)[r]{\strut{}(\ref{eq:critSU283.1})}}
  \put(185,125){\makebox(0,0)[r]{\strut{}(\ref{eq:critSU283.2})}}
  \put(195,175){\makebox(0,0)[r]{\strut{}(\ref{eq:critSU2.1})}}
  \put(165,235){\makebox(0,0)[r]{\strut{}(\ref{eq:critSU2.2})}}
 \end{picture}
 \caption{Criteria on spin-observables for the state (\ref{eq:lo}) on the $g$-$w$-plane.
(Red curves: the border of domains inside equations (\ref{eq:critSU2.1}) and~(\ref{eq:critSU2.2}) hold,
blue curves: the border of domains inside equations (\ref{eq:critSU283.1}) and~(\ref{eq:critSU283.2}) hold.
The inequalities hold on the side of the curves containing the origin.)}
 \label{fig:Spin}
\end{figure}
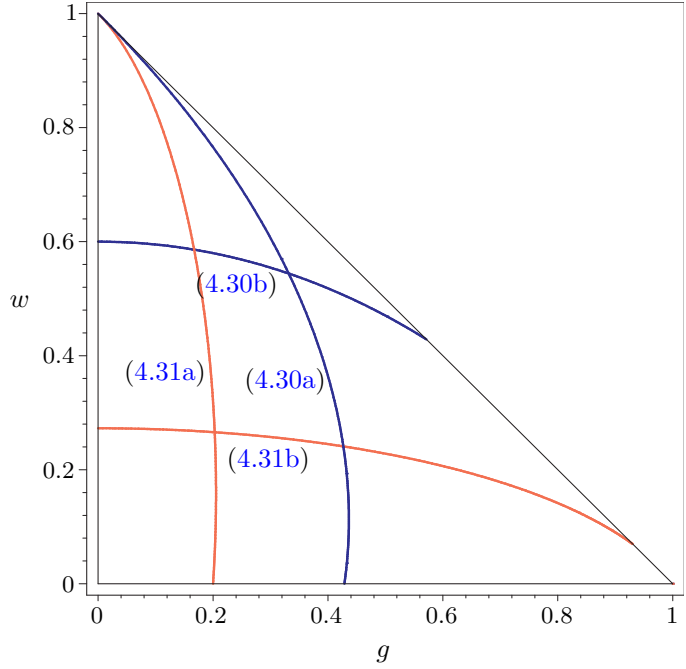
%%%%%%%%%%%%%%%%%%%%%%%%%%%%%
\begin{subequations}
\begin{align}
\varrho\in\mathcal{D}_\text{$2$-sep}\qquad\Longrightarrow\qquad
\label{eq:critSU2.1}
\text{I.}&\left\{
\begin{aligned}
\tg&\leq3\sqrt{\td(\td+\tw)}\\
0&\leq -7g^2 -6gw -15w^2 -18g +6w +9,
\end{aligned}
\right.\\
%&\text{(Setting I.)}\\
\label{eq:critSU2.2}
\text{II.}&\left\{
\begin{aligned}
3\tw&\leq\sqrt{(8\td+\tw)(8\td+4\tg+\tw)}\\
0&\leq -9g^2 -5w^2 -12w+9,
\end{aligned}
\right.\\
%&\text{(Setting II.)}\\
\label{eq:critSU2.3}
\text{III.}&\left\{
\begin{aligned}
\sqrt{4\tg^2+81\tw^2}&\leq3(8\td+2\tg+\tw)\\
0&\leq - g^2 -5w^2 -12w+9,
\end{aligned}
\right.
%&\text{(Setting III.)}
\end{align}
\end{subequations}
and
\begin{subequations}
\begin{align}
\varrho\in\bigcap_{\alpha_2}\mathcal{D}_{\alpha_2}\qquad\Longrightarrow\qquad
\label{eq:critSU283.1}
\text{I.}&\left\{
\begin{aligned}
\tg^2&\leq\td(\td+\tw)\\
0&\leq -45g^2 -2gw -5w^2 -6g +2w +3,
\end{aligned}
\right.\\
\label{eq:critSU283.2}
\text{II.}&\left\{
\begin{aligned}
81\tw^2&\leq(8\td+\tw)(8\td+4\tg+\tw)\\
0&\leq -9g^2 -77w^2 -12w+9,
\end{aligned}
\right.\\
\label{eq:critSU283.3}
\text{III.}&\left\{
\begin{aligned}
4\tg^2+81\tw^2&\leq(8\td+2\tg+\tw)^2\\
0&\leq -9g^2 -77w^2 -12w+9,
\end{aligned}
\right.
\end{align}
\end{subequations}
Clearly, the inequality of (\ref{eq:critSU2.3}) is weaker than the one of~(\ref{eq:critSU2.2}),
the inequality of (\ref{eq:critSU283.3}) is the same as the one of~(\ref{eq:critSU283.2}).
Moreover,
    the inequality of (\ref{eq:critSU283.1}) is the same as the one of~(\ref{eq:critPPT1}) of partial transposition criterion,
but the inequality of (\ref{eq:critSU283.2}) is strictly weaker than the other one of partial transposition criterion.
\emph{So these settings does not give stricter conditions for Classes 2.8 and 3 than partial transposition criterion,}
however, we get criteria for biseparability for the first time.
In figure~\ref{fig:Spin} we show the borderlines of the domains of the criteria belonging to Settings I and II.
These inequalities restrict Classes 2.1, 2.8 and 3 to be inside the domain enclosed by the blue curves
and Classes 2.8 and 3 to be inside the domain enclosed by the red curves.
\emph{We can conclude that Settings I and II are strong in detection of GHZ and W state respectively.}
One can check that for the $w=0$ GHZ-white noise mixture
the inequalities of (\ref{eq:critSU2.1}) and~(\ref{eq:critSU283.1}) of Setting I hold
if and only if the state is fully separable,
(\ref{eq:critSU2.1}) is violated but (\ref{eq:critSU283.1}) holds
if and only if the state is in Class 2.1
and
both of them are violated
if and only if the state is fully entangled.
For the $g=0$ W-white noise mixture,
if $3/11<w$ then $\varrho$ is in Class 2.1 or Class 1,
and if $3/5<w$ then $\varrho$ is fully entangled.

However, there are infinitely many criteria depending on the measurement settings
and we do not have a method to find a set of settings leading to the strictest criterion.
We have tried some other randomly chosen settings
which can be used to reduce the area where these criteria hold.
We could not find settings that give stronger criteria on the $w=0$ or $g=0$ axes of the simplex than Settings I and II, respectively.
We have found settings that exclude states from the corresponding classes,
but these states are far from these axes,
and we have not found settings which give stronger conditions for Classes 2.8 and 3 than the partial transposition criterion.
We have found settings by which the condition for biseparability can be strengthened,
but these conditions are just a little bit stronger far from the axes than the ones in section~\ref{subsec:SepCrit.3Part.Matrix}.

%******************************************************************************
\subsection{Criteria on matrix elements}
%\label{sec:Hub}
\label{subsec:SepCrit.3Part.Hub}

In a recent paper \cite{HuberkCrit},
Gabriel et.~al.~have given criterion for $k$-separability,
based on their previously derived framework for the detection of biseparability \cite{HuberCrit}.
It turns out that for the noisy GHZ-W mixture given in equation (\ref{eq:lo})
these criteria give the same results as the ones of quadratic Bell-inequalities, given in the previous section,
but these criteria have the advantage that they can be used \emph{in the same form} not only for qubits,
but for subsystems of arbitrary, even different dimensions.
To our knowledge, these were the first such criteria of $k$-separability.

Consider some permutation operators acting on $\mathcal{H}\otimes\mathcal{H}$,
that is, on the two copies of the $n$-partite Hilbert space 
$\mathcal{H}=\mathcal{H}_1\otimes\mathcal{H}_2\otimes\dots\otimes\mathcal{H}_n$.
Let $P_a$s be the operators which swap the $a$th subsystems of the two copies,
that is,
$P_a\cket{i_1i_2\dots i_n}\otimes\cket{j_1j_2\dots j_n}=
\cket{i_1i_2\dots i_{a-1}j_ai_{a+1} \dots i_n}\otimes\cket{j_1j_2\dots j_{a-1}i_aj_{a+1} \dots j_n}$
where $\{\cket{i_a}\}$ is basis in $\mathcal{H}_a$.
Now for a composite subsystem $K\subseteq L$ having the Hilbert space 
$\mathcal{H}_K=\otimes_{a\in K}\mathcal{H}_a$ let $P_K=\prod_{a\in K}P_a$.
The key fact is that for pure states $\pi$,
if the state of that subsystem can be separated from the rest of the state
then the corresponding $P_K$ leaves the two copies of the state invariant, $P_K(\pi\otimes\pi)P_K^\dagger=\pi\otimes\pi$.
With this and convexity arguments, one can get the following criteria for $k$-separability \cite{HuberkCrit}
\begin{equation}
\label{eq:critHub}
\varrho\in\mathcal{D}_{k-\text{sep}}\qquad\Longrightarrow\qquad
\sqrt{\bra{\phi}(\varrho\otimes\varrho)P_\text{tot}\cket{\phi}}
\leq\sum_i\left(\prod_{r=1}^k\bra{\phi}P_{L_r^{(i)}}(\varrho\otimes\varrho)P^\dagger_{L_r^{(i)}}\cket{\phi}\right)^{1/(2k)},
\end{equation}
where $\cket{\phi}\in\mathcal{H}\otimes\mathcal{H}$ is a \emph{fully separable} vector,
and the total swap operator is $P_\text{tot}=\prod_{a=1}^NP_a$.
Here $i$ runs over all posible $k$-partite splits $\alpha_k^{(i)}=L_1^{(i)}|L_2^{(i)}|\dots|L_k^{(i)}$.

The inequality in (\ref{eq:critHub}) is written on the matrix elements of $\varrho$
determined by the separable detection-vector $\cket{\phi}$.
For a given state, optimization on $\cket{\phi}$
is needed to achieve the violation of the right-hand side of (\ref{eq:critHub}).

To apply these criteria to the noisy GHZ-W mixture given in equation (\ref{eq:lo})
we have to choose a suitable detection-vector $\cket{\phi}$.
It turns out that
$\cket{\phi_{\text{GHZ}}}=\cket{000111}$ and its Hadamard-transformed (\ref{eq:Hadamard}) version
$\cket{\phi_{\text{W}}}=H^{\otimes6}\cket{\phi_{\text{GHZ}}}$
are good choices for states in the vicinity of GHZ and W states respectively,
as observed in~\cite{HuberkCrit}.
With these two vectors we get the same criteria
for $2$-separability as the ones in (\ref{eq:critSU2.1})   and~(\ref{eq:critSU2.2})  respectively,
and
for $3$-separability as the ones in (\ref{eq:critSU283.1}) and~(\ref{eq:critSU283.2})  respectively.
(These were obtained by the criteria on spin observables in the previous section.)
However, (\ref{eq:critSU283.1}) and~(\ref{eq:critSU283.2})
are conditions not only for Class 3, but for the union of Classes 2.8 and 3,
%so in this sense the criteria on spin observables are a bit stronger.
so in this sense the quadratic Bell inequalities are a bit stronger.

We can not be sure that the detection-vectors above give
the strongest conditions at least for the noisy GHZ and noisy W states.
However, it is an interesting observation
that the Hadamard transformation relates
not only the two ``strong'' detection-vectors $\cket{\phi_{\text{GHZ}}}$ and $\cket{\phi_{\text{W}}}$
but also the two ``strong'' measurement-settings of the previous section (by the transformation $\sigma_i\mapsto H\sigma_iH^\dagger$):
Setting I.~$(\sigma_1,\sigma_2,\sigma_3)$ and $(\sigma_3,-\sigma_2,\sigma_1)$, which is equivalent% to Setting II%
%%%%%%%%%%%%%%%%%%%%%%%%
\footnote{This equivalence holds only for \emph{permutation-invariant} three-qubit states,
when the three sets of observables $(X^{(1)},Y^{(1)},Z^{(1)})$ are the \emph{same for each subsystem.}
In this case one can check
that the quantities $\bracket{X_x^{(3)}}^2+\bracket{Y_x^{(3)}}^2$ for $x=0,1,2,3$
are invariant under the
transformation
$(X^{(1)},Y^{(1)},Z^{(1)}) \leftrightarrow (Y^{(1)},X^{(1)},Z^{(1)})$
and
$(X^{(1)},Y^{(1)},Z^{(1)}) \leftrightarrow (X^{(1)},-Y^{(1)},Z^{(1)})$.
These can be seen by writing out the definitions given in equations (\ref{eq:spin3}).}
%%%%%%%%%%%%%%%%%%%%%%%%
to Setting II.

We have tried some other randomly chosen detection-vectors
which can be used to reduce the area where the criteria hold,
and we get the same observations as at the end of the previous section.
One can strengthen the conditions only far from the $w=0$ or $g=0$ axes of the simplex,
we have not found detection-vectors which give stronger condition for full-separability than the partial transposition criterion,
and we have found settings by which the condition for biseparability can be strengthened,
but these conditions are just a little bit stronger far from the axes than the ones in section~\ref{subsec:SepCrit.3Part.Matrix}.

%******************************************************************************
\subsection{Criteria on matrix elements -- a different approach}
%\label{sec:Matrix}
\label{subsec:SepCrit.3Part.Matrix}

In~\cite{GuhneSevinckCrit} G\"uhne and Seevinck have given
some further biseparability and full-separability criteria
on the matrix elements.
The idea is that one can derive identities of matrix elements of pure separable states,
then these identities can be extended to inequalities on mixed states by a convexity argument.
For example, if $\cket{\psi}$ is separable under the $1|23$ split,
then $\cket{\psi}=\cket{\psi_1}\otimes\cket{\psi_{23}}$ has factorized coefficients $\psi^{ijk}=\psi_1^i\psi_{23}^{jk}$.
The pure state $\pi=\cket{\psi}\bra{\psi}$
has then factorizable matrix elements 
$\pi^{ijk}_{\phantom{ijk}i'j'k'}= \psi_1^i\psi_{23}^{jk} \psi_{1,i'}\psi_{23,j'k'}
\equiv \psi_1^i\psi_{23}^{jk}\cc{(\psi_1^{i'})}\cc{(\psi_{23}^{j'k'})}$
for which one can immediately check that
$\abs{\pi^{000}_{\phantom{000}111}} =\sqrt{\pi^{011}_{\phantom{011}011}\pi^{100}_{\phantom{100}100}}$ holds.
%For the $2|13$ and $3|12$ splits similar 
Now, we just need that these expressions behave well under
convex combination of the  $\pi^{ijk}_{\phantom{ijk}i'j'k'}$ matrix elements,
that is, the square root of the product of two diagonal (hence nonnegative) matrix elements is concave, 
while the absolute value is convex.
Similar reasonings lead to the following criteria:
\begin{subequations}
\begin{align}
\notag
\varrho\in\mathcal{D}_\text{$2$-sep}\qquad\Longrightarrow\quad&\\ 
\label{eq:critGS2a}
\abs{\varrho^{000}_{\phantom{000}111}}&\leq
\sqrt{\varrho^{110}_{\phantom{110}110}\varrho^{001}_{\phantom{001}001}}+
\sqrt{\varrho^{101}_{\phantom{101}101}\varrho^{010}_{\phantom{010}010}}+
\sqrt{\varrho^{011}_{\phantom{011}011}\varrho^{100}_{\phantom{100}100}},\\
\label{eq:critGS2b}
\begin{split}
 \abs{\varrho^{001}_{\phantom{001}010}}
+\abs{\varrho^{001}_{\phantom{001}100}}
+\abs{\varrho^{010}_{\phantom{010}100}}&\leq
\sqrt{\varrho^{000}_{\phantom{000}000}\varrho^{011}_{\phantom{011}011}}+
\sqrt{\varrho^{000}_{\phantom{000}000}\varrho^{101}_{\phantom{101}101}}+
\sqrt{\varrho^{000}_{\phantom{000}000}\varrho^{110}_{\phantom{110}110}}\\
&\quad+\left(\varrho^{001}_{\phantom{001}001}+\varrho^{010}_{\phantom{010}010}+\varrho^{100}_{\phantom{100}100}\right)/2.
\end{split}
\end{align}
\end{subequations}
The criterion (\ref{eq:critGS2a}) is necessary and sufficient for GHZ-diagonal states
and can also be obtained as a special case
of the criteria of section~\ref{subsec:SepCrit.3Part.Spin} (equation (\ref{eq:critSU2.1g})).
However, this criterion---and the others in this section---arises from a quite different approach
as the one in (\ref{eq:critSU2.1g}).
%since these criteria
%have been derived from direct investigation of the matrix elements of pure separable states
%with the use of convexity argument.
The criterion in (\ref{eq:critGS2b}) is independent of the first one
and it is quite strong in detection of W state mixed with white noise.

Of course these and the following inequalities
can be written on local unitary transformed density matrices,
and optimization under local unitaries might be necessary,
but this can lead to very complicated expressions in the original matrix elements.
An advantage of the method recalled in the previous section is that
it handles the matrix indices through the detection vector $\cket{\phi}$.

By the use of similar reasoning, criteria for full-separability%
%%%%%%%%%%%%%%%%%%%%%%%% 
\footnote{Criteria for $\mathcal{D}_{\alpha_k}$ can be obtained as well.}
%%%%%%%%%%%%%%%%%%%%%%%% 
can also be obtained \cite{GuhneSevinckCrit}, which are
\begin{subequations}
\begin{align}
\notag
\varrho\in\mathcal{D}_\text{$3$-sep}\qquad\Longrightarrow\quad&\\
\label{eq:critGS3.1g}
\abs{\varrho^{000}_{\phantom{000}111}}&\leq \left(
\varrho^{001}_{\phantom{001}001}\varrho^{010}_{\phantom{010}010}
\varrho^{011}_{\phantom{011}011}\varrho^{100}_{\phantom{100}100}
\varrho^{101}_{\phantom{101}101}\varrho^{110}_{\phantom{110}110}\right)^{1/6},\\
\label{eq:critGS3.2g}
\abs{\varrho^{001}_{\phantom{001}010}}+\abs{\varrho^{001}_{\phantom{001}100}}+\abs{\varrho^{010}_{\phantom{010}100}}&\leq
\sqrt{\varrho^{000}_{\phantom{000}000}\varrho^{011}_{\phantom{011}011}}+
\sqrt{\varrho^{000}_{\phantom{000}000}\varrho^{101}_{\phantom{101}101}}+
\sqrt{\varrho^{000}_{\phantom{000}000}\varrho^{110}_{\phantom{110}110}}.
\end{align}
\end{subequations}
Criterion (\ref{eq:critGS3.1g}) is
necessary and sufficient for GHZ state mixed with white noise,
and (\ref{eq:critGS3.2g}) is violated in the vicinity of the W state.
Moreover, one can obtain other conditions from (\ref{eq:critGS3.1g}) by making substitutions as follows.
Consider a fully separable state vector $\cket{\psi}=\cket{\psi_1}\otimes\cket{\psi_2}\otimes\cket{\psi_3}$,
which has factorizable coefficients $\psi^{ijk}=\psi_1^i\psi_2^j\psi_3^k$.
%Consider a fully separable pure state $\pi=\cket{\psi}\bra{\psi}$, where
%\begin{equation*}
%\cket{\psi}=(a^0\cket{0}+a^1\cket{1})\otimes(b^0\cket{0}+b^1\cket{1})\otimes(c^0\cket{0}+c^1\cket{1}).
%\end{equation*}
Then the diagonal elements of the pure state $\pi=\cket{\psi}\bra{\psi}$
are $\pi^{ijk}_{\phantom{ijk}ijk}=\abs{\psi_1^i}^2\abs{\psi_2^j}^2\abs{\psi_3^k}^2$,
%where we use the $ijk=000,001,\dots,111$ binary indexing.
leading to the identities
\begin{equation*} 
%\begin{split}
  \pi^{ijk}_{\phantom{ijk}ijk}\pi^{i'j'k'}_{\phantom{i'j'k'}i'j'k'}
=\pi^{ijk'}_{\phantom{ijk'}ijk'}\pi^{i'j'k}_{\phantom{i'j'k}i'j'k}
=\pi^{ij'k}_{\phantom{ij'k}ij'k}\pi^{i'jk'}_{\phantom{i'jk'}i'jk'}
=\pi^{i'jk}_{\phantom{i'jk}i'jk}\pi^{ij'k'}_{\phantom{ij'k'}ij'k'}.
%\end{split}
\end{equation*}
So these kinds of substitutions can be done for the $\varrho^{ijk}_{\phantom{ijk}ijk}$ diagonal
 matrix elements on the right-hand side of (\ref{eq:critGS3.1g}).
Moreover,
the right-hand side of the inequality of (\ref{eq:critGS3.1g})
can be written as
\begin{equation*}
\left((\varrho^{001}_{\phantom{001}001})^2(\varrho^{010}_{\phantom{010}010})^2(\varrho^{011}_{\phantom{011}011})^2(\varrho^{100}_{\phantom{100}100})^2(\varrho^{101}_{\phantom{101}101})^2(\varrho^{110}_{\phantom{110}110})^2\right)^{1/12},
\end{equation*}
and with these substitutions
we can obtain a third power of four matrix elements under the $12$th root.
So we can get expressions of four matrix elements
on the right-hand side~of (\ref{eq:critGS3.1g}),
for example $\left( \varrho^{001}_{\phantom{001}001}\varrho^{010}_{\phantom{010}010}\varrho^{100}_{\phantom{100}100}\varrho^{111}_{\phantom{111}111}\right)^{1/4}$.
With the substitutions above one can get $28$ different inequalities for (\ref{eq:critGS3.1g})
with an expression of sixth order under the sixth root on the right-hand side
and $12$ different ones with an expression of fourth order under the fourth root.
For permutation-invariant states we have
$\varrho^{001}_{\phantom{001}001}=\varrho^{010}_{\phantom{010}010}=\varrho^{100}_{\phantom{100}100}$ and
$\varrho^{110}_{\phantom{110}110}=\varrho^{101}_{\phantom{101}101}=\varrho^{011}_{\phantom{011}011}$,
and for our case (\ref{eq:MxR}) it also holds that $\varrho^{000}_{\phantom{000}000}=\varrho^{111}_{\phantom{111}111}$,
so the number of different inequalities reduces to $8$ and $5$ respectively.
The right-hand sides of inequality (\ref{eq:critGS3.1g}) which are different for permutation-invariant matrices
are as follows:
\begin{align*}
&\left( (\varrho^{000}_{\phantom{000}000})^3(\varrho^{111}_{\phantom{111}111})^3 \right)^{1/6},\\
&\left( \varrho^{001}_{\phantom{001}001}  \varrho^{010}_{\phantom{010}010}  \varrho^{011}_{\phantom{011}011}  \varrho^{100}_{\phantom{100}100}  \varrho^{101}_{\phantom{101}101}  \varrho^{110}_{\phantom{110}110}  \right)^{1/6},\\
&\left( \varrho^{000}_{\phantom{000}000}  (\varrho^{001}_{\phantom{001}001})^2(\varrho^{110}_{\phantom{110}110})^2\varrho^{111}_{\phantom{111}111}  \right)^{1/6},\\
&\left( (\varrho^{000}_{\phantom{000}000})^2\varrho^{001}_{\phantom{001}001}  \varrho^{110}_{\phantom{110}110}  (\varrho^{111}_{\phantom{111}111})^2\right)^{1/6},\\
&\left( \varrho^{000}_{\phantom{000}000}  \varrho^{001}_{\phantom{001}001}  \varrho^{010}_{\phantom{010}010}  \varrho^{100}_{\phantom{100}100}  (\varrho^{111}_{\phantom{111}111})^2\right)^{1/6},\\
&\left( (\varrho^{000}_{\phantom{000}000})^2\varrho^{011}_{\phantom{011}011}  \varrho^{101}_{\phantom{101}101}  \varrho^{110}_{\phantom{110}110}  \varrho^{111}_{\phantom{111}111}  \right)^{1/6},\\
&\left( (\varrho^{001}_{\phantom{001}001})^2\varrho^{010}_{\phantom{010}010}  \varrho^{100}_{\phantom{100}100}  \varrho^{110}_{\phantom{110}110}  \varrho^{111}_{\phantom{111}111}\right)^{1/6},\\
&\left( \varrho^{000}_{\phantom{000}000}  \varrho^{001}_{\phantom{001}001}  \varrho^{011}_{\phantom{011}011}  \varrho^{101}_{\phantom{101}101}  (\varrho^{110}_{\phantom{110}110})^2\right)^{1/6}
\end{align*}
and
\begin{align*}
&\left( (\varrho^{000}_{\phantom{000}000})^2(\varrho^{111}_{\phantom{111}111})^2 \right)^{1/4},\\
&\left( \varrho^{010}_{\phantom{010}010}  \varrho^{011}_{\phantom{011}011}  \varrho^{100}_{\phantom{100}100}  \varrho^{101}_{\phantom{101}101}   \right)^{1/4},\\
&\left( \varrho^{000}_{\phantom{000}000}  \varrho^{001}_{\phantom{001}001}  \varrho^{110}_{\phantom{110}110}  \varrho^{111}_{\phantom{111}111}   \right)^{1/4},\\
&\left( \varrho^{001}_{\phantom{001}001}  \varrho^{010}_{\phantom{010}010}  \varrho^{100}_{\phantom{100}100}  \varrho^{111}_{\phantom{111}111}   \right)^{1/4},\\
&\left( \varrho^{000}_{\phantom{000}000}  \varrho^{011}_{\phantom{011}011}  \varrho^{101}_{\phantom{101}101}  \varrho^{110}_{\phantom{110}110}   \right)^{1/4}.
\end{align*}
It turns out that the strongest conditions for the noisy GHZ-W mixture
can be given by the last one of these
and with the original one in (\ref{eq:critGS3.1g}).
(We could also make some substitutions
in the right-hand side~of (\ref{eq:critGS3.2g}) but these would not give stronger conditions than the original one.)

Writing out the criteria of biseparability and full separability we get:
\begin{subequations}
\begin{align}
\label{eq:critGS2.1}
\varrho\in\mathcal{D}_\text{$2$-sep}\quad\Longrightarrow\qquad 
%0&\leq -7g^2 -6gw -15w^2 -18g +6w +9,\\
\tg&\leq 3\sqrt{\td(\td+\tw)} ,\\
\label{eq:critGS2.2}
%0&\leq -117g^2 -85w^2 -138gw  +90g-6w +27,
\tw&\leq \sqrt{(\td+\tg)\td}+(\td+\tw)/2
\end{align}
\end{subequations}
and
\begin{subequations}
\begin{align}
\label{eq:critGS3.1}
\varrho\in\mathcal{D}_\text{$3$-sep}\quad\Longrightarrow\qquad 
%0&\leq -45g^2 -5w^2 -2gw -6g+2w+ 3,\\
%  \tg&\leq \td+\tg,\\
  \tg&\leq ((\td+\tw)\td)^{1/2},\\
%&&\tg&\leq ((\td+\tg)(\td+\tw))^{\frac12},\\
%&&\tg&\leq ((\td+\tg)\td)^{\frac12},\\
%&&\tg&\leq ((\td+\tg)(\td+\tw)\td)^{\frac13},\\
%&&\tg&\leq ((\td+\tg)^4(\td+\tw)\td)^{\frac16},\\
%&&\tg&\leq ((\td+\tg)(\td+\tw)^4\td)^{\frac16},\\
%&&\tg&\leq ((\td+\tg)(\td+\tw)\td^4)^{\frac16},\\
%&&\tg&\leq ((\td+\tg)^2(\td+\tw)\td)^{\frac14},\\
%&&\tg&\leq ((\td+\tg)(\td+\tw)^3)^{\frac14},\\
\label{eq:critGS3.11}
\tg&\leq ((\td+\tg)\td^3)^{1/4},\\
\label{eq:critGS3.2}
%0&\leq -27g^2 -18gw -55w^2 +18g -18w +  9,
\tw&\leq \sqrt{(\td+\tg)\td}.
\end{align}
\end{subequations}
(See in equation (\ref{eq:MxR}).)
Clearly, the biseparability condition of (\ref{eq:critGS2.1}) is the same as the one of~(\ref{eq:critSU2.1}) of the criterion on spin-observables
but condition of (\ref{eq:critGS2.2}) is strictly stronger than the one of~(\ref{eq:critSU2.2}).
(On the $g=0$ noisy W state it gives the bound $9/17$, which is the strongest for these states---to our knowledge.)
The full-separability condition of (\ref{eq:critGS3.1}) is the same as the one of~(\ref{eq:critSU283.1}) of the criterion on spin-observables
(and the one of (\ref{eq:critPPT1}) of partial transposition criterion as well)
but the condition of (\ref{eq:critGS3.2}) is weaker than
    the one of (\ref{eq:critSU283.2}) of the criterion on spin-observables.
%and the one of (\ref{eq:critPPT2}) of partial transposition criterion.
Hence at this point these criteria are
stronger for biseparability but weaker for full separability
than the criteria on spin-observables for our state.
But we have another full-separability condition, (\ref{eq:critGS3.11}), which
can be stronger in a region than the ones based on the partial transposition criterion.
The states of parameters in this region are entangled ones of positive partial transpose,
no pure state entanglement can be distilled from them.
The borders of the domains in which these conditions hold
and the region of PPTESs can be seen in figure~\ref{fig:Mxe}.
%%%%%%%%%%%%%%%%%%%%%%%%%%%%%%%%%%%%%%%%%%%%%%%%
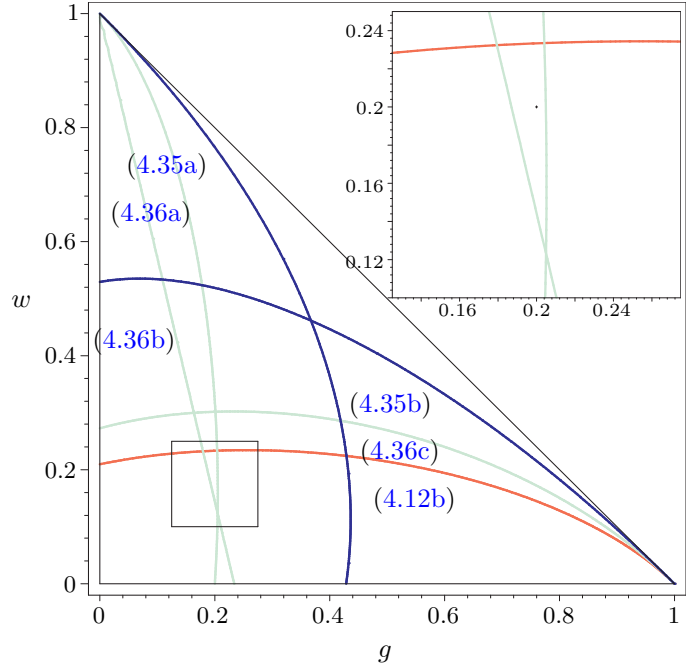
\begin{figure}
 \setlength{\unitlength}{0.001428571\textwidth}% this 1/(420)*0.6 = 0.001428571
 \begin{picture}(420,415)
  \put(0,0){\includegraphics[width=0.6\textwidth]{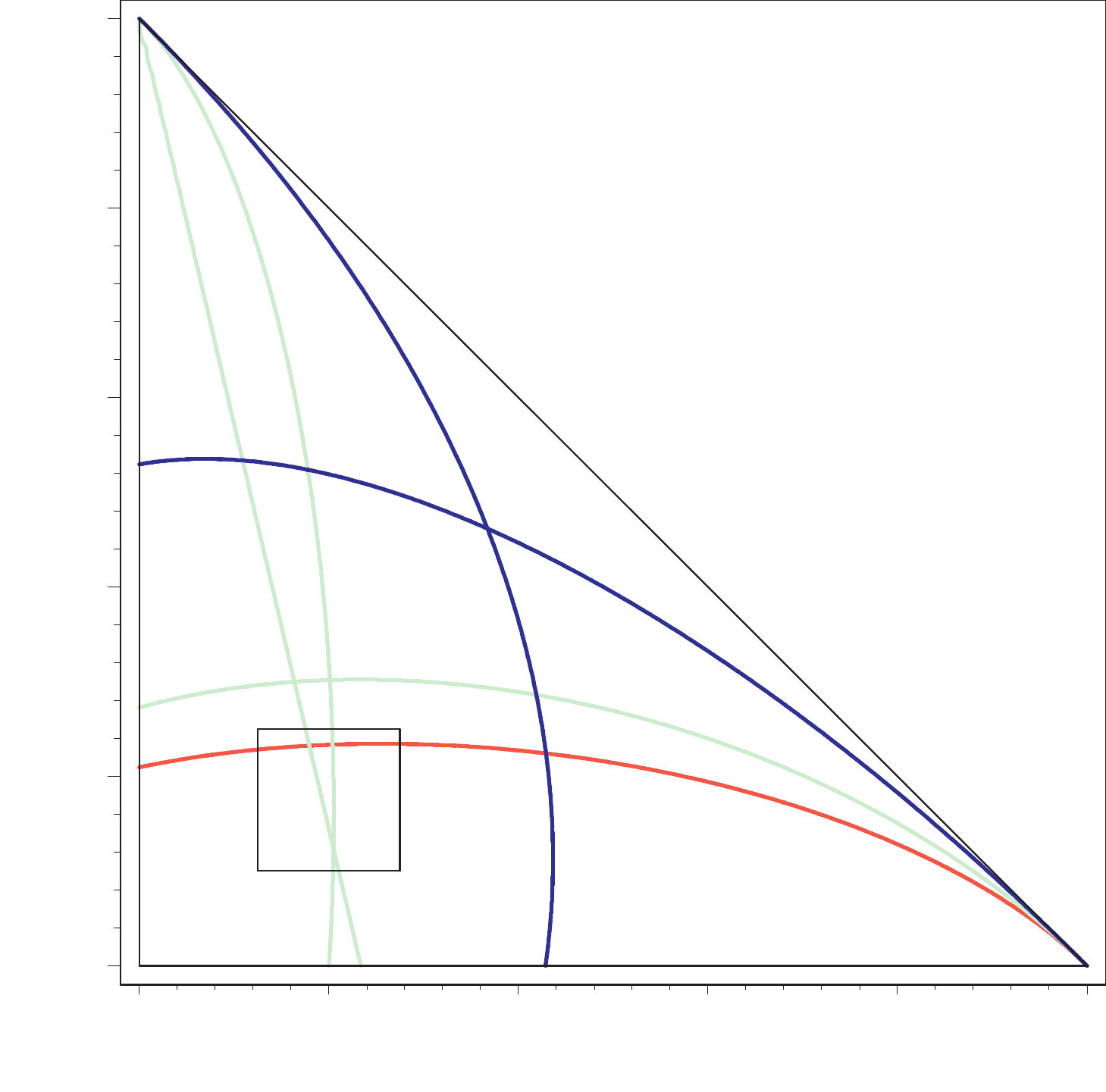}}
  \put(10,226){\makebox(0,0)[r]{\strut{}$w$}}
  \put(235,6){\makebox(0,0)[r]{\strut{}$g$}}
{\small
  \put(56,27){\makebox(0,0)[r]{\strut{}$0$}}
  \put(133,27){\makebox(0,0)[r]{\strut{}$0.2$}}
  \put(205,27){\makebox(0,0)[r]{\strut{}$0.4$}}
  \put(277,27){\makebox(0,0)[r]{\strut{}$0.6$}}
  \put(349,27){\makebox(0,0)[r]{\strut{}$0.8$}}
  \put(415,27){\makebox(0,0)[r]{\strut{}$1$}}
  \put(38,46){\makebox(0,0)[r]{\strut{}$0$}}
  \put(38,118){\makebox(0,0)[r]{\strut{}$0.2$}}
  \put(38,190){\makebox(0,0)[r]{\strut{}$0.4$}}
  \put(38,262){\makebox(0,0)[r]{\strut{}$0.6$}}
  \put(38,334){\makebox(0,0)[r]{\strut{}$0.8$}}
  \put(40,406){\makebox(0,0)[r]{\strut{}$1$}}
}
  \put(206,208){\includegraphics[width=0.3\textwidth]{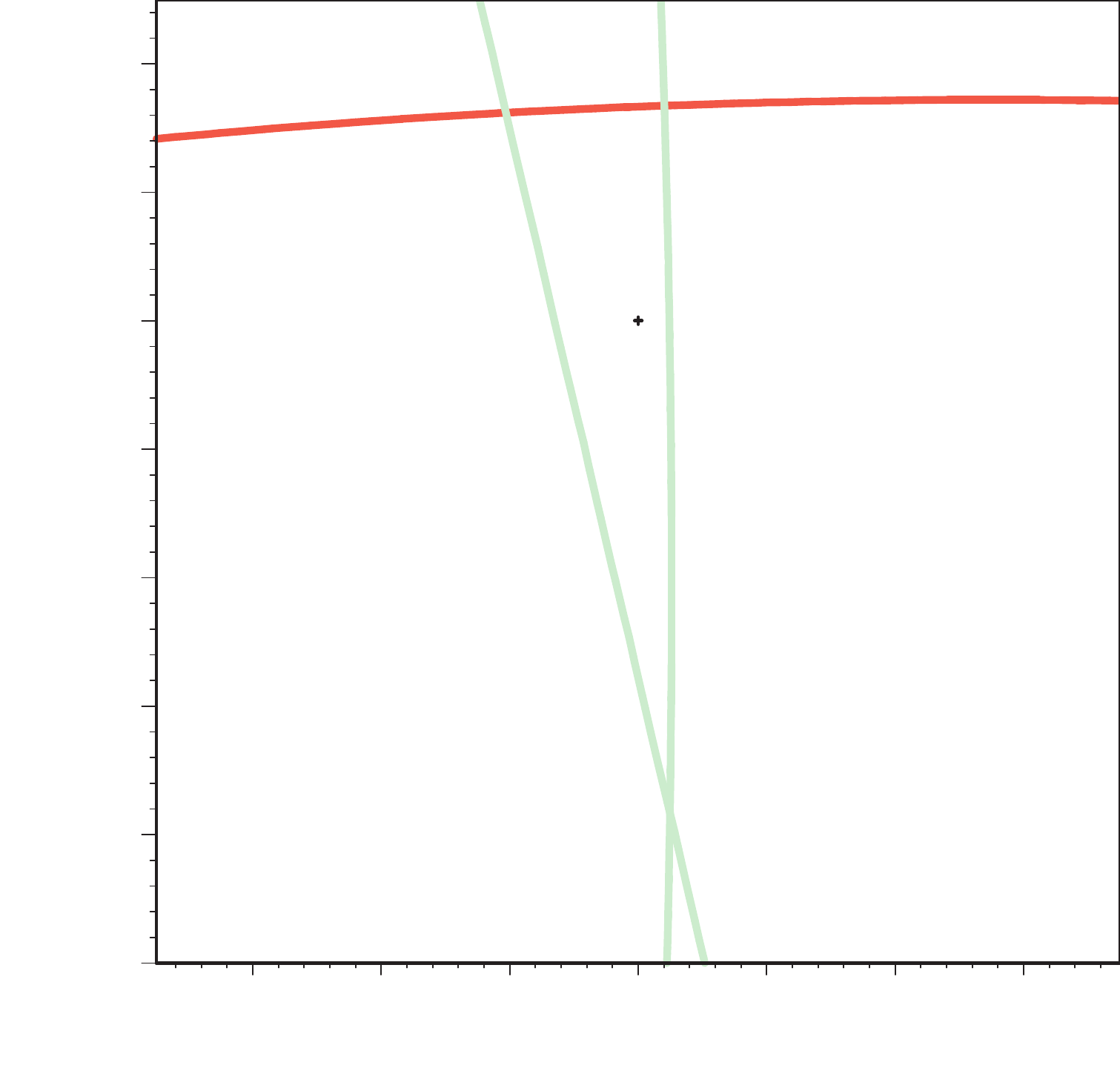}}
{\scriptsize
  \put(288,219){\makebox(0,0)[r]{\strut{}$0.16$}}
  \put(334,219){\makebox(0,0)[r]{\strut{}$0.2$}}
  \put(384,219){\makebox(0,0)[r]{\strut{}$0.24$}}

  \put(230,249){\makebox(0,0)[r]{\strut{}$0.12$}}
  \put(230,297){\makebox(0,0)[r]{\strut{}$0.16$}}
  \put(230,346){\makebox(0,0)[r]{\strut{}$0.2$}}
  \put(230,394){\makebox(0,0)[r]{\strut{}$0.24$}}
}
  \put(110,280){\makebox(0,0)[r]{\strut{}(\ref{eq:critGS3.1})}}
  \put(100,200){\makebox(0,0)[r]{\strut{}(\ref{eq:critGS3.11})}}
  \put(265,130){\makebox(0,0)[r]{\strut{}(\ref{eq:critGS3.2})}}
  \put(275,100){\makebox(0,0)[r]{\strut{}(\ref{eq:critPPT2})}}
  \put(120,310){\makebox(0,0)[r]{\strut{}(\ref{eq:critGS2.1})}}
  \put(260,160){\makebox(0,0)[r]{\strut{}(\ref{eq:critGS2.2})}}
 \end{picture}
 \caption{Criteria on matrix elements for the state (\ref{eq:lo}) on the $g$-$w$-plane.
(Green curves:
the borders of domains inside equations (\ref{eq:critGS3.1})-(\ref{eq:critGS3.2}) hold,
blue curves:
the borders of domains inside equations (\ref{eq:critGS2.1})-(\ref{eq:critGS2.2}) hold.
Red curve:
the border of domain inside equation (\ref{eq:critPPT2}) of partial transposition criterion hold,
copied from figure~\ref{fig:PPTRed}.
The point $g=1/5$, $w=1/5$ is also shown.
The inequalities hold on the side of the curves containing the origin.)}
 \label{fig:Mxe}
\end{figure}
%%%%%%%%%%%%%%%%%%%%%%%%%%%%%%%%%%%%%%%%%%%%%%%%
One can also show a representing matrix of the region of PPTESs determined by (\ref{eq:critPPT1}),~(\ref{eq:critPPT2})
and the violation of (\ref{eq:critGS3.11}).
It is easy to check that the state of parameters $(g=1/5,w=1/5)$ is contained by this set
and the explicit form of (\ref{eq:lo}) for this point is
\begin{equation} 
\label{eq:PPTESlo}
\varrho_{g=1/5,w=1/5}=\frac1{120}\begin{bmatrix}
 21  &\cdot&\cdot&\cdot&\cdot&\cdot&\cdot& 12  \\
\cdot& 17  &  8  &\cdot&  8  &\cdot&\cdot&\cdot\\
\cdot&  8  & 17  &\cdot&  8  &\cdot&\cdot&\cdot\\
\cdot&\cdot&\cdot&  9  &\cdot&\cdot&\cdot&\cdot\\
\cdot&  8  &  8  &\cdot& 17  &\cdot&\cdot&\cdot\\
\cdot&\cdot&\cdot&\cdot&\cdot&  9  &\cdot&\cdot\\
\cdot&\cdot&\cdot&\cdot&\cdot&\cdot&  9  &\cdot\\
 12  &\cdot&\cdot&\cdot&\cdot&\cdot&\cdot& 21  \\
\end{bmatrix}.
\end{equation}

%******************************************************************************
\section{Tripartite entanglement}
\label{sec:SepCrit.Class1}
Now we discuss some issues different from the previous ones related to separability criteria.
Namely, we investigate Class 1 of three-qubit entanglement,
and show that entangled two-qubit subsystems arise only in this class.

%******************************************************************************
\subsection{W and GHZ classes of three-qubit entanglement}
\label{subsec:SepCrit.Class1.WGHZ}

As we have seen in section \ref{subsec:QM.EntMeas.3QBPure},
a fully entangled three-qubit \emph{pure} state can be either of Class GHZ or of Class W
in the sense of SLOCC, that is,
vectors of these two different types can not be transformed into each other by local invertible operations.
These fully entangled vectors $\cket{\psi}$ can be classified by the $\tau(\psi)$ three-tangle (\ref{eq:tau}),
as $\tau(\psi)\neq0$ exactly for the GHZ-type vectors, see in table \ref{tab:SLOCC3Pure}.

In section~\ref{subsec:QM.EntMeas.3QBMix} we recalled a classification of \emph{mixed} three-qubit states
related to the pure-state SLOCC classes,
given by Ac{\'i}n et.~al.~\cite{Acinetal3QBMixClass}.
They have shown that Class 1 of fully entangled states
can be naturally divided into two subsets, namely Class W and Class GHZ.
%Recall that
%a state is of \emph{W-type} ($\mathcal{D}_{\text{W}}$) if it can be expressed as a mixture of projectors onto $2$-separable and W-type vectors
%(therefore $\mathcal{D}_{\text{W}}$ is also a convex set)
%and \emph{GHZ-type} vector is required for a GHZ-type mixed state.
%Let \emph{Class W} be the set $\mathcal{D}_{\text{W}}\setminus\mathcal{D}_{2-\text{sep}}$
%and \emph{Class GHZ} be the set $\mathcal{D}_{\text{GHZ}}\setminus\mathcal{D}_{\text{W}}$,
%so $\text{Class 1} = \text{Class W}\cup\text{Class GHZ}$.
Since Class GHZ is the set of states for the mixing of which GHZ-type vector is required,
the convex-roof extension of the three-tangle $\tau$
is a good indicator for Class GHZ, 
as $\cnvroof{\tau}(\varrho)\neq0$ exactly for Class GHZ.

\begin{figure}
 \setlength{\unitlength}{0.001428571\textwidth}% this 1/(420)*0.6 = 0.001428571
 \begin{picture}(420,415)
  \put(0,0){\includegraphics[width=0.6\textwidth]{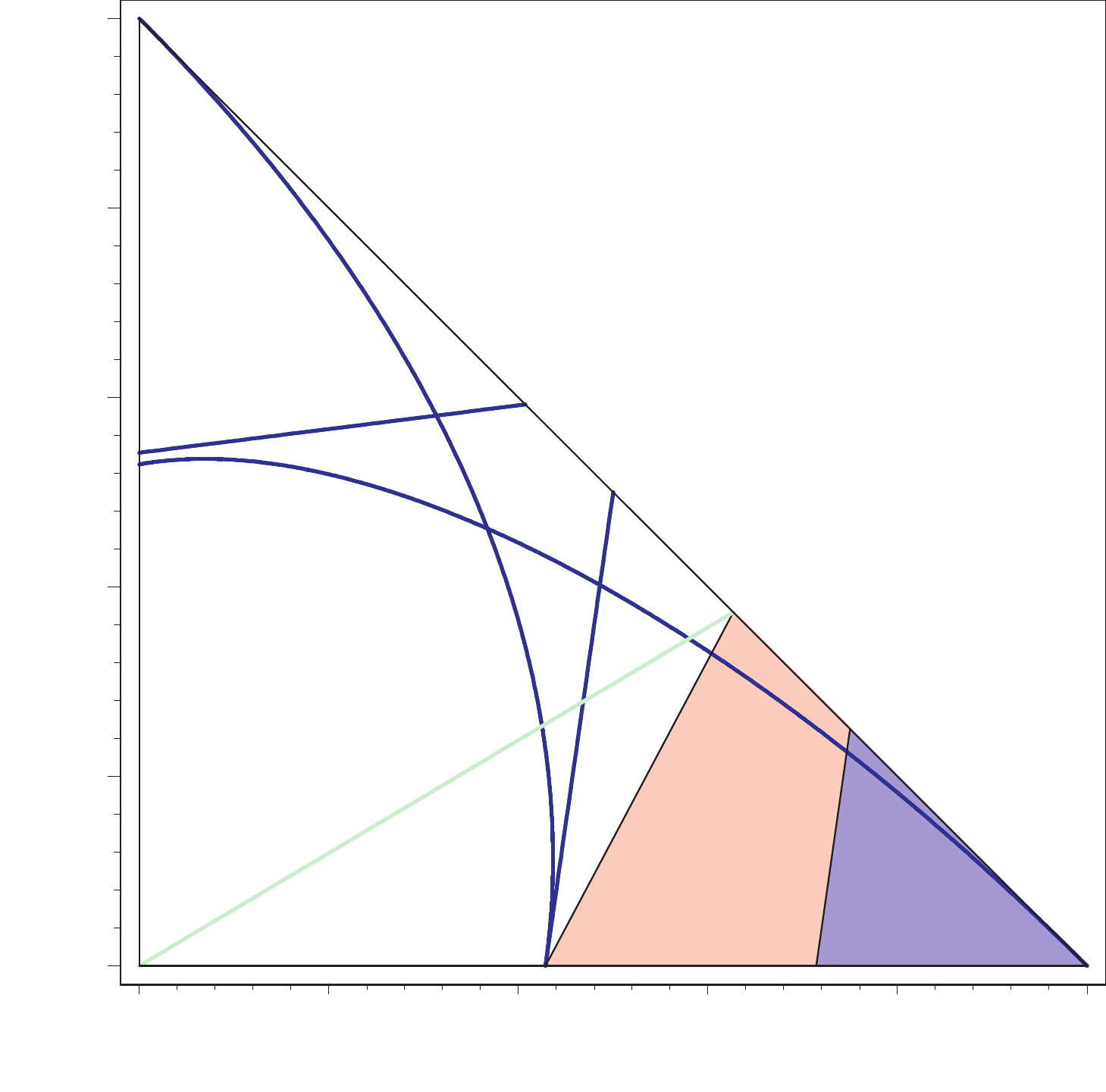}}
  \put(10,226){\makebox(0,0)[r]{\strut{}$w$}}
  \put(235,6){\makebox(0,0)[r]{\strut{}$g$}}
{\small
  \put(56,27){\makebox(0,0)[r]{\strut{}$0$}}
  \put(133,27){\makebox(0,0)[r]{\strut{}$0.2$}}
  \put(205,27){\makebox(0,0)[r]{\strut{}$0.4$}}
  \put(277,27){\makebox(0,0)[r]{\strut{}$0.6$}}
  \put(349,27){\makebox(0,0)[r]{\strut{}$0.8$}}
  \put(415,27){\makebox(0,0)[r]{\strut{}$1$}}
  \put(38,46){\makebox(0,0)[r]{\strut{}$0$}}
  \put(38,118){\makebox(0,0)[r]{\strut{}$0.2$}}
  \put(38,190){\makebox(0,0)[r]{\strut{}$0.4$}}
  \put(38,262){\makebox(0,0)[r]{\strut{}$0.6$}}
  \put(38,334){\makebox(0,0)[r]{\strut{}$0.8$}}
  \put(40,406){\makebox(0,0)[r]{\strut{}$1$}}
}
  \put(130,265){\makebox(0,0)[r]{\strut{}(\ref{eq:critWW1})}}
  \put(280,230){\makebox(0,0)[r]{\strut{}(\ref{eq:critWW2})}}
  \put(355,80){\makebox(0,0)[r]{\strut{}(\ref{eq:critWGHZ})}}
  \put(160,80){\makebox(0,0)[r]{\strut{}(\ref{eq:critGHZ2})}}
  \put(120,310){\makebox(0,0)[r]{\strut{}(\ref{eq:critGS2.1})}}
  \put(150,210){\makebox(0,0)[r]{\strut{}(\ref{eq:critGS2.2})}}
 \end{picture}
 \caption{Criteria on tripartite entanglement classes for the state (\ref{eq:lo}) on the $g$-$w$-plane.
(Blue straight lines: the borders of domains inside equations (\ref{eq:critWW1})-(\ref{eq:critWW2}) hold,
blue curves of second order: the borders of domains inside equations (\ref{eq:critGS2.1})-(\ref{eq:critGS2.2}) hold,
copied from figure~\ref{fig:Mxe}.
equation (\ref{eq:critWGHZ}) holds inside the blue domain,
and the border of Class GHZ is inside the red domain.
The inequalities hold on the side of the curves containing the origin.
equation (\ref{eq:critGHZ2}) holds under the green line.)}
 \label{fig:WGHZ}
\end{figure}
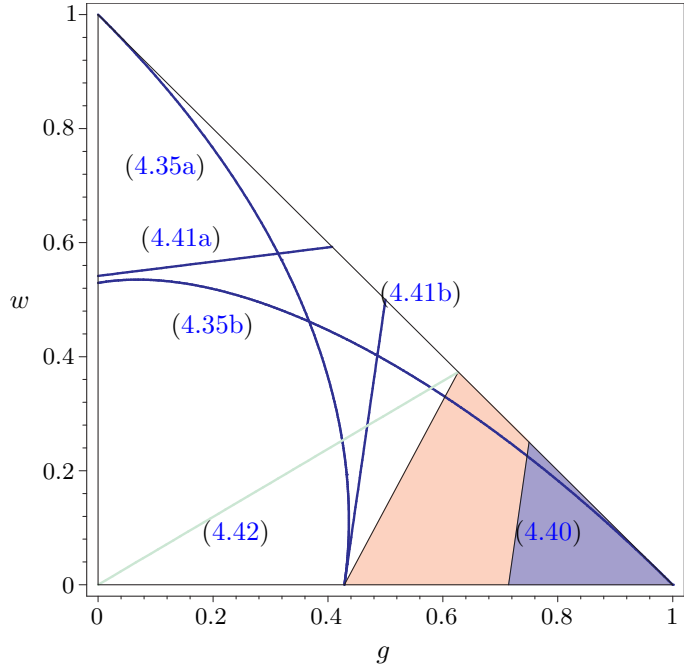

A different kind of method to determine to which class a given mixed state belongs is
the use of witness operators (section \ref{subsec:QM.Ent.2Part}).
In~\cite{Acinetal3QBMixClass} there have been given some witnesses
for $\mathcal{D}_{\text{W}}$ and $\mathcal{D}_{\text{GHZ}}$, namely,
\begin{equation}
W_{\text{GHZ}}=\frac34\Id\otimes\Id\otimes\Id-\cket{\text{GHZ}}\bra{\text{GHZ}}
\end{equation}
can detect $\mathcal{D}_{\text{GHZ}}$ and
\begin{subequations}
\begin{align}
W_{\text{W}_1}&=\frac23\Id\otimes\Id\otimes\Id-\cket{\text{W}}\bra{\text{W}},\\
W_{\text{W}_2}&=\frac12\Id\otimes\Id\otimes\Id-\cket{\text{GHZ}}\bra{\text{GHZ}}
\end{align}
\end{subequations}
can detect $\mathcal{D}_{\text{W}}$.
With these we have
\begin{equation}
\varrho\in \mathcal{D}_{\text{W}}\qquad\Longrightarrow\quad
\label{eq:critWGHZ}
0\leq\tr W_{\text{GHZ}}\varrho=(20\td-2\tg+9\tw)/4
 =(5-7g+w)/8
\end{equation}
and if the inequality is violated then $\varrho\in\text{Class GHZ}$,
as well as
\begin{subequations}
\begin{align}
\varrho\in \mathcal{D}_{2-\text{sep}}\quad\Longrightarrow\quad
\label{eq:critWW1}
0&\leq\tr W_{\text{W}_1}\varrho=(13\td+4\tg-3\tw)/3
 =(13+3g-21w)/24,\\
\label{eq:critWW2}
0&\leq\tr W_{\text{W}_2}\varrho=(6\td-2\tg+3\tw)/2
 =(3-7g+w)/8
\end{align}
\end{subequations}
and if either of the inequalities is violated then $\varrho\in\text{Class 1}$.
In figure~\ref{fig:WGHZ} we plot the lines on which these inequalities are saturated.
It can be checked that (\ref{eq:critWW1}) and~(\ref{eq:critWW2}) give
weaker conditions for biseparability than (\ref{eq:critGS2.1}) and~(\ref{eq:critGS2.2}) of the previous section.
We can conclude that all the states in the blue domain belong to Class GHZ,
and the biseparable states are enclosed by the blue curves,
however, both type of fully entangled states can be here too.

The equality in (\ref{eq:critWGHZ}) gives an ``upper bound'' for the border of Class GHZ
(blue domain in figure~\ref{fig:WGHZ}).
Fortunately, we have a possibility to give also a ``lower bound'' for that,
thanks to the results of Lohmayer et.~al.~\cite{MixedThreetangle}.
They have studied the GHZ-W mixture ($d=0$, $w=1-g$) and
they have found that there exists a decomposition of projectors onto vectors of vanishing three-tangle
if and only if $0\leq g\leq g_0=4\cdot2^{1/3}/(3+4\cdot2^{1/3})=0.626851\dots$,
hence for these parameters the convex-roof extension $\cnvroof{\tau}$ of the three-tangle is zero.
If we mix the states of this interval with white noise
then the three-tangle remains zero and neither of these states can belong to Class GHZ.
So we can state that
\begin{equation}
\label{eq:critGHZ2}
\varrho\in\mathcal{D}_{\text{GHZ}}\qquad\Longrightarrow\qquad
w<\frac3{4\cdot2^{1/3}}g,
\end{equation}
which holds under the green line of figure~\ref{fig:WGHZ}.
This condition is quite weak, but we can make it stronger.
Recall that on the $w=0$ line (noisy GHZ state)
$\varrho\in\text{Class 1}$ if and only if $3/7< g\leq1$ (section~\ref{sec:SepCrit.Rho}).
So the convexity of $\mathcal{D}_{\text{W}}$
restricts Class GHZ to be inside the triangle defined by the vertices
$(g=3/7,w=0)$, $(g=1,w=0)$ and $(g=g_0,w=1-g_0)$
(union of tinted domains in figure~\ref{fig:WGHZ}).
So we can conclude that
all the states in the blue domain belong to Class GHZ,
and \emph{the border of Class GHZ is in the red domain of figure~\ref{fig:WGHZ}.}

%******************************************************************************
\subsection{Wootters concurrence}
\label{subsec:SepCrit.Class1.Woott}

The Wootters concurrence (\ref{eq:WConc}) measures the entanglement inside the two-qubit subsystems.
Let us calculate that for the two-qubit reduced state $\varrho_{23}$, which is given in (\ref{eq:MxR23}).
Since the spin-flip for two-qubit density matrices means
transpose with respect to the antidiagonal
then multiplication of neither diagonal nor antidiagonal entries by $-1$,
one can easily get that
\begin{equation}
\label{eq:spectRho23W}
\Eigv\bigl(\tilde{\varrho}_{23}\varrho_{23}\bigr)=
\begin{aligned}[t]\bigl\{\qquad\qquad\qquad
4(\td+\tw)^2&= 4(3-3g+5w)^2/24^2,\\
(2\td+\tg)(2\td+\tg+\tw)
&= 12(1+g-w)(3+3g+w)/24^2,\\
(2\td+\tg)(2\td+\tg+\tw)
&= 12(1+g-w)(3+3g+w)/24^2,\\
%(2\td+\tg)(2\td+\tg+\tw)&= 12(1+g-w)(3+3g+w)/24^2,\\
4\td^2&= 36(1-g-w)^2/24^2\qquad\qquad\qquad\bigr\}.
\end{aligned}
\end{equation}
Clearly, the last eigenvalue is the smallest one.
If $4(\td+\tw)^2\leq(2\td+\tg)(2\td+\tg+\tw)$
then $\lambda^\downarrow_1=\lambda^\downarrow_2$ hence $\cnvroof{c}(\varrho_{23})=0$.
If $4(\td+\tw)^2\geq(2\td+\tg)(2\td+\tg+\tw)$,
then $\lambda^\downarrow_2=\lambda^\downarrow_3$ and $\cnvroof{c}(\varrho_{23})$ can be nonzero.
It turns out that the Wootters concurrence is
\begin{equation}
\label{eq:Conc}
\cnvroof{c}(\varrho_{23})=2\tw-2\sqrt{(2\td+\tg)(2\td+\tg+\tw)} 
=\frac{2}{3}w-\frac{1}{2\sqrt{3}}\sqrt{(1+g-w)(3+3g+w)}
\end{equation}
if $0\leq \tw^2 - (2\td+\tg)(2\td+\tg+\tw)
=(-9g^2+19w^2+6gw-18g+6w-9)/12^2$,
otherwise $\cnvroof{c}(\varrho_{23})=0$ (figure~\ref{fig:Woott}).
It takes its maximum $2/3$ in $(g=0,w=1)$, that is, for pure W-state.
For the GHZ-W mixture ($d=0$) the result is calculated also in~\cite{MixedThreetangle}.
On the other hand, it can be checked by (\ref{eq:critWW1}) and (\ref{eq:critGHZ2}) 
that all states in this simplex which have entangled two-qubit subsystems
are in Class W.

\begin{figure}
 \setlength{\unitlength}{0.001369863014\textwidth}% this 1/(365)*0.5 = 0.001369863014
 \begin{picture}(365,502)
  \put(0,0){\includegraphics[width=0.5\textwidth]{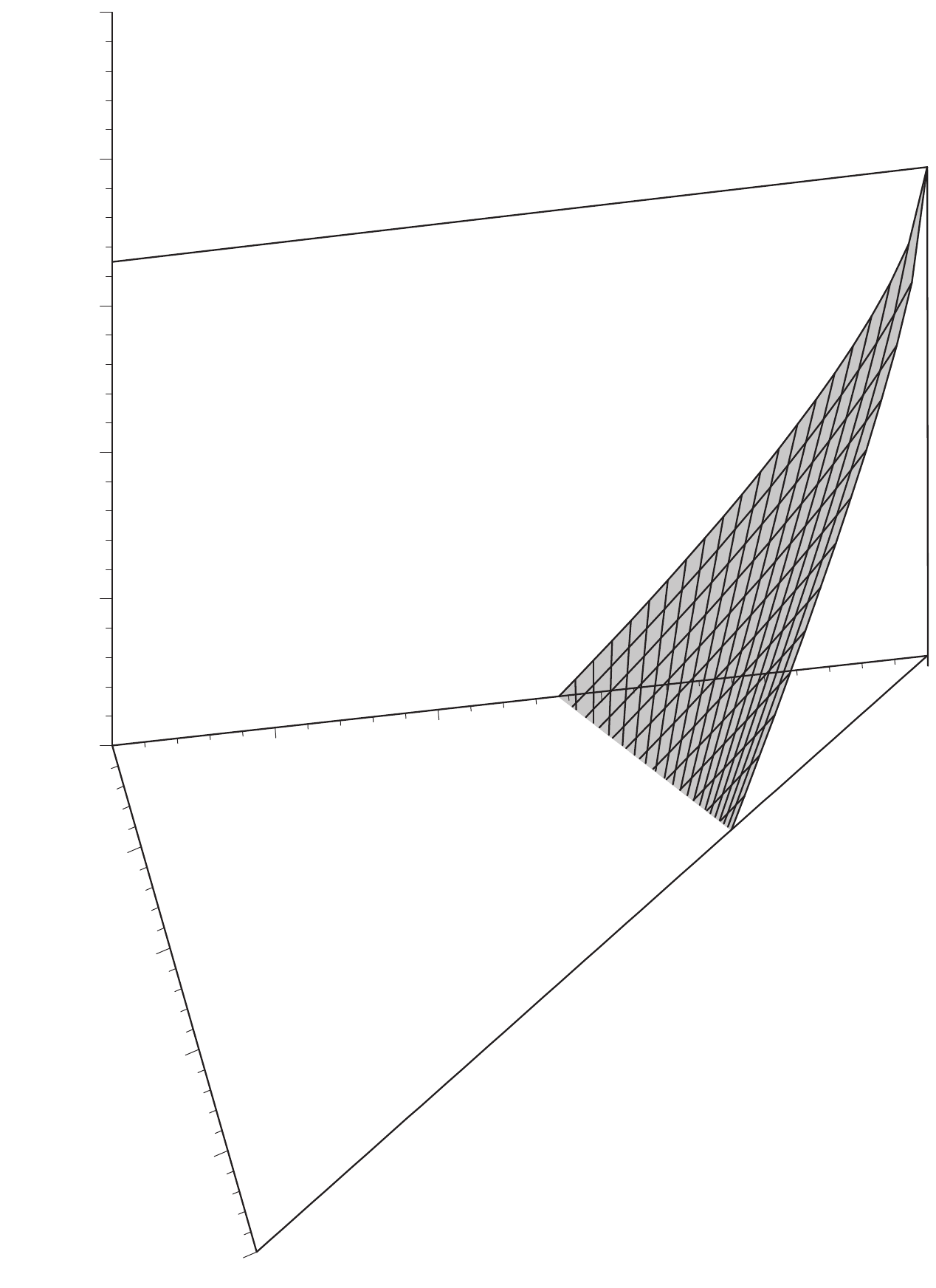}}
  \put(10,350){\makebox(0,0)[r]{\strut{}$\cnvroof{c}(\varrho_{23})$}}
  \put(30,90){\makebox(0,0)[r]{\strut{}$g$}}
  \put(210,195){\makebox(0,0)[r]{\strut{}$w$}}
{\small
  \put(50,160){\makebox(0,0)[r]{\strut{}$0.2$}}
  \put(60,122){\makebox(0,0)[r]{\strut{}$0.4$}}
  \put(70,83){\makebox(0,0)[r]{\strut{}$0.6$}}
  \put(80,44){\makebox(0,0)[r]{\strut{}$0.8$}}
  \put(90,4){\makebox(0,0)[r]{\strut{}$1$}}

  \put(118,205){\makebox(0,0)[r]{\strut{}$0.2$}}
  \put(180,212){\makebox(0,0)[r]{\strut{}$0.4$}}
  \put(244,219){\makebox(0,0)[r]{\strut{}$0.6$}}
  \put(308,226){\makebox(0,0)[r]{\strut{}$0.8$}}
  \put(366,233){\makebox(0,0)[r]{\strut{}$1$}}

  \put(38,209){\makebox(0,0)[r]{\strut{}$0$}}
  \put(37,266){\makebox(0,0)[r]{\strut{}$0.2$}}
  \put(37,323){\makebox(0,0)[r]{\strut{}$0.4$}}
  \put(37,380){\makebox(0,0)[r]{\strut{}$0.6$}}
  \put(37,437){\makebox(0,0)[r]{\strut{}$0.8$}}
  \put(40,495){\makebox(0,0)[r]{\strut{}$1$}}
}
 \end{picture}
 \caption{Wootters concurrence of $\varrho_{23}$ on the $g$-$w$-plane (\ref{eq:Conc}).}
 \label{fig:Woott}
\end{figure}
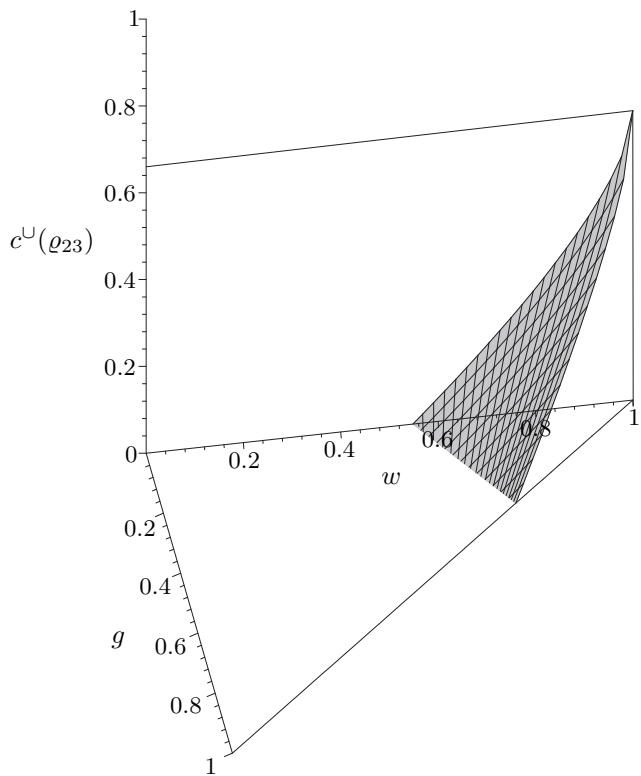

%******************************************************************************
%******************************************************************************
\section{Summary and remarks}
%\label{sec:Concl}
\label{sec:SepCrit.Concl}

In this chapter we have investigated the noisy GHZ-W mixture
and demonstrated some necessary but not sufficient criteria
for different classes of separability.
With these criteria we can restrict these classes into some domains of the 2-dimension simplex.
It has turned out that
the strongest conditions was
(\ref{eq:critPPT1}), (\ref{eq:critGS3.11}) and~(\ref{eq:critPPT2}) for full separability,
(\ref{eq:critPPT1}) and (\ref{eq:critPPT2}) for the union of Classes 2.8 and 3 and
(\ref{eq:critGS2.1}) and (\ref{eq:critGS2.2}) for biseparability.
These have been obtained from
the partial transposition criterion of Peres \cite{PeresCrit}
and the criteria of G\"uhne and Seevinck \cite{GuhneSevinckCrit} dealing with matrix elements.
Only these latter criteria have turned out to be strong enough
to reveal a set of entangled states of positive partial transpose.
(The set of these states can be given by the conditions of
(\ref{eq:critPPT1}), (\ref{eq:critPPT2}) and~(\ref{eq:critGS3.11}).
An example is given in equation (\ref{eq:PPTESlo}).)
We have also investigated the W and GHZ classes of fully entangled states
and we have given restrictions for Class GHZ.

\begin{remarks}
%%%%%%%%%%%%%%%%%%%%%%%%
\item 
%Besides this, some remarkable coincidences have also appeared:
Some parts of some bipartite separability criteria have proved to be necessary and sufficient
for separability classes of the GHZ-white noise mixture.
(These are the majorization criterion and the entropy criterion in the $\alpha\to\infty$ limit,
and reduction criterion.)
This is interesting because, e.g., the majorization criterion, to our knowledge, does not state anything about a density matrix
which is majorized by only either of its subsystems.
We do not think that this would be more than a coincidence,
however, the GHZ state is a very special one
so it is an interesting question
as to whether this can be generalized to the $n$-qubit noisy GHZ state,
or to some kinds of generalized GHZ states.
%%%%%%%%%%%%%%%%%%%%%%%%
\item Another interesting coincidence was that
the two settings and measurement vectors strong in detection of GHZ and W states
are related by local unitary Hadamard transformation
in the criteria on spin-observables of Seevinck and Uffink (section~\ref{subsec:SepCrit.3Part.Spin})
and also in the criteria of Gabriel et.~al.~(section~\ref{subsec:SepCrit.3Part.Hub}).
The transformation on settings and measurement vectors
can also be written on the state $\varrho\mapsto (H^{\otimes3})^\dagger\varrho H^{\otimes3}$,
which means that we can use the same measurements
on the transformed state
for the detection of W state as for the detection of GHZ state.
The interesting point here is that
this transformation, being local unitary, does not make a W state from a GHZ state.
%%%%%%%%%%%%%%%%%%%%%%%%
\item The issue of entangled two-qubit subsystems (section~\ref{subsec:SepCrit.Class1.Woott})
seems to be interesting.
Entangled two-qubit subsystems arise only in Class W, which is probably the consequence of the special mixture,
but it would be interesting to find some criteria for entangled two-qubit subsystems.
We will return to this issue in the next chapters,
and, as side results, we give conditions for this
for pure states (section \ref{subsec:ThreeQB.Pure.WConc})
and for mixed states (section \ref{subsec:PartSep.Threepart.Examples}).
\end{remarks}

\chapter{Partial separability classification}
\label{chap:PartSep}

In section \ref{subsec:QM.Ent.NPart} we have seen how the different partially separable pure states
give rise to a classification of mixed states by $\alpha_k$-separable and $k$-separable states \cite{SeevinckUffinkMixSep},
written out in the case of three subsystems.
Then in the previous chapter we gave some illustrations for restrictions of some of these classes
within a two-parameter simplex of permutation-invariant three-qubit mixed states.

In this chapter we show that this classification
does not take the problem of partial separability of mixed states in the full detail.
We extend this classification, moreover, 
we give necessary and sufficient criteria for the classes.
We call our extended classification \emph{PS classification}, which stands for Partial Separability,
because this classification is complete in the sense of partial separability,
that is, it utilizes all the possible combinations of different kinds of partially separable pure states.
We get this finding using the point of view that
a state is a mixture of an ensemble of pure states,
which leads us to a set of necessary and sufficient criteria for the classes.
We have seen in section \ref{subsec:QM.EntMeas.2Mix} that in the bipartite case,
where a state---either pure or mixed---can be either separable or entangled,
the vanishing of the convex roof extension of local entropies of pure states
is a \emph{necessary and sufficient criterion} of separability.
For us, this is the archetype of the general method for the detection of convex subsets by convex roof extensions.
However, for more-than-two-partite systems,
the partial separability properties have a complicated structure,
and, to our knowledge, this method was not used.
Instead of that,
the usual approach is the use of witness operators, as was done originally for three-qubit systems (section~\ref{sec:SepCrit.Class1}) %\cite{Acinetal3QBMixClass},
or other \emph{necessary but not sufficient criteria} for the detection of convex subsets,
some of them were reviewed in chapter~\ref{chap:SepCrit}.

The material of this chapter covers thesis statement \ref{statement:partsep}
(page \pageref{statement:partsep}).
\begin{organization}
\item[\ref{sec:PartSep.ThreePart}]
we elaborate the PS classification for tripartite mixed states.
We define the PS subsets (section~\ref{subsec:PartSep.ThreePart.PSsubsets})
and PS classes (section~\ref{subsec:PartSep.Threepart.PSclasses}),
and we give some examples for states which are contained in classes different only under the PS classification 
(section~\ref{subsec:PartSep.Threepart.Examples}).
Then we give necessary and sufficient criteria for the identification of the PS classes 
and obtain the functions by which these criteria can be formulated (section~\ref{subsec:PartSep.ThreePart.Indicators}).
\item[\ref{sec:PartSep.Gen}]
we generalize the construction for the case of arbitrary number of subsystems of arbitrary dimensions.
We work out the labelling of the PS subsets (section~\ref{subsec:PartSep.Gen.PSsubsets})
along with that of the PS classes,
and give a general conjecture about their non-emptiness (section~\ref{subsec:PartSep.Gen.PSclasses}).
Then we construct
the functions identifying the PS subsets and classes
with the minimal requirements (section~\ref{subsec:PartSep.Gen.Indicators}),
and also with stronger requirements leading to entanglement-monotone functions
(section~\ref{subsec:PartSep.Gen.monIndicators}).
\item[\ref{sec:PartSep.Sum}]
we give a summary and some remarks.
\end{organization}

%*******************************************************************************
%*******************************************************************************
\section{Partial separability of tripartite mixed states}
\label{sec:PartSep.ThreePart}

Here we introduce the PS classification for three subsystems.
We have already seen the main concept in section \ref{subsec:QM.Ent.NPart},
first given in \cite{DurCiracTarrach3QBMixSep,DurCiracTarrachBMixSep},
then used and extended in \cite{SeevinckUffinkMixSep,Acinetal3QBMixClass}, 
that we define a density matrix to be the element of a class
according to whether it can or can not be mixed by the use of
pure states of some given kinds. %, see in section \ref{subsec:QM.Ent.NPart}.
%Here we use more kinds

%*******************************************************************************
\subsection{PS subsets}
\label{subsec:PartSep.ThreePart.PSsubsets}

We have $\mathcal{H}=\mathcal{H}_1\otimes\mathcal{H}_2\otimes\mathcal{H}_3$ with
arbitrary $\tpl{d}=(d_1,d_2,d_3)$ local dimensions.
%Let us introduce some convenient notations
%for the disjoint subsets in the set of extremal points $\mathcal{P}$ of $\mathcal{D}$
%given by unit vectors of different partial separability%
Let us introduce some convenient notations
for the subsets in the set of extremal points $\mathcal{P}$ of $\mathcal{D}$
given by unit vectors% of different partial separability%
%%%%%%%%%%%%%%%%%%%%%%%%
\footnote{
Remember our convention:
The letters $a$, $b$ and $c$ are variables
taking their values in the set of labels $L=\{1,2,3\}$.
When these variables appear in a formula,
they form a partition of $\{1,2,3\}$,
so they take always different values
and \emph{the formula is understood for all the different values of these variables automatically.}
Although, sometimes a formula is symmetric under the interchange of two such variables
in which case we keep only one of the identical formulas.}
%%%%%%%%%%%%%%%%%%%%%%%%
\begin{subequations}
\label{eq:PSPsets}
\begin{align}
\mathcal{P}_{1|2|3} &=\bigl\{\pi\in\mathcal{P}\;\big\vert\; \pi=\pi_1\otimes\pi_2\otimes\pi_3 \bigr\},\\
\mathcal{P}_{a|bc}  &=\bigl\{\pi\in\mathcal{P}\;\big\vert\; \pi=\pi_a\otimes\pi_{bc}\bigr\},\\
\mathcal{P}_{123}   &\equiv\mathcal{P}.
\end{align}
\end{subequations}
We call these \emph{pure PS subsets}, and they are closed and \emph{contain} each other in a hierarchic way,
which is illustrated in figure~\ref{fig:PS3incl}.
The meaning of these sets is that their elements can contain \emph{``at most''} a given entanglement.
%(We define this ordering later.)
We can also introduce the \emph{disjoint} subsets given by unit vectors of different partial separability
\begin{subequations}
\label{eq:PSPclasses}
\begin{align}
\mathcal{Q}_{1|2|3} &=\bigl\{\pi\in\mathcal{P}\;\big\vert\; \pi=\pi_1\otimes\pi_2\otimes\pi_3 \bigr\} = \mathcal{P}_{1|2|3},\\
\mathcal{Q}_{a|bc}  &=\bigl\{\pi\in\mathcal{P}\;\big\vert\; \pi=\pi_a\otimes\pi_{bc},\; \pi_{bc}\neq\pi_b\otimes\pi_c \bigr\} 
= \mathcal{P}_{a|bc}\setminus \mathcal{P}_{1|2|3},\\
\mathcal{Q}_{123}   &=\bigl\{\pi\in\mathcal{P}\;\big\vert\; \pi\neq \pi_a\otimes\pi_{bc} \bigr\} 
= \mathcal{P}_{123} \setminus \bigl(\mathcal{P}_{1|23} \cup \mathcal{P}_{2|13} \cup \mathcal{P}_{3|12}\bigr).
\end{align}
We call these \emph{pure PS classes},
and they cover $\mathcal{P}$ entirely,
$\mathcal{P}=\mathcal{Q}_{1|2|3}\cup\mathcal{Q}_{1|23}\cup\mathcal{Q}_{2|13}\cup\mathcal{Q}_{3|12}\cup\mathcal{Q}_{123}$.
\end{subequations}
Except $\mathcal{Q}_{1|2|3}$, none of the above sets are closed.
The meaning of these sets is that their elements contain \emph{exactly} a given entanglement.

The notion of $k$-separability and $\alpha_k$-separability, 
given in \cite{SeevinckUffinkMixSep} and recalled in section \ref{subsec:QM.Ent.NPart},
can be formulated now as the convex hulls of some of the sets (\ref{eq:PSPclasses}).
%%%%%%
The \emph{$3$-separable states} ($\mathcal{D}_\text{$3$-sep}$), 
or equivalently $1|2|3$-separable states ($\mathcal{D}_{1|2|3}$) 
can be mixed from the pure states of $\mathcal{Q}_{1|2|3}$, that is, they are fully separable.
%%%%%%
The \emph{$a|bc$-separable states} ($\mathcal{D}_{a|bc}$)
can be written in the form $\sum_ip_i\varrho_{a,i}\otimes\varrho_{bc,i}$,
($\varrho_{a,i}\in\mathcal{D}(\mathcal{H}_a)$,
$\varrho_{bc,i}\in\mathcal{D}(\mathcal{H}_b\otimes\mathcal{H}_c)$,)
where we demand only the split between $a$ and $bc$,
but split between $b$ and $c$ can also occur in the pure state decompositions,
so they can be mixed from the pure states of $\mathcal{Q}_{1|2|3}$ and $\mathcal{Q}_{a|bc}$,
that is, of $\mathcal{P}_{a|bc}$.
%%%%%%
The \emph{$2$-separable states} ($\mathcal{D}_\text{$2$-sep}$) 
are of the form $\sum_ip_i\varrho_{a_i,i}\otimes\varrho_{b_ic_i,i}$,
so they can be mixed from the pure states of $\mathcal{Q}_{1|2|3}$,
$\mathcal{Q}_{1|23}$, $\mathcal{Q}_{2|13}$ and $\mathcal{Q}_{3|12}$,
that is, of $\mathcal{P}_{1|23}$, $\mathcal{P}_{2|13}$ and $\mathcal{P}_{3|12}$.
These states are also of relevance,
since although they are not separable under any $a|bc$ split,
but there is no need of genuine tripartite entangled pure state to mix them \cite{SeevinckUffinkMixSep}.
%%%%%%
From the point of view of convex hulls of extremal points,
it can be seen better than originally in \cite{SeevinckUffinkMixSep} that we can define
three new partial separability sets ``between'' the $a|bc$-separable
and $2$-separable ones.
For example, the \emph{$2|13$-$3|12$-separable states} ($\mathcal{D}_{2|13,3|12}$)
are the states which
can be mixed from the pure states of $\mathcal{Q}_{1|2|3}$, $\mathcal{Q}_{2|13}$, and $\mathcal{Q}_{3|12}$,
that is, of $\mathcal{P}_{2|13}$, and $\mathcal{P}_{3|12}$.
States of this kind are also of relevance,
since there is no need of $1|23$-separable pure states to mix them,
that is, entanglement within $23$ subsystem.
%%%%%%
Beyond these, 
we use the set of \emph{$123$-separable states} ($\mathcal{D}_{123}$),
or equivalently $1$-separable ($\mathcal{D}_\text{$1$-sep}$),
which is equal to the full set of states ($\mathcal{D}$).
Summarizing, we have the following 
\emph{PS subsets} in $\mathcal{D}$
arising as convex hulls of pure states of given kinds:
\begin{subequations}
\label{eq:PSsets}
\begin{align}
\label{eq:PSsets.1|2|3}
\mathcal{D}_{1|2|3} &= \Conv\bigl(\mathcal{P}_{1|2|3}\bigr)\equiv\mathcal{D}_\text{$3$-sep},\\
\mathcal{D}_{a|bc}  &= \Conv\bigl(\mathcal{P}_{a|bc}\bigr),\\
\label{eq:PSsets.bcacab}
\mathcal{D}_{b|ac,c|ab}  &= \Conv \bigl(\mathcal{P}_{b|ac}\cup\mathcal{P}_{c|ab}\bigr),\\
\label{eq:PSsets.2-sep}
\mathcal{D}_{1|23,2|13,3|12}  &= \Conv \bigl(\mathcal{P}_{1|23}\cup\mathcal{P}_{2|13}\cup\mathcal{P}_{3|12}\bigr)
\equiv\mathcal{D}_\text{$2$-sep},\\
\label{eq:PSsets.123}
\mathcal{D}_{123}  &= \Conv \bigl(\mathcal{P}_{123}\bigr)
\equiv\mathcal{D}_\text{$1$-sep}\equiv\mathcal{D}.
\end{align}
\end{subequations}
These sets are convex by construction,
and they contain each other in a hierarchic way,
which is illustrated in figure~\ref{fig:PS3incl}.
%%%%%%%%%%%%%%%%%%%%%%%%%%%%%%%%%%%%%%%%
\begin{figure}
 \includegraphics{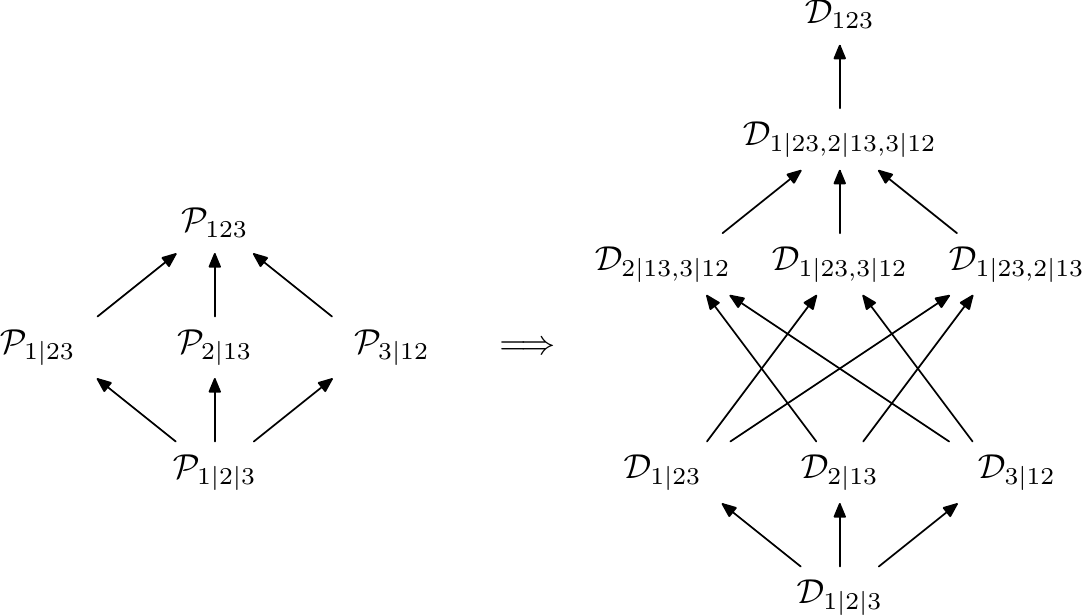}
 \caption{Inclusion hierarchy of the pure and mixed PS sets $\mathcal{P}_{\dots}$ and $\mathcal{D}_{\dots}$ given in (\ref{eq:PSPsets}) and (\ref{eq:PSsets}).}
\label{fig:PS3incl}
\end{figure}
%%%%%%%%%%%%%%%%%%%%%%%%%%%%%%%%%%%%%%%%

From an abstract point of view,
we form the convex hulls of \emph{closed} sets \cite{Acinetal3QBMixClass},
and the convex hulls of \emph{all the possible closed sets}
arising from the unions of the $\mathcal{Q}_{\dots}$ sets (\ref{eq:PSPclasses}) of extremal points
are listed in (\ref{eq:PSsets}) above.
We mean the PS classification
involving the PS subsets (\ref{eq:PSsets.1|2|3})-(\ref{eq:PSsets.123})
to be \emph{complete} in \emph{this} sense.
As special, non-complete cases, 
we get back the classification involving only the sets 
$\mathcal{D}_\text{$k$-sep}$ and $\mathcal{D}_{\alpha_k}$ 
(for any $k$-partite split $\alpha_k$)
obtained by Seevinck and Uffink \cite{SeevinckUffinkMixSep},
and also the classification involving only the sets $\mathcal{D}_{\alpha_k}$
obtained by D\"ur, Cirac and Tarrach \cite{DurCiracTarrach3QBMixSep,DurCiracTarrachBMixSep}.
%and also the classification involving only the sets $\mathcal{D}_\text{$k$-sep}$ and $\mathcal{D}_\text{W}$,
%obtained by Ac\'in, Bru\ss{}, Lewenstein and Sanpera \cite{Acinetal3QBMixClass}.

%*******************************************************************************
\subsection{PS classes}
\label{subsec:PartSep.Threepart.PSclasses}

Now we determine the \emph{PS classes} of tripartite mixed states.
The abstract definition of these classes \cite{SeevinckUffinkMixSep}
is that they are the possible non-trivial intersections of the 
$\mathcal{D}_{\dots}$ convex subsets listed in (\ref{eq:PSsets}).
Since we want to deal also with the sets $\mathcal{D}_{b|ac,c|ab}$,
we can not draw an expressive ``onion-like'' figure as is shown in figure \ref{fig:3part} 
% was done in \cite{SeevinckUffinkMixSep}
for the sets $\mathcal{D}_{1|2|3}$, $\mathcal{D}_{a|bc}$ and $\mathcal{D}_\text{$2$-sep}$.
We have to proceed in a formal manner.

If we have the sets $A_1,A_2,\dots,A_n$,
all of their possible intersections can be constructed
as the intersections for each $i$ the set $A_i$ or its complement $\cmpl{A}_i$.
We have $9$ PS subsets $\mathcal{D}_{\dots}$,
so we can formally write $2^9=512$ possible intersections in this way.
If $B\subseteq A$ then $B\cap \cmpl{A}=\emptyset$,
so some intersecions are automatically empty
(\emph{``empty by construction''}),
and, 
using the inclusion hierarchy of PS subsets in figure~\ref{fig:PS3incl},
we write only the intersections which are ``not empty by construction.''
The number of these turns out to be only $20$.
(Again, if $B\subseteq A$ then $B\cap A=B$ and $\cmpl{B}\cap \cmpl{A}=\cmpl{A}$,
so we can write these $20$ classes as intersection sequences
much shorter than $9$ terms.)
Since the appearance of the $\mathcal{D}_{b|ac,c|ab}$-type sets in the intersections
makes the meaning of the classes a little bit involved,
we write out the list of the PS classes with detailed explanations.
We list the classes in table~\ref{tab:PS3Classes}.

First, the class
\begin{subequations}
\label{eq:PSClasses}
\begin{equation}
%\mathcal{C}_3=\mathcal{D}_\text{$3$-sep}
\mathcal{C}_3=\mathcal{D}_{1|2|3}
\end{equation}
is the set of fully separable states.
(This is Class $3$ in section \ref{subsec:QM.Ent.NPart}.)

Then come the $18$ classes of $2$-separable entangled states,
that is, the subsets in $\mathcal{D}_{1|23,2|13,3|12}\setminus\mathcal{D}_{1|2|3}$.
%%%%%%
The class
\begin{equation}
\mathcal{C}_{2.8}=\cmpl{\mathcal{D}_{1|2|3}}\cap\mathcal{D}_{1|23}\cap\mathcal{D}_{2|13}\cap\mathcal{D}_{3|12}
\equiv\bigl(\mathcal{D}_{1|23}\cap\mathcal{D}_{2|13}\cap\mathcal{D}_{3|12}\bigr)\setminus\mathcal{D}_{1|2|3}
\end{equation}
is the set of states 
which can    be written as $1|23$-separable states
(that is, convex combinations of $\mathcal{Q}_{1|2|3}$ and $\mathcal{Q}_{1|23}$ pure states,
the formation is not unique)
and can also be written as $2|13$-separable states
and can also be written as $3|12$-separable states
but can not  be written as $1|2|3$-separable states.
%States of this class is sometimes called semiseparable. ..... cite ???
The existence of such states was counterintuitive,
since for pure states, if a tripartite pure state is separable under any $a|bc$ bipartition then it is fully separable.
For mixed states, however, explicit examples can be constructed \cite{BennettetalUPB,Acinetal3QBMixClass},
which can be written in the form $\sum_ip_i\varrho_{a,i}\otimes\varrho_{bc,i}$ for any $a|bc$ bipartition,
but can not be written in the form $\sum_ip_i\varrho_{1,i}\otimes\varrho_{2,i}\otimes\varrho_{3,i}$.
Alternatively, we can say that
states of this class can not be mixed without the use of bipartite entanglement,
but they can be mixed by the use of bipartite entanglement
within only one bipartite subsystem, it does not matter which one.
(This is Class $2.8$ in section \ref{subsec:QM.Ent.NPart}.)
%%%%%%
The next three classes are
\begin{equation}
\mathcal{C}_{2.7.a}=\cmpl{\mathcal{D}_{a|bc}}\cap\mathcal{D}_{b|ac}\cap\mathcal{D}_{c|ab}
\equiv\bigl(\mathcal{D}_{b|ac}\cap\mathcal{D}_{c|ab}\bigr)\setminus\mathcal{D}_{a|bc}.
\end{equation}
For eample, $\mathcal{C}_{2.7.1}$
is the set of states which can be written as $2|13$-separable states
and can also be written as $3|12$-separable states
but can not  be written as $1|23$-separable states.
Alternatively, we can say that
states of this class can not be mixed by the use of bipartite entanglement within only the $12$ subsystem,
but they can be mixed by the use of bipartite entanglement
within either the $23$ or the $13$ subsystems, %it does not matter which one.
both of them are equally suitable.
(These three classes are Classes $2.7$, $2.6$, and $2.5$ in section \ref{subsec:QM.Ent.NPart}.)
%%%%%%
The next three classes are
\begin{equation}
\mathcal{C}_{2.6.a}=\mathcal{D}_{a|bc}\cap\cmpl{\mathcal{D}_{b|ac}}\cap\cmpl{\mathcal{D}_{c|ab}}\cap\mathcal{D}_{b|ac,c|ab}
\equiv \mathcal{D}_{a|bc}\cap\bigl[\mathcal{D}_{b|ac,c|ab}\setminus\bigl(\mathcal{D}_{b|ac}\cup\mathcal{D}_{c|ab} \bigr)  \bigr].
\end{equation}
For eample, $\mathcal{C}_{2.6.1}$
is the set of states which can be written as $1|23$-separable states
and can also be written as states of a new kind,
where the state can be written as
$2|13$-$3|12$-separable states which are neither $2|13$-separable nor $3|12$-separable.
And this is the novelty here.
Alternatively, we can say that
to mix a state of this class
we need bipartite entanglement 
either within the $23$ subsystem,
or within both of the $12$ and $13$ subsystems.
(The latter seems like a roundabout connecting the $2$ and $3$ subsystems through the $1$ subsystem.)
%%%%%%
The next three classes are
\begin{equation}
\mathcal{C}_{2.5.a}=\mathcal{D}_{a|bc}\cap\cmpl{\mathcal{D}_{b|ac}}\cap\cmpl{\mathcal{D}_{c|ab}}\cap\cmpl{\mathcal{D}_{b|ac,c|ab}}
\equiv\mathcal{D}_{a|bc}\cap\cmpl{\mathcal{D}_{b|ac,c|ab}}
\equiv\mathcal{D}_{a|bc}\setminus\mathcal{D}_{b|ac,c|ab}.
\end{equation}
For eample, $\mathcal{C}_{2.5.1}$
is the set of states which can be written as $1|23$-separable states
but can not  be written as $2|13$-$3|12$-separable states.
Alternatively, we can say that
%This means that 
states of this class 
%an not be mixed without the use of bipartite entanglement within $23$ subsystem,
%ut in this case no bipartite entanglement within the other two bipartite subsystems are needed, contrary to $\mathcal{C}_{2.7.1}$.
%n the other hand, it 
can not be mixed by the use of bipartite entanglement only within both of the $13$ and $23$ subsystems, contrary to $\mathcal{C}_{2.6.1}$.
(The roundabout does not exist here.)
(The unions $\mathcal{C}_{2.6.a}\cup\mathcal{C}_{2.5.a}=\mathcal{D}_{a|bc}\cap\cmpl{\mathcal{D}_{b|ac}}\cap\cmpl{\mathcal{D}_{c|ab}}$
are Classes $2.4$, $2.3$, and $2.2$ in section \ref{subsec:QM.Ent.NPart}.)
%%%%%%%%%%%%%%%%%%%%%%%%%%%%%%%%%%%%%%%%
\begin{table}
\begin{tabu}{X[2,c]||X[1,c]|X[1,c]X[1,c]X[1,c]|X[1,c]X[1,c]X[1,c]|X[1,c]|X[1,c]||X[3,c]X[3,c]}
\hline
\begin{sideways}PS Class\end{sideways} & 
\begin{sideways}$\mathcal{D}_{1|2|3}$\end{sideways}  & 
\begin{sideways}$\mathcal{D}_{a|bc}$\end{sideways}  & 
\begin{sideways}$\mathcal{D}_{b|ac}$\end{sideways}  & 
\begin{sideways}$\mathcal{D}_{c|ab}$\end{sideways}  & 
\begin{sideways}$\mathcal{D}_{b|ac,c|ab}$\end{sideways}  & 
\begin{sideways}$\mathcal{D}_{a|bc,c|ab}$\end{sideways}  &
\begin{sideways}$\mathcal{D}_{a|bc,b|ac}$\end{sideways}  &
%$\mathcal{D}_\text{$2$-sep}$ &
\begin{sideways}$\mathcal{D}_{1|23,2|13,3|12}$\end{sideways}  &
\begin{sideways}$\mathcal{D}_{123}$\end{sideways}  &
%$\mathcal{D}_\text{GHZ}$ &
\begin{sideways}in \cite{SeevinckUffinkMixSep}\end{sideways}  &
\begin{sideways}in \cite{DurCiracTarrachBMixSep}\end{sideways}  \\
%$\mathcal{D}_\text{$1$-sep}$  \\
\hline
\hline
%                         || 1|2|3     |  a|bc         b|ac         c|ab      | b|ac-c|ab    c|ab-a|bc    a|bc-b|ac  |  2-sep     |  123
$\mathcal{C}_3$           & $\subset$  & $\subset$  & $\subset$  & $\subset$  & $\subset$  & $\subset$  & $\subset$  & $\subset$  & $\subset$  & 3       & 5       \\
\hline
$\mathcal{C}_{2.8}$       & $\nsubset$ & $\subset$  & $\subset$  & $\subset$  & $\subset$  & $\subset$  & $\subset$  & $\subset$  & $\subset$  & 2.8     & 4       \\
$\mathcal{C}_{2.7.a}$     & $\nsubset$ & $\nsubset$ & $\subset$  & $\subset$  & $\subset$  & $\subset$  & $\subset$  & $\subset$  & $\subset$  & 2.7,6,5 & 3.3,2,1 \\
$\mathcal{C}_{2.6.a}$     & $\nsubset$ & $\subset$  & $\nsubset$ & $\nsubset$ & $\subset$  & $\subset$  & $\subset$  & $\subset$  & $\subset$  & 2.4,3,2 & 2.3,2,1 \\
$\mathcal{C}_{2.5.a}$     & $\nsubset$ & $\subset$  & $\nsubset$ & $\nsubset$ & $\nsubset$ & $\subset$  & $\subset$  & $\subset$  & $\subset$  & 2.4,3,2 & 2.3,2,1 \\
$\mathcal{C}_{2.4}$       & $\nsubset$ & $\nsubset$ & $\nsubset$ & $\nsubset$ & $\subset$  & $\subset$  & $\subset$  & $\subset$  & $\subset$  & 2.1     & 1 \\
$\mathcal{C}_{2.3.a}$     & $\nsubset$ & $\nsubset$ & $\nsubset$ & $\nsubset$ & $\nsubset$ & $\subset$  & $\subset$  & $\subset$  & $\subset$  & 2.1     & 1 \\
$\mathcal{C}_{2.2.a}$     & $\nsubset$ & $\nsubset$ & $\nsubset$ & $\nsubset$ & $\subset$  & $\nsubset$ & $\nsubset$ & $\subset$  & $\subset$  & 2.1     & 1 \\
$\mathcal{C}_{2.1}$       & $\nsubset$ & $\nsubset$ & $\nsubset$ & $\nsubset$ & $\nsubset$ & $\nsubset$ & $\nsubset$ & $\subset$  & $\subset$  & 2.1     & 1 \\
\hline
$\mathcal{C}_1$           & $\nsubset$ & $\nsubset$ & $\nsubset$ & $\nsubset$ & $\nsubset$ & $\nsubset$ & $\nsubset$ & $\nsubset$ & $\subset$  & 1       & 1 \\
\hline
\end{tabu}
\bigskip
\caption{PS classes of tripartite mixed states.
Additionally, we show the 
classifications obtained by
Seevinck and Uffink \cite{SeevinckUffinkMixSep},
and D\"ur, Cirac and Tarrach \cite{DurCiracTarrachBMixSep}.} 
\label{tab:PS3Classes}
\end{table}
%%%%%%%%%%%%%%%%%%%%%%%%%%%%%%%%%%%%%%%%
%%%%%%
The next class is
\begin{equation}
\begin{split}
\mathcal{C}_{2.4}=&\cmpl{\mathcal{D}_{1|23}}\cap\cmpl{\mathcal{D}_{2|13}}\cap\cmpl{\mathcal{D}_{3|12}}
\cap\mathcal{D}_{2|13,3|12}\cap\mathcal{D}_{1|23,3|12}\cap\mathcal{D}_{1|23,2|13}\\
\equiv&\bigl(\mathcal{D}_{2|13,3|12}\cap\mathcal{D}_{1|23,3|12}\cap\mathcal{D}_{1|23,2|13}\bigr)
\setminus\bigl(\mathcal{D}_{1|23}\cup\mathcal{D}_{2|13}\cup\mathcal{D}_{3|12}\bigr)\\
\equiv&\bigl[\mathcal{D}_{2|13,3|12}\setminus\bigl(\mathcal{D}_{2|13}\cup\mathcal{D}_{3|12} \bigr)  \bigr]\\
  &\cap\bigl[\mathcal{D}_{1|23,3|12}\setminus\bigl(\mathcal{D}_{1|23}\cup\mathcal{D}_{3|12} \bigr)  \bigr]\\
  &\cap\bigl[\mathcal{D}_{1|23,2|13}\setminus\bigl(\mathcal{D}_{1|23}\cup\mathcal{D}_{2|13} \bigr)  \bigr],
\end{split}
\end{equation}
which is the set of states 
which can be mixed by the use of bipartite entanglement within any two bipartite subsystems,
but can not be mixed by the use of bipartite entanglement within only one bipartite subsystem.
%%%%%%
The next three classes are
\begin{equation}
\begin{split}
\mathcal{C}_{2.3.a}=&\cmpl{\mathcal{D}_{a|bc}}\cap\cmpl{\mathcal{D}_{b|ac,c|ab}}\cap
\mathcal{D}_{a|bc,c|ab}\cap\mathcal{D}_{a|bc,b|ac}\\
\equiv&\Bigl[\bigl[\mathcal{D}_{a|bc,c|ab}\setminus\bigl(\mathcal{D}_{c|ab}\cup\mathcal{D}_{a|bc} \bigr)  \bigr]
\cap   \bigl[\mathcal{D}_{a|bc,b|ac}\setminus\bigl(\mathcal{D}_{a|bc}\cup\mathcal{D}_{b|ac} \bigr)  \bigr]\Bigr]
\setminus\mathcal{D}_{b|ac,c|ab}.
\end{split}
\end{equation}
For eample, $\mathcal{C}_{2.3.1}$
is the set of states
which can be mixed by the use of bipartite entanglement within the $23$ subsystem
together with bipartite entanglement within either $12$ or $13$ subsystems,
but can not be mixed by the use of bipartite entanglement within $12$ \emph{and} $13$ subsystems only.
(Note that mixing by the use of only one kind of bipartite entanglement has already been excluded.)
%%%%%%
The next three classes are
\begin{equation}
\mathcal{C}_{2.2.a}=\mathcal{D}_{b|ac,c|ab}\cap\cmpl{\mathcal{D}_{a|bc,c|ab}}\cap\cmpl{\mathcal{D}_{a|bc,b|ac}}
\equiv\mathcal{D}_{b|ac,c|ab}\setminus\bigl(\mathcal{D}_{a|bc,c|ab}\cup\mathcal{D}_{a|bc,b|ac}\bigr).
\end{equation}
For eample, $\mathcal{C}_{2.3.1}$
is the set of states
which can be mixed by the use of bipartite entanglement within both of $12$ and $13$ subsystems together,
but can not be mixed by the use of bipartite entanglement within $23$ subsystem
together with bipartite entanglement within only one of $12$ or $13$ subsystems.
%%%%%%
The next class is
\begin{equation}
\begin{split}
\mathcal{C}_{2.1}&=\cmpl{\mathcal{D}_{2|13,3|12}}\cap\cmpl{\mathcal{D}_{1|23,3|12}}\cap\cmpl{\mathcal{D}_{1|23,2|13}}\cap\mathcal{D}_{1|23,2|13,3|12}\\
&\equiv\mathcal{D}_{1|23,2|13,3|12}\setminus\bigl(\mathcal{D}_{2|13,3|12}\cup\mathcal{D}_{1|23,3|12}\cup \mathcal{D}_{1|23,2|13}\bigr),
\end{split}
\end{equation}
which is the set of states
which can be mixed by the use of bipartite entanglement within all the three bipartite subsystems,
but can not be mixed by the use of bipartite entanglement within only two (or one) bipartite subsystems.
(The union
$\mathcal{C}_{2.4}
\cup\mathcal{C}_{2.3.1}\cup\mathcal{C}_{2.3.2}\cup\mathcal{C}_{2.3.3}
\cup\mathcal{C}_{2.2.1}\cup\mathcal{C}_{2.2.2}\cup\mathcal{C}_{2.2.3}
\cup\mathcal{C}_{2.1}= \mathcal{D}_{1|23,2|13,3|12}\setminus\bigl(\mathcal{D}_{1|23}\cup\mathcal{D}_{2|13}\cup\mathcal{D}_{3|12}\bigr)$
is Class $2.1$ in section \ref{subsec:QM.Ent.NPart}.)

Then comes the class of states containing tripartite entanglement
\begin{equation}
\mathcal{C}_1=\cmpl{\mathcal{D}_{1|23,2|13,3|12}}\cap\mathcal{D}_{123}
\equiv\mathcal{D}_{123}\setminus\mathcal{D}_{1|23,2|13,3|12},
\end{equation}
which is the set of states which can not be mixed
without the use of some tripartite entangled pure states.
(This is Class $1$ in section \ref{subsec:QM.Ent.NPart}.)
\end{subequations}

Except $\mathcal{C}_3$,
the $\mathcal{C}_\text{\dots}$ PS classes above are neither convex nor closed,
but, by construction, they cover $\mathcal{D}$ entirely.
Unfortunately, we can not draw an onion-like figure 
illustrating these classes like the one in figure \ref{fig:3part}, % \cite{SeevinckUffinkMixSep},
(maybe it could be drawn in $3$ dimensions),
we only summarize these $1+18+1$ classes in table~\ref{tab:PS3Classes}.
The non-emptiness of the PS classes above 
is not obvious,
since it depends on the arrangement of different kinds of extremal points.
(We know only that they are not empty \emph{by construction}.)
This issue has not been handled yet,
but experiences in the geometry of mixed states \cite{BengtssonZyczkowski} 
suggest that 
arrangement of different kinds of extremal points 
leading to some empty classes
would be very implausible.
In the next subsection, moreover, we show that some of the new classes are non-empty.

%*******************************************************************************
\subsection{Examples}
\label{subsec:PartSep.Threepart.Examples}

Here we collect some facts about the non-emptiness of the PS classes given in (\ref{eq:PSClasses}).
The classes given by Seevinck and Uffink (section \ref{subsec:QM.Ent.NPart}) are non-empty, 
which are $\mathcal{C}_{3}$,
$\mathcal{C}_{2.8}$,
$\mathcal{C}_{2.7.a}$,
the unions $\mathcal{C}_{2.6.a}\cup\mathcal{C}_{2.5.a}$,
the union $\mathcal{C}_{2.4}
\cup\mathcal{C}_{2.3.1}\cup\mathcal{C}_{2.3.2}\cup\mathcal{C}_{2.3.3}
\cup\mathcal{C}_{2.2.1}\cup\mathcal{C}_{2.2.2}\cup\mathcal{C}_{2.2.3}
\cup\mathcal{C}_{2.1}$,
and $\mathcal{C}_{1}$ \cite{SeevinckUffinkMixSep}.
On the other hand, the pure sets (\ref{eq:PSPclasses}) are contained in the following classes:
$\mathcal{Q}_{1|2|3}\subset\mathcal{C}_3$,
$\mathcal{Q}_{a|bc} \subset\mathcal{C}_{2.5.a}$,
$\mathcal{Q}_{123}  \subset\mathcal{C}_{1}$,
so we have additionally that $\mathcal{C}_{2.5.a}$ is non-empty.
In the next paragraphs,
we construct states contained in classes $\mathcal{C}_{2.2.a}$ and $\mathcal{C}_{2.1}$.
This justifies the use of $b|ac$-$c|ab$-separable sets in the classification,
since we can distinguish between $\mathcal{C}_{2.2.a}$ and $\mathcal{C}_{2.1}$ by the use of these,
although the nonemptiness of
$\mathcal{C}_{2.6.a}$,
$\mathcal{C}_{2.4}$, and 
$\mathcal{C}_{2.3.a}$
has not been shown yet.

From the point of view of ``mixtures of extremal points'', 
it is easy to check that
the bipartite subsystems are separable for states in some PS subsets as follows:
\begin{align*}
\varrho&\in\mathcal{D}_{1|2|3}
&\Longrightarrow&
&\text{$\varrho_{23}$ separable and $\varrho_{13}$ separable and $\varrho_{12}$ separable}\\
\varrho&\in\mathcal{D}_{a|bc}
&\Longrightarrow&
&\text{\phantom{$\varrho_{23}$ separable and } $\varrho_{ac}$ separable and $\varrho_{ab}$ separable}\\
\varrho&\in\mathcal{D}_{b|ac,c|ab}
&\Longrightarrow& 
&\text{$\varrho_{bc}$ separable \phantom{and $\varrho_{13}$ separable and $\varrho_{12}$ separable}}
\end{align*}
Unfortunately, the reverse implications are not true. 
For example, in the case of qubits, for the standard GHZ state (\ref{eq:GHZ}),
all bipartite subsystems are separable although $\cket{\text{GHZ}}\bra{\text{GHZ}}\notin\mathcal{D}_{1|2|3}$.
But the negations of the implications above turn out to be useful:
\begin{align*}
\varrho&\notin\mathcal{D}_{1|2|3}
&\Longleftarrow&
&\text{$\varrho_{23}$ entangled or $\varrho_{13}$ entangled or $\varrho_{12}$ entangled}\\
\varrho&\notin\mathcal{D}_{a|bc}
&\Longleftarrow&
&\text{\phantom{$\varrho_{23}$ entangled or } $\varrho_{ac}$ entangled or $\varrho_{ab}$ entangled}\\
\varrho&\notin\mathcal{D}_{b|ac,c|ab}
&\Longleftarrow& 
&\text{$\varrho_{bc}$ entangled \phantom{or $\varrho_{13}$ entangled or $\varrho_{12}$ entangled}}
\end{align*}
In the following we give examples for three-\emph{qubit} systems.
In this case, the entanglement of two-qubit subsystems can easily be checked
for example using the Peres-Horodecki criterion (\ref{eq:critPeres6}).

Now, take a $\varrho\in\mathcal{D}_{2|13,3|12}$. 
Then $\varrho_{23}$ is always separable,
but if both  $\varrho_{12}$ and $\varrho_{13}$ are entangled, then 
by the above observations we have that
$\varrho\notin\mathcal{D}_{1|23}$,
$\varrho\notin\mathcal{D}_{2|13}$, $\varrho\notin\mathcal{D}_{3|12}$,
moreover,
$\varrho\notin\mathcal{D}_{1|23,2|13}$ and
$\varrho\notin\mathcal{D}_{1|23,3|12}$.
This singles out exactly one class from table~\ref{tab:PS3Classes}, namely $\mathcal{C}_{2.2.1}$.
So if we can mix a state $\varrho$ from $\mathcal{Q}_{1|2|3}$, $\mathcal{Q}_{b|ac}$ and $\mathcal{Q}_{c|ab}$,
for which $\varrho_{ab}$ and $\varrho_{ac}$ are entangled,
then $\varrho\in\mathcal{C}_{2.2.a}$.
For example, such a state is the uniform mixture of projectors to the 
$\cket{0}_b\otimes\cket{\text{B}}_{ac}$ and
$\cket{0}_c\otimes\cket{\text{B}}_{ab}$
vectors, that is,
\begin{equation*}
\frac12 \cket{0}\bra{0}_b\otimes\cket{\text{B}}\bra{\text{B}}_{ac}+ 
\frac12 \cket{0}\bra{0}_c\otimes\cket{\text{B}}\bra{\text{B}}_{ab}
\in\mathcal{C}_{2.2.a}, 
\end{equation*}
where $\cket{\text{B}}$ is the usual Bell state (\ref{eq:B}).

Now, take a $\varrho\in\mathcal{D}_\text{$2$-sep}$. 
Then if the states of all the two-qubit subsystems are entangled, then
by the above observations we have
$\varrho\notin\mathcal{D}_{1|23}$,
$\varrho\notin\mathcal{D}_{2|13}$, $\varrho\notin\mathcal{D}_{3|12}$,
moreover,
$\varrho\notin\mathcal{D}_{2|13,3|12}$, 
$\varrho\notin\mathcal{D}_{1|23,3|12}$ and
$\varrho\notin\mathcal{D}_{1|23,2|13}$.
This singles out exactly one class from table~\ref{tab:PS3Classes}, namely $\mathcal{C}_{2.1}$.
So if we can mix a state $\varrho$ from $\mathcal{Q}_{1|2|3}$, $\mathcal{Q}_{1|23}$, $\mathcal{Q}_{2|13}$ and $\mathcal{Q}_{3|12}$,
whose all two-qubit subsystems are entangled,
then $\varrho\in\mathcal{C}_{2.1}$.
For example, such a state is the uniform mixture of the previous example with
the projector to the vector $\cket{1}_a\otimes\cket{\text{B}}_{bc}$, that is,
\begin{equation*}
 \frac14 \cket{0}\bra{0}_b\otimes\cket{\text{B}}\bra{\text{B}}_{ac} 
+\frac14 \cket{0}\bra{0}_c\otimes\cket{\text{B}}\bra{\text{B}}_{ab}
+\frac12 \cket{1}\bra{1}_a\otimes\cket{\text{B}}\bra{\text{B}}_{bc}
\in\mathcal{C}_{2.1}. 
\end{equation*}

%******************************************************************************
\subsection{Indicator functions for tripartite systems}
\label{subsec:PartSep.ThreePart.Indicators}

Now we give necessary and sufficient conditions for the PS subsets.
%We formulate the same method as in the bipartite case, given in the begining of this section.
To do this, we have to define nonnegative functions on pure states which vanish exactly for the pure states
from which the PS subsets (\ref{eq:PSsets}) are mixed.
Our basic quantities are the local entropies,
which vanish if and only if the given subsystems can be separated from the rest of the system (section \ref{subsec:QM.EntMeas.2Pure}).
This is the only property we need, 
so both R{\'e}nyi (\ref{eq:qRenyi}) and Tsallis~(\ref{eq:qTsallis}) entropies are equally suitable.
(Hence we drop the labels in the superscript, and write only $S_q$.)
Then, after some experimenting, we can define the following set of functions on pure states
\begin{subequations}
\label{eq:PS3PureInd}
\begin{align}
\label{eq:PS3PureInd.1|2|3}
f_{1|2|3}(\pi)          &= f_{1|23}(\pi) + f_{2|13}(\pi) + f_{3|12}(\pi),\\
\label{eq:PS3PureInd.a|bc}
f_{a|bc}(\pi)           &= S_q(\pi_a),\\
f_{b|ac,c|ab}(\pi)      &= f_{b|ac}(\pi)f_{c|ab}(\pi),\\
f_{1|23,2|13,3|12}(\pi) &= f_{1|23}(\pi)f_{2|13}(\pi)f_{3|12}(\pi),
\end{align}
\end{subequations}
which functions vanish for the given pure states
\begin{subequations}
\label{eq:PS3VanishingPure}
\begin{align}
\pi&\in\mathcal{P}_{1|2|3}&
\quad&\Longleftrightarrow&\quad f_{1|2|3}(\pi)&=0,\\
\pi&\in\mathcal{P}_{a|bc}&
\quad&\Longleftrightarrow&\quad f_{a|bc}(\pi)&=0,\\
\pi&\in\mathcal{P}_{b|ca}\cup\mathcal{P}_{c|ab}&
\quad&\Longleftrightarrow&\quad f_{b|ac,c|ab}(\pi)&=0,\\
\pi&\in\mathcal{P}_{1|23}\cup\mathcal{P}_{2|13}\cup\mathcal{P}_{3|12}&
\quad&\Longleftrightarrow&\quad f_{1|23,2|13,3|12}(\pi)&=0,
\end{align}
\end{subequations}
leading to conditions for the pure PS classes
as can be seen in table \ref{tab:PS3Pclasses}.
%%%%%%%%%%%%%%%%%%%%%%%%%%%%%%%%%%%%%%%%
\begin{table}
\begin{tabu}{X[c]||X[c]|X[c]X[c]X[c]|X[c]X[c]X[c]|X[c]}
\hline
\begin{sideways}Class\end{sideways} &
\begin{sideways}$f_{1|2|3}(\pi)$\end{sideways} &
\begin{sideways}$f_{1|23}(\pi)$\end{sideways} &
\begin{sideways}$f_{2|13}(\pi)$\end{sideways} &
\begin{sideways}$f_{3|12}(\pi)$\end{sideways} & 
\begin{sideways}$f_{2|13,3|12}(\pi)$\end{sideways} &
\begin{sideways}$f_{1|23,3|12}(\pi)$\end{sideways} &
\begin{sideways}$f_{1|23,2|13}(\pi)$\end{sideways} &
\begin{sideways}$f_{1|23,2|13,3|12}(\pi)$\end{sideways}  \\
\hline 
\hline
$\mathcal{Q}_{1|2|3}$     & $=0$     & $=0$     & $=0$     & $=0$     & $=0$     & $=0$     & $=0$     & $=0$        \\
\hline
$\mathcal{Q}_{1|23}$      & $>0$     & $=0$     & $>0$     & $>0$     & $>0$     & $=0$     & $=0$     & $=0$        \\
$\mathcal{Q}_{2|13}$      & $>0$     & $>0$     & $=0$     & $>0$     & $=0$     & $>0$     & $=0$     & $=0$        \\
$\mathcal{Q}_{3|12}$      & $>0$     & $>0$     & $>0$     & $=0$     & $=0$     & $=0$     & $>0$     & $=0$        \\
\hline
$\mathcal{Q}_{123}$       & $>0$     & $>0$     & $>0$     & $>0$     & $>0$     & $>0$     & $>0$     & $>0$       \\ 
\hline
\end{tabu}
\bigskip
\caption{Partial separability classes of tripartite pure states
identified by the vanishing of the pure state indicator functions
given in (\ref{eq:PS3PureInd}).}
\label{tab:PS3Pclasses}
\end{table}
%%%%%%%%%%%%%%%%%%%%%%%%%%%%%%%%%%%%%%%%
We call the functions obeying (\ref{eq:PS3VanishingPure}) \emph{pure state indicator functions} for the tripartite case.
We will give the exact definition of indicator functions for the general $n$-partite case later 
(in section~\ref{subsec:PartSep.Gen.Indicators}),
until that point we just use this name 
for nonnegative functions having these vanishing properties.

Now, generalizing (\ref{eq:convroof.discr}), it is easy to prove the following 
necessary and sufficient conditions for the PS subsets (\ref{eq:PSsets})
given by the convex roof extension (\ref{eq:cnvroofext}) of the indicator functions~(\ref{eq:PS3PureInd}):
\begin{subequations}
\label{eq:PS3VanishingMix}
\begin{align}
\varrho&\in\mathcal{D}_{1|2|3}&
\quad&\Longleftrightarrow&\quad \cnvroof{f}_{1|2|3}(\varrho)&=0,\\
\varrho&\in\mathcal{D}_{a|bc}&
\quad&\Longleftrightarrow&\quad \cnvroof{f}_{a|bc}(\varrho)&=0,\\
\varrho&\in\mathcal{D}_{b|ac,c|ab}&
\quad&\Longleftrightarrow&\quad \cnvroof{f}_{b|ac,c|ab}(\varrho)&=0,\\
\varrho&\in\mathcal{D}_{1|23,2|13,3|12}&
\quad&\Longleftrightarrow&\quad \cnvroof{f}_{1|23,2|13,3|12}(\varrho)&=0.
\end{align}
\end{subequations}
To see the \textit{$\Rightarrow$ implications,}
observe that all the $\mathcal{D}_\text{\dots}$ PS subsets are the convex hulls of
such pure states (\ref{eq:PSsets}) for which the given functions vanish 
%[see in table~\ref{tab:PS3Pclasses}].
(\ref{eq:PS3VanishingPure}).
Since these functions can take only nonnegative values,
the minimum in the convex roof extension is zero.
To see the \textit{$\Leftarrow$ implications,}
note that
if the convex roof extension of a nonnegative function vanishes
then there exists a decomposition for pure states for which the function vanishes.
Again, the vanishing of a given function
singles out the pure states (\ref{eq:PS3VanishingPure})
%[see in table~\ref{tab:PS3Pclasses}]
from which the states of the given $\mathcal{D}_\text{\dots}$ PS subset can be mixed (\ref{eq:PSsets}).

The necessary and sufficient conditions for the PS subsets (\ref{eq:PS3VanishingMix})
yields necessary and sufficient conditions for the PS classes,
and we can fill out table~\ref{tab:PS3Mix}
for the identification of the PS classes of table~\ref{tab:PS3Classes}, given for mixed states,
similar to table~\ref{tab:PS3Pclasses}, given for pure states.
Because of their vanishing properties, we call the convex roof extension of pure indicator functions
\emph{mixed indicator functions}.

Note that the convex roof extension is a non-linear operation,
$\cnvroof{(f+f')}\neq \cnvroof{f}+\cnvroof{{f'}}$.
But an inequality holds,
for example,  
$f_{1|2|3} = f_{1|23} + f_{2|13} + f_{3|12}$ and
$\cnvroof{f}_{1|2|3} = \cnvroof{(f_{1|23} + f_{2|13} + f_{3|12})}\geq \cnvroof{f}_{1|23} + \cnvroof{f}_{2|13} + \cnvroof{f}_{3|12}$,
so $\cnvroof{f}_{1|2|3}$ can be non-zero even if $\cnvroof{f}_{1|23}$, $\cnvroof{f}_{2|13}$ and $\cnvroof{f}_{3|12}$ are all zero.
This is why we could identify $20$ classes of mixed states
by the use of the convex roof extension of functions
which identify only $5$ classes of pure states.
On the other hand, if a classification does not involve all the PS subsets,
then, through (\ref{eq:PS3VanishingMix}), 
we have to use only some of the indicator functions,
for example, $f_{1|2|3}$, $f_{a|bc}$ and $f_{1|23,2|13,3|12}$ for the classification obtained by Seevinck and Uffink \cite{SeevinckUffinkMixSep},
or $f_{1|2|3}$ and $f_{a|bc}$ for the classification obtained by D\"ur, Cirac and Tarrach \cite{DurCiracTarrachBMixSep}.
%$y$, $t$ and $\tau^2$ for the classification obtained by Ac\'in, Bru\ss{}, Lewenstein and Sanpera \cite{Acinetal3QBMixClass}.

The structure of the formulas (\ref{eq:PS3PureInd}) give us a hint 
for the generalization for arbitrary number of subsystems of arbitrary dimensions: 
We just have to play a game with 
statements like ``being zero'', with the logical connectives ``and'' and ``or'',
parallel to the addition and multiplication,
and also parallel to the set-theoretical inclusion, union and intersection.
This will be carried out in the next section,
after the construction of the general definitions of PS classification.

%%%%%%%%%%%%%%%%%%%%%%%%%%%%%%%%%%%%%%%%
\begin{table}
\begin{tabu}{X[c]||X[c]|X[c]X[c]X[c]|X[c]X[c]X[c]|X[c]}
\hline
\begin{sideways}PS Class\end{sideways} &
\begin{sideways}$\cnvroof{f}_{1|2|3}(\varrho)$\end{sideways} &
\begin{sideways}$\cnvroof{f}_{1|23}(\varrho)$\end{sideways} &
\begin{sideways}$\cnvroof{f}_{2|13}(\varrho)$\end{sideways} &
\begin{sideways}$\cnvroof{f}_{3|12}(\varrho)$\end{sideways} &
\begin{sideways}$\cnvroof{f}_{2|13,3|12}(\varrho)$\end{sideways} &
\begin{sideways}$\cnvroof{f}_{1|23,3|12}(\varrho)$\end{sideways} &
\begin{sideways}$\cnvroof{f}_{1|23,2|13}(\varrho)$\end{sideways} &
\begin{sideways}$\cnvroof{f}_{1|23,2|13,3|12}(\varrho)$\end{sideways}  \\
\hline
\hline
%                         | y        | s_a        s_b        s_c      | g_a        g_b        g_c      | t   
$\mathcal{C}_3$           & $=0$     & $=0$     & $=0$     & $=0$     & $=0$     & $=0$     & $=0$     & $=0$ \\
\hline
$\mathcal{C}_{2.8}$       & $>0$     & $=0$     & $=0$     & $=0$     & $=0$     & $=0$     & $=0$     & $=0$ \\
$\mathcal{C}_{2.7.a}$     & $>0$     & $>0$     & $=0$     & $=0$     & $=0$     & $=0$     & $=0$     & $=0$ \\
$\mathcal{C}_{2.6.a}$     & $>0$     & $=0$     & $>0$     & $>0$     & $=0$     & $=0$     & $=0$     & $=0$ \\
$\mathcal{C}_{2.5.a}$     & $>0$     & $=0$     & $>0$     & $>0$     & $>0$     & $=0$     & $=0$     & $=0$ \\
$\mathcal{C}_{2.4}$       & $>0$     & $>0$     & $>0$     & $>0$     & $=0$     & $=0$     & $=0$     & $=0$ \\
$\mathcal{C}_{2.3.a}$     & $>0$     & $>0$     & $>0$     & $>0$     & $>0$     & $=0$     & $=0$     & $=0$ \\
$\mathcal{C}_{2.2.a}$     & $>0$     & $>0$     & $>0$     & $>0$     & $=0$     & $>0$     & $>0$     & $=0$ \\
$\mathcal{C}_{2.1}$       & $>0$     & $>0$     & $>0$     & $>0$     & $>0$     & $>0$     & $>0$     & $=0$ \\
\hline
$\mathcal{C}_1$           & $>0$     & $>0$     & $>0$     & $>0$     & $>0$     & $>0$     & $>0$     & $>0$ \\
\hline
\end{tabu}
\bigskip
\caption{PS classes of tripartite mixed states given in table~\ref{tab:PS3Classes},
identified by the vanishing of the mixed indicator functions
(convex roof extension of the indicator functions (\ref{eq:PS3PureInd})).}
\label{tab:PS3Mix}
\end{table}
%%%%%%%%%%%%%%%%%%%%%%%%%%%%%%%%%%%%%%%%

%******************************************************************************
%******************************************************************************
\section{Generalizations: Partial separability of multipartite systems}
\label{sec:PartSep.Gen}

In the previous section,
we have followed a didactic treatment
in order to illustrate the main concepts.
Now it is high time to turn to abstract definitions
to handle the PS classification and criteria
for an arbitrary number of subsystems.

For $n$-partite systems, the set of the labels of the singlepartite subsystems is $L=\{1,2,\dots,n\}$.
Let $\alpha=L_1|L_2|\dots|L_k$ denote a $k$-partite split,
that is a partition of the labels
into $k$ disjoint non-empty sets $L_r$,
where $L_1\cup L_2\cup\dots\cup L_k=L$.
For two partitions $\beta$ and $\alpha$,
$\beta$ is contained%
%%%%%%%%%%%%%%%%%%%%%%%%
\footnote{Instead of ``$\beta$ is contained in $\alpha$'',
it is sometimes said that ``$\beta$ is finer than $\alpha$'', or equivalently,
``$\alpha$ is coarser than $\beta$''.}
%%%%%%%%%%%%%%%%%%%%%%%%
 in $\alpha$,
denoted as $\beta\preceq\alpha$,
if $\alpha$ can be obtained from $\beta$ by joining some (maybe neither) of the parts of $\beta$.
This defines a partial order on the partitions.
(It is easy to see from the definition that 
$\alpha\preceq\alpha$ (reflexivity),
if $\gamma\preceq\beta$ and $\beta\preceq\alpha$ then $\gamma\preceq\alpha$ (transitivity),
if $\beta\preceq\alpha$ and $\alpha\preceq\beta$ then $\alpha=\beta$ (antisymmetry).)
For example, for the tripartite case $1|2|3\preceq a|bc\preceq 123$.
Since 
there is a greatest and a smallest element,
which are the full $n$-partite split and the trivial partition without split, respectively,
$1|2|\dots|n\preceq\alpha\preceq12\dots n$,
the set of partitions of $L$ for $\preceq$ forms a bounded lattice.

%******************************************************************************
\subsection{PS subsets in general}
\label{subsec:PartSep.Gen.PSsubsets}

The first point is the generalization of the PS subsets $\mathcal{D}_\text{\dots}$.
Let $\mathcal{Q}_\alpha\equiv\mathcal{Q}_\alpha(\mathcal{H})\subseteq\mathcal{P}(\mathcal{H})$ be the \emph{pure PS class of $\alpha$-separable states},
that is, the set of pure states which are separable under the partition $\alpha=L_1|L_2|\dots|L_k$
but not separable under any $\beta\prec\alpha$.
Then the \emph{pure PS subset of $\alpha$-separable states} is
\begin{subequations}
\begin{equation}
%\label{eq:genDsets.alpha}
\mathcal{P}_\alpha = \bigcup_{\beta\preceq\alpha}\mathcal{Q}_\beta,
\end{equation}
which is a special case %($l=1$) 
of the  \emph{pure PS subset of $\vs{\alpha}$-separable states}
\begin{equation}
%\label{eq:genDsets.alpha}
\mathcal{P}_{\vs{\alpha}} = \bigcup_{\alpha\in\vs{\alpha}} \mathcal{P}_\alpha
\equiv \bigcup_{\alpha\in\vs{\alpha}} \bigcup_{\beta\preceq\alpha}\mathcal{Q}_\beta,
\end{equation}
\end{subequations}
with the \emph{label} $\vs{\alpha}$ being an arbitrary \emph{set} of partitions.
%(In the writing we omit the $\{\dots\}$ set-brackets.)
Then the \emph{PS subset of $\alpha$-separable states} is
\begin{subequations}
\label{eq:genDsets}
\begin{equation}
\label{eq:genDsets.alpha}
\mathcal{D}_\alpha = \Conv \mathcal{P}_\alpha 
\equiv \Conv \bigcup_{\beta\preceq\alpha}\mathcal{Q}_\beta,
\end{equation}
which is a special case %($l=1$) 
of the \emph{PS subsets of $\vs{\alpha}$-separable states}
\begin{equation}
\label{eq:genDsets.alphal}
\mathcal{D}_{\vs{\alpha}} = \Conv \mathcal{P}_{\vs{\alpha}}
\equiv \Conv \bigcup_{\alpha\in\vs{\alpha}} \mathcal{P}_\alpha
\equiv \Conv \bigcup_{\alpha\in\vs{\alpha}} \bigcup_{\beta\preceq\alpha}\mathcal{Q}_\beta
\equiv \Conv \bigcup_{\alpha\in\vs{\alpha}} \mathcal{D}_{\alpha},
\end{equation}
\end{subequations}
with the \emph{label} $\vs{\alpha}$.
(In the writing we omit the $\{\dots\}$ set-brackets, as was seen, e.g., in (\ref{eq:PSsets.bcacab}).)
The set of $k$-separable states $\mathcal{D}_\text{$k$-sep}$ arises as a special case
where the $\alpha$ elements of $\vs{\alpha}$ are all the possible $k$-partite splits.
Note that in general the $\alpha\in\vs{\alpha}$ partitions are not required to be $k$-partite splits for the same $k$.
This freedom can not be seen in the case of three subsystems.

The $\mathcal{Q}_\beta$ sets are not closed if and only if $\beta$ is not the full $n$-partite split $1|2|\dots|n$,
but $\mathcal{P}_\alpha=\cup_{\beta\preceq\alpha}\mathcal{Q}_\beta$ is closed,
so the sets $\mathcal{D}_{\vs{\alpha}}$ are closed, and convex by construction.
Note that different $\vs{\alpha}$ labels can give rise to the same $\mathcal{D}_{\vs{\alpha}}$ sets,
in other words,
the $\vs{\alpha}\mapsto \mathcal{D}_{\vs{\alpha}}$ ``labelling map'' defined by (\ref{eq:genDsets.alphal})
is surjective but not injective.
For the full PS classification we need all the possible \emph{different} $\mathcal{D}_{\vs{\alpha}}$ sets.
%So first we work out a general framework for the labellin 
Because of the non-trivial structure of the lattice of partitions,
obtaining all the different PS sets is also a non-trivial task.
We can not provide a closed formula for that, but only an algorithm.
Before we do this, we need some constructions.

First, observe that
if $\beta\preceq\alpha$ then $\mathcal{D}_\beta\subseteq\mathcal{D}_\alpha$,
(from definition (\ref{eq:genDsets.alpha}), and the transitivity of $\preceq$)
from which it follows that 
for the labels $\vs{\beta}$ and $\vs{\alpha}$,
if for every $\beta\in\vs{\beta}$ there is an $\alpha\in\vs{\alpha}$ for which $\beta\preceq\alpha$ 
then $\mathcal{D}_{\vs{\beta}}\subseteq\mathcal{D}_{\vs{\alpha}}$.
(From definition (\ref{eq:genDsets.alphal}). Later we will prove the reverse too.)
These observations motivate the extension of %the partial order 
$\preceq$ from the partitions to the labels as
\begin{equation}
\label{eq:labelreldef}
\vs{\beta}\preceq\vs{\alpha}
\qquad\defn\qquad
\forall \beta\in\vs{\beta}, \exists \alpha\in\vs{\alpha}:\; \beta\preceq\alpha.
\end{equation}
Note that, at this point,
the relation $\preceq$ on the labels is not a partial order,
only the reflexivity and the transitivity properties hold.
The antisymmetry property fails, which is the consequence of
that the definition (\ref{eq:labelreldef}) was motivated by the inclusion of the PS sets,
and different $\vs{\alpha}$s can lead to the same PS set.
%rather than....
Independently of this problem, which will be handled later, the following is true.
\begin{prop} For $\vs{\alpha},\vs{\beta}$ labels
\begin{equation}
\label{eq:relationEquiv}
\vs{\beta}\preceq\vs{\alpha} \qquad\Longleftrightarrow\qquad
\mathcal{D}_{\vs{\beta}}\subseteq\mathcal{D}_{\vs{\alpha}}.
\end{equation}
\end{prop}

\begin{proof}
This can be shown in the following steps:
\begin{equation*}
\begin{split}
\mathcal{D}_{\vs{\beta}}\subseteq\mathcal{D}_{\vs{\alpha}}
\qquad&\overset{\text{(i)}}{\Longleftrightarrow}\qquad
          \Conv \bigcup_{\beta\in\vs{\beta}}   \bigcup_{\delta\preceq\beta }\mathcal{Q}_\delta
\subseteq \Conv \bigcup_{\alpha\in\vs{\alpha}} \bigcup_{\gamma\preceq\alpha}\mathcal{Q}_\gamma\\
&\overset{\text{(ii)}}{\Longleftrightarrow}\qquad
          \bigcup_{\beta\in\vs{\beta}}   \bigcup_{\delta\preceq\beta }\mathcal{Q}_\delta
\subseteq \bigcup_{\alpha\in\vs{\alpha}} \bigcup_{\gamma\preceq\alpha}\mathcal{Q}_\gamma\\
&\overset{\text{(iii)}}{\Longleftrightarrow}\qquad
\forall \beta\in\vs{\beta}, \forall \delta\preceq\beta, \exists \alpha\in\vs{\alpha}: \delta\preceq\alpha\\
&\overset{\text{(iv)}}{\Longleftrightarrow}\qquad
\forall \beta\in\vs{\beta}, \exists \alpha\in\vs{\alpha}:\; \beta\preceq\alpha\\
&\overset{\text{(v)}}{\Longleftrightarrow}\qquad
\vs{\beta}\preceq\vs{\alpha}.
\end{split}
\end{equation*}

The \textit{equivalences (i) and (v)} are by definition (\ref{eq:genDsets.alphal}) and~(\ref{eq:labelreldef}), respectively.

\textit{Equivalence (ii)} is the only one where 
it comes into the picture that the story is about quantum states.
The $\overset{\text{(ii)}}{\Longleftarrow}$ implication holds, since it is true in general that
$\Conv B \subseteq \Conv A \Leftarrow B \subseteq A$.
But to the $\overset{\text{(ii)}}{\Longrightarrow}$ implication 
we have to use some special properties coming from geometry.
Obviously, for the extremal points
\begin{equation*}
     \Extr\Conv \bigcup_{\beta\in\vs{\beta}}   \bigcup_{\delta\preceq\beta}\mathcal{Q}_\delta
\subseteq \Conv \bigcup_{\beta\in\vs{\beta}}   \bigcup_{\delta\preceq\beta}\mathcal{Q}_\delta
\subseteq \Conv \bigcup_{\alpha\in\vs{\alpha}} \bigcup_{\gamma\preceq\alpha }\mathcal{Q}_\gamma,
\end{equation*}
so $\pi\in\Extr\Conv\bigcup_{\beta\in\vs{\beta}} \bigcup_{\delta\preceq\beta}\mathcal{Q}_\delta$
is also the element of $\Conv \bigcup_{\alpha\in\vs{\alpha}} \bigcup_{\gamma\preceq\alpha  }\mathcal{Q}_\gamma$.
But $\pi$ is a projector of rank one,
so it is extremal also in $\Conv \bigcup_{\alpha\in\vs{\alpha}} \bigcup_{\gamma\preceq\alpha  }\mathcal{Q}_\gamma$.
This holds for all such $\pi$,
so we have that
\begin{equation*}
          \Extr\Conv\bigcup_{\beta\in\vs{\beta}}   \bigcup_{\delta\preceq\beta}\mathcal{Q}_\delta
\subseteq \Extr\Conv\bigcup_{\alpha\in\vs{\alpha}} \bigcup_{\gamma\preceq\alpha }\mathcal{Q}_\gamma.
\end{equation*}
Any $A$ sets of projectors of rank one 
have the property that $A=\Extr\Conv A$,
that is, they are all extremal points of their convex hulls,
which leads to the $\overset{\text{(ii)}}{\Longrightarrow}$ implication.
%so $\Extr\Conv\bigcup_{j'=1}^{l'} \bigcup_{\delta\preceq\beta_{j'}}\mathcal{Q}_\delta = \bigcup_{j'=1}^{l'} \bigcup_{\delta\preceq\beta_{j'}}\mathcal{Q}_\delta$
%and $\Extr\Conv\bigcup_{j=1}^l     \bigcup_{\gamma\preceq\alpha_j  }\mathcal{Q}_\gamma = \bigcup_{\gamma\preceq\alpha_j  }\mathcal{Q}_\gamma$.

\textit{Equivalence (iii)} comes from set algebra.
To see the $\overset{\text{(iii)}}{\Longrightarrow}$ implication,
we note that the $\mathcal{Q}_{\dots}$ sets are disjoint,
so every $\mathcal{Q}_\delta$ on the left-hand side of the inclusion
appears on the right-hand side as a $\mathcal{Q}_\gamma$,
which means that $\forall \beta\in\vs{\beta}$, $\forall \delta\preceq\beta$, $\exists \alpha\in\vs{\alpha}$ so that $\delta\preceq\alpha$.
To see the $\overset{\text{(iii)}}{\Longleftarrow}$ implication,
from the condition $\exists\alpha$ so that $\mathcal{Q}_\delta\subseteq \bigcup_{\gamma\preceq\alpha  }\mathcal{Q}_\gamma$,
but for different $\delta\preceq\beta$s there may exist different $\alpha$s,
so $\bigcup_{\delta\preceq\beta}\mathcal{Q}_\delta\subseteq \bigcup_{\alpha\in\vs{\alpha}}\bigcup_{\gamma\preceq\alpha  }\mathcal{Q}_\gamma$,
which holds for all $\beta\in\vs{\beta}$.

\textit{Equivalence (iv)} comes from the properties of partial ordering.
First, $\preceq$ is reflexive on partitions, that is, $\beta\preceq\beta$,
so the $\overset{\text{(iv)}}{\Longrightarrow}$ implication follows from the $\delta=\beta$ choice.
On the other hand,
$\preceq$ is transitive on partitions, which is just the $\overset{\text{(iv)}}{\Longleftarrow}$ implication:
for all $\delta$,
if $\delta\preceq\beta$ and $\beta\preceq\alpha$ then $\delta\preceq\alpha$.
\end{proof}

Again, note that the relations $\preceq$ and $\subseteq$
are defined on non-isomorphic sets,
so (\ref{eq:relationEquiv}) does not contradict the fact
that the latter is a partial order while the former is not.

The next step is to define 
those labels for which $\preceq$ is a partial order.
A label $\vs{\alpha}$ is called \emph{proper label}, if
\begin{equation}
\label{eq:properlabel}
%\forall \alpha,\alpha'\in\vs{\alpha}:\qquad
%\text{if $\alpha\neq\alpha'$ then $\alpha\npreceq\alpha'$}.
\forall \alpha,\alpha'\in\vs{\alpha},\; \alpha\neq\alpha'\qquad\Longrightarrow\qquad \alpha\npreceq\alpha'.
\end{equation}
With this, we have the following.
\begin{prop}
On the set of proper labels, the relation $\preceq$ defined in (\ref{eq:labelreldef})
is a partial order.
\end{prop}
\begin{proof}
%We show that the relation (\ref{eq:labelreldef})
%is a partial order on the set of proper labels (\ref{eq:properlabel}).
\textit{Reflexivity on labels:}
We need that $\vs{\alpha}\preceq\vs{\alpha}$, 
so by definition (\ref{eq:labelreldef})
$\forall\alpha\in\vs{\alpha}$ $\exists \alpha'\in\vs{\alpha}$
for which $\alpha\preceq\alpha'$.
This holds for the $\alpha'=\alpha$ choice,
since $\preceq$ is reflexive on partitions, $\alpha\preceq\alpha$.

\textit{Transitivity on labels:}
Suppose that $\vs{\beta}\preceq\vs{\alpha}$ and $\vs{\gamma}\preceq\vs{\beta}$,
so by definition (\ref{eq:labelreldef})
$\forall\gamma\in\vs{\gamma}$ $\exists \beta\in\vs{\beta}$
for which $\gamma\preceq\beta$,
and
for this $\beta$ $\exists \alpha\in\vs{\alpha}$
for which $\beta\preceq\alpha$.
Since $\preceq$ is transitive on partitions, we have that
$\forall\gamma\in\vs{\gamma}$ $\exists \alpha\in\vs{\alpha}$
for which $\gamma\preceq\alpha$,
which is $\vs{\gamma}\preceq\vs{\alpha}$  by definition (\ref{eq:labelreldef}).

\textit{Antisymmetry on proper labels:}
Let $\vs{\alpha}$ and $\vs{\beta}$ be proper labels.
Suppose that $\vs{\beta}\preceq\vs{\alpha}$ and $\vs{\alpha}\preceq\vs{\beta}$,
so by definition (\ref{eq:labelreldef})
$\forall\beta\in\vs{\beta}$ $\exists \alpha\in\vs{\alpha}$
for which $\beta\preceq\alpha$,
and
for this $\alpha$ $\exists \beta'\in\vs{\beta}$
for which $\alpha\preceq\beta'$.
Since $\preceq$ is transitive on partitions, we have that
$\beta\preceq\beta'$.
This can be true only if $\beta=\beta'$,
since $\vs{\beta}$ is a proper label,
so we have that $\beta\preceq\alpha$ and $\alpha\preceq\beta$.
Since $\preceq$ is antisymmetric on partitions, we have that
$\forall\beta\in\vs{\beta}$, $\exists \alpha\in\vs{\alpha}$
for which $\alpha=\beta$, which means that $\vs{\beta}\subseteq\vs{\alpha}$.
(The labels are \emph{sets}.)
Interchanging the roles of $\vs{\alpha}$ and $\vs{\beta}$,
we have that $\vs{\alpha}\subseteq\vs{\beta}$.
Since $\subseteq$ is antisymmetric on sets, we have that
$\vs{\beta}=\vs{\alpha}$.
\end{proof}

A corollary is that
the set of proper labels for $\preceq$ forms a bounded lattice,
its greatest and smallest elements are the one-element labels 
of full $n$-partite split and the trivial partition without split, respectively,
$1|2|\dots|n\preceq\vs{\alpha}\preceq12\dots n$.

Is it true that every PS subset can be labelled by proper label?
And different proper labels lead to different PS subsets?
In other words,
is the $\vs{\alpha}\mapsto \mathcal{D}_{\vs{\alpha}}$ ``labelling map''
from the \emph{set of proper labels} to the set of PS subsets
an isomorphism?
The injectivity is the $\Leftarrow$ implication from the next observation.
\begin{prop}
For $\vs{\alpha}$, $\vs{\beta}$ proper labels
\begin{equation}
\label{eq:propInj}
\vs{\beta}=\vs{\alpha}
\qquad\Longleftrightarrow\qquad
\mathcal{D}_{\vs{\beta}}=\mathcal{D}_{\vs{\alpha}}.
\end{equation}
%For the proof, see appendix~\ref{appsubsec:PartSep.Gen.propInj}.
%To see the injectivity,
\end{prop}
\begin{proof}
Let $\vs{\alpha}$ and $\vs{\beta}$ be proper labels.
The above can be shown in the following steps:
\begin{equation*}
\begin{split}
\mathcal{D}_{\vs{\beta}}=\mathcal{D}_{\vs{\alpha}}
\qquad&\overset{\text{(i)}}{\Longleftrightarrow}\qquad
\mathcal{D}_{\vs{\beta}}\subseteq\mathcal{D}_{\vs{\alpha}}
\;\;\text{and}\;\;
\mathcal{D}_{\vs{\alpha}}\subseteq\mathcal{D}_{\vs{\beta}}\\
&\overset{\text{(ii)}}{\Longleftrightarrow}\qquad
\vs{\beta}\preceq\vs{\alpha}
\;\;\text{and}\;\;
\vs{\alpha}\preceq\vs{\beta}\\
&\overset{\text{(iii)}}{\Longleftrightarrow}\qquad
\vs{\beta}=\vs{\alpha}.
\end{split}
\end{equation*}
\textit{Equivalence (i)} is the antisymmetry of $\subseteq$ on sets,
\textit{equivalence (ii)} is (\ref{eq:relationEquiv}) on labels,
\textit{equivalence (iii)} is the antisymmetry of $\preceq$ on proper labels.
\end{proof}
If $\vs{\beta}$ is a label, 
then we can obtain a unique proper label from that
if we drop every $\beta\in\vs{\beta}$ for which there is a $\beta'\in\vs{\beta}$ for which $\beta\preceq\beta'$.
The remaining partitions form a proper label which we denote $\vs{\alpha}$,
and the partitions which have been dropped out form a label which we denote $\vs{\gamma}$.
Then $\vs{\beta}=\vs{\alpha}\vs{\gamma}$, which means the union of labels $\vs{\alpha}$ and $\vs{\gamma}$.
(We omit the union sign too.)
Our next observation is useful for this case.
\begin{prop}
For the general labels $\vs{\alpha}$ and $\vs{\gamma}$,
\begin{equation}
\label{eq:propPart}
\vs{\gamma}\preceq\vs{\alpha}
\qquad\Longleftrightarrow\qquad
\mathcal{D}_{\vs{\alpha}\vs{\gamma}}=\mathcal{D}_{\vs{\alpha}},
\end{equation}
\end{prop}
\begin{proof}
This can be shown in the following steps:
\begin{equation*}
\begin{split}
\mathcal{D}_{\vs{\alpha}\vs{\gamma}}=\mathcal{D}_{\vs{\alpha}}
\qquad&\overset{\text{(i)}}{\Longleftrightarrow}\qquad
\mathcal{D}_{\vs{\alpha}\vs{\gamma}}\subseteq\mathcal{D}_{\vs{\alpha}}
\;\;\text{and}\;\;
\mathcal{D}_{\vs{\alpha}}\subseteq\mathcal{D}_{\vs{\alpha}\vs{\gamma}}\\
&\overset{\text{(ii)}}{\Longleftrightarrow}\qquad
\vs{\alpha}\vs{\gamma}\preceq\vs{\alpha}
\;\;\text{and}\;\;
\vs{\alpha}\preceq\vs{\alpha}\vs{\gamma}\\
&\overset{\text{(iii)}}{\Longleftrightarrow}\qquad
\vs{\gamma}\preceq\vs{\alpha}.
\end{split}
\end{equation*}
\textit{Equivalence (i)} is the antisymmetry of $\subseteq$ on sets,
\textit{equivalence (ii)} is (\ref{eq:relationEquiv}) on labels,
($\vs{\alpha}\preceq\vs{\alpha}\vs{\gamma}$ holds always)
\textit{equivalence (iii)} is from the observation that
$\vs{\beta}\preceq\vs{\alpha}$ and $\vs{\beta}'\preceq\vs{\alpha}$
if and only if $\vs{\beta}\vs{\beta}'\preceq\vs{\alpha}$, which can easily be seen from the definition (\ref{eq:labelreldef}).
\end{proof}
This means that when we obtain a proper label $\vs{\alpha}$
from a general label $\vs{\beta}$,
as was done above,
both of these lead to the same PS subset.
Since all PS subsets arise from general labels,
the above shows that they arise also from proper labels,
which is the surjectivity of the labelling by proper labels.

Now we have that the set of proper labels is isomorphic to the set of PS subsets.
The former one is much easier to handle.
Moreover, (\ref{eq:relationEquiv}) states now that
the lattice of $\vs{\alpha}$ proper labels with respect to the partial order $\preceq$
is isomorphic to
the lattice of $\mathcal{D}_{\vs{\alpha}}$ PS subsets with respect to the partial order $\subseteq$.
(This lattice is the generalization of the ``inclusion hierarchy'' in figure~\ref{fig:PS3incl}.)
To get all the PS subsets, we have to obtain all the proper labels.
A brute force method for this
is to form all the $\vs{\beta}$ labels (all the subsets of the set of all partitions),
then obtain the proper labels $\vs{\alpha}$ as before ($\vs{\beta}=\vs{\alpha}\vs{\gamma}$),
and keep the different proper labels obtained in this way.
This algorithm is very ineffective, %resource-wasting,
we give a much more optimized one.

%%%%%%%%%%%%%%%%%%%%%%%%%%%%%%%%%%%%%%%%
\begin{figure}
\centering
 \includegraphics{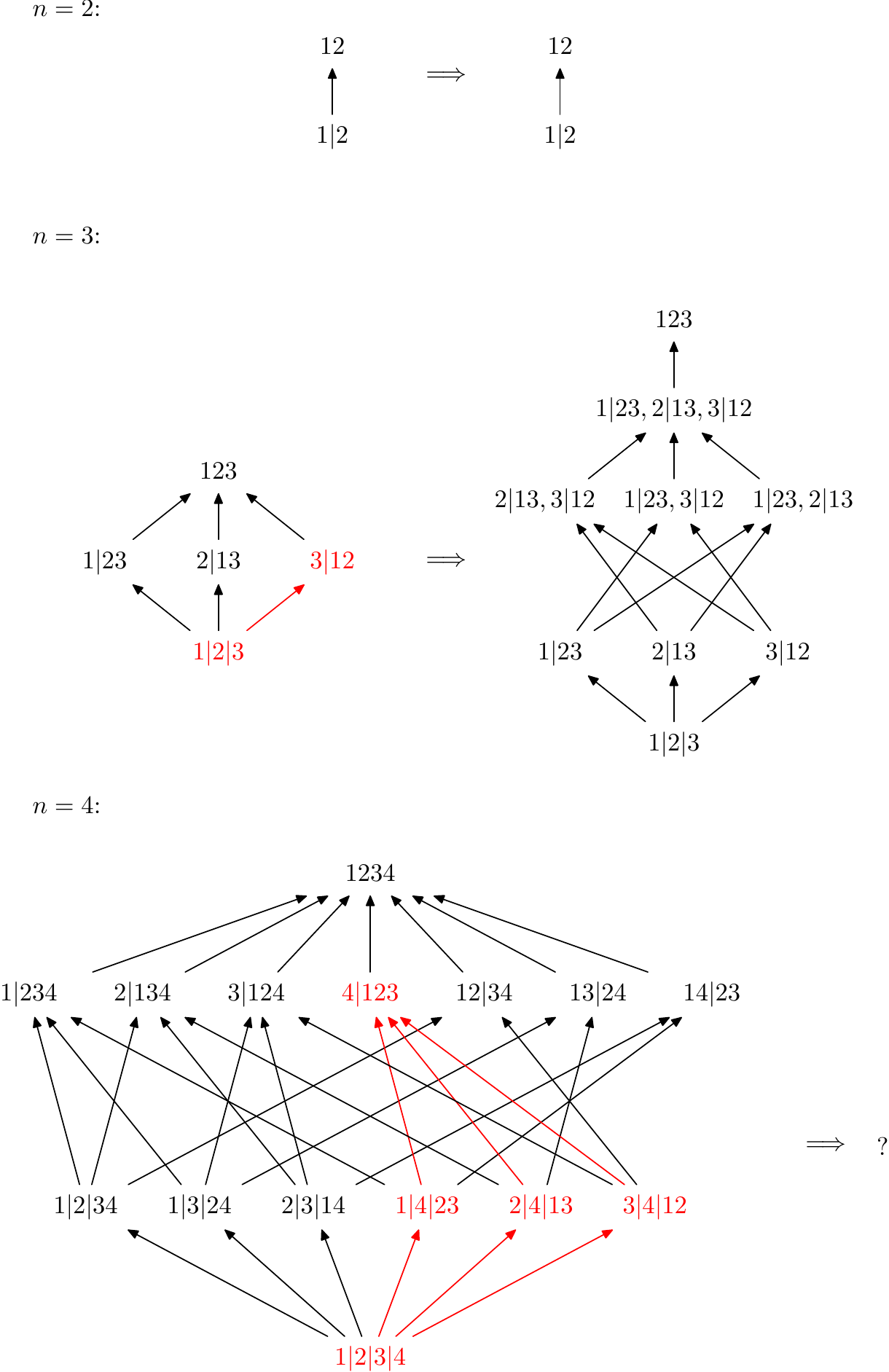}
 \caption{Lattices of partitions and proper labels in the case of two, three and four subsystems.}
\label{fig:lattices}
\end{figure}
%%%%%%%%%%%%%%%%%%%%%%%%%%%%%%%%%%%%%%%%

To obtain an efficient algorithm generating the proper labels of
all PS subsets,
it is necessary to consider the labels as \emph{$l$-tuples of partitions} instead of \emph{sets of partitions.}
In this case $\vs{\alpha}=\alpha_1,\alpha_2,\dots,\alpha_l$,
so the order of the elements is considered to be fixed when an $l$-tuple is given,
and the $\alpha_j$s are different for different $j$s.
(The $(\dots)$ $l$-tuple-brackets are omitted.
Note that, contrary to the notation used in \cite{SeevinckUffinkMixSep}, 
the lower index of the partitions $\alpha_j$ here does not refer to the $k$ number of $L_r$ sets in $\alpha_j$.)
Using ordered structure has further advantages beyond the obvious one that
a computer stores data sequentially, so implementing sets would mean additional difficulty.
Now the algorithm is the following.

\begin{alg}
~{}
\begin{enumerate}
\item{} [\textit{initialization}] 
Fix an order of the partitions,
this defines a lexicographical ordering for $l$-tuples of partitions.
This is denoted by $<$.
(This is to avoid obtaining an $l$-tuple more than once
and to make the algorithm more optimized.)
\item{} [\textit{level $1$}] Using this ordering,
we have all the $1$-tuples of partitions ordered lexicographically.
\item{} [\textit{induction step:} obtaining the $l+1$-tuples of partitions (level $l+1$)
from the $l$-tuples of partitions (level $l$)]
To every $\vs{\alpha}=\alpha_1,\alpha_2,\dots,\alpha_l$ $l$-tuples 
(coming in lexicographically ordered sequence)
we have to append 
any such partition $\alpha_{l+1}$
(coming in lexicographically ordered sequence) that\\
(i) $\alpha_{l+1}\npreceq \alpha_j$ and $\alpha_j\npreceq\alpha_{l+1}$ for all $j=1,2,\dots,l$, and\\
(ii) $\alpha_{l+1}>\alpha_l$. 
(Because of the lexicographical order $<$, it is enough to consider only the last ($l$th) partition.) \\
Then the resulting $\vs{\alpha}=\alpha_1,\alpha_2,\dots,\alpha_l,\alpha_{l+1}$ $l+1$-tuples, 
and also the partitions in every such $l+1$-tuple are ordered lexicographically.
\end{enumerate}
The algorithm stops when no new partition can be appended to any of the $l$-tuples,
which comes in finite steps, since the number of all the partitions is finite.
\end{alg}

This algorithm generates the lattice of proper labels from the lattice of partitions,
see in figure \ref{fig:lattices}.
We note, however, that the number of proper labels is so high even for four subsystems
(more than a hundred)
that the lattice of proper labels can be handled only in computer.

%******************************************************************************
\subsection{PS classes in general}
\label{subsec:PartSep.Gen.PSclasses}

The second point is the generalization of the PS classes $\mathcal{C}_{\dots}$,
which are the possible non-trivial intersections of the PS subsets $\mathcal{D}_{\dots}$.
Constructing these needs direct calculations for a given $n$, as was done in section~\ref{subsec:PartSep.Threepart.PSclasses}.

Let us divide the set of proper labels into two disjoint subsets, $\vvs{\alpha}$ and $\vvs{\beta}$,
then all the possible intersections of PS subsets can be labelled by such a pair,
which is called \emph{class-label}, as
% lehetne altalanosabban: 
% tetszoleges setek, aztan kulonbozok adnak azonos classokat, aztan mi lesz
\begin{equation}
\mathcal{C}_{\vvs{\alpha},\vvs{\beta}}=
\bigcap_{\vs{\alpha}\in\vvs{\alpha}} \cmpl{\mathcal{D}_{\vs{\alpha}}}\cap \bigcap_{\vs{\beta}\in\vvs{\beta}}\mathcal{D}_{\vs{\beta}}.
\end{equation}
It can happen that $\mathcal{C}_{\vvs{\alpha},\vvs{\beta}}=\emptyset$ \emph{by construction,}
by that we mean that its emptiness follows from the inclusion hierarchy of PS subsets.
For example, if $\mathcal{D}_{\vs{\beta}}\subseteq\mathcal{D}_{\vs{\alpha}}$ 
for some $\vs{\alpha}\in\vvs{\alpha}$ and $\vs{\beta}\in\vvs{\beta}$, then the intersection above is identically empty.
The PS classes for three subsystems in subsection~\ref{subsec:PartSep.Threepart.PSclasses} 
were obtained by the use of this observation.
In this general framework, this observation is formulated as follows:
\begin{equation*}
\begin{split}
\mathcal{C}_{\vvs{\alpha},\vvs{\beta}}=\emptyset
\qquad&\overset{\text{(i)}}{\Longleftrightarrow}\qquad
\cmpl{\bigcup_{\vs{\alpha}\in\vvs{\alpha}} \mathcal{D}_{\vs{\alpha}} } \cap \bigcap_{\vs{\beta}\in\vvs{\beta}}\mathcal{D}_{\vs{\beta}}=\emptyset\\
\quad&\overset{\text{(ii)}}{\Longleftrightarrow}\qquad
\bigcap_{\vs{\beta}\in\vvs{\beta}}\mathcal{D}_{\vs{\beta}} \subseteq \bigcup_{\vs{\alpha}\in\vvs{\alpha}} \mathcal{D}_{\vs{\alpha}}\\
\quad&\overset{\text{(iii)}}{\Longleftarrow}\qquad
\exists \vs{\alpha}\in\vvs{\alpha}, \exists\vs{\beta}\in\vvs{\beta}\;:\;\mathcal{D}_{\vs{\beta}}\subseteq\mathcal{D}_{\vs{\alpha}}\\
\quad&\overset{\text{(iv)}}{\Longleftrightarrow}\qquad
\exists \vs{\alpha}\in\vvs{\alpha}, \exists\vs{\beta}\in\vvs{\beta}\;:\; \vs{\beta}\preceq \vs{\alpha}.
\end{split}
\end{equation*}
\begin{proof}
\textit{Equivalence (i)} comes from De Morgan's law $\cmpl{A}\cap\cmpl{B}=\cmpl{A\cup B}$.
\textit{Equivalence (ii)} comes from the identity $B\subseteq A\Leftrightarrow (B\cap\cmpl{A}\equiv B\setminus A=\emptyset)$.
\textit{Implication (iii)} comes from $B\subseteq A\Rightarrow (B\cap B'\subseteq A\cup A')$.
%$\bigcap_{\vs{\beta}\in\vvs{\beta}}\mathcal{D}_{\vs{\beta}} 
%\subseteq \mathcal{D}_{\vs{\beta}}
%\subseteq \mathcal{D}_{\vs{\alpha}}
%\subseteq \bigcup_{\vs{\alpha}\in\vvs{\alpha}} \mathcal{D}_{\vs{\alpha}}$.
\textit{Equivalence (iv)} is (\ref{eq:relationEquiv}).
\end{proof}

Implication (iii) is the point which makes possible 
to formulate the emptiness of PS classes \emph{purely by the use of labels}. 
That is still a question
whether implication (iii) can be replaced by a stronger one,
which leads to condition involving only labels again.
(The problem is that we have no interpretations of $\cap$ and $\cup$ in the language of labels.)
\emph{Our first conjecture is that 
implication (iii) above is the strongest one which leads to a condition involving only labels.}

Summarizing, we have
\begin{subequations}
\begin{equation}
\label{eq:emptyByConstruction}
\mathcal{C}_{\vvs{\alpha},\vvs{\beta}}=\emptyset 
\qquad\Longleftarrow\qquad
\exists \vs{\alpha}\in\vvs{\alpha}, \exists\vs{\beta}\in\vvs{\beta}\;:\; \vs{\beta}\preceq \vs{\alpha}.
\end{equation}
If the right-hand side holds then we say, according to the conjecture above,
that $\mathcal{C}_{\vvs{\alpha},\vvs{\beta}}$ is \emph{empty by construction}.
Since this implication is only one-way,
it could happen that $\mathcal{C}_{\vvs{\alpha},\vvs{\beta}}=\emptyset$
for such class-label $\vvs{\alpha},\vvs{\beta}$ for which the right-hand side does not hold,
resulting in a class which is empty but not by construction.
But we think that this can not happen:
\emph{our second conjecture is that there is an equivalence in (\ref{eq:emptyByConstruction}),}
that is, all the PS classes which are not empty by construction are non-empty.%
%%%%%%%%%%%%%%%%%%%%%%%%
\footnote{
This implies the first conjecture above,
but it can still happen that implication (iii) can be replaced by a stronger condition,
so the first conjecture is false.
Then the (\ref{eq:emptyByConstruction}) definition of the emptiness-by-construction changes,
and the second conjecture concerns this new definition.}
%%%%%%%%%%%%%%%%%%%%%%%%
The motivation of this is the same as
in the tripartite case (see at the end of subsection~\ref{subsec:PartSep.Threepart.PSclasses}),
where the PS classes conjectured to be non-empty
was obtained under the same assumptions.

An advantage of the formulation by the labelling constructions
is, roughly speaking, that by the use of that
``we have separated the \emph{algebraic} and the \emph{geometric part}'' of the problem 
of non-emptiness of the classes.
At this point, it seems that 
we have tackled all the \emph{algebraic} issues of the problem,
and these conjectures can not be proven without the investigation of the \emph{geometry} of $\mathcal{D}$,
more precisely, the geometry of the different kinds of $\mathcal{Q}_{\alpha}$ extremal points.

The negation of (\ref{eq:emptyByConstruction}) leads to
\begin{equation}
\label{eq:nonEmptyByConstruction}
\mathcal{C}_{\vvs{\alpha},\vvs{\beta}}\neq\emptyset 
\qquad\Longrightarrow\qquad
\forall \vs{\alpha}\in\vvs{\alpha},
\forall \vs{\beta}\in\vvs{\beta}\;:\; \vs{\beta}\npreceq \vs{\alpha},
\end{equation}
\end{subequations}
so if we obtain all $\vvs{\alpha},\vvs{\beta}$ class-labels
for which the right-hand side of this holds (``\emph{non-emptiness-by-construction}'')
then we will have all the non-empty classes,
together with some empty classes if the second conjecture does not hold.
Because of the nontrivial structure of the lattice of proper labels,
obtaining all the class-labels leading to
not-empty-by-construction classes
is also a non-trivial task.
The number of all the partitions of $n$ grows rapidly \cite{oeisA000110,oeisA000041},
which is only the number of the PS subsets of $\alpha$-separability $\mathcal{D}_{\alpha}$.
So the number of all the PS subsets $\mathcal{D}_{\vs{\alpha}}$ grows more rapidly,
and the number of all the PS classes $\mathcal{C}_{\vvs{\alpha},\vvs{\beta}}$ grows even more rapidly.
But at least, it is finite.

%******************************************************************************
\subsection{Indicator functions in general}
\label{subsec:PartSep.Gen.Indicators}

The third point is the generalization of the indicator functions (\ref{eq:PS3PureInd}).
This is carried out in four steps.

(i) Let $F:\mathcal{D}(\mathcal{H}_K)\to\field{R}$ be a continuous function
for all $K\subset L$, that is, for all subsystems.
The only condition on $F$ is 
\begin{equation}
\label{eq:Fprop}
F(\varrho)\geq0, \quad\text{with equality if and only if $\varrho$ is pure},
\end{equation}
for example, the von Neumann entropy, any Tsallis or R\'enyi entropies are suitable.
(Note that the additional requirements of the features of
LU-invariance, convexity, Schur-concavity, %entanglement monotonity, 
additivity, being homogeneous polynomial, etc.,
are only optional, 
they are not needed for the construction.)
For all $K\subset L$ subsystems,
let the following functions on pure states be defined
\begin{equation}
\label{eq:fK}
\begin{split}
f_K&:\mathcal{P}(\mathcal{H})\longrightarrow \field{R},\\
f_K&(\pi) = F(\pi_K),
\end{split}
\end{equation}
where $\pi_K=\tr_{\pcmpl{K}}(\pi)$ with $\pcmpl{K}=L\setminus K$.
Then, for the $k$-partite split $\alpha=L_1|L_2|\dots|L_k$,
$f_{L_r}$ identifies the bipartite split $L_r|\pcmpl{L_r}$ (where $\pcmpl{L_r}=L\setminus L_r$),
\begin{equation}
f_{L_r}(\pi)=0 \qquad\Longleftrightarrow\qquad 
\pi\in \mathcal{P}_{L_r|\pcmpl{L_r}}
\equiv \bigcup_{\beta\preceq L_r|\pcmpl{L_r}}\mathcal{Q}_\beta,
\end{equation}
which is the consequence of (\ref{eq:Fprop}).
Note that $\alpha$ is the greatest element which is smaller than $L_r|\pcmpl{L_r}$ for all $r$.

(ii) Then the function
\begin{equation}
\label{eq:genIndicators}
f_\alpha(\pi)=
\sum_{r=1}^k f_{L_r}(\pi)
\end{equation}
has the ability to identify the $k$-partite split $\alpha$,
\begin{equation}
\label{eq:genIndicatorsdef}
f_\alpha(\pi)=0 \qquad\Longleftrightarrow\qquad 
\pi\in \mathcal{P}_\alpha
\equiv \bigcup_{\beta\preceq \alpha}\mathcal{Q}_\beta.
\end{equation}
All nonnegative $f_\alpha$ functions satisfying (\ref{eq:genIndicatorsdef})
are called \emph{$\alpha$-indicator functions},
not only the ones defined in (\ref{eq:genIndicators}).

(iii) The generalization of (\ref{eq:genIndicators}) for more-than-one partitions,
that is, for all labels, is defined as
\begin{equation}
\label{eq:genIndicatorsl}
f_{\vs{\alpha}}(\pi)= \prod_{\alpha\in\vs{\alpha}} f_{\alpha}(\pi),
\end{equation}
being the generalization of (\ref{eq:PS3PureInd}).
It vanishes exactly for the appropriate $\mathcal{Q}_{\alpha}$ sets,
\begin{equation}
\label{eq:genIndicatorsldef}
f_{\vs{\alpha}}(\pi)=0
\qquad\Longleftrightarrow\qquad 
\pi\in \mathcal{P}_{\vs{\alpha}}
\equiv \bigcup_{\alpha\in\vs{\alpha}}\mathcal{P}_\alpha.
%\equiv \bigcup_{\alpha\in\vs{\alpha}} \bigcup_{\beta\preceq\alpha}\mathcal{Q}_\beta.
\end{equation}
All nonnegative $f_{\vs{\alpha}}$ functions satisfying (\ref{eq:genIndicatorsldef})
are called \emph{$\vs{\alpha}$-indicator functions},
not only the ones defined in (\ref{eq:genIndicatorsl}).
For example, in the next chapter for the three-qubit case,
a special set of indicator functions will be given,
which will be constructed not by (\ref{eq:genIndicatorsl}),
but it will still satisfy (\ref{eq:genIndicatorsldef}).

(iv) Now, the vanishing of their convex roof extension 
\begin{equation*}
\cnvroof{f}_{\vs{\alpha}}(\varrho)=
\min\sum_i p_i f_{\vs{\alpha}}(\pi_i)
\end{equation*}
%(for all pure state decompositions $\sum_i p_i \pi_i=\varrho$)
identify the PS sets,
\begin{equation}
\label{eq:genvanishing}
\cnvroof{f}_{\vs{\alpha}}(\varrho)=0
\qquad\Longleftrightarrow\qquad 
\varrho\in \mathcal{D}_{\vs{\alpha}}
\equiv \Conv\mathcal{P}_{\vs{\alpha}},
%\equiv \Conv\bigcup_{\alpha\in\vs{\alpha}}\mathcal{P}_\alpha,
\end{equation}
being the generalization of (\ref{eq:PS3VanishingMix}).
Indeed,
$\cnvroof{f}_{\vs{\alpha}}(\varrho)=0$
if and only if there exists a decomposition $\varrho=\sum_ip_i\pi_i$
such that $f_{\vs{\alpha}}(\pi_i)=0$ for all $i$
($f_{\vs{\alpha}}$ is nonnegative),
%which means that $\pi_i\in\bigcup_{\alpha\in\vs{\alpha}} \bigcup_{\beta\preceq\alpha}\mathcal{Q}_\beta$
%which means that $\pi_i\in\bigcup_{\alpha\in\vs{\alpha}} \mathcal{P}_\alpha$
which means that $\pi_i\in \mathcal{P}_{\vs{\alpha}}$
(\ref{eq:genIndicatorsldef}),
which means that $\varrho\in\mathcal{D}_{\vs{\alpha}}$.

%******************************************************************************
\subsection{Entanglement-monotone indicator functions in general}
\label{subsec:PartSep.Gen.monIndicators}

As we have seen in section \ref{subsec:QM.EntMeas.EntMeas},
entanglement-monotonicity is a fundamental property of all entanglement measures.
Here we repeat the constructions above with a slight modification,
leading to indicator functions which are entanglement-monotones.
These are denoted by $m$ in contrast with the general $f$.

(i) It has been shown in \cite{VidalEntMon,HorodeckiEntMeas} that 
if $F:\mathcal{D}(\mathcal{H}_K)\to\field{R}$ is unitary-invariant and concave,
then the $f_K$ functions defined in (\ref{eq:fK}) are non-increasing on average for pure states,
that is, they obey (\ref{eq:averagePure}),
which is entanglement monotonity for pure states.
So let
\begin{equation}
\label{eq:mK}
m_K(\pi) = M(\pi_K)
\end{equation}
with $M:\mathcal{D}(\mathcal{H}_K)\to\field{R}$ \emph{vanishing if and only if the state is pure,} as before, 
but now we demand also \emph{unitary-invariace} and \emph{concavity.}
The von Neumann entropy (\ref{eq:Neumann}), 
the Tsallis entropies (\ref{eq:qTsallis}) for all $q>0$,
and the R\'enyi entropies (\ref{eq:qRenyi}) for all $0<q<1$
are known to be concave \cite{BengtssonZyczkowski}, and all of them are unitary-invariant.

(ii) Clearly, entanglement monotone functions form a cone, %the functions obeying (\ref{eq:averagePure}) form a cone,
that is, their sums and multiples by nonnegative real numbers are also entanglement monotones, %obey (\ref{eq:averagePure}),
so we can conclude that the sums of the functions $m_K$ are also entanglement monotones. %obey (\ref{eq:averagePure}).
Here, instead of the original sums in (\ref{eq:genIndicators}),
we introduce the \emph{arithmetic mean} of the $m_{L_r}$ functions,
\begin{equation}
\label{eq:monIndicators}
m_\alpha(\pi)=\frac1k \sum_{r=1}^k m_{L_r}(\pi),
\end{equation}
which are also indicator functions, since they obey (\ref{eq:genIndicatorsdef}).
(The factor $1/k$ is not really important,
but the next step and 
the three-qubit case in the next chapter
%$y=1/3(s_1+s_2+s_3)$ from
%(\ref{eq:PS3PureInd.y})-(\ref{eq:PS3PureInd.sa}) 
motivate the use of mean values.)

(iii) The only problem we face here is that
the set of entanglement monotone functions %obeying (\ref{eq:averagePure}) 
is not closed under multiplication, 
which is involved in the case of the $f_{\vs{\alpha}}$ functions of (\ref{eq:genIndicatorsl}).
This is related to the fact that the product of two concave functions is not concave in general.
Moreover, a recent result of Eltschka et.~al.~\cite{EltschkaetalEntMon} suggests
that entanglement monotone functions
%that functions obeying (\ref{eq:averagePure}) 
can not be of arbitrary high degree,%
%%%%%%%%%%%%%%%%%%%%%%%%
\footnote{See Theorem I.~in \cite{EltschkaetalEntMon},
concerning a special class of functions.}
%%%%%%%%%%%%%%%%%%%%%%%%
so we make a trial of such combination which does not change the degree,
but still fulfils the conditions (\ref{eq:genIndicatorsldef}).
The \emph{geometric mean} will be proven to be suitable, 
which is just a root of the product given in (\ref{eq:genIndicatorsl})
\begin{equation}
\label{eq:monIndicatorsl}
m_{\vs{\alpha}}(\pi)= \Bigl[\prod_{\alpha\in\vs{\alpha}} m_{\alpha}(\pi)\Bigr]^{1/\abs{\vs{\alpha}}},
\end{equation}
where $\abs{\vs{\alpha}}$ is the number of functions $m_{\alpha}$ in the product.
These functions obviously obey (\ref{eq:genIndicatorsldef}),
and they are also entanglement monotones, because the following proposition is true in general.
\begin{prop}
If the functions 
$\mu_j:\mathcal{P}(\mathcal{H})\to\field{R}$ ($j=1,2,\dots,q$)
are nonnegative and non-increasing on average,
\begin{subequations}
\begin{align}
\mu_j(\pi)&\geq0,\\
\label{eq:condAverage}
\sum_{i=1}^m p_i \mu_j(\pi_i) &\leq \mu_j(\pi),
\end{align}
\end{subequations}
then their geometric mean 
\begin{equation*}
\mu=(\mu_1\mu_2\dots\mu_q)^{1/q} 
\end{equation*}
is also nonnegative (obviously) and non-increasing on average,
\begin{subequations}
\begin{align}
\mu(\pi)&\geq0,\\
\sum_{i=1}^m p_i \mu(\pi_i) &\leq \mu(\pi).
\end{align}
\end{subequations}
\end{prop}
(We use this for functions defined on pure states,
although the following proof does not use that,
so the statement holds also for functions defined on all states.)

\begin{proof}
To obtain this, we need a Cauchy-Bunyakowski-Schwarz-like inequality,
for nonnegative $m$-tuples $\ve{x}^{(j)}\in\field{R}^m$, $x_i^{(j)}\geq0$,
which states that
\begin{equation}
\label{eq:CBSlike}
\sum_{i=1}^m  x_i^{(1)} x_i^{(2)}\dots x_i^{(q)} \leq \norm{\ve{x}^{(1)}}_q \norm{\ve{x}^{(2)}}_q\dots \norm{\ve{x}^{(q)}}_q,
\end{equation}
where the usual $q$-norm is 
\begin{equation}
\norm{\ve{x}}_q=\Bigl[\sum_{i=1}^m x_i^q\Bigr]^{1/q}.
\end{equation}
Indeed,
if $\ve{x}^{(j)}=\vs{0}$ for some $j$ then the inequality holds trivially,
else
\begin{equation*}
\begin{split}
\sum_{i=1}^m \frac{x_i^{(1)}}{\norm{\ve{x}^{(1)}}_q} \frac {x_i^{(2)}}{\norm{\ve{x}^{(2)}}_q} \dots\frac{x_i^{(q)}}{\norm{\ve{x}^{(q)}}_q}
\equiv&\sum_{i=1}^m \left[\frac{(x_i^{(1)})^q}{\norm{\ve{x}^{(1)}}_q^q} \frac{(x_i^{(2)})^q}{\norm{\ve{x}^{(2)}}_q^q}\dots
\frac{(x_i^{(q)})^q}{\norm{\ve{x}^{(q)}}_q^q}  \right]^{1/q}\\
\leq&\sum_{i=1}^m \frac1q\left[ \frac{(x_i^{(1)})^q}{\norm{\ve{x}^{(1)}}_q^q}+\frac{(x_i^{(2)})^q}{\norm{\ve{x}^{(2)}}_q^q}+\dots+
\frac{(x_i^{(q)})^q}{\norm{\ve{x}^{(q)}}_q^q} \right] = 1,
\end{split}
\end{equation*}
where the inequality follows
from the inequality of the arithmetic and geometric means, applied to all terms in the sum.
Using this,
\begin{equation*}
\begin{split}
 \sum_{i=1}^m p_i\mu(\pi_i)
=&\sum_{i=1}^m p_i \bigl[\mu_1(\pi_i)\mu_2(\pi_i)\dots \mu_q(\pi_i) \bigr]^{1/q}\\
=&\sum_{i=1}^m \bigl[p_i\mu_1(\pi_i)\bigr]^{1/q} \bigl[p_i\mu_2(\pi_i)\bigr]^{1/q}\dots \bigl[p_i\mu_q(\pi_i)\bigr]^{1/q}\\
\leq &\Bigl[\sum_{i=1}^mp_i\mu_1(\pi_i)\Bigr]^{1/q} \Bigl[\sum_{i=1}^mp_i\mu_2(\pi_i)\Bigr]^{1/q} \dots \Bigl[\sum_{i=1}^mp_i\mu_q(\pi_i)\Bigr]^{1/q}\\
\leq &\bigl[\mu_1(\pi)\big]^{1/q} \bigl[\mu_2(\pi)\bigr]^{1/q} \dots \bigl[\mu_q(\pi)\bigr]^{1/q}
= \mu(\pi),
\end{split}
\end{equation*}
where the first inequality is (\ref{eq:CBSlike}) for $x_i^{(j)}=\bigl[p_i\mu_j(\pi_i)\bigr]^{1/q}$,
and the second inequality is the condition (\ref{eq:condAverage}).
\end{proof}

(iv) Now, $m_{\vs{\alpha}}$ of (\ref{eq:monIndicatorsl}) 
is entanglement monotone, that is, non-increasing on average for pure states (\ref{eq:averagePure}),
so, thanks to (\ref{eq:averageConvRoof}), its convex roof extension 
\begin{equation}
\cnvroof{m}_{\vs{\alpha}}(\varrho)=
\min\sum_i p_i m_{\vs{\alpha}}(\pi_i)
\end{equation}
is also non-increasing on average (\ref{eq:meas.average})
so entanglement-monotone,
and also identifies the PS subsets
\begin{equation}
\cnvroof{m}_{\vs{\alpha}}(\varrho)=0
\quad\Longleftrightarrow\quad 
\varrho\in \mathcal{D}_{\vs{\alpha}},
\end{equation}
as in (\ref{eq:genvanishing}).

%******************************************************************************
%******************************************************************************
\section{Summary and remarks}
\label{sec:PartSep.Sum}

In this chapter we have constructed the complete PS classification of multipartite quantum states
by the PS classes arising from the PS subsets (\ref{eq:genDsets.alphal}),
together with necessary and sufficient conditions for the identification of the PS classes
through the necessary and sufficient conditions for the identification of the PS subsets (\ref{eq:genvanishing})
by indicator functions arising as convex roof extensions of the pure state indicator functions (\ref{eq:genIndicatorsl}).
The indicator functions can be constructed so as to be entanglement-monotone (subsection~\ref{subsec:PartSep.Gen.monIndicators}).
We have also discussed the PS classes in the tripartite case in detail.

\begin{remarks}
\item As was mentioned before,
this PS classification scheme is an extension of the classification based on $k$-separability and $\alpha_k$-separability
given by Seevinck and Uffink \cite{SeevinckUffinkMixSep},
which is the extension of the classification dealing only with $\alpha_k$-separability
given by D\"ur and Cirac \cite{DurCiracTarrach3QBMixSep,DurCiracTarrachBMixSep}.
\item The non-emptiness of the new classes was only conjectured.
More fully, we could not give neccessary and sufficient condition for the non-emptiness of the PS classes
in the purely algebraic language of labels.
Probably, methods from geometry or calculus would be needed to solve this puzzle
(subsection~\ref{subsec:PartSep.Gen.PSclasses}).
\item In close connection with this,
a further geometry-related conjecture could be drafted about the non-empty classes:
they are of non-zero measure.
It is known in the bipartite case that
the separable states are of non-zero measure \cite{Acinetal3QBMixClass,BengtssonZyczkowski}, %ChenDjokovicSemialg},
which can motivate this conjecture.
If the $\mathcal{D}_{1|2|\dots|n}$ PS subset of fully separable states is of non-zero measure,
then this is true for all $\mathcal{D}_{\vs{\alpha}}$ PS subsets.
However, the PS classes arise as intersections of these.
\item The necessary and sufficient criteria of the classes was given by convex roof extensions,
which has advantages and disadvantages.
\item First of all, convex roof extensions are hard to calculate \cite{RothlisbergerLehmannLossNumericalConvRoof,LibCreme}.
However, as we have seen in chapter \ref{chap:SepCrit},
neccessary and sufficient criteria for the detection of convex subsets 
seem always to be hard to calculate,
since they always contain an optimization problem,
such as finding a suitable witness,
or positive map, 
or symmetric extension,
or local spin mesurements,
or detection vector, 
or local bases, and so on.
These optimization problems have no solutions in a closed form in general cases.
\item Another disadvantage of convex roof extensions
is that this is a ``clearly theoretical'' method,
by that we mean that
the full tomography of the state is required, then the criteria are applied by computer.
The majority of the other known criteria share this disadvantage.
Exceptions are the criteria by witnesses (section \ref{subsec:QM.Ent.2Part}),
and by local spin mesurements (quadratic Bell inequalities, see in section \ref{subsec:SepCrit.3Part.Spin}, giving only necessary but not sufficient criteria)
where the criteria can be used in the laboratory, by the tuning of the measurement settings.
However, the optimization still has to be done by the measurement apparatus.
\item An advantage of the convex roof extension
is that it works independently of the dimensions of the subsystems,
so the criteria by that work for arbitrary dimensions.
However, the numerical optimization depends strongly on the rank of the state,
which can be high if the dimension is high,
resulting in slow convergence.
\item Last, but not least, the greatest advantage of our method is (at least for us) that
the criteria by convex roof extensions have a very transparent structure,
they reflect clearly the complicated structure of the PS classes by construction,
which can be seen most directly in (\ref{eq:genvanishing}).
\item
In closing, there is an important question, which can be of research interest as well.
The PS classification is about the issue:
``From which kinds of pure entangled states can a given state be \emph{mixed}?''
Another issue, 
which is equivalently important from the point of view of quantum computation
but which we have not dealt with, is 
``Which kinds of pure entangled states can be \emph{distilled out} from a given state?''
What can be said about this latter?
\end{remarks}

\chapter{Three-qubit systems and FTS approach}
\label{chap:ThreeQB}

In the previous chapter, we have obtained 
the partial separability classification of mixed states
together with necessary and sufficient conditions for the classes.
These conditions were formulated by convex roof extensions of
indicator functions defined on pure states.
Now we will study the case of three qubits,
in which a different set of indicator functions can be obtained.
The construction is based on a beautiful correspondence 
which was found between the three-qubit Hilbert space and a particular FTS (Freudenthal Triple System),
a correspondence which is ``compatible'' with the entanglement of pure three-qubit states \cite{BorstenetalFreudenthal3QBEnt}.
To our present knowledge, this construction works only for three qubits,
however, these results have advantages for the case of three qubits,
they have given us the main ideas for the constructions of the previous chapter,
and, besides these, they are beautiful and interesting in themselves.

Because the FTS approach is in connection with the SLOCC classification of three-qubit state vectors
rather than their partial separability only,
what we get for mixed states is a combination of
the PS classification given in the previous chapter
and the classification of three-qubit mixed states given by Ac{\'i}n et.~al.~(section \ref{subsec:QM.EntMeas.3QBMix}).
So we call this classification \emph{PSS classification,}
which stands for something like ``Partial Separability extended by pure-state SLOCC classes''.
This classification has the advantage of differentiating between different SLOCC classes of pure states,
and also between mixed states depending on which kind of pure entanglement is needed for the preparation of the state.
This is an important advantage, since the PS classification
does not make distinction between
pure states contained in
different SLOCC classes but having the same PS properties
(namely, classes W and GHZ),
although these states may be suitable for different tasks in quantum information processing.
However, in the majority of the cases there are continuously infinite SLOCC classes of pure states
labelled by more-than-one continuous parameters \cite{DurVidalCiracSLOCC3QB,VerstraeteetalSLOCC4QB,ChterentalDjokovicSLOCC4QB},
in which case it is not clear how this classification could be carried out,
if it could be at all.

The material of this chapter covers thesis statement~\ref{statement:threeqb}
(page \pageref{statement:threeqb}).
\begin{organization}
\item[\ref{sec:ThreeQB.Pure}]
we obtain a new set of indicator functions for the three-qubit case.
We recall the LSL-tensors of the FTS approach 
(section~\ref{subsec:ThreeQB.Pure.SLOCC3QBLSL})
by which the SLOCC classes can be identified.
Then we obtain a new set of LU-invariants
(section~\ref{subsec:ThreeQB.Pure.NewInvs})
which are suitable for the role of indicator functions.
Some of them are in connection with the Wootters concurrences of two-qubit subsystems
(section~\ref{subsec:ThreeQB.Pure.WConc}).
\item[\ref{sec:ThreeQB.Mixed}]
we define the PSS classification with the
PSS subsets (section~\ref{subsec:ThreeQB.Mixed.PSSsubsets}),
the PSS classes (section~\ref{subsec:ThreeQB.Mixed.PSSclasses}),
together with the indicator functions for these (section~\ref{subsec:ThreeQB.Mixed.CRoof}).
\item[\ref{sec:ThreeQB.GenThreePart}]
we try to generalize some functions coming from the FTS approach
for the case of three subsystems of arbitrary dimensions,
to obtain indicator functions.
We see that this is evident for full-separability and $a|bc$-separability
(section~\ref{subsec:ThreeQB.GenThreePart.ys}),
this depends on Raggio's conjecture for $b|ac$-$c|ab$-separability
(section~\ref{subsec:ThreeQB.GenThreePart.ga}),
and this does not work for $2$-separability
(section~\ref{subsec:ThreeQB.GenThreePart.t}).
\item[\ref{sec:ThreeQB.Sum}]
we give a summary and some remarks.
\end{organization}

%*******************************************************************************
%*******************************************************************************
\section{State vectors of three qubits}
\label{sec:ThreeQB.Pure}

In this section we review and use the FTS approach of three-qubit state vectors.%
%%%%%%%%%%%%%%%%%%%%%%%%
\footnote{
Remember our convention:
The letters $a$, $b$ and $c$ are variables
taking their values in the set of labels $L=\{1,2,3\}$.
When these variables appear in a formula,
they form a partition of $\{1,2,3\}$,
so they take always different values
and \emph{the formula is understood for all the different values of these variables automatically.}
Although, sometimes a formula is symmetric under the interchange of two such variables
in which case we keep only one of the identical formulas.}
%%%%%%%%%%%%%%%%%%%%%%%%
So we have the Hilbert space 
$\mathcal{H}=\mathcal{H}_1\otimes\mathcal{H}_2\otimes\mathcal{H}_3$
with the local dimensions $\tpl{d}=(2,2,2)$,
therefore, 
after the choice of an orthonormal basis $\{\cket{0},\cket{1}\}\subset\mathcal{H}_a$, $\mathcal{H}_a\isom\field{C}^2$.
The $\cket{\psi}\in\mathcal{H}$ state vectors
are not required to be normalized in this section,
and the $0\in\mathcal{H}$ zero-vector is also allowed.
(Physical states arise, however, from normalized vectors.)

%*******************************************************************************
\subsection{SLOCC Classification by LSL-covariants}
\label{subsec:ThreeQB.Pure.SLOCC3QBLSL}

In \cite{BorstenetalFreudenthal3QBEnt},
Borsten et.~al.~have revealed a very elegant correspondence
between 
the three-qubit Hilbert space $\mathcal{H}\isom\field{C}^2\otimes\field{C}^2\otimes\field{C}^2$
and 
the FTS (Freudenthal Triple System) $\mathfrak{M}(\mathcal{J})\isom\field{C}\oplus\field{C}\oplus\mathcal{J}\oplus\mathcal{J}$
over the cubic Jordan algebra $\mathcal{J}\isom\field{C}\oplus\field{C}\oplus\field{C}$.
The fundamental point of this correspondence is
that the automorphism group of this FTS is
$\LieGrp{Aut}\bigl(\mathfrak{M}(\field{C}\oplus\field{C}\oplus\field{C})\bigr)=\LieGrp{SL}(2,\field{C})^{\times3}$,
which is just the relevant LSL subgroup 
of $\LieGrp{GL}(2,\field{C})^{\times3}$, the LGL-group of SLOCC equivalence for three-qubit pure states
(section \ref{subsec:QM.Ent.LO}).
This group-theoretical coincidence arises only in the three-qubit case.
It has been shown \cite{BorstenetalFreudenthal3QBEnt} that
the vectors of \emph{different SLOCC classes of entanglement} in the three-qubit Hilbert space
(section \ref{subsec:QM.EntMeas.3QBPure})
are in one-to-one correspondence with
the elements of \emph{different rank} in the FTS.
The rank of an element of an FTS
is characterized by the vanishing of some associated elements, 
which are covariant (maybe invariant) under the action of the automorphism group,
resulting in conditions for the SLOCC classes in the Hilbert-space by the vanishing or non-vanishing of
$\LieGrp{SL}(2,\field{C})^{\times3}$ tensors.
Hence this classification is manifestly invariant under SLOCC-equivalence \cite{BorstenetalFreudenthal3QBEnt},
which can not be seen directly in the conventional classification (section \ref{subsec:QM.EntMeas.3QBPure}),
since the $c^2_a$ local entropies are scalars only under $\LieGrp{U}(2)^{\times3}$.
(However, the invariance of the vanishing of $c^2_a$s follows easily from the fact 
that the local rank is invariant under invertible transformations \cite{DurVidalCiracSLOCC3QB}.)

To a
\begin{equation*}
\cket{\psi}=\sum_{i,j,k=0}^1\psi^{ijk}\cket{ijk}
\end{equation*}
three-qubit state
we can assign an element $\psi\in\mathfrak{M}(\mathcal{\field{C}\oplus\field{C}\oplus\field{C}})$,
and calculate some associated quantities
needed for the identification of its rank.
Here we list these quantities 
in the form in which we will use them.
(For the basic definitions 
of Jordan algebras, Freudenthal triple systems
and the operations and maps defined on them,
see in \cite{BorstenetalFreudenthal3QBEnt} and in the references therein.)
\begin{subequations}
\label{eq:tensors}
\begin{align}
\label{eq:tensors.Upsilon}
\begin{split}
[\Upsilon_\phi(\psi)]^{ijk} = 
 &-\varepsilon_{ll'}\varepsilon_{mm'}\varepsilon_{nn'}\psi^{imn}\psi^{lm'n'}\phi^{l'jk}\\
 &-\varepsilon_{mm'}\varepsilon_{nn'}\varepsilon_{ll'}\psi^{ljn}\psi^{l'mn'}\phi^{im'k}\\
 &-\varepsilon_{nn'}\varepsilon_{ll'}\varepsilon_{mm'}\psi^{lmk}\psi^{l'm'n}\phi^{ijn'},
\end{split}\\
\label{eq:tensors.gamma1}
[\gamma_1(\psi)]^{ii'}=&\; \varepsilon_{jj'}\varepsilon_{kk'}\psi^{ijk}\psi^{i'j'k'},\\
\label{eq:tensors.gamma2}
[\gamma_2(\psi)]^{jj'}=&\; \varepsilon_{kk'}\varepsilon_{ii'}\psi^{ijk}\psi^{i'j'k'},\\
\label{eq:tensors.gamma3}
[\gamma_3(\psi)]^{kk'}=&\; \varepsilon_{ii'}\varepsilon_{jj'}\psi^{ijk}\psi^{i'j'k'},\\
\label{eq:tensors.T}
\begin{split}
[T(\psi,\psi,\psi)]^{ijk}
=&-\varepsilon_{ll'}\varepsilon_{mm'}\varepsilon_{nn'}\psi^{imn}\psi^{lm'n'}\psi^{l'jk},\\
=&-\varepsilon_{mm'}\varepsilon_{nn'}\varepsilon_{ll'}\psi^{ljn}\psi^{l'mn'}\psi^{im'k},\\
=&-\varepsilon_{nn'}\varepsilon_{ll'}\varepsilon_{mm'}\psi^{lmk}\psi^{l'm'n}\psi^{ijn'},
\end{split}\\
\begin{split}
\label{eq:tensors.q}
q(\psi) =&\;
\varepsilon_{ii'}\varepsilon_{jj'}
\varepsilon_{kk'}\varepsilon_{ll'}
\varepsilon_{mm'}\varepsilon_{nn'}
\psi^{ikl}\psi^{jk'l'}\psi^{i'mn}\psi^{j'm'n'}.
\end{split}
\end{align}
\end{subequations}
The summation for the pairs of indices occuring upstairs and downstairs are understood, as usual,
and $\varepsilon$ is the matrix of the 
$\LieGrp{Sp}(1)\isom\LieGrp{SL}(2)$-invariant non-degenerate antisymmetric bilinear form (\ref{eq:epsilon}):
Thanks to (\ref{eq:epstraf}),
index contraction by $\varepsilon$ is invariant under $\LieGrp{SL}(2,\field{C})$ transformations.
This shows that
if we regard $\psi$ and $\phi$ as tensors 
transform as a $(\mathbf{2},\mathbf{2},\mathbf{2})$
under $\LieGrp{SL}(2,\field{C})^{\times3}$,
then so do $\Upsilon_\phi(\psi)$ and $T(\psi,\psi,\psi)$,
while $\gamma_1(\psi)$, $\gamma_2(\psi)$ and $\gamma_3(\psi)$, being symmetric, transform as 
$(\mathbf{3},\mathbf{1},\mathbf{1})$,
$(\mathbf{1},\mathbf{3},\mathbf{1})$ and
$(\mathbf{1},\mathbf{1},\mathbf{3})$, respectively,
and $q(\psi)$ transforms as $(\mathbf{1},\mathbf{1},\mathbf{1})$,
which means that it is scalar.%
%%%%%%%%%%%%%%%%%%%%%%%%
\footnote{Note that for any $2\times2$ matrix $M$, the determinant
$2\det M =\varepsilon_{ii'}\varepsilon_{jj'}M^{ij}M^{i'j'}$,
thanks to (\ref{eq:epstraf}),
so $2\det \gamma_a(\psi) =q(\psi)$.}
%%%%%%%%%%%%%%%%%%%%%%%%
The main result of \cite{BorstenetalFreudenthal3QBEnt} is 
that the conditions for the SLOCC classes (section \ref{subsec:QM.EntMeas.3QBPure}) 
can be formulated by the vanishing of these tensors
in the way which can be seen in table~\ref{tab:SLOCC3Pure2}.
%%%%%%%%%%%%%%%%%%%%%%%%%%%%%%%%%%%%%%%%
\begin{table}
\begin{tabu}{X[c]||X[c]|X[c]|X[c]X[c]X[c]|X[c]|X[c]}
\hline
Class  & $\psi$  & $\Upsilon_\phi(\psi)$    & $\gamma_1(\psi)$ & $\gamma_2(\psi)$ & $\gamma_3(\psi)$ &  $T(\psi,\psi,\psi)$  & $q(\psi)$ \\
\hline
\hline
$\mathcal{V}_\text{Null}$ & $=0$    & $=0,\forall\phi$         & $=0$     & $=0$     & $=0$     & $=0$     & $=0$ \\
\hline
$\mathcal{V}_{1|2|3}$     & $\neq0$ & $=0,\forall\phi$         & $=0$     & $=0$     & $=0$     & $=0$     & $=0$ \\
\hline
$\mathcal{V}_{1|23}$      & $\neq0$ & $\neq0,\exists\phi$      & $\neq0$  & $=0$     & $=0$     & $=0$     & $=0$ \\
$\mathcal{V}_{2|13}$      & $\neq0$ & $\neq0,\exists\phi$      & $=0$     & $\neq0$  & $=0$     & $=0$     & $=0$ \\
$\mathcal{V}_{3|12}$      & $\neq0$ & $\neq0,\exists\phi$      & $=0$     & $=0$     & $\neq0$  & $=0$     & $=0$ \\
\hline
$\mathcal{V}_\text{W}$    & $\neq0$ & $\neq0,\exists\phi$      & $\neq0$  & $\neq0$  & $\neq0$  & $\neq0$  & $=0$ \\
$\mathcal{V}_\text{GHZ}$  & $\neq0$ & $\neq0,\exists\phi$      & $\neq0$  & $\neq0$  & $\neq0$  & $\neq0$  & $\neq0$ \\
\hline
\end{tabu}
\bigskip
\caption{SLOCC classes of three-qubit state vectors
 identified by the vanishing of LSL-covariants~(\ref{eq:tensors}).}
\label{tab:SLOCC3Pure2}
\end{table}
%%%%%%%%%%%%%%%%%%%%%%%%%%%%%%%%%%%%%%%%
%In the light of the conditions by the norm and the four determinants of the conventional SLOCC-classification, 
%(see section \ref{subsec:QM.EntMeas.SLOCC3QBLU} and in particular in table~\ref{tab:SLOCC3Pure},)
%this scheme constructed by seven quantities seems to be redundant.
%Indeed, it is redundant for the SLOCC-classification of pure states,
%%but it will turn out that this shows us the way towards the generalization for mixed states.
%%but it will turn out that this way leads to the generalization for mixed states.
%but it is still too small for the PSS-classification of mixed states.

%*******************************************************************************
\subsection{SLOCC Classification by a new set of LU-invariants}
\label{subsec:ThreeQB.Pure.NewInvs}

Now, we need quantities (indicator functions) which can be extended from pure states to mixed states
by the convex roof construction. %, see in (\ref{eq:cnvroofext}). %\cite{BennettetalMixedStates,UhlmannFidelityConcurrence,UhlmannConvRoofs}.
There is no natural ordering on the tensors of~(\ref{eq:tensors})
so convex roof construction does not work directly for them,
but we can form quantities from them taking values in the field of real numbers.
During this, we will 
lose the covariance under $\LieGrp{SL}(2,\field{C})^{\times3}$,
but gain the invariance under the group $\LieGrp{U}(2)^{\times3}$.

Returning from the FTS language to the Hilbert space language,
we have vectors
\begin{align*}
\cket{\Upsilon_\phi(\psi)}&=\sum_{i,j,k=0}^1[\Upsilon_\phi(\psi)]^{ijk}\cket{ijk}\quad\in\mathcal{H},\\
\cket{T(\psi,\psi,\psi)}&=\sum_{i,j,k=0}^1[T(\psi,\psi,\psi)]^{ijk}\cket{ijk}\quad\in\mathcal{H},
\intertext{and local operators}
\gamma_1(\psi)\varepsilon&=\sum_{i,i'=0}^1[\gamma_1(\psi)\varepsilon]^i_{\;i'}\cket{i}\bra{i'}\quad\in\Lin(\mathcal{H}_1),\\
\gamma_2(\psi)\varepsilon&=\sum_{j,j'=0}^1[\gamma_2(\psi)\varepsilon]^j_{\;j'}\cket{j}\bra{j'}\quad\in\Lin(\mathcal{H}_2),\\
\gamma_3(\psi)\varepsilon&=\sum_{k,k'=0}^1[\gamma_3(\psi)\varepsilon]^k_{\;k'}\cket{k}\bra{k'}\quad\in\Lin(\mathcal{H}_3)
\end{align*}
associated to $\cket{\psi}\in\mathcal{H}$ through~(\ref{eq:tensors}) of the FTS language.%
%%%%%%%%%%%%%%%%%%%%%%%%
\footnote{Note that 
$\varepsilon  \in \mathcal{H}_a^*\otimes\mathcal{H}_a^*
              \isom\Lin(\mathcal{H}_a\to\mathcal{H}_a^*)
              \isom\BiLin(\mathcal{H}_a\times\mathcal{H}_a\to\field{C})$,
while
$\gamma_a(\psi)\in  \mathcal{H}_a\otimes\mathcal{H}_a
              \isom\Lin(\mathcal{H}_a^*\to\mathcal{H}_a)$,
so $\gamma_a(\psi)\varepsilon\in\Lin(\mathcal{H}_a\to\mathcal{H}_a)$.}
%%%%%%%%%%%%%%%%%%%%%%%%
These are just computational auxiliaries coming from considerations related to the LSL-tensors,
not state vectors and local observables of any systems,
because they depend nonlinearly on the state vector $\cket{\psi}$, 
moreover, $\gamma_a(\psi)\varepsilon\notin\mathcal{A}(\mathcal{H}_a)$.

Using these,
the vanishing conditions of the tensors~(\ref{eq:tensors}) in table~\ref{tab:SLOCC3Pure2} can be reformulated.
Clearly, $\psi=0$ if and only if $\norm{\psi}^2=0$.
Taking a look at $\Upsilon_\phi(\psi)$ in~(\ref{eq:tensors.Upsilon}) it turns out that
$\Upsilon_\phi(\psi)$ can be written in the Hilbert space language as
\begin{equation*}
\cket{\Upsilon_\phi(\psi)} = Y(\psi) \cket{\phi}
\end{equation*}
with the operator
\begin{equation*}
Y(\psi) = 
-\gamma_1(\psi)\varepsilon\otimes\Id\otimes\Id
-\Id\otimes\gamma_2(\psi)\varepsilon\otimes\Id
-\Id\otimes\Id\otimes\gamma_3(\psi)\varepsilon
\quad\in\Lin(\mathcal{H}).
\end{equation*}
Using this,
the vanishing condition of $\Upsilon_\phi(\psi)$ for all $\phi$:
\begin{equation*}
\begin{split}
\cket{\Upsilon_\phi(\psi)}=0\;\;\forall\cket{\phi}\qquad
&\Longleftrightarrow\qquad Y(\psi) \cket{\phi}=0\;\;\forall\cket{\phi}\\
&\Longleftrightarrow\qquad Y(\psi)=0\\
&\Longleftrightarrow\qquad \norm{Y(\psi)}^2=0, 
% \;\text{for any norm},
\end{split}
\end{equation*}
so we can eliminate the quantifiers and $\phi$ from the condition.
Since the last implication holds for any norm,
using the usual complex matrix $2$-norm $\norm{M}^2=\tr(M^\dagger M)$
we have
\begin{equation*}
\norm{Y(\psi)}^2=4\bigl(\norm{\gamma_1(\psi)}^2 + \norm{\gamma_2(\psi)}^2 + \norm{\gamma_3(\psi)}^2\bigr).
\end{equation*}
This formula has a remarkably structure. 
Namely, note that the local concurrence-squared (\ref{eq:conc2}),
which can be used as $a|bc$-indicator function (\ref{eq:PS3PureInd.a|bc}),
arises here as
\begin{equation*}
c^2_a(\psi)=C^2(\tr_{bc}\cket{\psi}\bra{\psi})=\norm{\gamma_b(\psi)}^2 + \norm{\gamma_c(\psi)}^2,
\end{equation*}
%(see in \cite{BorstenetalFreudenthal3QBEnt},
%%%%%%%%%%%%%%%%
%\footnote{Unfortunately,
%there are some misprints in the constant factors
%given in \cite{BorstenetalFreudenthal3QBEnt}.
%(6):  $K=3\tr(\rho_A\otimes\rho_B\rho_{AB})-\tr(\rho_A^3)-\tr(\rho_B^3)=\dots$
%(29): $S_A=\tr\gamma^{B\dagger}\gamma^B+\tr\gamma^{C\dagger}\gamma^C$,
%(30): $\tr\gamma^{A\dagger}\gamma^A=\frac12[S_B+S_C-S_A]$,
%(33): $||T||^2 =\frac{2}{3}(K-\abs{\psi}^6)+\frac{1}{4}\abs{\psi}^2(S_A+S_B+S_C)$.
%%%(33): $\rangle T\vert T\langle =\frac{2}{3}(K-\abs{\psi}^6)+\frac{1}{4}\abs{\psi}^2(S_A+S_B+S_C)$.
%% bug: labjegyzetbe nem engedi ezeket
%These have no influence on the results of the paper,
%we list them only to prevent misunderstandings.}%
%%%%%%%%%%%%%%%%
%)
and $\gamma_a(\psi)=0$ if and only if $\norm{\gamma_a(\psi)}^2=0$.
Now turn to the vanishing of $T(\psi,\psi,\psi)$, given in~(\ref{eq:tensors.T}).
Again, this vanishes if and only if its norm $\norm{T(\psi,\psi,\psi)}^2$ does.%
%%%%%%%%%%%%%%%%%%%%%%%%
\footnote{
The quantity $\norm{T(\psi,\psi,\psi)}^2$ appears also in the
twistor-geometric approach of three-qubit entanglement, 
it is proportional to $\omega_{\text{ABC}}$ in \cite{PeterGeom3QBEnt}.}
%%%%%%%%%%%%%%%%%%%%%%%%
This can be calculated by the use of the form
\begin{equation*}
\cket{T(\psi,\psi,\psi)}
= -\gamma_1(\psi)\varepsilon\otimes\Id\otimes\Id \cket{\psi}
= -\Id\otimes\gamma_2(\psi)\varepsilon\otimes\Id \cket{\psi}
= -\Id\otimes\Id\otimes\gamma_3(\psi)\varepsilon \cket{\psi}
= \frac13 Y(\psi) \cket{\psi}.
\end{equation*}
About the scalar $q$, notice that $q(\psi)=-2\Det(\psi)$ with Cayley's hyperdeterminant (\ref{eq:HDet}),
and it vanishes if and only if the three-tangle~(\ref{eq:tau}) does.

Summarizing the observations above,
it will be useful to define the following set of real-valued functions on $\mathcal{H}$:
\begin{subequations}
\label{eq:newPureLUinvs}
\begin{align}
n(\psi)     &= \norm{\psi}^2,\\
\label{eq:newPureLUinvs.y}
y(\psi)     &= \frac23\bigl(g_1(\psi) + g_2(\psi) + g_3(\psi)\bigr),\\
\label{eq:newPureLUinvs.sa}
c^2_a(\psi)   &= g_b(\psi) + g_c(\psi),\\
\label{eq:newPureLUinvs.ga}
g_a(\psi)   &= \norm{\gamma_a(\psi)}^2,\\
\label{eq:newPureLUinvs.t}
t(\psi)     &= 4\norm{T(\psi,\psi,\psi)}^2,\\
\label{eq:newPureLUinvs.tau2}
\tau^2(\psi)&= 4\abs{q(\psi)}^2.
\end{align}
\end{subequations}
%%%%%%%%%%%%%%%%%%%%%%%%%%%%%%%%%%%%%%%%
\begin{table}
\begin{tabu}{X[c]||X[c]|X[c]|X[c]X[c]X[c]|X[c]X[c]X[c]|X[c]|X[c]}
\hline
Class  & $n(\psi)$ & $y(\psi)$ & $c^2_1(\psi)$ & $c^2_2(\psi)$ & $c^2_3(\psi)$ & $g_1(\psi)$ & $g_2(\psi)$ & $g_3(\psi)$ & $t(\psi)$ & $\tau^2(\psi)$  \\
\hline
\hline
$\mathcal{V}_\text{Null}$ & $=0$     & $=0$     & $=0$     & $=0$     & $=0$     & $=0$     & $=0$     & $=0$     & $=0$     & $=0$ \\
\hline
$\mathcal{V}_{1|2|3}$     & $>0$     & $=0$     & $=0$     & $=0$     & $=0$     & $=0$     & $=0$     & $=0$     & $=0$     & $=0$ \\
\hline
$\mathcal{V}_{1|23}$      & $>0$     & $>0$     & $=0$     & $>0$     & $>0$     & $>0$     & $=0$     & $=0$     & $=0$     & $=0$ \\
$\mathcal{V}_{2|13}$      & $>0$     & $>0$     & $>0$     & $=0$     & $>0$     & $=0$     & $>0$     & $=0$     & $=0$     & $=0$ \\
$\mathcal{V}_{3|12}$      & $>0$     & $>0$     & $>0$     & $>0$     & $=0$     & $=0$     & $=0$     & $>0$     & $=0$     & $=0$ \\
\hline
$\mathcal{V}_\text{W}$    & $>0$     & $>0$     & $>0$     & $>0$     & $>0$     & $>0$     & $>0$     & $>0$     & $>0$     & $=0$ \\
$\mathcal{V}_\text{GHZ}$  & $>0$     & $>0$     & $>0$     & $>0$     & $>0$     & $>0$     & $>0$     & $>0$     & $>0$     & $>0$ \\
\hline
\end{tabu}
\bigskip
\caption{SLOCC classes of three-qubit state vectors
identified by the vanishing of the pure state indicator functions 
given in~(\ref{eq:newPureLUinvs}).}
\label{tab:PSS3Pclasses}
\end{table}
%%%%%%%%%%%%%%%%%%%%%%%%%%%%%%%%%%%%%%%%
(It is shown in [\ref{pub:partsep}] of the list on page \pageref{chap:publist}
that the constant factors have been chosen so that 
$0\leq y(\psi),c^2_a(\psi),g_a(\psi),t(\psi),\tau^2(\psi) \leq1$ for normalized states.%
%%%%%%%%%%%%%%%%%%%%%%%%
\footnote{
Otherwise, for unnormalized states
their maxima is scaling by the corresponding power of the norm,
$0\leq y(\psi),c^2_a(\psi),g_a(\psi) \leq n^2(\psi)$,
$0\leq t(\psi) \leq n^3(\psi)$,
$0\leq \tau^2(\psi) \leq n^4(\psi)$,
since these functions are homogeneous ones.}%
%%%%%%%%%%%%%%%%%%%%%%%%
)
These quantities are obtained 
by index-contraction of $\psi^{ijk}$s and complex conjugated $\cc{(\psi^{i'j'k'})}=\psi_{i'j'k'}$s by $\delta^{i'}_{\phantom{i'}i}$s
from the tensors in~(\ref{eq:tensors}), which were obtained by index-contraction of $\psi^{ijk}$s and $\psi^{i'j'k'}$s by $\varepsilon_{ii'}$s.
From the contractions of free indices of the thesors in~(\ref{eq:tensors}), 
we have $U^\dagger\delta U= \delta$ for $U\in\LieGrp{U}(2)$.
From the contractions inside the tensors of~(\ref{eq:tensors}),
we have $U^\transp\varepsilon U= \varepsilon \det U$
but for every factor $\det U$ there is a conjugated $\cc{(\det U)}=1/\det U $
from $U^*\varepsilon U^\dagger= \varepsilon \det\cc{U}$.
Consequently, all the functions in~(\ref{eq:newPureLUinvs}) are LU-invariant ones,
while their vanishing are still LSL-invariant.%
%%%%%%%%%%%%%%%%%%%%%%%%
\footnote{
Moreover, $n$ is invariant under the larger group $\LieGrp{U}(8)$,
and $\tau^2$ under $\left[\LieGrp{U}(1)\times\LieGrp{SL}(2,\field{C})\right]^{\times3}$,
and $n$, $y$, $t$ and $\tau^2$ under the discrete group of non-local transformations of the permutations of subsystems.}
%%%%%%%%%%%%%%%%%%%%%%%%
Being LU-invariant homogeneous polynomials, 
we can express them in the standard basis of
three-qubit LU-invariant homogeneous polynomials (\ref{eq:3QBcanonPureLUinvs}) as
\begin{subequations}
\label{eq:newPureLUinvsCanon}
\begin{align}
n   &= I_0,\\
y   &= 2I_0^2-\frac23\bigl(I_1+I_2+I_3\bigr),\\
c^2_a &= 2\bigl(I_0^2-I_a\bigr),\\
g_a &= I_0^2+I_a-I_b-I_c,\\
\label{eq:newPureLUinvsCanon.t}
t   &= \frac83 I_4 + \frac{10}{3}I_0^3 - 2I_0\bigl(I_1+I_2+I_3\bigr),\\
\tau^2&= 4I_5.
\end{align}
\end{subequations}

Now, the conditions for the SLOCC classes 
by the vanishing of the tensors in~(\ref{eq:tensors})
(see in table~\ref{tab:SLOCC3Pure2})
can be reformulated 
by the vanishing of the functions in~(\ref{eq:newPureLUinvs})
in the way which can be seen in table~\ref{tab:PSS3Pclasses}.
Comparing this with table \ref{tab:PS3Pclasses}
we have that the functions $y$, $c^2_a$, $g_a$ and $t$ are proper indicator functions
for the PS classification of the states of three-qubit systems.
The function $\tau^2$ completes those for the PSS classification.
But, before turning to this, let us discuss the relation
of the new functions $g_a$ to the entanglement of the two-qubit subsystems.

%*******************************************************************************
\subsection{Entanglement of two-qubit subsystems}
\label{subsec:ThreeQB.Pure.WConc}
The entanglement inside the two-qubit subsystems 
can be calculated by the use of the Wootters concurrence (\ref{eq:WConc}).
Recall that, for a two-qubit state $\omega$,
it can be written by the eigenvalues of $\sqrt{\sqrt{\omega}\tilde{\omega}\sqrt{\omega}}$ as
\begin{equation*}
\cnvroof{c}(\omega)=\bigl(\lambda_1^\downarrow-\lambda_2^\downarrow-\lambda_3^\downarrow-\lambda_4^\downarrow\bigr)^+.
\end{equation*}
If the two-qubit mixed state
for which the Wootters concurrence-squared is calculated
is reduced from a pure three-qubit state,
then $\omega=\pi_{bc}$ is at the most of rank $2$,
and 
\begin{equation*}
{\cnvroof{c}}^2(\pi_{bc})=(\lambda_1-\lambda_2)^2 =\tr\pi_{bc}\tilde{\pi}_{bc}-2\lambda_1\lambda_2.
\end{equation*}
One can check that
\begin{subequations}
\begin{equation}
\label{eq:Conc3QB1}
\tr\pi_{bc}\tilde{\pi}_{bc}=\tr\gamma_a(\psi)^\dagger\gamma_a(\psi),
\end{equation}
which is just $g_a(\psi)$ of (\ref{eq:newPureLUinvs.ga}), and
\begin{equation}
\label{eq:Conc3QB2}
\lambda_1\lambda_2= \abs{\det\gamma_a(\psi)}=\abs{\Det\psi},
\end{equation}
\end{subequations}
see in \cite{CKWThreetangle}.
The Wootters concurrence is then given by
\begin{equation}
\label{eq:Conc3QB}
{\cnvroof{c}}^2(\pi_{bc})= g_a(\psi)-\frac12 \tau(\psi).
\end{equation}
As it can be seen, 
$g_a(\psi)$ measures the entanglement in two-qubit subsystems
for all vectors which are not of Class GHZ.
Note that this is only a zero-measured subset of all state vectors.
Note also that
${\cnvroof{c}}^2(\pi_{bc})$ can not be used instead of $g_a$ for the role of an indicator function,
because it can be zero for Class GHZ vectors%
%%%%%%%%%%%%%%%%%%%%%%%%
\footnote{For example the standard GHZ state (\ref{eq:GHZ}) has separable two-qubit subsystems,
for which the Wootters concurrence vanishes.
Such vectors, having maximal three-tangle, form an important subclass of three-qubit entanglement.}
%%%%%%%%%%%%%%%%%%%%%%%%
hence does not obey the last line of table \ref{tab:PSS3Pclasses}.
On the other hand, the CKW equality (\ref{eq:CKW}) about entanglement monogamy,
\begin{equation}
\label{eq:recallCKW}
c^2_a(\psi)= {\cnvroof{c}}^2(\pi_{ab})+ {\cnvroof{c}}^2(\pi_{ac})+\tau(\psi),
\end{equation}
is then equivalent to (\ref{eq:newPureLUinvs.sa}).

The roof extension relates the concurrence
with another important quantity, the fidelity \cite{UhlmannFidelityConcurrence}.
The \emph{fidelity between two density matrices} $\omega$ and $\sigma$
is $F(\omega,\sigma)=\tr\sqrt{\sqrt{\omega}\sigma\sqrt{\omega}}$,
which is the square root of the transition probability,
and it is in connection with distances and distinguishability measures on the space of density matrices \cite{BengtssonZyczkowski}.
The \emph{fidelity of a state} with respect to the spin flip is $F(\omega,\tilde{\omega})$,
which is just the \emph{concave roof extension}%
%%%%%%%%%%%%%%%%%%%%%%%%
\footnote{The concave roof extension is the maximization of the weighted average over the decompositions
instead of the minimization in (\ref{eq:cnvroofext}).} 
%%%%%%%%%%%%%%%%%%%%%%%%
of $c$. For the two-qubit case this takes the form \cite{UhlmannFidelityConcurrence}
\begin{equation*}
\cncroof{c}(\omega)=F(\omega,\tilde{\omega})=\lambda_1+\lambda_2+\lambda_3+\lambda_4.
\end{equation*}
Again, for the mixed states of two-qubit subsystems
arising from a three-qubit system being in a pure state,
the fidelity is
\begin{equation*}
{\cncroof{c}}^2(\pi_{bc})
=(\lambda_1+\lambda_2)^2=\tr\pi_{bc}\tilde{\pi}_{bc}+2\lambda_1\lambda_2.
\end{equation*}
Using (\ref{eq:Conc3QB1}) and (\ref{eq:Conc3QB2}),
it is of the form similar to the concurrence (\ref{eq:Conc3QB})
\begin{equation}
\label{eq:Fid3QB}
{\cncroof{c}}^2(\pi_{bc})= g_a(\psi)+\frac12 \tau(\psi).
\end{equation}
Then, using (\ref{eq:newPureLUinvs.sa}),
we get a CKW-like equality for the fidelities,
\begin{equation}
\label{eq:CKWFid}
c^2_a(\psi)= {\cncroof{c}}^2(\pi_{ab})+ {\cncroof{c}}^2(\pi_{ac})-\tau(\psi).
\end{equation}
On the other hand, from (\ref{eq:Conc3QB}) and (\ref{eq:Fid3QB}),
$g_a(\psi)$ is just the average of the concave and convex roofs
and $\tau(\psi)$ is their difference,
\begin{subequations}
\begin{align}
 g_a(\psi) &= \frac12\bigl( {\cncroof{c}}^2(\pi_{bc}) + {\cnvroof{c}}^2(\pi_{bc}) \bigr),\\
\tau(\psi) &=               {\cncroof{c}}^2(\pi_{bc}) - {\cnvroof{c}}^2(\pi_{bc}).
\end{align}
\end{subequations}
Hence, on the zero-measured set of non-GHZ pure states,
the (two-qubit) convex and concave roof extensions of local concurrence $c$ are equal,
and both of them are equal to the indicator function $g_a(\psi)$.

%*******************************************************************************
%*******************************************************************************
\section{Mixed states of three qubits}
\label{sec:ThreeQB.Mixed}
Now, we have the indicator functions $y$, $c^2_a$, $g_a$ and $t$ 
for the PS classification of the states of three-qubit systems,
and we will also have $\tau^2$ for Class W of the classification 
of Ac{\'i}n et.~al.~(section \ref{subsec:QM.EntMeas.3QBMix}).
The PSS classification arises as the combination of these two.
Since only the PS class $\mathcal{C}_1$ is divided into two PSS classes,
obtaining the PSS classification is straightforward.
Only for the sake of completeness, we write out this PSS classification
for which these indicator functions can be used.

%*******************************************************************************
\subsection{PSS subsets}
\label{subsec:ThreeQB.Mixed.PSSsubsets}

Starting with the extremal points of the space of states,
we have the pure states given by the vectors of different SLOCC classes
given in section \ref{subsec:QM.EntMeas.3QBPure}.
These are the same as those of the general tripartite case, given in (\ref{eq:PSPclasses}),
only the set of tripartite entangled pure states $\mathcal{Q}_{123}$ 
is divided into two disjoint subsets according to the W and GHZ-type three-qubit entanglement.
These are denoted with $\mathcal{Q}_{123}=\mathcal{Q}_\text{W}\cup\mathcal{Q}_\text{GHZ}$,
and with these we get the \emph{pure PSS classes}
\begin{subequations}
\label{eq:PSSPclasses}
\begin{align}
\mathcal{Q}_{1|2|3}     &=\bigl\{\pi=\cket{\psi}\bra{\psi}\;\big\vert\; \cket{\psi}\in \mathcal{V}_{1|2|3},\; \norm{\psi}^2=1 \bigr\},\\
\mathcal{Q}_{a|bc}      &=\bigl\{\pi=\cket{\psi}\bra{\psi}\;\big\vert\; \cket{\psi}\in \mathcal{V}_{a|bc},\;  \norm{\psi}^2=1 \bigr\},\\
\mathcal{Q}_\text{W}    &=\bigl\{\pi=\cket{\psi}\bra{\psi}\;\big\vert\; \cket{\psi}\in \mathcal{V}_\text{W},\; \norm{\psi}^2=1 \bigr\},\\
\mathcal{Q}_\text{GHZ}  &=\bigl\{\pi=\cket{\psi}\bra{\psi}\;\big\vert\; \cket{\psi}\in \mathcal{V}_\text{GHZ},\; \norm{\psi}^2=1 \bigr\}.
\end{align}
\end{subequations}
Again, these are disjoint sets covering $\mathcal{P}$ entirely.
From these, we can obtain the \emph{pure PSS sets}
\begin{subequations}
\label{eq:PSSPsets}
\begin{align}
\mathcal{P}_{1|2|3}    %&=\bigl\{\pi=\cket{\psi}\bra{\psi}\;\big\vert\; \cket{\psi}\in \mathcal{V}_{1|2|3},\; \norm{\psi}^2=1 \bigr\}
&= \mathcal{Q}_{1|2|3},\\
\mathcal{P}_{a|bc}     %&=\bigl\{\pi=\cket{\psi}\bra{\psi}\;\big\vert\; \cket{\psi}\in \mathcal{V}_{1|2|3}\cup\mathcal{V}_{a|bc},\;  \norm{\psi}^2=1 \bigr\}
&= \mathcal{Q}_{1|2|3}\cup\mathcal{Q}_{a|bc},\\
\mathcal{P}_\text{W}   %&=\bigl\{\pi=\cket{\psi}\bra{\psi}\;\big\vert\; \cket{\psi}\in \mathcal{V}_{1|2|3}\cup\mathcal{V}_{a|bc}\cup\mathcal{V}_\text{W},\;  \norm{\psi}^2=1 \bigr\}
&= \mathcal{Q}_{1|2|3}\cup\mathcal{Q}_{a|bc}\cup\mathcal{Q}_\text{W},\\
\mathcal{P}_\text{GHZ} &= \mathcal{Q}_{1|2|3}\cup\mathcal{Q}_{a|bc}\cup\mathcal{Q}_\text{W}\cup\mathcal{Q}_\text{GHZ} \equiv\mathcal{P}_{123}\equiv\mathcal{P}.
\end{align}
\end{subequations}
These sets are closed and contain each other in a hierarchic way,
which is illustrated in figure~\ref{fig:PSS3incl}.

Using these, we can obtain the \emph{PSS subsets},
which are the same as the PS subsets of the general tripartite case, given in (\ref{eq:PSsets}),
except the new set $\mathcal{D}_\text{W}$,
which can be mixed without the use of GHZ-type entanglement:
\begin{subequations}
\label{eq:PSSsets}
\begin{align}
\label{eq:PSSsets.1|2|3}
\mathcal{D}_{1|2|3} &= \Conv\bigl(\mathcal{P}_{1|2|3}\bigr)\equiv\mathcal{D}_\text{$3$-sep},\\
\mathcal{D}_{a|bc}  &= \Conv\bigl(\mathcal{P}_{a|bc}\bigr),\\
\label{eq:PSSsets.bcacab}
\mathcal{D}_{b|ac,c|ab}  &= \Conv \bigl(\mathcal{P}_{b|ac}\cup\mathcal{P}_{c|ab}\bigr),\\
\label{eq:PSSsets.2-sep}
\mathcal{D}_{1|23,2|13,3|12}  &= \Conv \bigl(\mathcal{P}_{1|23}\cup\mathcal{P}_{2|13}\cup\mathcal{P}_{3|12}\bigr)
\equiv\mathcal{D}_\text{$2$-sep},\\
\label{eq:PSSsets.123}
\mathcal{D}_\text{W}  &= \Conv \bigl(\mathcal{P}_\text{W}\bigr),\\
\begin{split}
\mathcal{D}_{123}  &= \Conv \bigl(\mathcal{P}_{123}\bigr)
\equiv\mathcal{D}_\text{$1$-sep}\equiv\mathcal{D}_\text{GHZ}\equiv\mathcal{D}.
\end{split}
\end{align}
\end{subequations}
Again, these sets are convex and they
contain each other in a hierarchic way,
which is illustrated in figure~\ref{fig:PSS3incl}.
%%%%%%%%%%%%%%%%%%%%%%%%%%%%%%%%%%%%%%%%
\begin{figure}
 \includegraphics{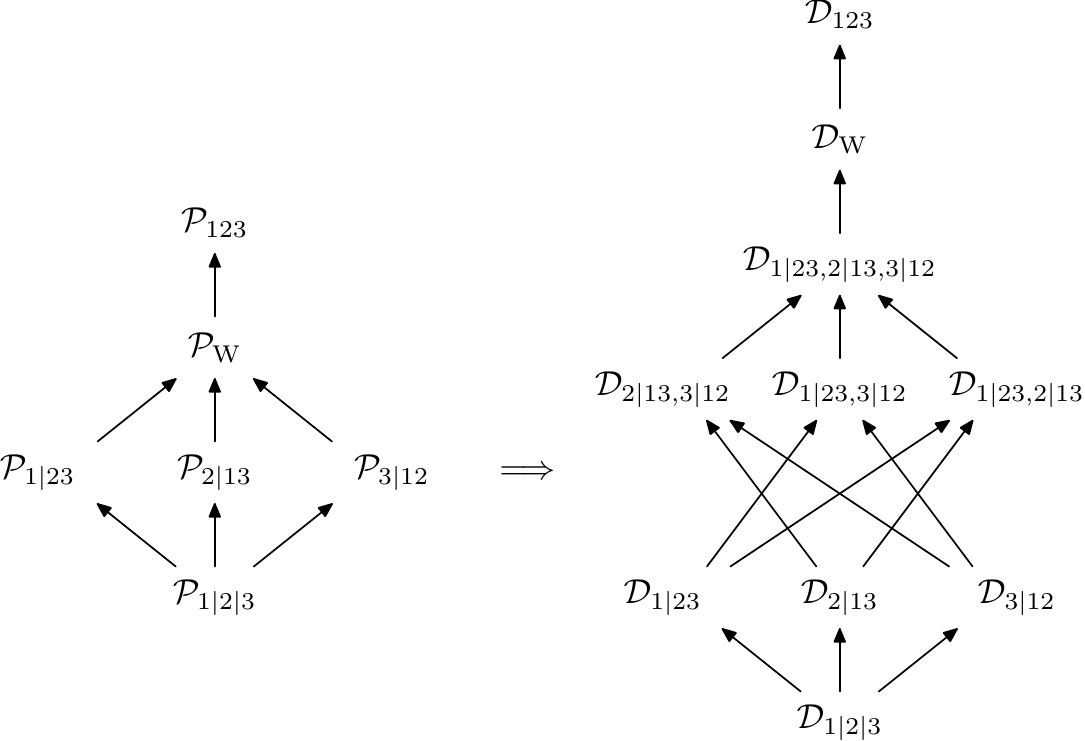}
 \caption{Inclusion hierarchy of the pure and mixed PSS sets $\mathcal{P}_{\dots}$ and $\mathcal{D}_{\dots}$ given in (\ref{eq:PSSPsets}) and (\ref{eq:PSSsets}).}
\label{fig:PSS3incl}
\end{figure}
%%%%%%%%%%%%%%%%%%%%%%%%%%%%%%%%%%%%%%%%
And, again, these sets 
are the convex hulls of \emph{all the possible closed sets}
arising from the unions of the $\mathcal{Q}_{\dots}$ sets~(\ref{eq:PSSPclasses}) of extremal points.

%*******************************************************************************
\subsection{PSS classes}
\label{subsec:ThreeQB.Mixed.PSSclasses}

The PSS classes of three-qubit mixed states arise
as the possible intersection of the PSS subsets (\ref{eq:PSSsets}).
Since the inclusion hierarchy of these (figure \ref{fig:PSS3incl})
is just a slight extension of that of the PS subsets of tripartite mixed states (figure \ref{fig:PS3incl}),
the situation does not become more complicated.
The arising PSS classes are the same as the PS classes, given in (\ref{eq:PSClasses}),
except that class $\mathcal{C}_1=\mathcal{D}\setminus\mathcal{D}_{1|23,2|13,3|12}$, containing tripartite entanglement,
is divided into two PSS classes.
%%%%%
The first one of them is
\begin{subequations}
\begin{equation}
\mathcal{C}_\text{W}=\cmpl{\mathcal{D}_{1|23,2|13,3|12}}\cap\mathcal{D}_\text{W}
\equiv\mathcal{D}_\text{W}\setminus\mathcal{D}_{1|23,2|13,3|12},
\end{equation}
which is the set of states which can not be mixed
without the use of some tripartite entangled pure states,
but there is no need of GHZ-type entanglement \cite{Acinetal3QBMixClass}.
%%%%%
The second one is
\begin{equation}
\mathcal{C}_\text{GHZ}=\cmpl{\mathcal{D}_\text{W}}\cap\mathcal{D}_{123}
\equiv\mathcal{D}_{123}\setminus\mathcal{D}_\text{W},
\end{equation}
\end{subequations}
which is the set of states which can not be mixed 
without the use of GHZ-type entanglement.
We summarize these $1+18+2$ PSS classes in table~\ref{tab:PSS3Classes}.
%%%%%%%%%%%%%%%%%%%%%%%%%%%%%%%%%%%%%%%%
%\begin{landscape}
\begin{table}
%\centering\begin{landscape*}
\begin{tabu}{X[2,c]||X[1,c]|X[1,c]X[1,c]X[1,c]|X[1,c]X[1,c]X[1,c]|X[1,c]|X[1,c]X[1,c]||X[3,c]X[3,c]X[3,c]}
\hline
\begin{sideways}PSS Class\end{sideways} & 
\begin{sideways}$\mathcal{D}_{1|2|3}$\end{sideways} & 
\begin{sideways}$\mathcal{D}_{a|bc}$\end{sideways} & 
\begin{sideways}$\mathcal{D}_{b|ac}$\end{sideways} & 
\begin{sideways}$\mathcal{D}_{c|ab}$\end{sideways} & 
\begin{sideways}$\mathcal{D}_{b|ac,c|ab}$\end{sideways} & 
\begin{sideways}$\mathcal{D}_{a|bc,c|ab}$\end{sideways} &
\begin{sideways}$\mathcal{D}_{a|bc,b|ac}$\end{sideways} &
\begin{sideways}$\mathcal{D}_{1|23,2|13,3|12}$\end{sideways} &
\begin{sideways}$\mathcal{D}_\text{W}$\end{sideways} &
\begin{sideways}$\mathcal{D}_{123}$\end{sideways} &
%$\mathcal{D}_\text{GHZ}$ &
\begin{sideways}in \cite{SeevinckUffinkMixSep}\end{sideways} &
\begin{sideways}in \cite{DurCiracTarrachBMixSep}\end{sideways} &
\begin{sideways}in \cite{Acinetal3QBMixClass}\end{sideways} \\
%$\mathcal{D}_\text{$1$-sep}$  \\
\hline
\hline
%                                                     || 1|2|3     |  a|bc         b|ac         c|ab      | b|ac/c|ab    c|ab/a|bc    a|bc/b|ac  |  2-sep     |  W         |  GHZ
$\mathcal{C}_3$           & $\subset$  & $\subset$  & $\subset$  & $\subset$  & $\subset$  & $\subset$  & $\subset$  & $\subset$  & $\subset$  & $\subset$  & 3       & 5       & S \\
\hline
$\mathcal{C}_{2.8}$       & $\nsubset$ & $\subset$  & $\subset$  & $\subset$  & $\subset$  & $\subset$  & $\subset$  & $\subset$  & $\subset$  & $\subset$  & 2.8     & 4       & B \\
$\mathcal{C}_{2.7.a}$     & $\nsubset$ & $\nsubset$ & $\subset$  & $\subset$  & $\subset$  & $\subset$  & $\subset$  & $\subset$  & $\subset$  & $\subset$  & 2.7,6,5 & 3.3,2,1 & B \\
$\mathcal{C}_{2.6.a}$     & $\nsubset$ & $\subset$  & $\nsubset$ & $\nsubset$ & $\subset$  & $\subset$  & $\subset$  & $\subset$  & $\subset$  & $\subset$  & 2.4,3,2 & 2.3,2,1 & B \\
$\mathcal{C}_{2.5.a}$     & $\nsubset$ & $\subset$  & $\nsubset$ & $\nsubset$ & $\nsubset$ & $\subset$  & $\subset$  & $\subset$  & $\subset$  & $\subset$  & 2.4,3,2 & 2.3,2,1 & B \\
$\mathcal{C}_{2.4}$       & $\nsubset$ & $\nsubset$ & $\nsubset$ & $\nsubset$ & $\subset$  & $\subset$  & $\subset$  & $\subset$  & $\subset$  & $\subset$  & 2.1     & 1       & B \\
$\mathcal{C}_{2.3.a}$     & $\nsubset$ & $\nsubset$ & $\nsubset$ & $\nsubset$ & $\nsubset$ & $\subset$  & $\subset$  & $\subset$  & $\subset$  & $\subset$  & 2.1     & 1       & B \\
$\mathcal{C}_{2.2.a}$     & $\nsubset$ & $\nsubset$ & $\nsubset$ & $\nsubset$ & $\subset$  & $\nsubset$ & $\nsubset$ & $\subset$  & $\subset$  & $\subset$  & 2.1     & 1       & B \\
$\mathcal{C}_{2.1}$       & $\nsubset$ & $\nsubset$ & $\nsubset$ & $\nsubset$ & $\nsubset$ & $\nsubset$ & $\nsubset$ & $\subset$  & $\subset$  & $\subset$  & 2.1     & 1       & B \\
\hline
$\mathcal{C}_\text{W}$    & $\nsubset$ & $\nsubset$ & $\nsubset$ & $\nsubset$ & $\nsubset$ & $\nsubset$ & $\nsubset$ & $\nsubset$ & $\subset$  & $\subset$  & 1       & 1       & W \\
$\mathcal{C}_\text{GHZ}$  & $\nsubset$ & $\nsubset$ & $\nsubset$ & $\nsubset$ & $\nsubset$ & $\nsubset$ & $\nsubset$ & $\nsubset$ & $\nsubset$ & $\subset$  & 1       & 1       & GHZ \\
\hline
\end{tabu}
%\end{landscape*}
\bigskip
\caption{PSS classes of mixed three-qubit states.
Additionally, we show the 
classifications obtained by
Seevinck and Uffink \cite{SeevinckUffinkMixSep},
D\"ur, Cirac and Tarrach \cite{DurCiracTarrachBMixSep},
and Ac\'in, Bru\ss{}, Lewenstein and Sanpera \cite{Acinetal3QBMixClass}.} 
\label{tab:PSS3Classes}
\end{table}
%\end{landscape}
%%%%%%%%%%%%%%%%%%%%%%%%%%%%%%%%%%%%%%%%

%*******************************************************************************
\subsection{Indicator functions}
\label{subsec:ThreeQB.Mixed.CRoof}

The indicator functions (\ref{eq:newPureLUinvs}) are given for state vectors $\cket{\psi}$
but they can be written for pure states $\pi=\cket{\psi}\bra{\psi}$ as well.%
%%%%%%%%%%%%%%%%%%%%%%%%
\footnote{This is because they are LU-invariants.
Note that the original LSL-tensors (\ref{eq:tensors}) can not be written directly for pure states.}
%%%%%%%%%%%%%%%%%%%%%%%%
During this, the index contractions with the $\varepsilon_{ii'}$s and $\delta^{i'}_{\phantom{i'}i}$s
lead to partial traces, partial transposes and spin-flips,
and it turns out that%
%%%%%%%%%%%%%%%%%%%%%%%%
\footnote{A slight abuse of the notation is that we use the same symbols
for the $\mathcal{H}\supset S^{2d-1}\to\field{R}$ functions given in (\ref{eq:newPureLUinvs})  
and the $\mathcal{P}(\mathcal{H})\to\field{R}$ functions given in (\ref{eq:newPureLUinvsPi}),
which take the same values for $\cket{\psi}$ and $\pi=\cket{\psi}\bra{\psi}$.}
%%%%%%%%%%%%%%%%%%%%%%%%
\begin{subequations}
\label{eq:newPureLUinvsPi}
\begin{align}
y(\pi)     &= \frac23\bigl(g_1(\pi) + g_2(\pi) + g_3(\pi)\bigr),\\
c^2_a(\pi) &= g_b(\pi) + g_c(\pi),\\
g_a(\pi)   &= \tr\pi_{bc}\widetilde{\pi_{bc}},\\
t(\pi)     &= 4\tr(\pi_a\otimes\pi_{bc})\tilde{\pi},\\
\tau^2(\pi)&= 4\tr(\pi^{\transp_a}\widetilde{\pi^{\transp_a}})^2,
\end{align}
\end{subequations}
where the spin-flipped states
$\widetilde{\pi_{bc}}=\bigl(\varepsilon\otimes\varepsilon\pi_{bc}\varepsilon^\dagger\otimes\varepsilon^\dagger\bigr)^*$
and
$\tilde{\pi}=\bigl(\varepsilon\otimes\varepsilon\otimes\varepsilon\pi\varepsilon^\dagger\otimes\varepsilon^\dagger\otimes\varepsilon^\dagger\bigr)^*$
arise, together with the partial transposition.

Now, we have the PSS subsets~(\ref{eq:PSSsets})
and the functions~(\ref{eq:newPureLUinvsPi})
with the vanishing properties 
\begin{subequations}
\label{eq:PSS3VanishingPure}
\begin{align}
\pi&\in\mathcal{P}_{1|2|3}&
\quad&\Longleftrightarrow&\quad y(\pi)&=0,\\
\pi&\in\mathcal{P}_{a|bc}&
\quad&\Longleftrightarrow&\quad c^2_a(\pi)&=0,\\
\pi&\in\mathcal{P}_{b|ca}\cup\mathcal{P}_{c|ab}&
\quad&\Longleftrightarrow&\quad g_a(\pi)&=0,\\
\pi&\in\mathcal{P}_{1|23}\cup\mathcal{P}_{2|13}\cup\mathcal{P}_{3|12}&
\quad&\Longleftrightarrow&\quad t(\pi)&=0,\\
\pi&\in\mathcal{P}_\text{W}&
\quad&\Longleftrightarrow&\quad \tau^2(\pi)&=0,
\end{align}
\end{subequations}
given also in table~\ref{tab:PSS3Pclasses}.
From these, the following holds for their convex-roof extension
in the same way as in~(\ref{eq:PS3VanishingMix}):
\begin{subequations}
\label{eq:PSS3VanishingMix}
\begin{align}
\label{eq:PSS3VanishingMix.y}
\varrho&\in\mathcal{D}_{1|2|3}&
\quad&\Longleftrightarrow&\quad \cnvroof{y}(\varrho)&=0,\\
\label{eq:PSS3VanishingMix.sa}
\varrho&\in\mathcal{D}_{a|bc}&
\quad&\Longleftrightarrow&\quad \cnvroof{{c^2_a}}(\varrho)&=0,\\
\label{eq:PSS3VanishingMix.ga}
\varrho&\in\mathcal{D}_{b|ac,c|ab}&
\quad&\Longleftrightarrow&\quad \cnvroof{g}_a(\varrho)&=0,\\
\label{eq:PSS3VanishingMix.t}
\varrho&\in\mathcal{D}_{1|23,2|13,3|12}&
\quad&\Longleftrightarrow&\quad \cnvroof{t}(\varrho)&=0,\\
\varrho&\in\mathcal{D}_\text{W}&
\quad&\Longleftrightarrow&\quad \cnvroof{{\tau^2}}(\varrho)&=0.
\end{align}
\end{subequations}

These necessary and sufficient conditions for the PSS subsets~(\ref{eq:PSS3VanishingMix})
yields necessary and sufficient conditions for the PSS classes,
and we can fill out table~\ref{tab:PSS3Mix} 
for the identification of the PSS classes of table~\ref{tab:PSS3Classes}, given for mixed states,
similar to table~\ref{tab:PSS3Pclasses}, given also for pure states.
%%%%%%%%%%%%%%%%%%%%%%%%%%%%%%%%%%%%%%%%
\begin{table}
\begin{tabu}{X[2,c]||X[1,c]|X[1,c]X[1,c]X[1,c]|X[1,c]X[1,c]X[1,c]|X[1,c]|X[1,c]}
\hline
PSS Class % & $\cnvroof{n}(\varrho)$ 
& $\cnvroof{y}(\varrho)$ & 
$\cnvroof{{c^2_a}}(\varrho)$ & $\cnvroof{{c^2_b}}(\varrho)$ & $\cnvroof{{c^2_c}}(\varrho)$ & 
$\cnvroof{g}_a(\varrho)$ & $\cnvroof{g}_b(\varrho)$ & $\cnvroof{g}_c(\varrho)$ & 
$\cnvroof{t}(\varrho)$ & $\cnvroof{{\tau^2}}(\varrho)$  \\
\hline
\hline
%                         | y        | c^2_a        c^2_b        s_c      | g_a        g_b        g_c      | t        | tau
$\mathcal{C}_3$           & $=0$     & $=0$     & $=0$     & $=0$     & $=0$     & $=0$     & $=0$     & $=0$     & $=0$ \\
\hline
$\mathcal{C}_{2.8}$       & $>0$     & $=0$     & $=0$     & $=0$     & $=0$     & $=0$     & $=0$     & $=0$     & $=0$ \\
$\mathcal{C}_{2.7.a}$     & $>0$     & $>0$     & $=0$     & $=0$     & $=0$     & $=0$     & $=0$     & $=0$     & $=0$ \\
$\mathcal{C}_{2.6.a}$     & $>0$     & $=0$     & $>0$     & $>0$     & $=0$     & $=0$     & $=0$     & $=0$     & $=0$ \\
$\mathcal{C}_{2.5.a}$     & $>0$     & $=0$     & $>0$     & $>0$     & $>0$     & $=0$     & $=0$     & $=0$     & $=0$ \\
$\mathcal{C}_{2.4}$       & $>0$     & $>0$     & $>0$     & $>0$     & $=0$     & $=0$     & $=0$     & $=0$     & $=0$ \\
$\mathcal{C}_{2.3.a}$     & $>0$     & $>0$     & $>0$     & $>0$     & $>0$     & $=0$     & $=0$     & $=0$     & $=0$ \\
$\mathcal{C}_{2.2.a}$     & $>0$     & $>0$     & $>0$     & $>0$     & $=0$     & $>0$     & $>0$     & $=0$     & $=0$ \\
$\mathcal{C}_{2.1}$       & $>0$     & $>0$     & $>0$     & $>0$     & $>0$     & $>0$     & $>0$     & $=0$     & $=0$ \\
\hline
$\mathcal{C}_\text{W}$    & $>0$     & $>0$     & $>0$     & $>0$     & $>0$     & $>0$     & $>0$     & $>0$     & $=0$ \\
$\mathcal{C}_\text{GHZ}$  & $>0$     & $>0$     & $>0$     & $>0$     & $>0$     & $>0$     & $>0$     & $>0$     & $>0$ \\  
\hline
\end{tabu}
\bigskip
\caption{PSS classes of mixed three-qubit states given in table~\ref{tab:PSS3Classes}
identified by the vanishing of the mixed indicator functions
(convex roof extension of the indicator functions~(\ref{eq:newPureLUinvs})).}
\label{tab:PSS3Mix}
\end{table}
%%%%%%%%%%%%%%%%%%%%%%%%%%%%%%%%%%%%%%%%
On the other hand, if a classification does not involve all the PS(S) subsets,
then, through~(\ref{eq:PSS3VanishingMix}), 
we have to use only some of the indicator functions,
for example, $y$, $c^2_a$ and $t$ for the classification obtained by Seevinck and Uffink \cite{SeevinckUffinkMixSep},
$y$ and $c^2_a$ for the classification obtained by D\"ur, Cirac and Tarrach \cite{DurCiracTarrachBMixSep},
$y$, $t$ and $\tau^2$ for the classification obtained by Ac\'in, Bru\ss{}, Lewenstein and Sanpera \cite{Acinetal3QBMixClass}.

%******************************************************************************
%******************************************************************************
\section{Generalizations: Three subsystems}
\label{sec:ThreeQB.GenThreePart}

In section~\ref{subsec:PartSep.ThreePart.Indicators} of the previous chapter 
we obtained indicator functions for tripartite systems from a general approach (\ref{eq:PS3PureInd}),
then in section~\ref{subsec:ThreeQB.Pure.NewInvs} 
we got indicator functions for three-qubit systems from the FTS approach (\ref{eq:newPureLUinvs}).
In this section we break up with qubits and consider general tripartite systems, 
and we try to obtain indicator functions as the generalization of the ones coming from the FTS approach.
Note that the FTS approach has led to an ``additive'' definition of indicator functions (\ref{eq:newPureLUinvs}),
while the general approach has led to a ``multiplicative'' one (\ref{eq:PS3PureInd}).

So, we need the generalizations of the pure state indicator functions in~(\ref{eq:newPureLUinvs.y})-(\ref{eq:newPureLUinvs.t}).
Apart from continuity,
the main and only requirement for these is
to satisfy the vanishing requirements for pure states given in the relevant part of table~\ref{tab:PSS3Pclasses},
which is copied here in table \ref{tab:PSS3PclassesAdd}.
%%%%%%%%%%%%%%%%%%%%%%%%%%%%%%%%%%%%%%%%
\begin{table}
\begin{tabu}{X[c]||X[c]|X[c]X[c]X[c]|X[c]X[c]X[c]|X[c]}
\hline
Class  & $y(\pi)$ & $s_1(\pi)$ & $s_2(\pi)$ & $s_3(\pi)$ & $g_1(\pi)$ & $g_2(\pi)$ & $g_3(\pi)$ & $t(\pi)$  \\
\hline
\hline
$\mathcal{Q}_{1|2|3}$     & $=0$     & $=0$     & $=0$     & $=0$     & $=0$     & $=0$     & $=0$     & $=0$     \\
\hline
$\mathcal{Q}_{1|23}$      & $>0$     & $=0$     & $>0$     & $>0$     & $>0$     & $=0$     & $=0$     & $=0$     \\
$\mathcal{Q}_{2|13}$      & $>0$     & $>0$     & $=0$     & $>0$     & $=0$     & $>0$     & $=0$     & $=0$     \\
$\mathcal{Q}_{3|12}$      & $>0$     & $>0$     & $>0$     & $=0$     & $=0$     & $=0$     & $>0$     & $=0$     \\
\hline
$\mathcal{Q}_{123}$       & $>0$     & $>0$     & $>0$     & $>0$     & $>0$     & $>0$     & $>0$     & $>0$     \\
\hline
\end{tabu}
\bigskip
\caption{Required vanishing properties of additive indicator functions for tripartite pure states.}
\label{tab:PSS3PclassesAdd}
\end{table}
%%%%%%%%%%%%%%%%%%%%%%%%%%%%%%%%%%%%%%%%
Then the convex roof extensions of them identifies the corresponding PS classes $\mathcal{D}_{\dots}$,
since these vanishing properties are the only ones which are needed for e.g.~(\ref{eq:PS3VanishingMix}).

%******************************************************************************
\subsection{\texorpdfstring{Indicator functions for $1|2|3$- and $a|bc$-separability}{Indicator functions for 1|2|3- and a|bc-separability}}
\label{subsec:ThreeQB.GenThreePart.ys}
The pure state indicator functions of~(\ref{eq:newPureLUinvs}) have been obtained from the FTS approach,
which works only for the qubit case.
However, some parts of the definitions can be generalized.
To do this, our basic quantities will be the~(\ref{eq:qTsallis}) local Tsallis entropies 
\begin{equation*}
s_a(\psi)=S^\text{Ts}_q(\pi_a) 
\end{equation*}
instead of the functions $g_a(\psi)$ given in~(\ref{eq:newPureLUinvs.ga}),
since the former ones are defined for all dimensions.%
%%%%%%%%%%%%%%%%%%%%%%%%
\footnote{Previously in this chapter we have used $c^2_a(\pi)=C^2(\pi_a)$ with the concurrence squared (\ref{eq:conc2}),
which is a normalized version of the $q=2$ Tsallis entropy,
but now we relax $q$, and try other entropies as well.}
%%%%%%%%%%%%%%%%%%%%%%%%
Obviously, for all Tsallis entropies of the subsystems,
$s_a(\psi)=S^\text{Ts}_q(\pi_a)$ fulfils the corresponding column of table~\ref{tab:PSS3PclassesAdd},
since it vanishes if and only if the subsystem is pure,
which means the separability of that subsystem from the rest of the system
if the whole system is in pure state.
From~(\ref{eq:newPureLUinvs.sa}) and~(\ref{eq:newPureLUinvs.ga}),
it turns out that $y$, given in~(\ref{eq:newPureLUinvs.y}),
is just the average of the local entropies $y=1/3(s_1+s_2+s_3)$,
vanishing if and only if no entanglement is present.
This works well not only for qubits, so we can keep this definition of $y$. 
This is not a novelty, 
since this is equivalent to the general construction (\ref{eq:PS3PureInd.1|2|3})-(\ref{eq:PS3PureInd.a|bc}).

%******************************************************************************
\subsection{\texorpdfstring{Indicator functions for $b|ac$-$c|ab$-separability}{Indicator functions for b|ac-c|ab-separability}}
\label{subsec:ThreeQB.GenThreePart.ga}
What is more interesting is the $b|ab$-$c|ab$-indicator functions $g_a$ in~(\ref{eq:newPureLUinvs.ga}).
These can also be expressed by the local entropies~(\ref{eq:newPureLUinvs.sa})
for qubits as $g_a=1/2(s_b+s_c-s_a)$.
Can this definition be kept for subsystems of arbitrary dimensions?
For $\mathcal{Q}_{1|2|3}$, $s_a=s_b=s_c=0$ so $g_a=0$.
For $\mathcal{Q}_{a|bc}$, the subsystem $a$ can be separated from the others
so the subsystems $a$ and $bc$ are in pure states, $s_a=0$ and $s_b=s_c\neq0$,
from which $g_a\neq0$ and $g_b=g_c=0$.
So the first four rows of the $g_a$ columns of table~\ref{tab:PSS3PclassesAdd} is fulfilled by $g_a=1/2(s_b+s_c-s_a)$.
For the last row, we need that $g_a>0$
when tripartite entanglement is present.
And this is the problematic point.
%First note that, for tripartite pure states, 
%$g_a$s are related to the failure of additivity of the entropies of the bipartite subsystems.
This question can be traced back to the subadditivity of the Tsallis entropies.
Raggio's conjecture \cite{RaggioTsallis} 
about that is twofold: For%
%%%%%%%%%%%%%%%%%%%%%%%%
\footnote{Note that for $0<q<1$, there is no definite relation
between $S^\text{Ts}_q(\varrho)$ and $S^\text{Ts}_q(\varrho_1) + S^\text{Ts}_q(\varrho_2)$.} 
%%%%%%%%%%%%%%%%%%%%%%%%
$q>1$,
\begin{subequations}
\label{eq:Raggio}
\begin{align}
\label{eq:Raggio.subadd}
S^\text{Ts}_q(\varrho) &\leq S^\text{Ts}_q(\varrho_1) + S^\text{Ts}_q(\varrho_2),\\
\label{eq:Raggio.add}
S^\text{Ts}_q(\varrho) &= S^\text{Ts}_q(\varrho_1) + S^\text{Ts}_q(\varrho_2) 
\qquad\Longleftrightarrow\qquad
\bigl(\varrho = \varrho_1\otimes \varrho_2 \;\;\text{and}\;\;
\text{$\varrho_1$ or $\varrho_2$ pure}\bigr).
\end{align}
\end{subequations}
%(Note that for $0<q<1$, there is no definite relation 
%between $S^\text{Ts}_q(\varrho)$ and $S^\text{Ts}_q(\varrho_1) + S^\text{Ts}_q(\varrho_2)$.)
Both statements hold for the classical scenario \cite{RaggioTsallis},
which can be modelled in the quantum scenario 
by density matrices being LU-equivalent to diagonal ones.
The first part~(\ref{eq:Raggio.subadd}) of the conjecture
has been proven by Audenaert \cite{AudenaertTsallisSubadd}.
This guarantees the nonnegativity of our $g_a$s, since
for pure states,
$S^\text{Ts}_q(\pi_a)=S^\text{Ts}_q(\pi_{bc})\leq S^\text{Ts}_q(\pi_{b}) + S^\text{Ts}_q(\pi_{c})$,
so $0\leq 1/2(s_b+s_c-s_a) = g_a$.
On the other hand,~(\ref{eq:Raggio.add}) is exactly what we need.
That is, $\cket{\psi}\in\mathcal{Q}_{123}$ if and only if
neither of its subsystems are pure, which means that 
there is subadditivity in a strict sense,
so $0 < 1/2(s_b+s_c-s_a) = g_a$.
The $\Leftarrow$ implication in~(\ref{eq:Raggio.add}) holds obviously,
but the whole second part~(\ref{eq:Raggio.add}) of the conjecture,
to our knowledge, has not been proven yet.
A very little side-result of our work is that 
Raggio's conjecture holds
for the very restricted case of two-qubit mixed states which are at the most of rank two.

We note that the (\ref{eq:Neumann}) von Neumann entropies of the subsystems 
\emph{are not suitable} for the role of $s_a$s,
if we want to write $g_a$s by that as $1/2(s_b+s_c-s_a)$,
since the von Neumann entropy is additive for product states 
without any reference to the purity of subsystems, % in the equality condition of subadditivity:
\begin{subequations}
\label{eq:NeumannProp}
\begin{align}
\label{eq:NeumannProp.subadd}
S(\varrho) &\leq S(\varrho_1) + S(\varrho_2),\\
\label{eq:NeumannProp.add}
S(\varrho) &= S(\varrho_1) + S(\varrho_2) 
\qquad\Longleftrightarrow\qquad
\varrho = \varrho_1\otimes \varrho_2.
\end{align}
\end{subequations}
Indeed, it is easy to construct a tripartite state,
which is not separable under any partition, but has vanishing $g_a$ (defined by the von Neumann entropy).
For example, let $\dim\mathcal{H}_a=4$,
then for the state
\begin{equation*}
\cket{\psi}=\frac12\bigl(\cket{000}+\cket{101}+\cket{210}+\cket{311}\bigr)
\end{equation*}
$\pi_{23}=\pi_2\otimes\pi_3$, so $g_1(\psi)=1/2\bigl(S(\pi_2)+S(\pi_3)-S(\pi_1)\bigr)=0$,
while $S(\pi_1)=\ln4$, and $S(\pi_2)=S(\pi_3)=\ln2$, so neither of the subsystems are pure,
the state is tripartite-entangled.

The~(\ref{eq:qRenyi}) \emph{R\'enyi entropy} 
has the advantage of additivity,
\begin{equation}
S^\text{R}_q(\varrho)=S^\text{R}_q(\varrho_1)+S^\text{R}_q(\varrho_2) 
\qquad\Longleftarrow\qquad \varrho = \varrho_1\otimes \varrho_2.
\end{equation}
This is an advantage when entanglement is studied in the asymptotic regime, 
when the state is present in multiple copies
and properties are investigated against the number of copies. 
Again, this advantage is a disadvantage from our point of view,
the R\'enyi entropies of the subsystems \emph{are not suitable} for the role of $s_a$s
if we want to write $g_a$s by that as $1/2(s_b+s_c-s_a)$.
Moreover, subadditivity does not hold for R\'enyi entropy,
so the non-negativity of $g_a$s defined by R\'enyi entropies
does not even guaranteed.%
%%%%%%%%%%%%%%%%%%%%%%%%
\footnote{For further properties and references on the quantum entropies,
see e.g.~\cite{BengtssonZyczkowski,OhyaPetzQEntr,Petzfdivergence,FuruichiTsallis}.}
%%%%%%%%%%%%%%%%%%%%%%%%

%******************************************************************************
\subsection{\texorpdfstring{Indicator functions for $1|23$-$2|13$-$3|12$-separability}{Indicator functions for 1|23-2|13-3|12-separability}}
\label{subsec:ThreeQB.GenThreePart.t}
We have seen that the possibility of the additive definition of the $b|ac$-$c|ab$-indicator functions $g_a$
depends on Raggio's conjecture.
What can we say about $t$, the indicator for $1|23$-$2|13$-$3|12$-separability?
Since in the three-qubit case we need also the Kempe invariant $I_4$ to write $t$ (\ref{eq:newPureLUinvsCanon.t}),
it follows from the independency of the $I_0,\dots,I_5$ polynomials (\ref{eq:3QBcanonPureLUinvs})
that $t$ can not arise as linear combination of $c^2_a$ local Tsallis entropies of parameter $q=2$.
But we can also use $s_a$ local Tsallis entropies for $q\neq2$,
so it seems as if finding expression for $t$ with entropies could be possible.
%Moreover, there are probably infinitely many such expressions.
However, $t$, being in connection with the Kempe invariant,
contains nonlocal information, which can not be extracted from the local spectra.
So the direct generalization through entropies, 
which works for $y$, $s_a$ and $g_a$ (this latter depends on Raggio's conjecture),
does not work for $t$.
However, note that finding indicator functions and the generalization of $t$ of the FTS approach
are independent issues, 
since the indicator functions of the general construction (\ref{eq:PS3PureInd}) 
could be constructed only from local informations.

%******************************************************************************
%******************************************************************************
\section{Summary and remarks}
\label{sec:ThreeQB.Sum}

In this chapter we have introduced the PSS classification of mixed three-qubit states,
which is the combination 
of the PS classification of thripartite states (section \ref{sec:PartSep.ThreePart})
and the three-qubit classification given by Ac{\'i}n et.~al.~(section \ref{subsec:QM.EntMeas.3QBMix}).
We have constructed the relevant indicator functions by the use of the FTS approach of three-qubit entanglement.

\begin{remarks}
\item First, note that the FTS approach of three-qubit entanglement 
\cite{BorstenetalFreudenthal3QBEnt}
is coming from the famous Black Hole/Qubit Correspondence \cite{BorstenDuffLevayBHQB}.
The FTS approach has turned out to be fruitful 
also in the description of the structure of entanglement in some other particular composite systems
\cite{LevayVranaThreeFermionSix,VranaLevayFTS}.
\item Since the convex roof extensions of polynomials are known to be semi-algebraic functions \cite{PetiPriv,ChenDjokovicSemialg},
it can be useful to use LU-invariant homogeneous polynomials for the identification of the classes.
Then we have polynomials of this kind 
from~(\ref{eq:newPureLUinvs}) coming from the FTS-approach,
and from~(\ref{eq:PS3PureInd}) with the Tsallis entropy for $q=2$ coming from the general constructions.
The former ones are of lower degree, which may lead to simpler convex roof extensions.
\item Moreover, 
this holds also for the  $g_a$ functions in the general tripartite case
if Raggio's conjecture holds (subsection~\ref{subsec:ThreeQB.GenThreePart.ga}).
\item A little side-result of our work is that
Raggio's conjecture holds
for two-qubit mixed states which are at the most of rank two.
\item An interesting question is 
as to whether all pure state indicator functions for $n$-partite systems
can be obtained without products of local entropies, but using only linear combinations of them.
Some issues in connection with this were discussed in section \ref{sec:ThreeQB.GenThreePart} for tripartite systems, 
but some ideas or hints are still missing (section \ref{subsec:ThreeQB.GenThreePart.t}).
The problem here is that
we have to find such linear combination of different entropies 
which fulfils \emph{every} line of the last column of table \ref{tab:PSS3PclassesAdd},
while this can depend on other entropic inequalities, which can be unknown at this time.
\item In the light of chapter \ref{chap:Deg6},
looking for convex roof extensions in the language of LU invariant polynomials
would be an interesting research direction.
\item\label{rem:threeqb.notmon} As a disadvantage of the FTS approach,
we have to mention that
some of the indicator functions coming from the FTS approach 
are not non-increasing on average~(\ref{eq:averagePure}), 
namely $g_a$ and $t$ given in~(\ref{eq:newPureLUinvs.ga}) and~(\ref{eq:newPureLUinvs.t}).
(Counter-examples for~(\ref{eq:averagePure}) can be constructed for these functions by direct calculation.)
\item
The functions (\ref{eq:newPureLUinvs})
give us some abstract motivations for the relevance of the extension of the
Seevinck-Uffink classification, done in the previous chapter.
%\item 
Although we get back the classification given by Seevinck and Uffink
if we simply forget about the sets $\varrho\in\mathcal{D}_{b|ac,c|ab}$,
and the functions $\cnvroof{g}_a(\varrho)$,
but the appearance of the $g_a(\psi)$ polynomials is a natural
in the light of the formulae~(\ref{eq:newPureLUinvs.y}),~(\ref{eq:newPureLUinvs.sa}), and~(\ref{eq:newPureLUinvs.ga}).
This is another motivation of the introduction of the sets $\varrho\in\mathcal{D}_{b|ac,c|ab}$
to the classification.
\item The $g_a$ functions are interesting in themselves.
First, they mean the failure of additivity of the Tsallis entropies (section \ref{subsec:ThreeQB.GenThreePart.ga}),
which is a quantum mutual information like quantity, defined by Tsallis entropy instead of von Neumann entropy.
\item Second, for all non-GHZ vectors, the $g_a$ functions coincide with
the squared Wootters concurrences of two-qubit subsystems~(section \ref{subsec:ThreeQB.Pure.WConc}).
However, note that the Wootters concurrences of two-qubit subsystems
are not suitable for being indicator functions,
since they can be zero also for GHZ-type vectors,
so they do not fulfil the last row of $g_a$ columns of table~\ref{tab:PSS3Pclasses}.
For example for the usual GHZ state~(\ref{eq:GHZ}),
the Wootters concurrences of two-qubit subsystems are zero.
\item Third, by reason of (\ref{eq:newPureLUinvs.sa}), (\ref{eq:Conc3QB}) and (\ref{eq:recallCKW}),
the entanglement of subsystem $a$ with $bc$ given by the concurrence is $C^2(\pi_a)\equiv c^2_a(\pi)=g_b(\pi)+g_c(\pi)$,
and due to (\ref{eq:Conc3QB}), $g_a(\pi)={\cnvroof{c}}^2(\pi_{bc})+1/2\tau(\pi)$.
So $g_a$ seems like some kind of ``effective entanglement''
within the $bc$ subsystem, containing also tripartite entanglement.
However, note that $g_a$ is not entanglement monotone (item (\ref{rem:threeqb.notmon})),
so it does not express the amount of any kind of entanglement in general.
%\item In section~\ref{subsec:PartSep.Threepart.Examples} 
%we have shown states
%which are definitely in classes that are different in the extended classification.
%This is another reason for using also the sets $\varrho\in\mathcal{D}_{b|ac,c|ab}$
%in the classification.
\end{remarks}

\chapter*{Epilogue}

In the Prologue, we mentioned the nonclassical correlations arising in quantum systems being in entangled state,
and their utilization for nonclassical tasks.
Then in the main part of the dissertation, we reviewed and studied some questions in connection with entanglement.
But, is the reverse also true, namely, 
is the presence of entanglement completely equivalent to the nonclassical correlations?
For pure states that holds, that is,
a pure state is classically correlated if and only if it is separable.
However, in this case, it is actually uncorrelated not only quantum mechanically but also classically.
To model classical correlations, we need to use density matrices.
For density matrices, however, separability is not equivalent to classicality.
Classically correlated density matrices are the ones which are local unitary equivalent to diagonal ones,
correlations of no other density matrices have classical counterpart, hence are regarded nonclassical \cite{DornerVedralCorrelations}.
Entangled states are of course nonclassical, but the most of the separable states are also nonclassical in this sense.
Moreover, there are quantum algorithms which use states of this latter kind to achieve quantum speed-up over classical algorithms.
These separable states exhibiting nonclassical behaviour can be created by the use of
local operations and classical correlations only.
In this sense, entanglement is regarded as a stronger form of nonclassicality.
There are different quantities measuring the amount of nonclassical correlations,
maybe the most important one of them is the \emph{quantum discord} \cite{ModietalDiscord}.
 
Nowadays, investigating the correlations arising in quantum systems is an active field of research,
helping us to uncover the intriguing differences between the classical and quantum worlds.

%\appendix

%*******************************************************************************
%	Back Matter
%
\backmatter

%*******************************************************************************
%	Bibliography
%
\bibliographystyle{amsalpha}
%\bibliographystyle{amsplain}
%\bibliographystyle{harvard}
%\cleardoublepage
\phantomsection
\bibliography{dissertation.bib}

%*******************************************************************************
%	Index
%
%\printindex	% See note above about multiple indexes.

%*******************************************************************************
\end{document}